\documentclass[12pt]{article}
\input{packages}
\addauthor{lm}{blue}
\addauthor{sv}{red}
\addauthor{done}{green}
\addauthor{TODO}{red}

\newcommand*{\Scale}[2][4]{\scalebox{#1}{$#2$}}%
\DeclareMathOperator{\mme}{\mathbb{E}}
\DeclareMathOperator{\mmc}{\mathrm{Cov}}
\DeclareMathOperator{\mmv}{\mathrm{Var}}

\newcommand{\buyer}{{buyer}}
\newcommand{\Buyer}{{Buyer}}
\newcommand{\buyers}{{{\buyer}s}}
\newcommand{\seller}{{seller}}
\newcommand{\Seller}{{Seller}}
\newcommand{\sellers}{{{\seller}s}}
\newcommand{\RomanNumeralCaps}[1]
    {\MakeUppercase{\romannumeral #1}}

\newcommand{\eg}{{\textit{e.g.}}}

\newcommand{\mmi}{[I]}
\newcommand{\mmj}{[J]}

\newcommand{\wt}{{\rm T}}
\newcommand{\ct}{{\rm T}}
\newcommand{\rt}{{\rm T}}

\newcommand{\wc}{{\rm C}}
\newcommand{\cc}{{\rm C}}
\newcommand{\rc}{{\rm C}}

\newcommand{\ccc}{{\textcolor{red}{\rm cc}}}
\newcommand{\icb}{{\textcolor{green}{\rm ib}}}
\newcommand{\ics}{{\textcolor{blue}{\rm is}}}
\newcommand{\ctt}{{\rm tr}}

\newcommand{\redc}{{\textcolor{red}{\rm C}}}
\newcommand{\orangec}{\textcolor{orange}{\cc}}
\newcommand{\greenc}{{\textcolor{green}{\rm C}}}
\newcommand{\bluec}{{\textcolor{blue}{\rm C}}}
\newcommand{\magc}{\textcolor{magenta}{\cc}}
\newcommand{\brownc}{{\textcolor{brown}{\rm C}}}
\newcommand{\blackt}{{\rm T}}

\newcommand{\Type}{\Gamma}
\newcommand{\type}{\gamma}
\newcommand{\types}{\{\ccc,\icb,\ics,\ctt\}}
\newcommand{\ATE}{\mathrm{ATE}}

\newcommand{\oybi}{\overline{y}^{\rb}_{i}}
\newcommand{\oysj}{\overline{y}^{\rs}_{j}}

\newcommand{\coefvec}{\vec{\bm{{\beta}}}}

\newcommand{\dyij}{\axisvariationmeanpopulation{ij}{\rb\rs}} 
\newcommand{\dyiijj}{\axisvariationmeanpopulation{i'j'}{\rb\rs}} 
\newcommand{\dyijj}{\axisvariationmeanpopulation{ij'}{\rb\rs}} 
\newcommand{\dyiij}{\axisvariationmeanpopulation{i'j}{\rb\rs}} 
\newcommand{\dyi}{\axisvariationmeanpopulation{i}{\rb}} 
\newcommand{\dyii}{\axisvariationmeanpopulation{i'}{\rb}} 
\newcommand{\dyj}{\axisvariationmeanpopulation{j}{\rs}} 
\newcommand{\dyjj}{\axisvariationmeanpopulation{j'}{\rs}} 

\newcommand{\dydy}[2]{\axisvariationmeanpopulation{#1}{#2}} 

\newcommand{\taudirect}{\tau_{\rm direct}}
\newcommand{\betaspillb}{\bm{\vec{\beta}}_{\rm spill}^\rb}
\newcommand{\betaate}{\bm{\vec{\beta}}_{\ATE}}
\newcommand{\betadirect}{\bm{\vec{\beta}}_{\rm direct}}

\newcommand{\betaspills}{\bm{\vec{\beta}}_{\rm spill}^\rs}
\newcommand{\tauspillb}{\tau_{\rm spill}^\rb}

\newcommand{\hattauspillb}{\widehat{\tau}_{\rm spill}^\rb}

\newcommand{\hattau}{\hat{\tau}}

\newcommand{\bw}{\mathbf{W}}
\newcommand{\bww}{\mathbf{w}}
\newcommand{\mmw}{\mathbb{W}}

\newcommand{\rs}{{\rm S}}
\newcommand{\rb}{{\rm B}}

\newcommand{\been}{\mathbf{1}}
\newcommand{\owwi}{\overline{w}^\rb_{i}}
\newcommand{\owwj}{\overline{w}^\rs_{j}}
\newcommand{\ooww}{\overline{\overline{w}}}

\newcommand{\II}{{\mathcal{I}}} 
\newcommand{\JJ}{{\mathcal{J}}} 

\newcommand{\NN}{{\mathcal{N}}}

\newcommand{\mmr}{\mathbb{R}}

\newcommand{\meanpopulation}[1]{{\overline{\overline{y}}}_{#1}}

\newcommand{\meanestimate}[1]{\widehat{\overline{\overline{Y}}}_{#1}}

\newcommand{\meanonefixedindex}[2]{\overline{\overline{{y}}}_{#1,#2}}
\newcommand{\meanonefixedindexvec}[2]{\overline{\overline{\bm{y}}}_{#1,#2}}

\newcommand{\estimatefixedindex}[2]{\widehat{\overline{\overline{Y}}}_{#1,#2}}

\newcommand{\axismeanestimate}[2]{\widehat{\overline{Y}}_{#1}^{#2}}
\newcommand{\axismeanpopulation}[2]{\overline{y}_{#1}^{#2}}

\newcommand{\axisvariationmeanpopulation}[2]{\delta_{#1}^{#2}}

\newcommand{\axisvariancepopulation}[2]{\sigma^{#2}_{#1}}
\newcommand{\axisvarianceestimate}[2]{\widehat{\Sigma}^{#2}_{#1}}

\newcommand{\axisvariancecrosspopulation}[2]{\xi^{#2}_{#1}}

\newcommand{\axismeanfixedindex}[3]{\overline{y}_{#1,#2}^{#3}}

\newcommand{\axisbias}[2]{\eta_{#1}^{#2}}
\newcommand{\axisbiasestimate}[2]{\widehat{\eta}_{#1}^{#2}}

\newcommand{\axisvariancepopulationscaled}[2]{\sigma_{#1}^{#2}}
\newcommand{\axisvarianceestimatescaled}[2]{\hat{\Sigma}_{#1}^{#2}}

\newcommand{\typevariance}{\hat{\Sigma}}

\newcommand{\irt}{I_{\rm T}}
\newcommand{\irn}{I_{\rm C}}
\newcommand{\jrt}{J_{\rm T}}
\newcommand{\jrn}{J_{\rm C}}
\newcommand{\ooe}{\overline{\overline{\varepsilon}}}
\newcommand{\oee}{\overline{\varepsilon}}

\newcommand{\db}{D^{\rm \rb}}
\newcommand{\ds}{D^{\rm \rs}}

\newcommand{\Opfin}{O_p^{\mathrm{fin}}}

\newcommand{\boldb}{\bm{b}}
\newtheorem{theorem}{Theorem}
\newtheorem*{theorem*}{Theorem}

\newtheorem{definition}[theorem]{Definition}
\newtheorem{assumption}[theorem]{Assumption}

\newtheorem{lemma}[theorem]{Lemma}

\newtheorem{example}[theorem]{Example}
\newtheorem{corollary}[theorem]{Corollary}

\theoremstyle{definition}
\newtheorem{remark}[theorem]{Remark}

\DeclareRobustCommand{\t}{\dot{\tau}_i^{\pi}}
\DeclareRobustCommand{\T}{\dot{\tau}_i^{\Pi}}

\DeclareRobustCommand{\J}{\mathcal{J}}

\DeclareRobustCommand{\Set}[2][]{\left\{#1 \, \middle|\, #2 \right\}}
\DeclareRobustCommand{\bb}{\mathbb}
\DeclareRobustCommand{\beef}{\vphantom{\sum}}
\DeclareRobustCommand{\f}{\frac}
\DeclareRobustCommand{\d}{\delta}

\DeclareRobustCommand{\PP}{\mathfrak{P}}
\DeclareRobustCommand{\Unif}{\mathrm{Unif}}

\DeclareRobustCommand{\revision}[1]{{\textcolor{black}{#1}}}

 \appto\appendix{\counterwithin{equation}{section}}
\makeatletter

\begin{document}
\title{Multiple Randomization Designs:\\ Estimation and Inference with Interference}
\author[1]{Lorenzo Masoero}
\author[1]{Suhas Vijaykumar}
\author[3]{Thomas S.~Richardson}
\author[1]{James McQueen}
\author[1]{Ido Rosen}
\author[4]{Brian Burdick}
\author[4]{Pat Bajari}
\author[2]{Guido Imbens}

\affil[1]{Amazon, US}
\affil[2]{Corresponding author, Graduate School of Business and Department of Economics, Stanford University, US}
\affil[3]{Department of Statistics, University of Washington, US}
\affil[4]{Work done while at Amazon, US}
\maketitle

\newif\ifarxiv
\arxivtrue
\newif\ifresponse
\responsefalse

\bigskip


\begin{abstract}
    Completely randomized experiments, originally developed by Fisher and   Neyman in the 1930s, are still widely used in practice, even in online experimentation. 
    However, such designs are of limited value for answering standard questions in marketplaces, where multiple populations of agents interact strategically, leading to complex patterns of spillover effects. 
    In this paper, we derive the finite-sample properties of tractable estimators for ``Simple Multiple Randomization Designs'' (SMRDs), a new class of experimental designs which account for complex spillover effects in randomized experiments. Our derivations are obtained under a natural and general form of cross-unit interference, which we call ``local interference.'' 
    We discuss the estimation of main effects, direct effects, and spillovers, and present associated central limit theorems.
\end{abstract}

\noindent%
\vfill

\numberwithin{theorem}{section}



\section{Introduction} \label{sec:introduction}

Randomized experiments, introduced in the 1920s \citep{neyman1923applications, fisher1937design}, are an indispensable tool for estimating causal effects across many disciplines. 
For example, the Food and Drug Administration in the United States requires such experiments as part of the drug approval process. 
Recently, online experimentation has also become an integral part of  product development in the private sector.
\citet{gupta2019top} list some online businesses that collectively run hundreds of thousands of experiments annually. 

Modern experimental contexts differ markedly from those that inspired early experimental designs: experiments are carried out in marketplaces, often online, where multiple populations of units interact strategically (\eg, {\buyers} and {\sellers}; riders and drivers; renters and  property managers; viewers,  content creators  and advertisers). A challenge posed by these settings is that cross-unit interactions often lead to interference or spillovers. In our running example of buyers and sellers, the treatment assigned to one unit in a population ({\eg}, a {\seller}) might affect the outcome for a different unit of the same  population (another {\seller}). If present, this interference invalidates conventional analyses of standard experimental designs.

We study the finite sample properties of tractable estimators for a new class of experimental designs, ``Multiple Randomization Designs'' (MRDs for short), which are tailored for experimentation in marketplaces \citep{bajari2023experimental, johari2022experimental}. The distinguishing feature of these designs is that they involve multiple populations of units:  treatment assignments and outcomes are measured at the level of a tuple of units, one from each population ({\it e.g.}, the impact of providing additional information on a buyer's past expenditure to a seller, measured at the {\buyer}-{\seller} level).
The experimental designs correspond to distributions of assignments for these tuples of units, {\it e.g.,} over the {\buyer}-{\seller} pairs.

In the leading case we consider, a ``Simple MRD'' or SMRD, a  subset of {\buyer}s is selected at random, and a subset of {\seller}s is selected at random, and only the {\buyer}-{\seller} pairs where both the {\buyer} and the {\seller} was selected are exposed to the binary treatment. 
This two-level randomization serves to isolate and measure interference between units, thus distinguishing it from classical designs with multi-level randomization (\eg, Latin square and split-plot designs). 

This paper provides the first formal analysis of MRDs, showing that they may be used to
\begin{inumerate}
    \item test for the presence of spillovers,  
    \item estimate and conduct inference for the overall treatment effect in the presence of a large class of spillovers, and
    \item obtain\textemdash even without interference or spillovers\textemdash more precise estimates of the average causal effect than standard, single-sided randomization designs.
\end{inumerate}

Our work contributes to the rapidly growing literature on causal inference under interference \citep{hong2006evaluating, hong2008causal, hudgens2008toward, rosenbaum:2007, aronow2012general, vanderweele2014interference, ogburn2014causal, athey2018exact, ugander2013graph, blake2014marketplace, basse2019randomization}.  
Recent research has focused on experimental design in settings with complex spillovers, differing mainly in the settings they consider and the corresponding assumptions placed on cross-unit interference. 
Some work considers cases where spillovers between units are mediated by low-dimensional measures, such as prices in a marketplace or shares of treated units in a peer group
\citep{wager2021experimenting, munro2021treatment, aronow2017estimating}. 
Another line of work focuses on the role of clustering to mitigate interference, \eg, 
\citet{viviano2023causal}.
A separate approach models interference in terms of a bipartite graph between units and treatment sites  (\eg, advertisers bidding on the same keywords, as in \citealp{zigler2021bipartite}, and \citealp{ harshaw2022design}).
Others consider crossover or switchback designs in dynamic contexts where treatments vary over time and have lasting effects \citep{cox2000theory, bojinov2020design, xiong2023optimal, shi2023behavioral}. 
Finally, some work has modeled spatial or network spillovers in order to improve precision in survey experiments
\citep{savitz2009exploiting}.

Multiple Randomization Designs  were informally introduced by \citet{bajari2023experimental} and  \citet{johari2022experimental}.
The key feature of MRDs is the presence of two or more populations, \eg 
\ {\buyer}s and {\seller}s, where interventions can be assigned and outcomes measured at the level of the {\buyer-\seller} pair.
We provide exact characterizations of the design-based variance, together with corresponding variance estimators, and central limit theorems that allow for inference under these designs.

On the surface, MRDs share common features with Latin squares \citep{welch1937z} and split-plot designs \citep{fisher1928statistical,zhao2018randomization, zhao2022reconciling}, but they are fundamentally quite different.
In all three cases, the experimental units are organized in a matrix or clustered structure.
\revision{However, they differ in important ways: for example, }
Latin square designs are aimed at reducing variance through balance of the location of experimental units in a geographic space, whereas MRDs address interference and spillovers between experimental units. 
\revision{Meanwhile, although split-plot designs have been used to study  spillovers (e.g., \citealt{hudgens2008toward, zhao2022reconciling}), they consider units grouped  into  clusters as opposed to a two-dimensional array.}

Multiple randomization allows us to account for interference in ways not possible with completely randomized experiments, but in doing so they complicate estimation and inference. 
Challenges arise from the intrinsic dependence 
structure in the assignment process across the two populations: {\buyer}s\  and {\seller}s in our generic example. 
We address these using a randomization-based approach, where we take the potential outcomes under different treatment regimes as fixed. We exactly characterize the finite-sample variances of the proposed estimators with respect to the random design. 
We also propose conservative variance estimators, similar to those available for conventional randomized experiments. 
Finally, we prove design-based central limit theorems, extending the recent results of \citet{li2017general,shi2022berry} for single population experiments to our setting with multiple-population experiments, under appropriate side assumptions.

Most similar to our work is   \citet{johari2022experimental}, who studied how spillover effects caused by interference can lead to bias in standard experimental designs, and analyzed a special case of the MRDs we consider in this paper. 
\citet{johari2022experimental} produce a dynamic, stochastic model of a two-sided marketplace with cross-unit interference. Following a detailed analysis of the model, they use it to illustrate the favorable properties of SMRDs in comparison to standard experimental designs.
\section{Experiments in Marketplaces: Interference} \label{sec:interference}

We start by introducing a framework for randomized experiments in marketplaces with multiple populations of agents. 
We use the two-population {\buyer}-{\seller} (or {customer}-{product}) case as our generic example, but we emphasize that the ideas we present apply to other settings and extend to higher-order unit tuples, {\eg,} subscriber-creator-advertiser, customer-restaurant-driver or passenger-airline-travel agent.
Interference  or spillover effects arise naturally in these settings: treatment of one unit can impact the outcomes of other units, invalidating  assumptions that serve as the basis for analyzing standard experiments.
An example of the treatment is  the presentation of additional information ({\it e.g.,} in the form of more detailed reviews) shown to {\buyer} $i$ when viewing products from
{\seller} $j$.

In our generic {\buyer}-{\seller} example, one of the populations consists of  $I$ {\buyer}s, indexed  by $i~\in~\mmi~:=~\{1,\ldots,I\}$. The {\buyer}s interact with members of the second population, consisting of $J$ {\seller}s indexed by $j\in\mmj:=\{1,\ldots,J\}$.
Over a fixed period of time, say a week or a month, we measure for each {\buyer-\seller} pair an outcome metric of engagement $Y_{ij}$ (\eg, the amount of money paid by {\buyer} $i$  to {\seller} $j$). 
The experimenter performs an intervention at the level of the {\buyer}-{\seller} pair $(i,j)$, via the randomized treatment assignment $W_{ij} \in \{\cc,\ct\}$. 
Critically, the treatment might not be offered to all buyers who interact with a particular seller, nor to all sellers for any given buyer. Let $\bw \in \{\wc, \wt\}^{I \times J}$ denote a random $I\times J$ matrix of treatment assignments with typical element $W_{ij}\in\{\wc,\wt\}$, and $\bww$ a realization of this matrix. 

We adopt the potential outcome framework \citep[see e.g.,][]{fisher1937design, neyman1923, imbens2015causal}: for each value ${\bww}$ of the $I \times J$ assignment matrix, $y_{ij}({\bf w})$ is the corresponding potential outcome for unit $(i,j)$, which is non-stochastic.
An example assignment matrix is shown in \eqref{equilibrium110}, where rows identify five {\buyers} and columns identify six {\sellers}.
Colors highlight four sets of experimental units ({\buyer}-{\seller} pairs), instead of the usual two. Three of these groups (\textcolor{magenta}{pink}, \textcolor{cyan}{blue}, and \textcolor{olive}{yellow}) are assigned to the control treatment,  and are differentiated by the fraction of ``neighboring'' {\buyer}-{\seller} pairs\textemdash units in the same row or column\textemdash which are assigned to treatment. For example, pairs $(i,j)$ colored in pink are exposed to control, with $4/4$ of units $(i',j)$, $i'\ne i$ in the same column also exposed to control, and $5/6$ of units $(i,j')$, $j' \ne j$ in the same row exposed to treatment. Similarly, yellow pairs are also exposed to control, again with $4/4$ of units in the same column also exposed to control, but only $2/6$ of units in the same row  exposed to treatment.
Prior to intervention, all four groups are comparable due to randomization. After the intervention, even the three groups of control units need not be comparable: their outcomes might differ systematically due to spillovers.

\begin{equation}\label{equilibrium110}
\begin{split}
     \arraycolsep=14pt\def\arraystretch{0.65}
  \begin{NiceArray}{rl@{\hspace{.1em}}cccccc@{\hspace{1em}}c}[%
 name=superwmat,%
 create-extra-nodes,%
 margin,%
 extra-margin=0pt,%
 cell-space-limits=3pt,%
 baseline=7]
  & & & & & & & & {\rm Buyers:}\ i \\
  &{\rm Sellers:}\ j \rightarrow &1 & 2 & 3 & 4 & 5 & 6 & \rotate \leftarrow \\ 
  & & \textcolor{cyan}{\cc} & \textcolor{cyan}{\cc} & \ct& \textcolor{olive}{\cc} &\ct&\textcolor{olive}{\cc} & 1 \\
  & & \textcolor{cyan}{\cc} & \textcolor{cyan}{\cc} & \ct& \textcolor{olive}{\cc} &\ct& \textcolor{olive}{\cc} & 2 \\
  & &\ct&\ct&\ct& \textcolor{magenta}{\cc} &\ct&\textcolor{magenta}{\cc} & 3 \\
  & &\textcolor{cyan}{\cc} & \textcolor{cyan}{\cc}  & \ct&\textcolor{olive}{\cc} &\ct& \textcolor{olive}{\cc} & 4 \\
  & &\ct&\ct&\ct& \textcolor{magenta}{\cc} &\ct&\textcolor{magenta}{\cc}  & 5
 \CodeAfter
    \SubMatrix({3-3}{7-8})[name=wmat]
    \begin{tikzpicture}
        \node[left] (W) at (wmat-left.west) {$\bw =$} ;
    \end{tikzpicture}
 \end{NiceArray}
 \begin{tikzpicture}[remember picture,overlay]
 \end{tikzpicture}
 \end{split}
 \end{equation}

Consider the {\buyer}-{\seller} example in which for pairs of {\buyer}s and {\seller}s assigned to the treatment group the {\buyer} gets to see more information about the {\seller} or product in the form of additional reviews. {\Buyer} 1 gets the additional information when interacting with {\seller}s 3 and 5, but not when interacting with {\seller}s 1, 2, 4 and 6. If the information is generally helpful, this may lead buyer 1 to switch engagement from {\seller}s 1, 2, 4 and 6 to {\seller}s 3 and 5, a common form of spillover. {\Seller}s 3 and 5 are in the treatment group for all buyers. If the information raises the engagement with those sellers relative to, say sellers 4 and 6 who are always in the control group, this may lead {\seller}s 3 and 5 to change other behaviors, such as their marketing strategy, leading to a different type of  spillover.

Formally, spillovers are present whenever potential outcomes $y_{ij}(\bww)$ and $y_{ij}(\bww')$ differ for assignments $\bww$ and $\bww'$ where the treatment for the pair $(i,j)$ is identical, $w_{ij}=w'_{ij}$, but some other elements of the assignment matrices $\bww$ and $\bww'$ differ. Obtaining unbiased estimates of  causal effects in the presence of spillovers is challenging: classical causal analyses typically impose strong assumptions that rule out any form of cross-unit interference (\eg, the stable unit value assumption or SUTVA, \citealp{rubin1974estimating}).

We now introduce different assumptions on the potential outcomes, leading to different structures for the interference. We later discuss in \cref{sec:mrds} how alternative forms of interference can be effectively addressed using specific experimental designs. The simplest possibility is to rule out  \emph{any} type of interference (a version of SUTVA where the experimental unit is given by a {\buyer}-{\seller} pair).

\begin{assumption}[Strong No-Interference]\label{sutva}
    Potential outcomes satisfy the strong no-interference assumption if
        $y_{ij}(\bww)=y_{ij}(\bww')$,
    for all $(i,j)$ such that $w_{ij}=w'_{ij}$.
\end{assumption}
Under \cref{sutva}, a natural approach is to randomize all pairs, subject to treatment balance within {\buyer}s and {\seller}s.
This generally allows for more efficient estimation than designs which randomize only buyers or only sellers.

A natural way to weaken \cref{sutva} is to allow the outcome for a given {\buyer}-{\seller} pair to additionally depend on the treatment assignments involving the same {\buyer} but different {\seller}s (but not to depend on the assignments received by other {\buyer}s). Let $\bww,\bww'$ be assignment matrices where the treatment for the pair $(i,j)$ coincides, so $w_{ij}=w'_{ij}$, but there is a {\seller} $j'$ for which $w_{ij'}\neq w'_{ij'}$. Under this type of interference, it may be that $y_{ij}(\bww)\neq y_{ij}(\bww')$. However, for any assignment $\bww''$ with $w_{ij'}''=w_{ij'}, \forall j' \in \mmj$, $y_{ij}(\bww)=y_{ij}(\bww'')$. We formalize this form of interference in \cref{assp:no_buyer_interference}.

\begin{assumption}[No-Interference for {\Buyer}s] \label{assp:no_buyer_interference}
     Potential outcomes satisfy the no-interference for {\buyer}s assumption if
        $y_{ij}(\bww)=y_{ij}(\bww')$
    for all $(i,j)$ such that
     $w_{ij'}=w'_{ij'}$ for all $j' \in \mmj$.
\end{assumption}

Under \cref{assp:no_buyer_interference}, changing one or more of the treatment assignments for a different {\buyer} $i'$ does not change the outcomes for {\buyer}-{\seller} pair $(i,j)$.
But, changing one or more of the treatments for a different {\seller} $j'$ may affect the outcome $y_{ij}$. 
Under this assumption a {\buyer}-randomized experiment\revision{, corresponding to the matrix assignment later introduced in \cref{eq:srd},} is a natural strategy. Similarly, a {\seller}-randomized experiment is natural if we expect the following ``no-interference for {\seller}s'' assumption to hold. 
\begin{assumption}[No-Interference for {\Seller}s] \label{assp:no_seller_interference}
     Potential outcomes satisfy the no-interference for {\seller}s assumption if
        $y_{ij}(\bww)=y_{ij}(\bww')$
    for all $(i,j)$ such that
     $w_{i'j}=w'_{i'j}$ for all $i' \in \mmi$.
\end{assumption}

Next, we consider an assumption first introduced in \citet{bajari2023experimental} that allows for some forms of interference across both  buyers and sellers. 
This is a key assumption in our paper. It attempts to balance competing interests: allowing for a substantial degree of interference and at the same time imposing enough structure so that questions of interest are answerable. 

\begin{assumption}[Local Interference]\label{sutva_local}
Potential outcomes satisfy the local interference assumption if
$y_{ij}(\bww)=y_{ij}(\bww')$,
     for any pair $(i,j)$,  such that (a) the assignments for the pair $(i,j)$ coincide, $w_{ij} = w'_{ij}$, (b) the fraction of  treated {\seller}s for  {\buyer} $i$ coincide under $\bww$ and $\bww'$,  
    and (c) the fraction of treated {\buyer}s for  {\seller} $j$ coincide under  $\bww$ and $\bww'$.  
\end{assumption}  
Consider the following two assignment matrices $\bww, \bww'$:
\begin{equation*}
    \arraycolsep=10pt\def\arraystretch{0.65}
    \bww =
        \left({\begin{array}{ccccc}
        \cc & \ct  & \textcolor{purple}{\cc} & \cc& \cc  \\
        \ct & \cc   &\textcolor{purple}{\ct} & \cc& \ct \\
        \textcolor{purple}{\ct} & \textcolor{purple}{\cc}  & \textcolor{purple}{\ct} &\textcolor{purple}{\ct}& \textcolor{purple}{\ct}\\
        \cc& \cc   & \textcolor{purple}{\cc} &\cc& \cc
        \end{array}}  \right),
        \; 
    \bww' = \left(
        {\begin{array}{ccccc}
         \cc & \ct  &  \textcolor{purple}{\ct} & \cc& \cc  \\
        \cc & \ct   & \textcolor{purple}{\cc} & \cc& \cc \\
         \textcolor{purple}{\ct} &  \textcolor{purple}{\cc}  & \textcolor{purple}{\ct} & \textcolor{purple}{\ct}&  \textcolor{purple}{\ct}\\
        \cc& \cc   &  \textcolor{purple}{\cc} &\ct& \ct 
        \end{array}}
        \right).
\end{equation*} 
Under local interference, the outcome for {\buyer}-{\seller} pair $(3,3)$ must be identical for the assignment matrices $\bww$ and $\bww'$ (that is, $y_{33}(\bww)=y_{33}(\bww')$), because (a) the (3,3) elements of $\bww$ and $\bww'$  are identical, and (b) the third columns of the assignment matrices (given in purple) have the same fraction of treated pairs $(1/2)$, and
(c) the third rows of the assignment matrices   (also given in purple) have the same fraction of treated pairs (4/5).

Although obviously weaker than \cref{sutva} which rules out all interference, and more flexible than \cref{assp:no_buyer_interference} which rules out interference between buyers while allowing for interference within sellers, local interference does still substantially restrict the possible forms of interference between units. In particular, for a given unit pair $(i,j)$ only $I+J-1$ of the total $IJ$ unit-level assignments defining $\bww$ are relevant to the realized outcome: those of pairs $(i,j')$ and $(i',j)$. Further, the unit-level outcome is a function of only three sufficient statistics: the unit's own treatment assignment ($w_{ij})$, and the averages of the (same) row and column to which the pair belongs ($\sum_{i'} w_{i'j}/I$, and $\sum_{j'} w_{ij'}/J$). 

\revision{Similar forms of interference were previously proposed by \citet{manski2013identification} (cf.\ ``anonymous interactions'') and \citet{hudgens2008toward} (cf.\ ``stratified interference'').} Despite its simplicity, we believe that this assumption is a natural starting point for  approximating many types of interference that arise due to strategic behavior in a two-sided market. To illustrate, we now provide a simple example of a two-sided marketplace in which---at Nash equilibrium---potential outcomes exhibit both buyer and seller interference, and satisfy local interference. Later we show that under some designs, including the leading Simple MRD, local interference has no testable implications. 
\revision{We also show in \cref{sec:extensions} that more complex MRDs do lead to testable implications on the conditional expectations (over treatment assignments) of the outcomes.}

\begin{example}\label{example:strategic-local-interference} \normalfont
    Consider a two-sided platform where content creators $i ~\in~ [I]$ and advertisers $j \in [J]$ interact. 
    Each content creator $i$ produces corresponding content with score $q^c_i$, and each advertiser places ads with corresponding advertisement quality $q^a_j$. 
    In this model, each creator-advertiser pair generates revenue $y_{ij}$. In the absence of any intervention, revenue generated by $(i,j)$ is given by 
    \[
        y_{ij} = m_{ij} \{q^c_i + q^a_j\},
    \]
    where the (fixed) scalar factor $m_{ij} \in \mathbb{R}$ reflects the compatibility between $i$ and $j$ (e.g., footwear ads might have higher compatibility with content produced by a creator focusing on sports). 
    Creators and advertisers are compensated by the platform according to a contract: for each pair $(i,j)$, creator $i$ is compensated $r^c_i y_{ij}$ and advertiser $j$ is compensated $r^a_j y_{ij}$, and the platform keeps $(1 - r^c_i - r^a_j)y_{ij}$; the platform negotiates $r^c_i$, $r^a_j$ with each creator and advertiser. 
    In practice, generating high-quality content requires costly effort. In particular, we suppose that both creators and advertisers maximize their total compensation minus the cost of effort:
    \[
        U^c_i =  
        \left(
            \sum_{j=1}^J r^c_i y_{ij}
        \right) 
        - 
        \frac{(q^c_i)^2}{2},
        \qquad\text{and}\qquad
        U^a_j =  
        \left(
            \sum_{i=1}^I r^a_j y_{ij}
        \right) 
        - 
        \frac{(q^a_j)^2}{2}.
    \]
    In the static Nash equilibrium, each creator and advertiser solves the maximization problem treating the other agents' inputs $q_{i}^c, q_j^a$ as fixed and known. This leads to the equilibrium actions
    \[
    q_i^c = \sum_{j=1}^J r^c_i y_{ij}, \quad\text{and}\quad q^a_j =  \sum_{i=1}^I r^a_j y_{ij}.
    \]
        The platform hosting the content creators and advertisers tests the impact of a subsidy via a binary intervention $\bww$ affecting the revenue as follows:
    \[
        y_{ij}(\bww) = (m_{ij} + \eta w_{ij})\{q^c_i(\mathbf{w}) + q^a_j(\mathbf{w})\}.
    \]
    Here, for $\eta \in \mathbb{R}$, the factor $\eta w_{ij} \in\{0,\eta\}$ represents an extra incentive paid by the platform ($\eta$ is the incentive, and $w_{ij}$ is a binary treatment variable).
    Notice that each agent's incentives depends on the average treatment status of their interactions. This influences their action, which creates precisely a local interference structure. At Nash equilibrium, the revenue $y_{ij}$ and profit $\pi_{ij}$ both satisfy local interference; they are given by  
    \begin{equation}
    \begin{split}
        y_{ij}(\bww)
        &= (m_{ij} + \eta w_{ij})
        \left[
            Jr^c_i\left( \bar m^c_i+\eta\bar w^c_i \right)
            +Ir^a_j\left( \bar m^a_j+\eta\bar w^a_j\right)
        \right], \\
        \pi_{ij}(\bww)
        &= \{(1 - r_i^c - r_j^a)m_{ij} - ( r_i^c + r_j^a)\eta w_{ij})\}
        \left[
            Jr^c_i\left(\bar m^c_i+\eta\bar w^c_i)\right)
            +Ir^a_j\left(\bar m^a_j+\eta\bar w^a_j\right)
        \right], 
        \end{split}
        \label{eq:potential-outcomes-local-interference}
    \end{equation}
    where $\bar w^c_i = \frac{1}{J}\sum_{j'=1}^J w_{ij'}$, and $\bar m^c_i ~=~ \frac{1}{J}\sum_{j'=1}^J m_{ij'}$  and $\bar m^a_j$, $\bar w^a_j$ are defined symmetrically. 
\end{example}
\Cref{example:strategic-local-interference} shows a two-sided-marketplace with strategic agents in which agents' equilibrium actions lead potential outcomes (revenue or profits) to satisfy local interference (as in \cref{eq:potential-outcomes-local-interference}). 
Local interference arises somewhat naturally, as it assumes the outcome of an interaction between two agents will depend non-parametrically on the interaction-level treatment, as well as both agents' cumulative exposure to treatment. More generally, local interference may be viewed as a natural, tractable first approximation to the complex spillover effects arising in a two-sided marketplace. 
In \cref{sec:experiment}, we simulate the above example to show that agents' strategic responses can lead to large spillover effects, which are neglected by traditional designs.
In this way, our results are closely related to but distinct from the work of \cite{munro2021treatment} on treatment effects in market equilibrium: for example, the above Nash equilibrium in a finite marketplace is not captured by that work. It is also related to the works of \citet{harshaw2022design} and \citet{aronow2017estimating} in that potential outcomes depend on low-dimensional measures of ``exposure,'' though distinct in that we place agents on both sides\textemdash as opposed to one side \textemdash of the bipartite network.

\section{Multiple Randomization Designs} \label{sec:mrds}

Multiple Randomization  Designs (MRDs) are a generalization of standard A/B tests to allow for spillover effects common in marketplaces \citep{bajari2023experimental, johari2022experimental}. These designs can provably detect and measure spillover effects of the type introduced in \cref{sec:interference}, as we will discuss in \cref{sec:estimation}.  Let $\mmw$ denote the set of $2^{IJ}$ values that the random binary assignment matrix $\bw$ can take. We now formally define MRDs.

\begin{definition}[Multiple Randomization Designs]\label{def:MED}
A Multiple Randomization Design (MRD) is a probability distribution over $\mmw$, $p:\mmw\mapsto [0,1)$, such that 
\begin{inumerate}
\item \label{s3d1-1} $p(\cdot)$ is row and column exchangeable, and 
\item \label{s3d1-2} there exists  $\ooww\in(0,1)$ such that for any $\bm{w} = (w_{ij}) \in \{0,1\}^{I \times J}$ in the support of $p$,
\end{inumerate}
\[
    \frac{1}{IJ}\sum_{i=1}^I\sum_{j=1}^J \been({w_{ij}=\wt})=\ooww.
\]
\end{definition}
Note that a probability distribution $p(\cdot)$ over matrices $\mathbf{w}$ is said to be row (or column) exchangeable if, under $p(\cdot)$, any two assignments which differ by a permutation of the rows (or columns) are assigned the same probability. By imposing exchangeability of $p(\cdot)$ through \cref{def:MED}\ref{s3d1-1} we rule out the possibility of degenerate experiments in which a single value $\bww$ has probability one.
Condition \ref{def:MED}\ref{s3d1-2} ensures that all assignments with positive probability have the same fraction $\ooww$ of treated {\buyer}-{\seller} pairs. It is not strictly necessary, but it helps us to derive exact finite-sample results in \cref{sec:estimation}, clarifying what can be learned without  large sample approximations. 

Given an assignment matrix $\bww$, for each {\buyer} $i$ let $\owwi$ be the fraction of  {\seller}s $j$ for which $(i,j)$ received the treatment, and let $\owwj$ be the symmetric quantity for {\seller} $j$:
\begin{equation*}
    \owwi:=\sum_{j=1}^J \frac{\been({w_{ij}=\wt})}{J}, 
    \qquad\text{and}\qquad 
    \owwj:=\sum_{i=1}^I \frac{\been({w_{ij}=\wt})}{I}. 
\end{equation*}
\Cref{def:MED} implies that $\ooww=\sum _i \owwi/I=\sum_j \owwj/J$. A key feature of an MRD is that it allows both {\buyer}s and {\seller}s to be exposed to different treatments within the same experiment. We refer to the presence of such variation in the assignment as {\it inhomogeneity} of the {\buyer} or {\seller} experience.
\begin{definition}[Homogeneous and Inhomogeneous Experiences] 
Assignment ${\bf w}$ induces a homogeneous experience for {\buyer}\ $i$ if $\owwi\in\{0,1\}$, and an inhomogeneous experience for {\buyer}\ $i$ if $\owwi\in(0,1)$. 
Similarly, it induces a homogeneous experience for {\seller} $j$ if $\owwj\in\{0,1\}$ and an inhomogeneous experience for {\seller} $j$ if $\owwj\in(0,1)$.
\end{definition}
In assignment matrix \eqref{equilibrium110}, {\seller}s 3, 4, 5 and 6 have a homogeneous experience while {\seller}s 1 and 2 and all {\buyer}s have an inhomogeneous experience. Inhomogeneous experiences are at the heart of spillover concerns in our set-up. 
Suppose that the treatment corresponds to offering more information to some {\buyer}-{\seller} pairs. Buyers with an inhomogeneous experience may shift their engagement from sellers in the control group to sellers in the treatment group, without changing their overall engagement or expenditure.

Next, we showcase the flexibility of MRDs by defining three classes of experimental designs that fit within the general Definition~\ref{def:MED}. 
These three classes do not exhaust the possibilities, but make specific points: they show that MRDs
\begin{inumerate}
\item encompass standard experimental designs, (\cref{sec:srd}), 
\item can increase efficiency 
(\cref{sec:crd}) 
and 
\item most importantly, in certain cases can answer questions that standard designs cannot answer, as we discuss in \cref{sec:smrd}. 
\end{inumerate}
We conclude the section by discussing connections between these designs and the local interference assumption introduced in \cref{sec:interference}.

\subsection{Single Randomization Designs}  \label{sec:srd}

A Single Randomization Design (SRD) is an MRD where each {\buyer} or {\seller} has a homogeneous experience with probability one: i.e.\ a {\buyer} experiment  ($\owwi\in\{0,1\}$ and $\owwj=\ooww$), or a {\seller} experiment ($ \owwi=\ooww$ and $\owwj\in\{0,1\}$).
A {\buyer} experiment is a simple {\buyer}-randomized A/B test, where assignment matrices are of the form of \eqref{eq:srd}, with identical columns and constant rows: 
\begin{equation}
    \bww = 
    {{
    \left(
    \begin{array}{cccccccccc}
        \wc  & \wc  & \wc   & \wc  & \wc & \wc  &\wc &\wc \\
        \wt   & \wt  & \wt    &\wt  & \wt & \wt  &\wt &\wt \\
        \wc   & \wc  & \wc    & \wc  &\wc & \wc  &\wc &\wc \\
        \wc   & \wc  & \wc    & \wc  & \wc &\wc  &\wc &\wc \\
    \end{array}
    \right)
    }.} 
    \label{eq:srd}
\end{equation}
Here {\buyer}s $1,3,4$ are in the control group, and {\buyer} $2$ is in treatment. All buyers here have a homogeneous experience, whereas none of the sellers have a homogeneous experience.

\subsection{Crossover Designs} \label{sec:crd}

In contrast to standard (buyer or seller) experiments, MRDs include experiments in which neither all {\buyers} nor all {\sellers} have homogeneous experiences.
The simplest such an MRD is one in which all interactions $(i,j)$ are randomly assigned.
This design is widely used in settings where the second dimension is time, and where such designs have been referred to 
as rotation experiments \citep{cochran1939long}, crossover experiments \citep{brown1980crossover}, or switchback experiments \citep{bojinov2020design}, although it is not limited to settings where time is one of the dimensions. An example is given in assignment matrix \eqref{eq:staggered}:
\begin{equation} \label{eq:staggered}
\bww =
    \begin{NiceArray}{c>{\color{gray}}c@{\hspace{1em}}cccccccc@{\hspace{1em}}}[%
    margin]
    \\
    \RowStyle[color=gray]{}
    &   & 1   & 2     & 3     & 4     & 5     & 6     & 7     & 8     \\
    & 1 & \wt & \wt   & \wc   & \wc   & \wc   & \wt   &\wc    & \wt   \\
    & 2 & \wt & \wt   & \wt   & \wc   & \wt   & \wc   &\wc    & \wc   \\
    & 3 & \wc & \wc   & \wc   & \wt   & \wt   & \wt   &\wc    & \wt   \\
    & 4 & \wc & \wc   & \wt   & \wc   & \wc   & \wt   &\wt    & \wt   \\
    & 5 & \wc & \wt   & \wc   & \wt   & \wt   & \wc   &\wt    & \wc   \\
    & 6 & \wt & \wc   & \wt   & \wt   & \wc   & \wc   &\wt    & \wc   \\
    \\
    \CodeAfter
        \SubMatrix({3-3}{8-10})
        \begin{tikzpicture}
            \draw[decorate,thick,color=gray] (2.1) node[above,xshift=1em] (S) {\text{Time Period}};
            \draw[-latex,thick,color=gray] (S.south) |- (2-3.west);
            \draw[decorate,thick,color=gray] (2.1) node[below,xshift=-3em,yshift=2pt] (B) {Individual Unit};
            \draw[-latex,thick,color=gray] (B.east) -| (3-2.north);
        \end{tikzpicture}
    \end{NiceArray}.
\end{equation}
In assignment matrix \eqref{eq:staggered} we consider a balanced design, where each unit is in the treatment group for four periods, and in every period exactly three units are in the treatment group. It is particularly attractive in settings where strong no-interference is reasonable (\cref{sutva}), where, under additional assumptions on the potential outcomes, it can be shown to improve efficiency \citep{masoero2023efficient}.

\begin{remark}
    A related experimental design is that of staggered adoption, or the ``stepped wedge'' design. There, units are assigned to the treatment at different points in time, but once assigned to the treatment they never exit. See
    \citet{athey2022design,
    hemming2015stepped} for analyses of these experiments, and
     \citet{xiong2023optimal} for optimal design.
\end{remark}

\subsection{Simple Multiple Randomization Designs} \label{sec:smrd}

The next design we consider introduces systematic variation in  
$\owwi$ over {\buyer}s and variation  in $\owwj$ over {\seller}s. 
Such variation allows for the detection of spillovers, as well as for estimation of their magnitude.
To accomplish this goal, we randomize {\buyer}s and {\seller}s separately: we select at random $I_T$ buyers, with $1< I_T < I-1$ and assign them $W_i^{\rb}=1$. For the remaining \buyers, $W_i^{\rb}=0$, so that we have a {\buyer}-assignment random vector $\bm{\vec{W}}^{\rb} \in \{0,1\}^I$ with $\sum_i W_i^{\rb} = I_T$. Symmetrically, we select $J_T$ \sellers\ at random, with $1<J_T<J-1$ and assign them $W_j^{\rs} = 1$. The remaining sellers are assigned $W_j^{\rs} = 0$, yielding a {\seller}-assignment random vector $\bm{\vec{W}}^{\rs} \in \{0,1\}^J$ with $\sum_j W_j^{\rs} = J_T$. 
Then the assignment for the  pair $(i,j)$ is a function of the {\buyer} and {\seller} assignments $W_i^{\rb}$ and $W_j^{\rs}$.

\begin{definition}[Simple Multiple Randomization Designs] \label{smrd}
Given a population of $I$ buyers and $J$ sellers, a  Simple Multiple Randomization Design (SMRD) is an MRD in which, for fixed proportions $p^{\rb} = I_T/I \in (0,1)$ and $p^{\rs} = J_T/J \in (0,1)$, we randomly assign to each buyer $W_i^{\rb} \in \{0,1\}$ such that $\sum_{i} W_{i}^{\rb} = I_T$, and independently randomly assign each seller $W_j^{\rs} \in \{0,1\}$ such that 
$\sum_{j} W_j^{\rs} = J_T$. 
\revision{The pair $(i,j)$ is exposed to treatment via
\begin{equation}
    W_{ij}= 
    \begin{cases}
            \wt &\mbox{ \rm{if} } \min\left(w^\rb,w^\rs\right) = 1, \\
        \wc &\mbox{ \rm{otherwise}}.
    \end{cases}\label{eq:conjunctive} 
\end{equation}
}
\end{definition}
While SMRDs do not have the richness of the full class of MRDs, they contain many of the insights that apply to the general case. This special case of MRDs has also been discussed in \citet{johari2022experimental}, where the focus is on the bias  of the difference in means estimator  for the average treatment effect. 
See also \citet{bajari2023experimental, li2021interference}.

An assignment example for an SMRD  is given in matrix \eqref{eq:smrd}, where the {\buyer}-assignment vector $\bm{\vec{w}}^\rb = [0,0,1,1]$ and {\seller}-assignment vector $\bm{\vec{w}}^\rs = [0,0,0,0,1,1,1,1]$ lead to:
\begin{equation}
    \bww={
    \left(\begin{array}{cccccccccc}
     \textcolor{red}{\cc} & \textcolor{red}{\cc} & \textcolor{red}{\cc} & \textcolor{red}{\cc}  &\textcolor{blue}{\cc} &\textcolor{blue}{\cc} & \textcolor{blue}{\cc}& \textcolor{blue}{\cc} \\
    \textcolor{red}{\cc}  & \textcolor{red}{\cc} & \textcolor{red}{\cc}   &\textcolor{red}{\cc}  &\textcolor{blue}{\cc} &\textcolor{blue}{\cc} & \textcolor{blue}{\cc}& \textcolor{blue}{\cc} \\
    \textcolor{green}{\cc}  & \textcolor{green}{\cc} & \textcolor{green}{\cc}    & \textcolor{green}{\cc}    & \ct & \ct &\ct& \ct\\
    \textcolor{green}{\cc}   & \textcolor{green}{\cc}  & \textcolor{green}{\cc}    & \textcolor{green}{\cc}    & \ct & \ct &\ct& \ct\\
    \end{array}\right)}. \label{eq:smrd}
\end{equation}
In these SMRDs, the pairs of binary values  $(w^\rb_i,w^\rs_j)$ induce  four assignment types of {\buyer}-{\seller}  pairs (each type identified by a different color in the  assignment matrix \eqref{eq:smrd}):
\begin{align} \label{eq:types}
    \type_{ij}= 
     \begin{cases}
         \ccc &\mbox{ if } w^\rb_i=0,w^\rs_j=0 \; ({\rm so}\ w_{ij}=0), \\
         \icb &\mbox{ if } w^\rb_i=1,w^\rs_j=0 \; ({\rm so}\ w_{ij}=0), \\
         \ics &\mbox{ if } w^\rb_i=0,w^\rs_j=1\; ({\rm so}\ w_{ij}=0), \\
         \ctt &\mbox{ if } w^\rb_i=1,w^\rs_j=1 \; ({\rm so}\ w_{ij}=1).
     \end{cases}
\end{align}
Here, $\ccc$ is ``homogeneous control'', $\icb$ ``inhomogeneous {\buyer} control'', $\ics$ ``inhomogeneous {\seller} control'', and $\ctt$ ``treated''. Consistent with \cref{eq:conjunctive}, $w_{ij}=\ct$ if ${\type_{ij}=\ctt}$ and $w_{ij}=\cc$ otherwise. The values $w^\rb_i$ and $w^\rs_j$ can be inferred from the assignment matrix $\bww$, hence the type can be inferred from the assignment matrix, $\type_{ij}=\type_{ij}(\bww)$. 
These assignment types play an important role under the local interference assumption (\ref{sutva_local}), as \revision{highlighted} in \cref{lemma:local_interference}.
\begin{lemma} \label{lemma:local_interference}
    For $\bww, \bww'$ consistent with an SMRD and assuming that potential outcomes satisfy local interference (\cref{sutva_local}), potential outcomes can be written as a function of the assignment types only: for $\bww, \bww'$ it holds that
    \[ 
        \type_{ij}(\bww)=\type_{ij}(\bww')
        \Rightarrow 
        y_{ij}(\bww)=y_{ij}(\bww').
    \]
\end{lemma}
This simplification, where potential outcomes depend only on a function of their original argument, is related to the exposure mapping concept in \citet{aronow2017estimating}.

Of the four groups of {\buyer}-{\seller} pairs induced by an SMRD\textemdash all of which are comparable prior to treatment due to physical randomization\textemdash types $\ccc,\icb,\ics$ are all exposed to control. 
 Having multiple sets of pairs which are
\begin{inumerate}
    \item comparable prior to treatment, 
    \item all exposed to the same treatment (control) and 
    \item not comparable post-treatment, 
\end{inumerate}
gives SMRDs the ability to detect interference.
This ability is based on comparisons of average outcomes for these three groups in which pairs are all exposed to  control. Under a simple {\buyer} or {\seller} experiment, where only two types are present, and only one is exposed to the control treatment, spillovers could not be detected.

An interesting feature of the SMRD is that the local interference assumption is not testable here: differences between expected outcomes for the comparison groups can always be rationalized in a way that is consistent with local interference. 
We observe outcomes for four types of pairs, $\type_{ij}\in\{\ccc,\icb,\ics,\ctt\}$. The local interference assumption does not restrict the distribution of the outcomes for these four types. 
In contrast the no-interference for buyers assumption,  \cref{assp:no_buyer_interference}, does have testable implications in settings with a large number of buyers and sellers: it would imply that the distribution of outcomes for the $\ccc$ pairs
is the same as the distribution of outcomes for the $\icb$ pairs.

\section{Estimation and Inference for SMRDs} \label{sec:estimation}

We now describe methods that make the experimental designs introduced in the previous section practically useful by enabling statistical inference. Specifically, we provide five results.
First, we introduce estimands and estimators for causal effects in the presence of local interference for the SMRD (\cref{sec:estimands}). 
Second, we show the proposed estimators are unbiased (\cref{sec:estimators}). 
Third, we characterize the exact finite sample variance of these estimators (\cref{sec:variance_characterization}). 
Fourth, we derive, in the tradition of the causal inference literature, conservative estimators for their variances (\cref{sec:variance_estimators}). 
Finally, we provide central limit theorems that allow for the construction of confidence intervals (\cref{sec:CLTs}). 
Proofs are deferred to the appendix.
While seemingly standard, our results require a non-trivial amount of technical complexity due to the fact that randomization acts jointly on the multiple dimensions through which potential outcomes are indexed.

In what follows, for a given type $\type$, we let $I_\type$ ($J_\type$) denote the number of \buyers\ $i$ (\sellers\ $j$) for which there is at least one pair $(i,j)$ such that $\type_{ij} = \type$.
For example, in the assignment of \cref{eq:smrd}, $I_\type=2$  and $J_\type=4$ for all $\type \in \types$.
This is because the first two \buyers have pairs exposed to $\ccc, \ics$ (so that $I_\cc = I_\ics = 2$) and the last two have pairs exposed to $\icb, \ctt$ (so that $I_\icb = I_\ctt = 2$). A symmetric argument holds for sellers.
Moreover, whenever we consider an SMRD for which local interference holds,
we leverage \cref{lemma:local_interference} and --- with some abuse of notation --- write $y_{ij}(\type)$ instead of $y_{ij}(\bww)$.

\subsection{Causal Estimands and Spillover Effects}  \label{sec:estimands}
Under the local interference \cref{sutva_local}, \cref{lemma:local_interference} proves that the potential outcomes $y_{ij}$ are indexed by type $\type_{ij}\in\types$.
Define the population averages by type:
\begin{equation}
    \meanpopulation{\type} :=\frac{1}{IJ}\sum_{i=1}^I\sum_{j=1}^J y_{ij}(\type),\; \text{for}\;  \type \in \types. \label{eq:meantypepopulation}
\end{equation}

For $\coefvec = [\beta_{\ccc},\beta_{\icb}, \beta_{\ics},\beta_{\ctt}]^\top$, we consider causal estimands that can be written as linear combinations of the $ \meanpopulation{\type} $  defined in \cref{eq:meantypepopulation}:
\begin{equation}
    \tau(\coefvec) :=
    \beta_{\ccc}\meanpopulation{\ccc} +
    \beta_{\icb}\meanpopulation{\icb} +
    \beta_{\ics}\meanpopulation{\ics} +
    \beta_{\ctt}\meanpopulation{\ctt}. \label{eq:causal_estimands}
\end{equation}
This class of estimands includes many interesting quantities that shed light on the direct effect of the treatment, the spillover effects on untreated units stemming from applying treatment to other pairs, and the total effect. For example,  $\betaate := [-1,0,0,1]^\top$ corresponds to $\tau_{\ATE}:= \tau(\betaate) = \meanpopulation{\ctt} - \meanpopulation{\ccc}$, which is the average treatment effect of assigning both buyer $i$ and seller $j$ to treatment versus both being assigned to control under an SMRD design. 
Like all other estimands in our setting, $\tau_{\ATE}$ is implicitly parametrized by the fractions $p^\rb \in (0,1)$ of treated {\buyers}, $p^\rs \in (0,1)$ of treated {\sellers}. 
For $\betaspillb:=[-1,1,0,0]^\top$, $\tauspillb:=\tau(\betaspillb)=\meanpopulation{\icb} - \meanpopulation{\ccc}$ measures a ``{\buyer}''-spillover effect.
If there are no spillovers within buyers (\cref{assp:no_buyer_interference}), this average causal effect is equal to zero. Thus, the estimated counterpart of this estimand sheds light on the presence of {\buyer} spillovers.
Similarly, for $\betaspills := [-1,0,1,0]^\top$, $\tauspillb := \meanpopulation{\ics} - \meanpopulation{\ccc}$ measures a ``{\seller}''-spillover effect. $\betadirect:=[1,-1,-1,1]^\top$, which induces the effect 
$\taudirect$, is a measure of something closer to the direct effect of the treatment, removing the spillover effects.

To elaborate and be more precise about the value of these estimands for decision making, note that within the class of SMRD's indexed by the probabilities $p^\rb$ and $p^\rs$, the population averages $\meanpopulation{\type}$ depend on the values of these probabilities, other than 
$\meanpopulation{\ccc}$.  A natural object of interest for a decision maker is the average effect of switching from no exposure to all buyer/seller pairs exposed. This can be written as
\[
    \meanpopulation{\ctt}(p^\rb = 1, p^\rs = 1) - \meanpopulation{\ccc}(p^\rb = 0, p^\rs = 0). 
\]
This cannot be estimated directly from an SMRD experiment with a single pair of values $(p^\rb,p^\rs),$ as it requires extrapolation to $p^\rs=1$ and $p^\rb=1$. Either doing an experiment with $p^\rb$ and $p^\rs$ close enough to one or carrying out a more complex experiment with variation in $p^\rs$ and $p^\rb$ would facilitate this.
A second goal for the decision maker may be to assess the magnitude of the spillovers relative to direct effects. SMRD experimentation lowers precision relative to completely randomized experiments, and if one finds the the spillovers are modest, one may not need to be concerned about the spillovers in future experimentation. 

Note that our analysis is richer than that presented in \citet{johari2022experimental}, where the focus is only on the estimand  defined as the average outcome for the treated, $\meanpopulation{\ctt}$, and  the average outcome for all pairs exposed to the control group, not adjusting for any spillovers.

\subsection{Unbiased Estimators for the Causal Effects}  \label{sec:estimators}
In what follows, we use capital letters to denote stochastic counterparts of the corresponding population quantities. 
In particular, we use $\Type_{ij}$ to denote the random ``type'' assigned to pair $(i,j)$ in the context of an SMRD. 
Define the realized counterpart of the population average of the  {\buyer}-{\seller} pairs by type introduced in \cref{eq:meantypepopulation}:
\begin{equation}
     \meanestimate{\type} := \frac{1}{I_{\type}J_{\type}}\sum_{i=1}^I\sum_{j=1}^J  y_{ij}(\type)\been({\Type_{ij}=\type}), \label{eq:meantypeestimate}
\end{equation}
\cref{lemma:unbiasedness} shows that in an SMRD under \cref{sutva_local}, \cref{eq:meantypeestimate} provides an unbiased estimator of the corresponding population average $ \meanpopulation{\type} $ defined in  \cref{eq:meantypepopulation}.

\begin{lemma} \label{lemma:unbiasedness}
Consider an SMRD in which local interference (\cref{sutva_local}) holds. 
The plug-in estimators in  \cref{eq:meantypeestimate} satisfy
\[ 
    \mme\left[\meanestimate{\type}\right] =\meanpopulation{\type},\;\forall\; \type\in\types.
\]
\end{lemma}

In light of \cref{lemma:unbiasedness}, a direct application of the linearity of the expectation implies that simple plug-in estimators of causal effects $\tau(\coefvec)$ of the form of \cref{eq:causal_estimands} are unbiased.
\begin{theorem} \label{thm:spillover_unbiasedness}
Consider an SMRD where \cref{sutva_local} holds. 
The plug-in estimators 
\(
	\hat{\tau}(\coefvec) =     
	\beta_{\ccc}\meanestimate{\ccc} +
    	\beta_{\icb}\meanestimate{\icb} +
    	\beta_{\ics}\meanestimate{\ics} +
   	 \beta_{\ctt}\meanestimate{\ctt}
\)
 for  $\tau(\coefvec)$ defined in \cref{eq:causal_estimands} satisfy
\begin{equation}
    \mme\left[ \hat{\tau}(\coefvec) \right] = {\tau}(\coefvec). \label{eq:linear_estimator}
\end{equation}

\end{theorem}

\subsection{Variances of Linear Estimators}  \label{sec:variance_characterization}

We now characterize the variances of linear estimators $\hat{\tau}(\coefvec)$  (\cref{thm:covariance_characterization}) and provide conservative estimates for their variances (\cref{thm:variance_bounds}). Our results generalize classic results for SRDs, but their derivation is more complex because of the double summation over {\buyer}s and {\seller}s, and requires additional notation.
Define the (population) average outcome for each {\buyer} and each {\seller}, for a given type $\type$:
\begin{equation}
    \axismeanpopulation{i}{\rb}(\type) 
    :=\frac{1}{J}\sum_{j=1}^J y_{ij}(\type), 
    \qquad\text{and}\qquad
    \axismeanpopulation{j}{\rs}(\type) 
    :=\frac{1}{I}\sum_{i=1}^I y_{ij}(\type). \label{eq:axismeanpopulation}
\end{equation}
Define the deviations from population averages  for {\buyer} $i$, {\seller} $j$, and interactions $(i,j)$:
\begin{align*}
    \axisvariationmeanpopulation{i}{\rb}(\type) 
    := 
        \axismeanpopulation{i}{\rb}(\type)-\meanpopulation{\type}
        ,
    \qquad  
    \axisvariationmeanpopulation{j}{\rs}(\type) 
    :=
        \axismeanpopulation{j}{\rs}(\type)-\meanpopulation{\type},
     \intertext{and}
	\axisvariationmeanpopulation{ij}{\rb\rs}(\type) 
	:= 
        y_{ij}(\type)
        -\axismeanpopulation{i}{\rb}(\type)
        -\axismeanpopulation{j}{\rs}(\type)
        +\meanpopulation{\type}.
\end{align*}
Next define the population variances for each type at the \buyer, \seller, and interaction level:
\[ 
    \axisvariancepopulation{\type}{\rb}
    := \frac{\sum_{i=1}^I\left[\axisvariationmeanpopulation{i}{\rb}(\type)\right]^2}{I}, 
    \;
    \axisvariancepopulation{\type}{\rs}
    := \frac{\sum_{j=1}^J\left[\axisvariationmeanpopulation{j}{\rs}(\type)\right]^2}{J},
    \;
    \axisvariancepopulation{\type}{\rb\rs}
    ~:=~
    \frac{\sum_{i=1}^I 
    \sum_{j=1}^J \left[
    \axisvariationmeanpopulation{ij}{\rb\rs}(\type)
    \right]^2}{IJ}.\]
We additionally define for all $\type,\type'\in\types$ the following quantities, which can be interpreted as the average square deviation from the mean at the \buyer, \seller, and interaction level:
\begin{align}
\begin{split}
    \axisvariancecrosspopulation{\type,\type'}{\rb}
    := 
    &\sum_{i=1}^I\frac{\left[
        \axisvariationmeanpopulation{i}{\rb}(\type) 
        - \axisvariationmeanpopulation{i}{\rb}(\type')
    \right]^2}{I},
    \qquad
    \axisvariancecrosspopulation{\type,\type'}{\rs}
    :=
    \sum_{j=1}^J 
    \frac{\left[
        \axisvariationmeanpopulation{j}{\rs}(\type)
        - \axisvariationmeanpopulation{j}{\rs}(\type')
    \right]^2}{J},\\
 \quad \axisvariancecrosspopulation{\type,\type'}{\rb\rs}
    &:=    
    \frac{1}{IJ}\sum_{i=1}^I \sum_{j=1}^J
    \left[
        \axisvariationmeanpopulation{ij}{\rb\rs}(\type)
        -\axisvariationmeanpopulation{ij}{\rb\rs}(\type')
    \right]^2.
    \label{eq:xis}
\end{split}
\end{align}
Last, define for $\type \in \types$ the weights
\begin{equation}
    \alpha^\rb_{\type} := \frac{1}{I-1}\frac{I-I_{\type}}{I_{\type}}
    \qquad\text{and}\qquad 
    \alpha^\rs_{\type} := \frac{1}{J-1}\frac{J-J_{\type}}{J_{\type}}.
    \label{eq:alpha_weights}
\end{equation}
Let
\[
        \nu^{\rb}_{\type,\type'} := 
        \begin{cases}
            \alpha^\rb_{\type}/2\mbox{ if } \type=\type',\text{ or } (\type,\type') \in \{(\ccc,\ics), (\ics,\ccc), (\icb,\ctt), (\ctt,\icb)\} \\
            - 1/(2(I-1)) \mbox{ otherwise,}
            \end{cases}
\]
and
\[
        \nu^{\rs}_{\type,\type'}:=
        \begin{cases}
             \alpha^\rs_{\type}/2 \mbox{ if } \type = \type' \text{ or } (\type,\type') \in \{(\ccc,\icb),(\icb,\ccc), (\ics,\ctt), (\ctt,\ics)\} \\
            - 1/(2(J-1)) \mbox{ otherwise.}
        \end{cases}
\]
We now characterize variances and covariances of all the estimators of the sample average defined in \cref{eq:meantypeestimate}. 
\begin{theorem} \label{thm:covariance_characterization}
For an SMRD where \cref{sutva_local} holds, and for all $\type,\type'$,
\begin{equation*}
    \mmc\left[ \meanestimate{\type}, \meanestimate{\type'} \right]
    = \nu^{\rb}_{\type,\type'} \zeta^{\rb}_{\type,\type'} + \nu^{\rs}_{\type,\type'} \zeta^{\rs}_{\type,\type'} + \nu^{\rb}_{\type,\type'}\nu^{\rs}_{\type,\type'} \zeta^{\rb\rs}_{\type,\type'},
\end{equation*}
where for $x \in \{\rb,\rs,\rb\rs\}$,
\(
    \zeta^{x}_{\type,\type'} :=  \axisvariancepopulation{\type}{x}+ \axisvariancepopulation{\type'}{x}- \axisvariancecrosspopulation{\type,\type'}{x}.
\)
\end{theorem}
Variances for the type estimator $\meanestimate{\type}$ are obtained using the formula above whenever $\type'=\type$.
Exact variances of estimators $\hat{\tau}(\coefvec)$ can be directly obtained by noting that $\hat{\tau}(\coefvec)$ is a linear estimator, for which the following decomposition holds:
\[
    \mmv(aX+bY) = a^2\mmc(X,X)+b^2\mmc(Y, Y)+2ab\mmc(X,Y).
\]

\subsection{Variance Estimation} \label{sec:variance_estimators}
We now present unbiased estimators for the variance of the sample average of potential outcomes defined in \cref{eq:meantypeestimate}; these are given in \cref{thm:sample_variance_avg_effs}. We then give lower and upper bounds on the variance of the linear estimators $\hat{\tau}(\coefvec)$ in \cref{thm:variance_bounds}. 

Towards this goal, we proceed to define the sample counterparts of the population quantities introduced in \cref{sec:variance_characterization}. 
Given a randomly drawn SMRD assignment matrix $\bw \in \mmw$, inducing corresponding types $\bm{\Gamma}$, let $\II_{\type}:=\{ i \in \mmi \;\text{s.t.}\; \Type_{ij}=\type\text{ for some }j\}$ with size $|\II_{\type}|=I_{\type}$ and $\JJ_{\type}:=\{ j \in \mmj \;\text{s.t.}\; \Type_{ij}=\type\text{ for some }i\}$ with size $|\JJ_{\type}|=J_{\type}$. From \cref{eq:types},  each $i\in\mmi$ belongs to exactly two sets $\II_{\type}$: if $W_i^{\rb}=0$, $i \in \II_\ccc$ and $i \in \II_{\ics}$. 
If $W_i^{\rb}=1$, $i \in \II_{\icb}$ and $i \in \II_{\ctt}$. 
Symmetrically, each $j \in \mmj$ belongs in exactly two sets $\JJ_{\type}$: if $W_j^{\rs}=0$, $j \in \JJ_\ccc$ and $j \in \JJ_{\icb}$, and if $W_j^{\rs}=1$, $j \in \JJ_{\ics}$ and $i \in \JJ_{\ctt}$. For $i\in\II_{\type}, j \in \JJ_{\type}$ the sample counterparts $\axismeanestimate{i}{\rb}(\type)$ of $\axismeanpopulation{i}{\rb}(\type)$ and $\axismeanestimate{j}{\rs}(\type)$ of $\axismeanpopulation{j}{\rs}(\type)$ are: 
\[
    \axismeanestimate{i}{\rb}(\type):=\frac{1}{J_{\type}}\sum_{j \in \JJ_{\type}} y_{ij}(\type),\qquad
    \axismeanestimate{j}{\rs}(\type):=\frac{1}{I_{\type}}\sum_{i \in \II_{\type}} y_{ij}(\type).
\]
We define estimator counterparts $\axisvarianceestimate{\type}{\rb}$ for $\axisvariancepopulation{\type}{\rb}$ ({\buyers}) and $\axisvarianceestimate{\type}{\rs}$ for $\axisvariancepopulation{\type}{\rs}$ ({\sellers}):
\[
    \axisvarianceestimate{\type}{\rb}
    :=
        \frac{1}{I_{\type}} \sum_{{i\in\II_{\type}}}
        \left[
            \axismeanestimate{i}{\rb}(\type) 
            - \meanestimate{\type}
        \right]^2
        ,
        \quad
    \axisvarianceestimate{\type}{\rs} 
    :=
        \sum_{{j \in \JJ_{\type}}}
            \frac{1}{J_{\type}}\left[\axismeanestimate{j}{\rs}(\type) 
            - \meanestimate{\type}\right]^2.
\]
For the interactions, we define the estimator counterpart $\axisvarianceestimate{\type}{\rb\rs}$ for $\axisvariancepopulation{\type}{\rb\rs}$:
\[
    \axisvarianceestimate{\type}{\rb\rs}
        ~:=~ \sum_{i\in\II_{\type}, j\in\JJ_{\type}} 
        \frac{\left(
            y_{i,j}(\type) 
            - \axismeanestimate{i}{\rb}(\type) 
            - \axismeanestimate{j}{\rs}(\type) 
            + \meanestimate{\type}
        \right)^2}{I_{\type}J_{\type}}.
\]    

\begin{theorem} \label{thm:sample_variance_avg_effs} 
For an SMRD where \cref{sutva_local} holds, for all $\type\in \types$, 
    \[
        \mme\left[\widehat{\Sigma}({\type})\right]= \mmv\left(\meanestimate{\type}\right), \quad\text{where}
    \]
\begin{align*}
    \hat{\Sigma}({\type})
    &:= \frac{
        \alpha^{\rb}_{\type}\axisvarianceestimate{\type}{\rb}
        + \alpha^{\rs}_{\type}\axisvarianceestimate{\type}{\rs}
        +\alpha^{\rb}_{\type}\alpha^{\rs}_{\type}\axisvarianceestimate{\type}{\rb\rs}
        }{
        1 - \alpha^\rb_{\type} - \alpha^\rs_{\type} + \alpha^\rb_{\type} \alpha^\rs_{\type}
        } 
        \\
        &- \frac{\frac{\alpha^\rb_{\type}}{(1-\alpha^\rb_{\type})}}{I_\type J_\type}  
       \sum_{i \in \II_{\type}, j \in \JJ_{\type}}
        \frac{
            \left(y_{i,j}(\type) 
            - \axismeanestimate{i}{\rb}(\type)\right)^2
        }{
        \frac{(J-1)(J_{\type}-1)}{(J-J_{\type})}
        } 
        - \frac{\frac{\alpha^\rs_{\type}}{(1-\alpha^\rs_{\type})}}{I_\type J_\type} 
        \sum_{i \in \II_{\type}, j \in \JJ_{\type}}
        \frac{
            \left(y_{i,j}(\type) 
            - \axismeanestimate{j}{\rs}(\type)\right)^2
        }{
        \frac{(I-1)(I_{\type}-1)}{(I-I_{\type})}
        }. 
\end{align*}
\end{theorem}
Young's inequality yields a conservative estimator for the variance of $\hat{\tau}(\coefvec)$:
\begin{equation}
    \widehat{\mmv}\left(\hat{\tau}^{\rm hi}(\coefvec)\right) = 
    \sum_{\type \in \types} \beta_\type^2 \hat{\Sigma}_\type + \sum_{\type \neq \type'} \beta_\type \beta_{\type'} \left(\hat{\Sigma}_\type + \hat{\Sigma}_{\type'}\right)
    \label{eq:young_cov}
\end{equation}
This result mirrors the case of SRDs \citep{neyman1923}. 
We provide the result for $\hattauspillb$ in \cref{thm:sample_variance_spillover}. See  \cref{thm:general_variance_sample_causal} in the Appendix for the case of a generic $\hat{\tau}(\coefvec)$.
\begin{theorem} \label{thm:variance_bounds}
Under the assumptions of \cref{thm:sample_variance_avg_effs} a conservative estimator of $ \mmv({{\hattauspillb}})$ is:
  \[
        \widehat{\mmv}^{\rm hi}(\hattauspillb) := 
        2 \left( \widehat{\Sigma}({\icb}) + 
        \widehat{\Sigma}({\ccc}) \right).
    \] 
    $\widehat{\mmv}^{\rm hi}(\hattauspillb) $ is conservative in the usual sense that $\mathbb{E}\left[\widehat{\mmv}^{\rm hi}(\hattauspillb)\right] \ge {\mmv}(\hattauspillb)$.
    \label{thm:sample_variance_spillover}
\end{theorem}
We emphasize that, while it is possible to provide an unbiased estimator for the variance of $\meanestimate{\type}$ (\cref{thm:sample_variance_avg_effs}), one \emph{cannot} provide an unbiased estimator for the covariance of $\meanestimate{\type}$ and $\meanestimate{\type'}$ for $\type\neq\type'$ without stronger assumptions on the potential outcomes. The same phenomenon occurs for conventional randomized experiments. This is because the terms $ \axisvariancecrosspopulation{\type,\type'}{x}$ introduced in \cref{eq:xis} depend on covariances of potential outcomes for the same {\buyer}-{\seller} pair, which cannot be identified from the observed data. 

It is however possible to show that the variance estimator $\hat{\Sigma}({\type})$ converges to the true underlying variance $\Sigma_\type$, i.e.\ $\mmv\big(\meanestimate{\type}\big)$, under relatively weak assumptions. By the continuous mapping theorem, this implies convergence of the general estimator $\widehat{\mmv}^{\rm hi}[ \hat{\tau}(\coefvec)]$ to its (conservative) limit.
A stronger version of this result was communicated to us by \citet{sudijono2025private}; a proof is given in \cref{sec:proofs_variance_estimation} for completeness.

We now state the result. To do so, we introduce two additional assumptions which will also be used in \cref{sec:CLTs} to derive a central limit theorem.
\revision{\begin{assumption} \label{ass:regularity}
Consider an SMRD with $I$ \buyers{} and $J$ \sellers{}, and assume that the local interference assumption \cref{sutva_local} holds. We impose the following regularity conditions. 
\begin{enumerate}[label=(\roman*)]
    \item[(a)] \label{eq:assumption_a} Balance: for all $\type \in \types$, a valid assignment is characterized by fixed $I_\type$ and $J_{\type}$, with $I/I_{\type}, J/J_{\type} \le C_1$.
    \item[(b)] \label{eq:assumption_b} Boundedness: for all buyers and seller interactions $(i,j)$ and all types $\type$, $|y_{ij}(\type)| \le C_2$.
\end{enumerate}
\end{assumption}
\begin{theorem}\label{thm:var-consistent}
    Let $\hat{\tau}(\coefvec)$ be the linear estimator given in \Cref{eq:linear_estimator}, and let $\widehat{\mmv}^{\rm hi}[ \hat{\tau}(\coefvec)]$ be its conservative variance estimator given in \cref{thm:sample_variance_avg_effs}.
    Then, in any sequence of SMRDs satisfying \cref{ass:regularity} and in which $(I^{-2} + J^{-2})/\mme\{ \widehat{\mmv}^{\rm hi}[ \hat{\tau}(\coefvec)]\} \to 0$, we have
    \[
        \frac{\widehat{\mmv}^{\rm hi}[ \hat{\tau}(\coefvec)]}{\mme\left\{\widehat{\mmv}^{\rm hi}[ \hat{\tau}(\coefvec)]\right\}} 
        = 
        1 + o_p(1).
    \]
\end{theorem}
}



\subsection{Finite Population Central Limit Theorem} \label{sec:CLTs}

We conclude this section by providing a quantitative central limit theorem for the estimators introduced in \cref{sec:estimation}. 
Notably, we do not assume that the observed units are drawn from an underlying ``super-population,'' nor do we consider a sequence of experiments.
Instead, our approach quantifies the distribution of our estimates using only the randomness of the design, in terms of well-defined properties of the finite population. 
This approach allows us to limit assumptions imposed on the potential outcomes. Our contribution can be seen as an extension to the multi-population setting of recent advances in the causal inference literature, and in particular of the works of \citet{li2017general} and \citet{shi2022berry} for single-sided experiments. 
Our setting presents additional technical challenges, as the outcomes exhibit a complex dependence structure. 
\Cref{thm:clt} serves as the basis for statistical inference in the context of multiple randomization designs.
\begin{theorem} \label{thm:clt}
Consider an SMRD where \cref{sutva_local} and \cref{ass:regularity} hold.
Then we have
\begin{equation}
    \sup_{t \in \mathbb{R}} \left| \mathbb{P}\left\{ \frac{\hat{\tau}(\coefvec) - \tau(\coefvec)}{\sqrt{ \mmv\left[\hat{\tau}(\coefvec)\right]}} \le t \right\} - \Phi(t) \right| \le C\Delta^{\frac{1}{3}} \log\left(\frac{C}{\Delta}\right), \label{eq:berry-esseen}
\end{equation}
with $\Delta \coloneqq \frac{ C_1^2C_2(I^{-1} + J^{-1})}{\mmv\{\hattau(\coefvec)\}^{\frac{1}{2}}/\|\coefvec\|}$ and 
where $\Phi$ denotes the standard normal cumulative density function (CDF), and $C > 0$ is a universal constant.\footnote{We are very grateful to \citet{sudijono2025private}, who communicated an important idea that led to the correction of an error in the proof of \Cref{thm:clt}.}
\end{theorem}

\begin{remark}[Boundedness and sparsity]
In addition to ruling out heavy-tailed potential outcome distributions, an important limitation of \cref{thm:clt} is \textit{sparsity}, when a large fraction of unit potential outcomes $y_{ij}(\type)$, or their differences $y_{ij}(\type)-y_{ij}(\type')$, are zero. 
Sparsity can also cause problems in CLTs for conventional randomized experiments such as the ones cited above. 
It may be especially relevant in our setting, however, where units correspond to pairwise interactions between large populations.

Since $\hattau(\coefvec)$ is linear in observed outcomes $Y_{ij}$, the quantity
\(
C_{\coefvec} = {C_2}/ \{\mmv\{\hattau(\coefvec)\}^{\frac{1}{2}}/\|\coefvec\|\}
\)
appearing in our bound \eqref{eq:berry-esseen} is invariant to re-scaling observations (\eg, to ensure non-degeneracy of $\tau(\coefvec)$).  
\Cref{thm:clt} requires that $C_{\coefvec}$ be small in comparison to $\{I^{-1}+J^{-1}\}^{-1}$ for $I$ and $J$ large, allowing a limited degree of sparsity. For example, if potential outcomes are binary, if $I$ and $J$ are of the same order, and if half of the rows and columns have a fraction $2\mu$ of non-zero entries (while the rest are all zero), then our result requires $\mu$ to be much larger than $I^{-1/2}$.
Generalizing \cref{thm:clt} to better accommodate heavy-tailed and sparse potential outcomes is an important direction for future work. 

\end{remark}
We articulate our proof in three main steps, described in detail in \cref{proof:clt}. 
First, we prove that if we fix the assignment of one of the two populations (\eg, sellers), an analogous version of the results proved by \citet{li2017general} and \citet{shi2022berry} holds for the multi-population setting, where the parameters of the CLT are indexed by the  seller assignment (\cref{sec:proof_clt_conditional}). 
Second, we show that with high probability, these fixed parameters are either themselves normally distributed, or else are close to their expected value (\cref{sec:concentration}). 
Last, we combine these results to prove a CLT for simple double randomized experiments (\cref{sec:tying}).

The main challenge in proving our result is that separate randomization of the two populations creates two-way dependence in the realized outcomes, complicating the application of standard techniques. Similar settings have been studied using Stein's method of exchangeable pairs, although the proofs are quite complex \citep{zhao1997error}. 
Interestingly, the proof of \cref{thm:clt} treats the two populations asymmetrically, although the final bound is symmetric in $I$ and $J$.

Finally, we comment on the application of \cref{thm:clt} in practice. It is natural to replace the variance $\mmv[\hat{\tau}(\coefvec)]$ by its estimated upper bound $\widehat{\mmv}^{\rm hi}[\hat{\tau}(\coefvec)]$. Roughly speaking, the Studentized statistic $\hat z_\tau = \{\hat{\tau}(\coefvec)-\tau\}/\widehat{\mmv}^{\rm hi}[\hat{\tau}(\coefvec)]^{1/2}$ will be approximately normally distributed with variance at most $1$ provided the denominator converges, which follows by \cref{thm:var-consistent}.
One can then test one- and two-sided hypotheses on $\tau(\coefvec)$ by comparing $\hat z_\tau$ to standard normal critical values. We empirically verify normality of the Studentized statistic and illustrate the resulting tests with synthetic data in \cref{sec:experiment}.

\section{Simulations} \label{sec:experiment}

We now verify the results of \cref{sec:estimation} for SMRDs under local interference. 
Our simulations follow the model of strategic agents in a two-sided marketplace introduced in \cref{example:strategic-local-interference}, which naturally produces local interference. 
Additional experiments from a simple additive Gaussian model satisfying local interference are provided in \cref{sec:additional_experiments}. 
\texttt{Python} code to replicate all our simulations is available at \url{https://github.com/lorenzomasoero/MultipleRandomizationDesigns}.

The simulations following \cref{example:strategic-local-interference} also illustrate the practical value of SMRDs. 
In the underlying model, higher quality ads increase the incentive to produce high quality content, and vice-versa. 
This is an example of \emph{strategic complementarity}, a prominent and well-studied feature of many real-world marketplaces \citep{milgrom1990rationalizability}.
In this model, it leads to significant positive spillovers for both advertisers and creators.
These spillovers are neatly captured by the MRD, but cause conventional, single randomized experiments to underestimate the treatment effect, possibly leading to sub-optimal policies.

To empirically validate the results presented in \cref{sec:estimation}, we first instantiate the model from \cref{example:strategic-local-interference} by fixing the incentive level $\eta = 5\%$ and drawing independent and identically distributed parameters $m_{ij} \sim \mathrm{Exp}(1)$, $r^c_i, r^a_j \sim \mathrm{Unif}([0,1/5])$ across advertisers $i \in [I]$ and creators $j \in [J]$ (notice: here creators and advertisers have roles analogous to that of buyers and sellers in the discussion of \cref{sec:interference,sec:mrds,sec:estimation}). 
In our simulation, we let $I = 200$ and $J = 150$. 
Taking these parameters\textemdash which determine the fixed population\textemdash as given, we fix the treatment group size $I_T = 100$ and $J_T = 80$. We then sample  treatment assignment matrices $\bw$ at random from the SMRD $\mmw$, which determine realized equilibrium outcomes $Y_{ij}(\bw)$ to be the platform's profit following \cref{eq:potential-outcomes-local-interference}. 

Since local interference (\cref{sutva_local}) is satisfied, $Y_{ij}(\bw)$ depends only upon the type $\Type_{ij} \in \types$ of unit $(i,j)$, conditional upon the parameters $I_T$,  $J_T$, and the fixed population.
Each assignment $\bw$ then corresponds to an observed matrix of $I \times J$ realized potential outcomes.
We use the collection of outcomes from 10{,}000 independent re-randomizations to empirically verify the properties of the proposed estimators. 
\begin{figure}
    \centering   \includegraphics[width=.75\textwidth]{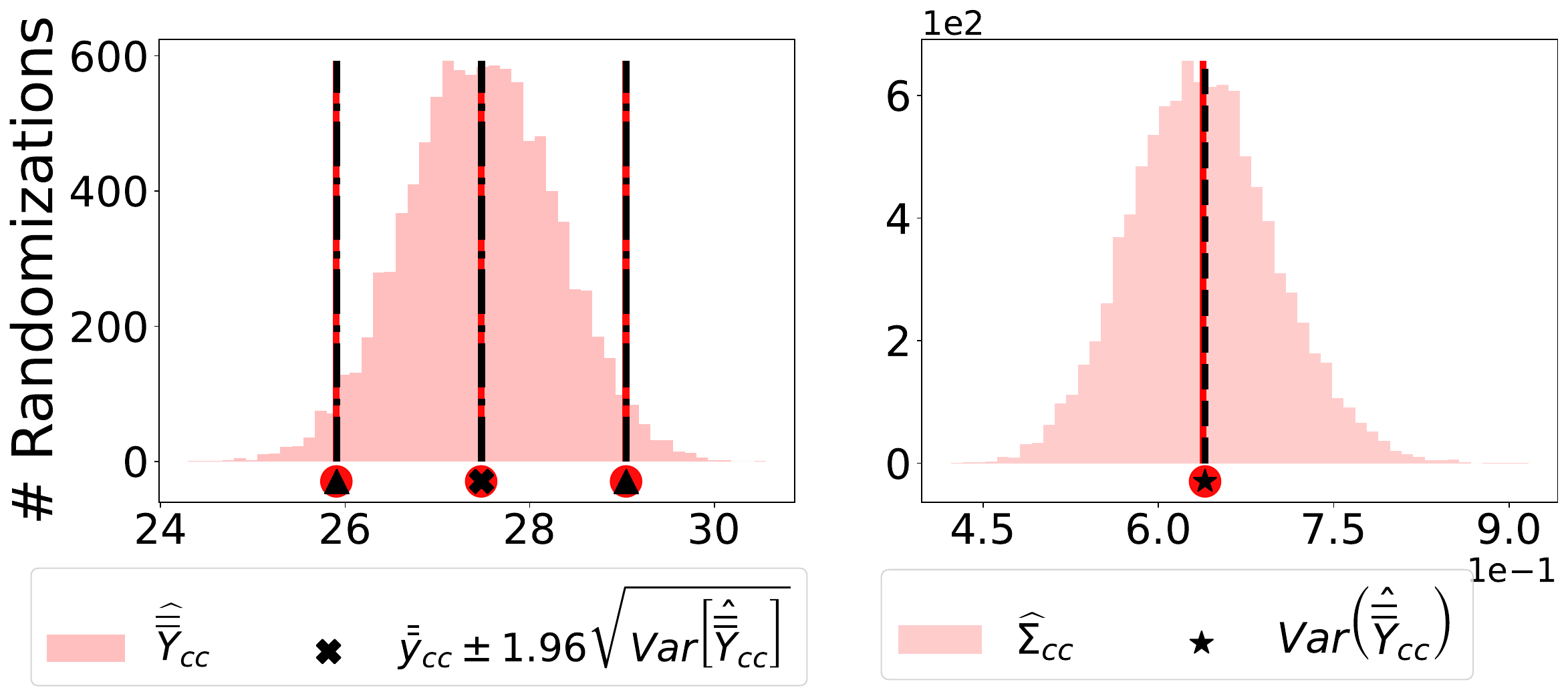}
    \caption{
   Distribution of $\meanestimate{\ccc}$ (left) and of the variance estimator $\widehat{\Sigma}_{\ccc}$ (right). Black lines correspond to the population quantities $\meanpopulation{\ccc}$, $\mmv\left(\meanestimate{\ccc}\right)$.
    }
    \label{fig:type_effs}
\end{figure}
\Cref{fig:type_effs} reports the histogram of the values attained by $\meanestimate{\ccc}$ (left) and $\widehat{\Sigma}_{\ccc}$ (right) across the 10{,}000 Monte Carlo replicates. 
As follows from \cref{lemma:unbiasedness}, $\meanestimate{\ccc}$ is centered at the true population value $\meanpopulation{\type}$, and, using \cref{thm:clt}, under mild conditions $\meanestimate{\type}$ is approximately normally distributed. 
Moreover, the distance between the 2.5\% and 97.5\% quantiles of the distribution of the type estimator (red vertical lines) is close to the length of the 95\% confidence interval around the population value $\meanpopulation{\type}$, formed by using the true variance  of $\meanestimate{\type}$.
In the right panel of \cref{fig:type_effs}, we show that $\widehat{\Sigma}_{\ccc}$ is an unbiased estimator for the variance of the type estimator, as proved in \cref{thm:sample_variance_avg_effs}. Analogous results hold for $\icb, \ics, \ctt$.

\begin{figure}
    \centering   \includegraphics[width=.75\textwidth]{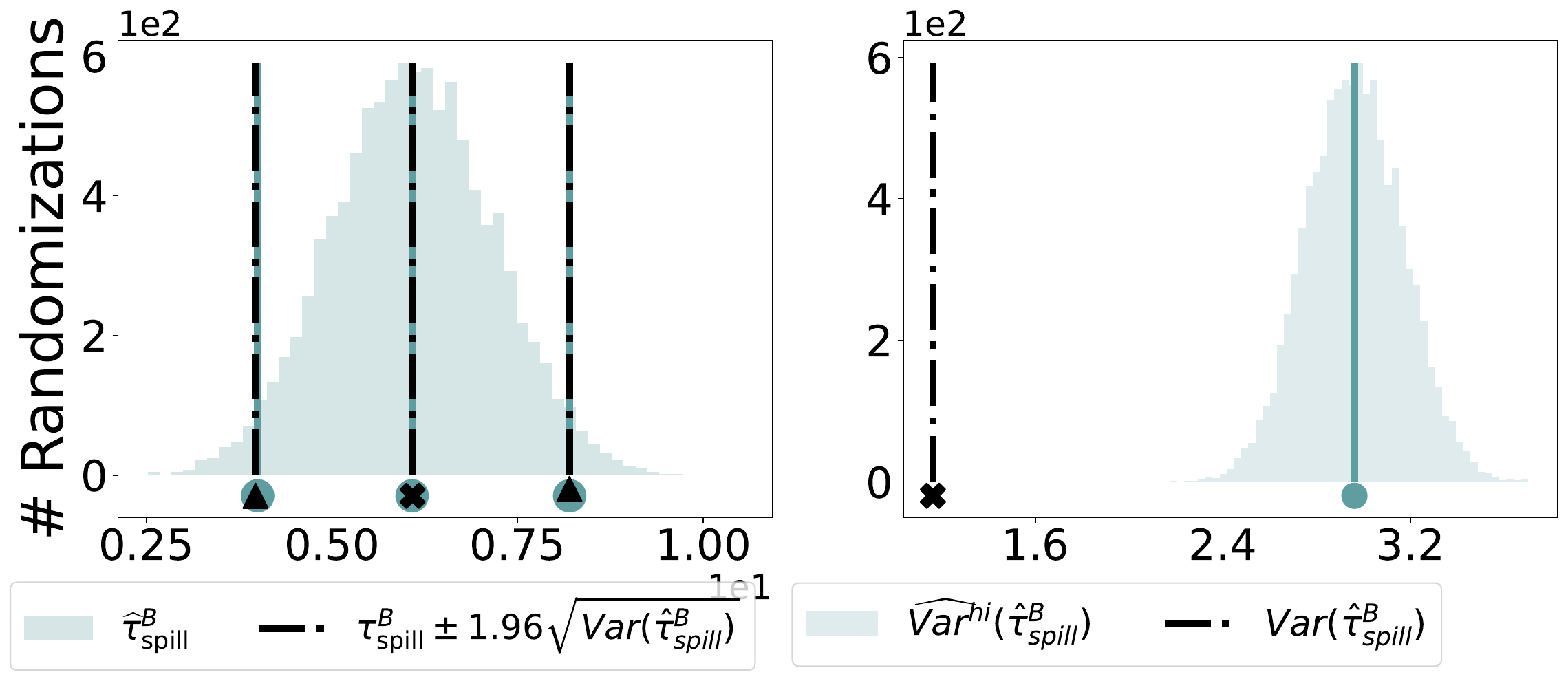}
    \caption{
   Distribution of the estimator for the spillover effect $\hat{\tau}_{\rm{spill}}^B$ (left) and corresponding variance estimator $\widehat{\mmv}^{\rm{hi}}(\hat{\tau}_{\rm{spill}}^B)$ (right). Black lines correspond to the population quantities.
    }
    \label{fig:spill}
\end{figure}
We focus on the spillover effect $\tauspillb$ in \cref{fig:spill}: the left panel shows the distribution of the unbiased estimator $\hattauspillb$ (\cref{thm:spillover_unbiasedness}). 
$\hattauspillb$ is Gaussian (as shown in \cref{thm:clt}), and conservative confidence intervals can be derived. 
The right panel contains the distribution of the upper bound $\widehat{\mmv}^{\rm hi}(\hattauspillb)$ for the variance ${\mmv}(\hattauspillb)$ (\cref{thm:sample_variance_spillover}). 
Additional plots and implementation details are provided in \cref{sec:additional_experiments}. 

\begin{figure}[t]
    \centering   
    \includegraphics[width=.75\textwidth]{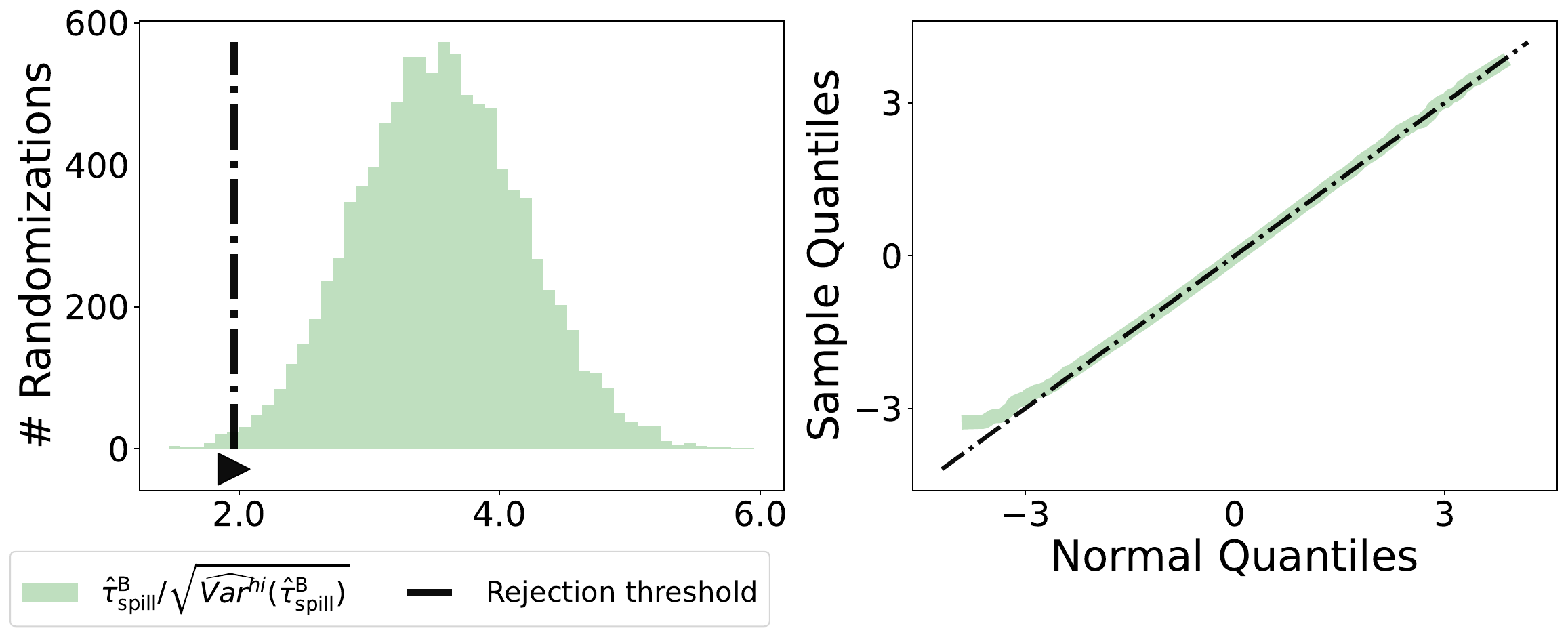}
    \caption{
    Left: distribution of the Studentized statistic $\hattauspillb/\{\widehat{\mmv}^{\rm hi}(\hattauspillb)\}^{1/2}$ and resulting conservative two-sided tests of the null hypothesis of no effect $H_0=\{\tauspillb = 0\}$ (statistics to the right of the black line reject $H_0$). Right: QQ plot comparing the Studentized statistic to a Gaussian law with the same mean and variance. 
    }
    \label{fig:t-test}
\end{figure}

Under mild conditions laid out in \cref{thm:clt} and the following discussion, one can practically test for the presence of positive spillover effects by constructing the Studentized statistic $\hat{z}_0:=\hattauspillb/\{\widehat{\mmv}^{\rm hi}(\hattauspillb)\}^{1/2}$ and comparing it to standard normal critical values. 
For our model, the conservative test $\hat{z}_0$ rejects the null hypothesis of no effect $99.5\%$ of the time (Type-II error is $0.5\%$), showing substantial power to detect positive spillovers.
\begin{figure}[ht]
    \centering   \includegraphics[width=.75\textwidth]{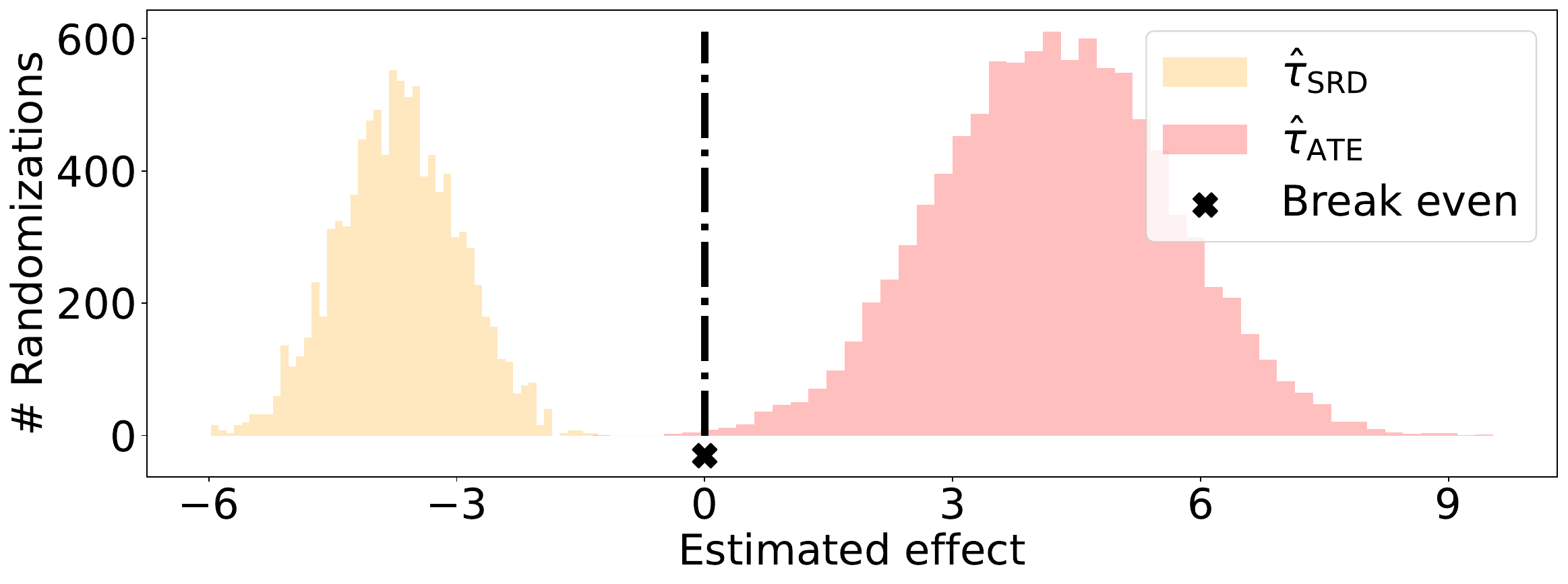}
    \caption{
   Distribution of the standard difference-in-means estimator $\hat\tau_{\rm SRD}$ across 10,000 single randomized experiments (yellow), compared to the distribution of $\hat\tau_{\rm ATE}$ in as many SDRDs (red). The estimators are produced by re-drawing different randomization designs for the same underlying finite population, with potential outcomes given by \Cref{example:strategic-local-interference}.
    }
    \label{fig:comparison}
\end{figure}

Finally, we compare MRDs to the standard practice of single randomization\textemdash randomizing exactly $50\%$ of creators $i \in [I]$ into treatment, and treating all of their interactions as in \cref{eq:srd}, and then using the standard difference-in-means estimator $\hat\tau_{\rm SRD}$. 
In our model, such an estimator neglects positive spillovers mediated by advertisers' strategic responses. 
By comparing the distribution of the difference-in-means estimator $\hat\tau_{\rm SRD}$ under the standard creator-randomized design given by \cref{eq:srd} to the distribution of  $\hat{\tau}_{\ATE}$ under the SMRD, we illustrate in \cref{fig:comparison} that in our model the standard design usually produces the incorrect sign of the platform's profit relative to that which would be obtained by treating the whole population, while the SMRD usually produces the correct sign.


\section{Extensions and future work} \label{sec:extensions}

The designs discussed in \cref{sec:mrds} are a few of many possible designs that fit into the MRD framework. 
While we have focused in detail on the ``Simple'' MRD case, many other designs fit the MRD paradigm\textemdash including clustered experiments, experiments involving three or more populations, etc. These generalizations also include time-randomized experiments: e.g., recently \citet{masoero2023efficient} used the MRD framework to show that under certain assumptions on the potential outcomes, switchback designs based upon the MRD framework can lead to more efficient estimates of causal effects.
MRDs have also been used in practice in the context of  online marketplaces, to quantify the direct and indirect effects of certain interventions; see, e.g., \citet{masoero2024measuring,zhu2024seller,bright2024reducing}.

Additionally, as highlighted in the discussion following \cref{sutva_local}, we emphasize that the local interference assumption is only a starting point from which to rigorously study causal inference with MRDs. 
We envision that future work will study how MRDs can be used in conjunction with more complicated interference structures. 
Characterizing minimal restrictions on interference under which similar, design-based inference results  can be derived is an open question beyond the scope of this paper.

To illustrate the richness of our framework, we conclude by describing four additional designs which fit within the MRD setting. First, instead of partitioning {\buyers{}} and {\sellers{}} into two groups each, we can assign them to a finite number of groups, with the assignment a function of this finer partition. This allows to generate more variation in $ \owwi$ and $ \owwj$ and in turn to build models for the dependence of the potential outcomes on the share of treated {\buyer}s and {\seller}s that will permit more credible extrapolation to full exposure to treatment or control. 
As a simple example, we could endow each \buyer{} $i$ and \seller{} $j$ with scalar scores $w_i^\rb$ and $w_j^\rs$ (as opposed to binary values), and let the treatment assignment be defined by a modified version of  \cref{eq:conjunctive}, \eg, $f(w_i^\rb, w_j^\rs) = \been(w_i^\rb + w_j^\rs) > \kappa$ for a given threshold $\kappa$ (\eg, $\kappa = 0.5$ in \ref{eq:extension_score}).

\begin{equation} \label{eq:extension_score}
\bww = 
    \begin{NiceArray}{c@{\hspace{0.5em}}cccccc@{\hspace{0.5em}}cccc@{\hspace{0.15em}}>{\color{gray}}cc}[margin,baseline=5]
    \RowStyle[color=gray]{}
        & 1 & 2 & 3 & 4 & 5 && 6 & 7 & 8 &&& \\
        \RowStyle[color=black]{}
        ~& 0 & 0 & 0.2 & 0.2 & 0.4 && 0.4 & 0.6 & 0.6 &&& \\
        & \redc & \redc & \bluec & \bluec  & \magc && \magc & \ct &\ct && 1 & 0\\
        & \redc & \redc & \bluec & \bluec  & \magc && \magc & \ct &\ct && 2 & 0\\
        & \orangec  &\orangec  & \brownc & \brownc   &\ct && \ct & \ct &\ct && 3 & 0.2\\
        & \orangec & \orangec & \brownc & \brownc  & \ct&&  \ct & \ct &\ct && 4 & 0.2 \\
        & \greenc  &\greenc  & \ct & \ct   & \ct && \ct & \ct &\ct && 5& 0.4\\
        & \greenc  &\greenc  & \ct & \ct   & \ct && \ct & \ct &\ct && 6 & 0.4 \\
        \CodeAfter
        \SubMatrix({3-2}{8-10})
        \begin{tikzpicture}
            \draw[decorate,thick,color=gray] (1-2) node[left,xshift=-1.5em] (S) {\text{{\Seller}}};
            \draw[-latex,thick,color=gray] (S.east) |- (1-2.west);
            \draw[decorate,thick,color=black] (2-2) node[left,xshift=-1.5em] (E) {Score};
            \draw[decorate,thick,color=black] (3-12) node[above,xshift=2em,yshift=1em] (B) {\text{{Score}}};
            \draw[decorate,thick,color=gray] (3-12) node[above,xshift=2em,yshift=2em] (B) {\text{{\Buyer}}};
            \draw[-latex,thick,color=gray] (B.west) -| (3-12.north);
        \end{tikzpicture}
        \end{NiceArray}
\end{equation}

Second, one can first partition one of the groups ({\it e.g.}, {\sellers{}}) into two random groups (A, B), and run a {\buyer{}} experiment for one group and a {\seller{}} experiment for the other.

\begin{equation*}
\bww = 
    \begin{NiceArray}{c@{\hspace{0.5em}}cccccc@{\hspace{0.5em}}cccc@{\hspace{0.15em}}>{\color{gray}}c}[margin,baseline=5]
\RowStyle[color=gray]{}
& 1 & 2 & 3 & 4 & 5 && 6 & 7 & 8 && \\
\RowStyle[color=gray]{}
~& A & A & A & A & A && B & B & B && \\
&\greenc & \greenc  & \greenc & \greenc  & \greenc&& \ct& \redc &\ct && 1 \\
& \greenc &\greenc & \greenc & \greenc  & \greenc&& \ct& \redc &\ct && 2 \\
& \ct   &\ct  & \ct & \ct   &\ct && \ct & \textcolor{blue}{\cc} &\ct && 3 \\
& \greenc&\greenc & \greenc & \greenc  & \greenc&&  \ct&\redc &\ct && 4 \\
& \ct  &\ct  & \ct & \ct   & \ct && \ct & \textcolor{blue}{\cc} &\ct && 5 \\
    \CodeAfter
    \SubMatrix({3-2}{7-10})
    \begin{tikzpicture}
        \draw[decorate,thick,color=gray] (1-2) node[left,xshift=-1.5em] (S) {\text{{\Seller}}};
        \draw[-latex,thick,color=gray] (S.east) |- (1-2.west);
        \draw[decorate,thick,color=gray] (2-2) node[left,xshift=-1.5em] (E) {\text{{Exp}}};
        \draw[-latex,thick,color=gray] (E.east) |- (2-2.west);
        \draw[decorate,thick,color=gray] (3-12) node[above,xshift=2em,yshift=1em] (B) {\text{{\Buyer}}};
        \draw[-latex,thick,color=gray] (B.west) -| (3-12.north);
    \end{tikzpicture}
    \end{NiceArray}
\end{equation*}

Third, when one wants to do a {\seller}-clustered experiment, one may partition the {\buyer} population into two groups, $A$ and $B$, and then run a {\seller} clustered experiment in one group and a regular seller experiment in the second group. This would allow the researchers to infer within the context of a single experiment the within-cluster spillovers, as well as get estimates of the overall average effect.

\begin{equation*} \arraycolsep=14pt\def\arraystretch{0.65}
\bww = 
    \begin{NiceArray}{c@{\hspace{0.5em}}ccccccc@{\hspace{0.5em}}>{\color{gray}}c>{\color{gray}}c}[margin,baseline=7]
\RowStyle[color=gray]{}
    & 1                           & 2                           & 3                           & 4                           & 5                           & 6                           && \text{~}      & \text{~}        \\
    \RowStyle[color=gray]{}
    ~& \text{\RomanNumeralCaps{1}} & \text{\RomanNumeralCaps{1}} & \text{\RomanNumeralCaps{2}} & \text{\RomanNumeralCaps{2}} & \text{\RomanNumeralCaps{3}} & \text{\RomanNumeralCaps{3}} && ~      & ~        \\
    & \redc                       & \bluec                      & \blackt                     & \blackt                     &  \redc                      & \bluec                      &&      1 & \text{A} \\
    & \redc                       & \bluec                      & \blackt                     & \blackt                     &  \redc                      & \bluec                      &&      2 & \text{A} \\
    & \redc                       & \bluec                      & \blackt                     & \blackt                     & \redc                       & \bluec                      &&      3 & \text{A} \\
    & \redc                       & \bluec                      & \blackt                     & \blackt                     & \redc                       & \bluec                      &&      4 & \text{A} \\
    \\[-0.75em]
    & \blackt                     & \brownc                     & \blackt                     & \blackt                     & \blackt                     & \brownc                     &&      5 & \text{B} \\
    & \blackt                     & \brownc                     &  \blackt                    & \blackt                     & \blackt                     & \brownc                     &&      6 & \text{B} \\
    & \blackt                     & \brownc                     & \blackt                     & \blackt                     & \blackt                     & \brownc                     &&      7 & \text{B} \\
    & \blackt                     & \brownc                     & \blackt                     & \blackt                     & \blackt                     & \brownc                     &&      8 & \text{B} \\
    \CodeAfter
    \SubMatrix({3-2}{11-7})
    \begin{tikzpicture}
        \draw[decorate,thick,color=gray] (1-2) node[left,xshift=-1.5em] (S) {\text{{\Seller}}};
        \draw[-latex,thick,color=gray] (S.east) |- (1-2.west);
        \draw[decorate,thick,color=gray] (2-2) node[left,xshift=-1.5em] (C) {\text{Cluster}};
        \draw[-latex,thick,color=gray] (C.east) |- (2-2.west);
        \draw[decorate,thick,color=gray] (3-9) node[above,xshift=2.2em,yshift=2em] (B) {\text{{\Buyer}}};
        \draw[-latex,thick,color=gray] (B.west) -| (3-9.north);
        \draw[decorate,thick,color=gray] (3-10) node[above,xshift=2em,yshift=1em] (E) {\text{Exp.}};
        \draw[-latex,thick,color=gray] (E.west) -| (3-10.north);
    \end{tikzpicture}
    \end{NiceArray}
\end{equation*}

Fourth, we can consider designs where the local interference assumption is testable.
\begin{equation*}
\bww = 
    \begin{NiceArray}{c@{\hspace{0.5em}}cccccc@{\hspace{0.5em}}cccc@{\hspace{0.15em}}>{\color{gray}}c}[margin,baseline=5]
\RowStyle[color=gray]{}
    & 1 & 2 & 3 & 4 & 5 && 6 & 7 & 8 && \\
    & \redc & \redc & \bluec & \bluec  & \bluec && \bluec & \bluec &\bluec && 1 \\
    & \redc & \redc & \bluec & \bluec  & \bluec && \bluec & \bluec &\bluec && 2 \\
    & \ct   &\ct  & \brownc & \brownc   &\brownc && \brownc & \brownc &\brownc && 3 \\
    & \ct & \ct & \brownc & \brownc  & \brownc&&  \brownc &\brownc &\brownc && 4 \\
    & \greenc  & \greenc & \ct & \ct & \ct && \ct & \ct &\ct && 5 \\
    & \greenc  & \greenc & \ct & \ct & \ct && \ct & \ct &\ct && 6 \\
    \CodeAfter
    \SubMatrix({2-2}{7-10})
    \begin{tikzpicture}
        \draw[decorate,thick,color=gray] (1-2) node[left,xshift=-1.5em] (S) {\text{{\Seller}}};
        \draw[-latex,thick,color=gray] (S.east) |- (1-2.west);
        \draw[decorate,thick,color=gray] (3-12) node[above,xshift=2em,yshift=2em] (B) {\text{{\Buyer}}};
        \draw[-latex,thick,color=gray] (B.west) -| (2-12.north);
    \end{tikzpicture}
    \end{NiceArray}
\end{equation*}

Consider the red $\redc $ and the blue $\bluec$. 
In both cases they correspond to buyers who are in the control group for all sellers, and in both cases they correspond to sellers who are in  the treatment group for 1/3 of the buyers. 
However, sellers in the red $\redc $  pairs are in the treatment group for buyers who are very rarely in the treatment group, whereas the sellers in the blue $\bluec$ pairs are in the treatment group for buyers who are often in the treatment group. 
When local interference holds, that should not matter, but if local interference is violated, it may matter.
 
\newpage
\appendix
\centerline{\sc \large{Appendix}}

\section{Proofs for Multiple Randomization Designs}

We here prove the results presented in \Cref{sec:estimation}. 
\ifarxiv
\else 
Note: a longer and more detailed appendix containing the same results can be found in the extended online supplementary material \citep{masoero2024multiple}.
\fi
We consider conjunctive SMRDs (as per \Cref{smrd}) where local interference holds (\cref{sutva_local}), with a total of $I$ {\buyers}, $J$ {\sellers}, and $I\times J$ units. 
All \buyers{} and \sellers{} are endowed with random variables $W_i^\rb, W_j^\rs \in \{0,1\}$, so that $I>I_{\rt}>1$ and $J>J_{\rt}>1$, where $I_\rt := \sum_i W_i^\rb$, $J_\rt:=\sum_j W_j^\rs$.

\begin{lemma}[\Cref{lemma:local_interference}]  
    Under local interference (\Cref{sutva_local}), potential outcomes can be written as a function of the assignment types only: for $\bww, \bww' \in \{0,1\}^{I \times J}$ it holds that
    \[ 
        \type_{ij}(\bww)=\type_{ij}(\bww')\Rightarrow y_{ij}(\bww)=y_{ij}(\bww').
    \]
\end{lemma}

\begin{proof}
    Under \Cref{sutva_local}, for any $(i,j)$ and any pair of assignment matrices $\bm{w}, \bm{w}' \in \{0,1\}^{I\times J}$
    $y_{ij}(\bm{w}) = y_{ij}(\bm{w}')$    
    whenever (a) $w_{ij}=w_{ij}'$, (b) the fraction of treated sellers for buyer $i$ coincides in $\bm{w}, \bm{w}'$ and (c) the fraction of treated buyers for {\seller} $j$ coincides in $\bm{w}, \bm{w}'$. If (a), (b) and (c) hold, it must be the case that $\type_{ij}(\bm{w}) = \type_{ij}(\bm{w}')$, yielding the thesis.
\end{proof}

\subsection{Useful definitions} \label{app-sec-def}
Recall the definitions of the average outcomes for each {\buyer} and each {\seller}:
\begin{equation*} 
    \oybi(\type) := \frac{1}{J}\sum_{j=1}^J y_{ij}(\type), \quad
    \oysj(\type) := \frac{1}{I}\sum_{i=1}^I y_{ij}(\type) \quad\text{and}\quad
    \meanpopulation{\type} := \frac{1}{IJ}\sum_{i=1}^I \sum_{j=1}^J y_{ij}(\type).
\end{equation*}
For each type $\type\in \types$, {\buyer} $i$ and {\seller} $j$, define the following deviations:
\begin{equation*}
    \dyi(\type)  := \oybi(\type)-\meanpopulation{\type},
    \quad
    \dyj(\type)  := \oysj(\type)-\meanpopulation{\type},
    \quad 
    \dyij(\type) := y_{ij}(\type)-\oybi(\type)-\oysj(\type)+\meanpopulation{\type}.
\end{equation*}
By definition, the sum of these deviations is equal to zero:
\ifarxiv
\begin{equation*}
    \sum_{i=1}^I \dyi(\type)    = 0, \quad
    \sum_{i=1}^I \dyij(\type)   = 0, \quad
    \sum_{j=1}^J \dyj(\type)    = 0, \quad
    \sum_{j=1}^J \dyij(\type)   = 0.
\end{equation*}
We decompose $y_{ij}(\type)$  as
\begin{equation*}                 
    y_{ij}(\type)=\meanpopulation{\type}+\dyi(\type)+\dyj(\type)+\dyij(\type).
\end{equation*}
\else 
\(
    \sum_{i=1}^I \dyi(\type)    = 0, 
    \sum_{i=1}^I \dyij(\type)   = 0, 
    \sum_{j=1}^J \dyj(\type)    = 0, 
    \sum_{j=1}^J \dyij(\type)   = 0.
\)
We decompose $y_{ij}(\type)$  as
\(                
    y_{ij}(\type)=\meanpopulation{\type}+\dyi(\type)+\dyj(\type)+\dyij(\type).
\)
\fi 
Last, for $\type \in \types$ we let $I_\type$ be the number of {\buyers} eligible for type $\type$ and $J_\type$ be the number of {\sellers} eligible for type $\type$. Define $I_C := I - I_{\rt}$ and $J_C := J - J_{\rt}$, then
    $I_{\ccc} = I_{\rc}$, $J_\ccc = J_{\rc}$,
    $I_{\icb} = I_{\rt}$, $J_\icb = J_{\rc}$,
    $I_{\ics} = I_{\rc}$, $J_\ics = J_{\rt}$,
    $I_{\ctt} = I_{\rt}$, $J_\ctt = J_{\rt}$.

\subsection{Linear representation of the type estimators}
\ifarxiv
Recall from \Cref{smrd} that $W^\rb_i$ and $W^\rs_j$ are random variables which determine whether {\buyer} $i$ and {\seller} $j$ are eligible to be exposed to the treatment. 
\else
\fi
\begin{lemma}\label{lemma_appendix0} 
The (doubly averaged) sample mean estimator 
\ifarxiv
$\meanestimate{\type}$ can be decomposed as
\begin{align}
    \meanestimate{\ctt} &= \meanpopulation{\ctt} +\frac{1}{I_\rt}\sum_{i=1}^IW^\rb_i\dyi(\ctt) +\frac{1}{J_\rt} \sum_{j=1}^J W^\rs_j\dyj(\ctt) 
    + \frac{1}{I_\rt J_\rt}\sum_{i=1}^I\sum_{j=1}^J W^\rb_i W^\rs_j\dyij(\ctt), 
    \nonumber 
    \\
    \meanestimate{\icb}   &= \meanpopulation{\icb} + \frac{1}{I_\rt}\sum_{i=1}^I W^\rb_i\dyi(\icb) + \frac{1}{J_\rc} \sum_{j=1}^J(1-W^\rs_j)\dyj(\icb)
    + \frac{1}{I_\rt J_\rc}\sum_{i=1}^I\sum_{j=1}^JW^\rb_i(1-W^\rs_j)\dyij(\icb), \nonumber \\
    \meanestimate{\ics} &= \meanpopulation{\ics}+\frac{1}{I_\rc}\sum_{i=1}^I (1-W^\rb_i)\dyi(\ics)+\frac{1}{J_\rt} \sum_{j=1}^J W^\rs_j\dyj(\ics)
    + \frac{1}{I_\rc J_\rt}\sum_{i=1}^I\sum_{j=1}^J (1-W^\rb_i) W^\rs_j \dyij(\ics), \nonumber \\
    \begin{split}
        \meanestimate{\ccc}   &= \meanpopulation{\ccc} + \frac{1}{I_\rc} \sum_{i=1}^I(1-W^\rb_i) \dyi(\ccc) 
    + \frac{1}{J_\rc} \sum_{j=1}^J(1-W^\rs_j)\dyj(\ccc)  
    \\
    & 
    \qquad \quad 
    + \frac{1}{I_\rc J_\rc}\sum_{i=1}^I\sum_{j=1}^J (1-W^\rb_i)(1-W^\rs_j)\dyij(\ccc)
    .
    \label{eq:ycc}
    \end{split}
\end{align}
\else
$\meanestimate{\ctt}$ can be decomposed as
\begin{align}
    \meanestimate{\ctt} &= \meanpopulation{\ctt} +\frac{1}{I_\rt}\sum_{i=1}^IW^\rb_i\dyi(\ctt) +\frac{1}{J_\rt} \sum_{j=1}^J W^\rs_j\dyj(\ctt) 
    + \frac{1}{I_\rt J_\rt}\sum_{i=1}^I\sum_{j=1}^J W^\rb_i W^\rs_j\dyij(\ctt).
    \label{eq:ycc}
\end{align}
\fi
\end{lemma}
\begin{proof}[Proof of Lemma \ref{lemma_appendix0}]
Consider the case of $\meanestimate{\ctt}$: leveraging the decomposition of $y_{ij}(\ctt)$,
\begin{align*}
    \meanestimate{\ctt}   &= 
    \frac{1}{I_\rt J_\rt}\sum_{i=1}^I\sum_{j=1}^J W^\rb_i W^\rs_j y_{ij}(\ctt) 
    = \frac{1}{I_\rt J_\rt}\sum_{i=1}^I\sum_{j=1}^J W^\rb_i W^\rs_j 
    \left( \meanpopulation{\ctt}
    + 
    \dyi(\ctt) 
    +
    \dyj(\ctt) 
    + 
    \dyij(\ctt)
    \right)\\
    &=
    \meanpopulation{\ctt} +\frac{1}{I_\rt}\sum_{i=1}^IW^\rb_i\dyi(\ctt) +\frac{1}{J_\rt} \sum_{j=1}^J W^\rs_j\dyj(\ctt) 
    + \frac{1}{I_\rt J_\rt}\sum_{i=1}^I\sum_{j=1}^J W^\rb_i W^\rs_j\dyij(\ctt).
\end{align*}
\ifarxiv
Results for $\type \neq \ctt$ are similar and are omitted.
\fi
\end{proof}
\ifarxiv
\else
Analogous characterizations hold for $\type \neq \ctt$. They are similar and omitted.
\fi
\subsection{Moment characterization}

We use \cref{lemma_appendix0} to re-write the estimator $\meanestimate{\type}$ of $\meanpopulation{\type}$ as a linear combination of the random labels $W^\rb_i$, $W^\rs_j$ with non-stochastic coefficients.
We use this to derive the first two moments of $(\meanestimate{\type}, \meanestimate{\type'})$ under the SMRD design. 
To do so, we define the demeaned treatment 
\(
    D_i^{\rb}=W^{\rb}_i- {I_\rt}/{I}, \quad \text{and} \quad D_j^{\rs}=W^{\rs}_j- {J_\rt}/{J}. 
\)
\begin{lemma} \label{lemma:demeaned}
    For $i\neq i' \in\mmi $,
    \(
        \mme\left[\db_i\right]=0, \quad 
        \mmv\left(\db_i\right)= \frac{\irn\irt}{I^2}, \quad
        \mmc(\db_i,\db_{i'})= - \frac{\irn\irt}{I^2(I-1)}.
    \)
    For $j,j'\in\mmj$, $j\neq j'$,
    \(
        \mme[\ds_j]=0, \quad
        \mmv(\ds_j)=\frac{\jrn\jrt}{J^2}, \quad
        \mmc(\ds_j,\ds_{j'}) = - \frac{\jrn\jrt}{J^2(J-1)}.
    \)
    Finally, because $\db_i$ and $\ds_j$ are independent, we have
    \(
        \mmc(\db_i,\ds_j)=0, \quad \forall\; i,j.
    \)
\end{lemma}

\begin{proof}[Proof of \Cref{lemma:demeaned}]
    $W_i^{\rb}$ is a Bernoulli random variable with bias given by $p^{\rb}=I_{\rt}/I$, hence $\mme\left[\db_i\right]=0$. Moreover,
    \(
        \mmv\left(\db_i\right)=\mmv\left(W^{\rb}_i\right) = \frac{I_{\rt}}{I}\left(1-\frac{I_{\rt}}{I}\right) = \frac{\irn\irt}{I^2}.
    \)
    Last,
    \begin{align*}
        \mmc(\db_i,\db_{i'}) &=  
                            \mme[W_i^{\rb}W_{i'}^{\rb}] - \mme[W_i^{\rb}]\mme[W_{i'}^{\rb}] 
                            = \frac{I_{\rt}}{I}\frac{I_{\rt}-1}{I-1} - \frac{I_{\rt}^2}{I^2} 
                            = - \frac{\irn\irt}{I^2(I-1)}.
    \end{align*}
    Corresponding proofs for $\ds_j$ are analogous and omitted. 
\end{proof}
\ifarxiv
Note that the covariance between $\db_i$ and $\db_{i'}$ for $i\neq i'$ differs from zero because we fix the number of selected {\buyer}s at $I_\rt$, rather than tossing a coin for each {\buyer}.
 Fixing the number of selected {\buyer}s is important for getting exact finite sample results for the variances.
\fi
Define the average residuals by assignment type, for $\type\in \types$:
\[ 
    \oee_{\type}^\rb=\frac{1}{I_\rt}\sum_{i=1}^I D^\rb_i\dyi(\type),
    \quad
    \oee_{\type}^\rs=\frac{1}{J_\rt} \sum_{j=1}^JD^\rs_j\dyj(\type),
    \quad 
    \ooe_{\type}^{\rb\rs}=\frac{1}{I_\rt J_\rt}\sum_{i=1}^I\sum_{j=1}^JD^\rb_i D^\rs_j\dyij(\type).
\]
\ifarxiv
These representations allow us to split the averages of observed values $\meanestimate{\type}$ into deterministic and stochastic components.
\fi
\begin{lemma}\label{lemma_appendix1}
\begin{itemize}
    \item[(a)] The sample estimates $\meanestimate{\type}$, $\type \in \types$ can be written as the sums of four terms:
        \[
            \meanestimate{\type}=\meanpopulation{\type}+ \oee_{\type}^{\rb}+\oee_{\type}^{\rs}+\ooe_{\type}^{\rb\rs},
        \]
    \item[(b)] $\forall \; \type\in\{\ccc,\icb,\ics,\ctt\}$, the $\epsilon$ in the decomposition above are mean-zero error terms:
    \[ 
        \mme\left[\oee_{\type}^{\rb}\right]=\mme\left[\oee_{\type}^{\rs}\right]=\mme\left[\ooe_{\type}^{\rb\rs}\right]=0.
    \]
    \item[(c)] For all $\type \neq \type'\in\{\ccc,\icb,\ics,\ctt\}$, the error terms above are uncorrelated:
    \[ 
        \mmc\left(\oee_{\type}^{\rb},\oee_{\type'}^{\rs}\right)=\mmc\left(\oee_{\type}^{\rb},\ooe_{\type'}^{\rb\rs}\right)
        =\mmc\left(\oee_{\type}^{\rs},\ooe_{\type'}^{\rb\rs}\right)=0.
    \]
\end{itemize}
\end{lemma}
\ifarxiv
Before proving this lemma, let us just provide an intuition about the decomposition of the four averages $\meanestimate{\ccc}$, $\meanestimate{\icb}$, $\meanestimate{\ics}$, and $\meanestimate{\ctt}$ described above, as this is a key step to obtaining the variance of the estimator for the average treatment effect. In particular, looking at $(i)$, the first term $\meanpopulation{\type}$ is deterministic (the unweighted average of potential outcomes over all pairs $(i,j)$, not depending on the assignment). The other three terms, $\oee_{\type}^\rb$, $\oee_{\type}^\rs$, and $\ooe_{\type}^{\rb\rs}$, are mutually  uncorrelated stochastic terms with expectation equal to zero. 
The variances of the four averages will depend on the variances of the three stochastic terms, and the covariances will depend on the covariances of the corresponding stochastic terms, {\it e.g.,} the covariance of $\oee_{\ctt}^\rb$ and $\oee_{\icb}^\rb$, or the covariance of $\ooe_{\ccc}^{\rb\rs}$ and $\ooe_{\ics}^{\rb\rs}$.
\fi
\begin{proof}[Proof of Lemma \ref{lemma_appendix1}]
For part $(a)$ consider $\meanestimate{\ctt}$.
Now consider for the treated type the average of the observed outcomes, decomposed as in Lemma \ref{lemma_appendix0}:
\begin{align*}
    \meanestimate{\ctt}&=\meanpopulation{\ctt}+\frac{1}{I_\rt}\sum_{i=1}^I W^\rb_i\dyi(\ctt)+\frac{1}{\jrt}
    \sum_{j=1}^JW^\rs_j\dyj(\ctt) +\frac{1}{I_\rt J_\rt}\sum_{i=1}^I\sum_{j=1}^JW^\rb_i W^\rs_j\dyij(\ctt).
\end{align*} 
Via \Cref{lemma:demeaned}, substituting $\db_i+\irt/I$ for $W^\rb_i$ and  $\ds_j+\jrt/J$ for $W^\rs_j$, we can write
\begin{align*}
    \meanestimate{\ctt}&=\meanpopulation{\ctt}+\frac{1}{I_\rt}\sum_{i=1}^I \left(D^\rb_i+\frac{I_\rt}{I}\right)\dyi(\ctt)
    +\frac{1}{J_\rt}
    \sum_{j=1}^J\left(D^\rs_j+\frac{J_\rt}{J}\right)\dyj(\ctt)\\
    &+\frac{1}{I_\rt J_\rt}\sum_{i=1}^I\sum_{j=1}^J\left(D^\rb_i+\frac{I_\rt}{I}\right)\left( D^\rs_j+\frac{J_\rt}{J}\right)\dyij(\ctt).
\end{align*}
By definition, $\dyij(\ctt)$, $\dyi(\ctt)$ and $\dyj(\ctt)$ sum to zero. Hence the equation above simplifies to
\begin{align*}
    \meanestimate{\ctt}
    &= 
    \meanpopulation{\ctt}
    + \sum_{i=1}^I \frac{D^\rb_i\dyi(\ctt)}{I_\rt}
    + \sum_{j=1}^J \frac{D^\rs_j\dyj(\ctt)}{J_\rt}
    + \sum_{i=1}^I\sum_{j=1}^J \frac{D^\rb_i D^\rs_j\dyij(\ctt)}{I_\rt J_\rt} 
    =\meanpopulation{\ctt}+ \oee_{\ctt}^{\rb}+\oee_{\ctt}^{\rs}+\ooe_{\ctt}^{\rb\rs}.
\end{align*}

This concludes the proof of the first part of  $(a)$. The proofs of the other parts of $(a)$ follow the same argument and are omitted. Given part $(a)$, $(b)$ follows immediately  because $\db_i$ and $\ds_j$ have expectation equal to zero. The same holds for the covariances in $(c)$.
\end{proof}
Unbiasedness results in \Cref{lemma:unbiasedness} and \Cref{thm:spillover_unbiasedness} follow directly from \Cref{lemma_appendix1}. 
\begin{lemma}[\Cref{lemma:unbiasedness} in the main paper]
Consider a SMRD in which \Cref{sutva_local} holds. 
\ifarxiv 
The plug-in estimators in  \Cref{eq:meantypeestimate} satisfy
\[ 
    \mme\left[\meanestimate{\type}\right] =\meanpopulation{\type},\;\forall\; \type\in\types.
\]
\else
The plug-in estimators in  \Cref{eq:meantypeestimate} satisfy
\(
    \mme\left[\meanestimate{\type}\right] =\meanpopulation{\type},\;\forall\; \type\in\types.
\)
\fi 
\end{lemma}
\begin{proof}[Proof of \Cref{lemma:unbiasedness}]
    Apply \Cref{lemma_appendix1}, and linearity of the expectation operator.
\end{proof}

\begin{theorem}[Already \Cref{thm:spillover_unbiasedness} in the main paper]
Consider a SMRD where \Cref{sutva_local} holds. The plug-in estimators $\hat{\tau}(\coefvec)$ for  $\tau(\coefvec)$ defined in \Cref{eq:causal_estimands} satisfy
\begin{equation*}
    \mme\left[ \hat{\tau}(\coefvec) \right] = {\tau}(\coefvec),
\quad
\textrm{with }\quad
    \hat{\tau}(\coefvec) := \beta_{\ccc}\meanestimate{\ccc} +
    \beta_{\icb}\meanestimate{\icb} +
    \beta_{\ics}\meanestimate{\ics} +
    \beta_{\ctt}\meanestimate{\ctt}.
\end{equation*}
\end{theorem}
\begin{proof}[Proof of \Cref{thm:spillover_unbiasedness}]
    Apply \Cref{lemma:unbiasedness}, and linearity of the expectation operator.
\end{proof}
We now move to the variance characterization. 
For $\type\in\{\ccc,\icb,\ics,\ctt\}$, recall the definitions of the population variances of $\dyi(\type)$ and  $\dyj(\type)$ given in \Cref{sec:estimation}:
\[ 
    \axisvariancepopulation{\type}{\rb}:= \sum_{i=1}^I \frac{\left(\dyi(\type)\right)^2}{I},
    \quad 
    \axisvariancepopulation{\type}{\rs}:=\sum_{j=1}^J \frac{\left(\dyj(\type)\right)^2}{J},
    \quad 
    \axisvariancepopulation{\type}{\rb\rs}:= \sum_{i=1}^I \sum_{j=1}^J \frac{\left(\dyij(\type)\right)^2}{IJ}  .
\]

\begin{lemma}\label{lemma_appendix2}
For $\type \in \types$, the variance of $\meanestimate{\type}$ is:
\begin{align}
    \mmv_\type&:= \mmv\left(\meanestimate{\type}\right)
    =
    \frac{I-I_\type}{I_\type} \frac{1}{I-1} \axisvariancepopulation{\type}{\rb} 
    + \frac{J - J_\type}{J_\type} \frac{1}{J-1} \axisvariancepopulation{\type}{\rs} 
    + \frac{I-I_\type}{I_\type} \frac{1}{I-1} \frac{J-J_\type}{J_\type} \frac{1}{J-1} \axisvariancepopulation{\type}{\rb\rs} \nonumber \\
    &= \alpha_\type^\rb \axisvariancepopulation{\type}{\rb}  + \alpha_\type^\rs \axisvariancepopulation{\type}{\rs}  + \alpha_\type^\rb \alpha_\type^\rs \axisvariancepopulation{\type}{\rb \rs} \nonumber,
\end{align}
where $\alpha_\type^\rb$ and $\alpha_\type^\rs$ where defined in \cref{eq:alpha_weights} in the main text.
\end{lemma}
\begin{proof}[Proof of Lemma \ref{lemma_appendix2}]
We consider $\type = \ctt$, (i.e., $I_\type = I_{\rt}$, $J_\type = J_{\rt}$). 
For $\type = \ctt$, $I_{\rc} = I-I_{\type}$ and $J_{\rc} = J - J_{\type}$.
We show the three following equalities hold:
\begin{equation}\label{eq:three_eq}
    \mmv\left(\oee_{\ctt}^{\rb}\right) =\frac{I_\rc}{I_\rt } \frac{1}{I-1} \axisvariancepopulation{\ctt}{\rb}, 
    \quad
    \mmv\left( \oee_{\ctt}^{\rs}\right)=\frac{J_\rc}{J_\rt } \frac{1}{J-1} \axisvariancepopulation{\ctt}{\rs},
    \quad
    \mmv\left( \ooe_{\ctt}^{\rb\rs}\right)=\frac{I_\rc}{I_\rt } \frac{1}{I-1} \frac{J_\rc}{J_\rt } \frac{1}{J-1} \axisvariancepopulation{\ctt}{\rb\rs}.
\end{equation}
Because  \Cref{lemma_appendix1} implies that
\(
    \mmv_\ctt=\mmv\left( \meanestimate{\ctt}\right)=\mmv\left( \oee_{\ctt}^{\rb}\right)+\mmv\left( \oee_{\ctt}^{\rs}\right) +\mmv\left(\ooe_{\ctt}^{\rb\rs}\right),
\)
showing the three equalities in \cref{eq:three_eq} yields the thesis. 
\ifarxiv
\begin{align*}
    \mmv\left( \oee_{\ctt}^{\rb}\right)
    &=\mme\left[\left(\frac{1}{I_\rt}\sum_{i=1}^I D^\rb_i\dyi(\ctt)\right)^2
    \right] 
    =\frac{1}{I_\rt^2}\mme\left[\sum_{i=1}^I\sum_{i'=1}^I D^\rb_i D^\rb_{i'} \dyi(\ctt) \dydy{i'}{\rb}(\ctt)\right]
    \\ &
    =\frac{1}{I_\rt^2}\sum_{i=1}^I\sum_{i'=1}^I \mme\left[ D^\rb_i D^\rb_{i'}
    \right] \dyi(\ctt) \dydy{i'}{\rb}(\ctt) \\
    & = \frac{1}{I_\rt^2}\sum_{i=1}^I \mme\left[ (D^\rb_i )^2
    \right] \dyi(\ctt) + \frac{1}{I_\rt^2}\sum_{i=1}^I \sum_{i'\neq i} \mme\left[ D^\rb_i D^\rb_{i'} 
    \right] \dyi(\ctt) \dydy{i'}{\rb}(\ctt) \\
    &=\frac{1}{I^2_\rt}\sum_{i=1}^I \frac{I_{\rc}I_{\rt}}{I^2}\left(\dyi(\ctt)\right)^2
    -\frac{1}{I_\rt^2}\sum_{i=1}^I\sum_{i'\neq i} \frac{I_\rt I_\rc}{I^2 (I-1)} \dyi(\ctt) \dydy{i'}{\rb}(\ctt) \\
    &=
    \frac{1}{I^2_\rt}\sum_{i=1}^I \frac{I_{\rc}I_{\rt}}{I^2}\left(\dyi(\ctt)\right)^2
    -\frac{1}{I_\rt^2}\sum_{i=1}^I\sum_{i'=1}^I \frac{I_\rt I_\rc}{I^2 (I-1)} \dyi(\ctt) \dydy{i'}{\rb}(\ctt)
    + \frac{1}{I^2_\rt}\sum_{i=1}^I \frac{I_\rt I_\rc}{I^2 (I-1)} (\dyi(\ctt))^2.
\end{align*}
\else
\begin{align*}
    \mmv\left( \oee_{\ctt}^{\rb}\right)
    &=\mme\left[\left(\frac{1}{I_\rt}\sum_{i=1}^I D^\rb_i\dyi(\ctt)\right)^2
    \right] 
    =\frac{1}{I_\rt^2}\sum_{i=1}^I\sum_{i'=1}^I \mme\left[ D^\rb_i D^\rb_{i'}
    \right] \dyi(\ctt) \dydy{i'}{\rb}(\ctt) \\
    &=
    \frac{1}{I^2_\rt}
        \left[
            \sum_{i=1}^I \frac{I_{\rc}I_{\rt}}{I^2}\left(\dyi(\ctt)\right)^2
            -\sum_{i=1}^I\sum_{i'=1}^I \frac{I_\rt I_\rc}{I^2 (I-1)} \dyi(\ctt) \dydy{i'}{\rb}(\ctt)
            +\sum_{i=1}^I \frac{I_\rt I_\rc}{I^2 (I-1)} (\dyi(\ctt))^2
        \right].
\end{align*}
\fi
Because $\sum_i\dyi(\ctt)=0$, the term above involving the double sum is equal to zero: 
\begin{align*}
    \mmv\left( \oee_{\ctt}^{\rb}\right)&=\frac{1}{I^2_\rt}\frac{I_\rt I_\rc}{I^2 (I-1)} \sum_{i=1}^I (\dyi(\ctt))^2
    +\frac{1}{I^2_\rt} \frac{I_\rt I_\rc}{I^2} \sum_{i=1}^I \left(\dyi(\ctt)\right)^2  
    =\frac{I_\rc}{I_\rt} \frac{1}{I-1} \axisvariancepopulation{\ctt}{\rb}. 
\end{align*}
The second equality in \cref{eq:three_eq} is proved analogously. 
For the last equality in \cref{eq:three_eq},
\begin{align*}
    \mmv_{\ctt}^{\rb\rs}&:=\mmv\left( \ooe_{\ctt}^{\rb\rs}\right)= \mmv\left(\frac{1}{I_\rt J_\rt}\sum_{i=1}^I\sum_{j=1}^JD^\rb_i D^\rs_j\dyij(\ctt) \right) \\
    &= \mme\left[ 
    \frac{1}{I^2_\rt J^2_\rt}\sum_{i=1}^I\sum_{i'=1}^I
    \sum_{j=1}^J\sum_{j'=1}^J D^\rb_i D^\rb_{i'}D^\rs_j D^\rs_{j'}\dyij(\ctt)
    \delta_{i', j'}^{\mathrm{BS}}(\ctt)\right].
\end{align*}
\ifarxiv
By independence of $D^\rb_i$ and $D^\rs_j$, this is equal to
\[
    \mmv_{\ctt}^{\rb\rs}=
    \frac{1}{I^2_\rt J^2_\rt}\sum_{i=1}^I\sum_{i'=1}^I \mme\left[ D^\rb_i D^\rb_{i'}\right]
    \sum_{j=1}^J\sum_{j'=1}^J\mme\left[ D^\rs_j D^\rs_{j'}\right]\dyij(\ctt)
    \delta_{i', j'}^{\mathrm{BS}}(\ctt).
\]
Now we expand the four-way sum above, noting that it is either the case that (a) : $i=i'$ and $j=j'$, (b) : $i=i'$ and $j \neq j'$, (c) : $i\neq i'$ and $j = j'$ or (d) : $i\neq i'$ and $j \neq j'$.
\begin{align*}
    \mmv_{\ctt}^{\rb\rs}&\overset{(a)}{=} \frac{1}{I^2_\rt J^2_\rt}\sum_{i=1}^I \sum_{j=1}^J \mme[(D_i^\rb)^2] \mme[(D_j^\rs)^2] \left(\dyij(\ctt)\right)^2  \\
    &\overset{(b)}{+}\frac{1}{I^2_\rt J^2_\rt}\sum_{i=1}^I \sum_{j=1}^J \sum_{j'\neq j}^J \mme[(D_i^\rb)^2] \mme[D_j^\rs D_{j'}^{\rs}] \dyij(\ctt) \delta_{i, j'}^{\mathrm{BS}}(\ctt) \\
    &\overset{(c)}{+}\frac{1}{I^2_\rt J^2_\rt}\sum_{i=1}^I \sum_{i'\neq i}^I\sum_{j=1}^J \mme[D_i^\rb D_{i'}^{\rb}] \mme[(D_j^\rs)^2] \dyij(\ctt) \delta_{i', j}^{\mathrm{BS}}(\ctt) \\
    &\overset{(d)}{+} \frac{1}{I^2_\rt J^2_\rt} \sum_{i=1}^I \sum_{i'\neq i }^I \sum_{j=1}^J \sum_{j' \neq j }^J \mme[D_i^\rb D_{i'}^{\rb}] \mme[D_j^\rs D_{j'}^{\rs}] \dyij(\ctt) \delta_{i', j'}^{\mathrm{BS}}(\ctt).
\end{align*}

Now we ``complete'' each of the last ``incomplete'' sums (b), (c), (d). For (b):
\begin{align*}
    \sum_{i=1}^I \sum_{j=1}^J \sum_{j'\neq j}^J \frac{\mme[(D_i^\rb)^2] \mme[D_j^\rs D_{j'}^{\rs}]}{I^2_\rt J^2_\rt} &\dyij(\ctt) \delta_{i, j'}^{\mathrm{BS}}(\ctt) = - \frac{\left( \frac{I_\rt I_\rc}{I^2} \right) \left( \frac{J_\rt J_\rc}{J^2(J-1)} \right)}{I^2_\rt J^2_\rt}  \sum_{i=1}^I \sum_{j=1}^J \sum_{j'\neq j}^J  \dyij(\ctt) \dydy{ij'}{\rb\rs}(\ctt) \\
    &= - \frac{\left( \frac{I_\rt I_\rc}{I^2} \right) \left( \frac{J_\rt J_\rc}{J^2(J-1)} \right)}{I^2_\rt J^2_\rt} \sum_{i=1}^I \sum_{j=1}^J \sum_{j'=1}^J  \dyij(\ctt) \dydy{ij'}{\rb\rs}(\ctt) \\
    &+ \frac{\left( \frac{I_\rt I_\rc}{I^2} \right) \left( \frac{J_\rt J_\rc}{J^2(J-1)} \right)}{I^2_\rt J^2_\rt}\sum_{i=1}^I \sum_{j=1}^J  (\dyij(\ctt))^2 \\
    &= \frac{1}{I^2_\rt J^2_\rt} \left( \frac{I_\rt I_\rc}{I^2} \right) \left( \frac{J_\rt J_\rc}{J^2(J-1)} \right) \sum_{i=1}^I \sum_{j=1}^J  (\dyij(\ctt))^2,
\end{align*}
where we observe that $\sum_{i=1}^I \sum_{j=1}^J \sum_{j'=1}^J  \dyij(\ctt) \dydy{ij'}{\rb\rs}(\ctt) = 0$.
A similar derivation allows us to ``complete'' (c), yielding:
\begin{align*}
   \sum_{i=1}^I \sum_{i'\neq i}^I\sum_{j=1}^J \frac{\mme[D_i^\rb D_{i'}^{\rb}] \mme[(D_j^\rs)^2] \dyij(\ctt) \dydy{i'j}{\rb\rs}(\ctt)}{{I^2_\rt J^2_\rt}} = \frac{\left( \frac{I_\rt I_\rc}{I^2(I-1)} \right) \left( \frac{J_\rt J_\rc}{J^2} \right) }{I^2_\rt J^2_\rt} \sum_{i=1}^I \sum_{j=1}^J  (\dyij(\ctt))^2.
\end{align*}
Last, for (d), 
\begin{align*}
   \sum_{i=1}^I \sum_{i'\neq i}^I\sum_{j=1}^J \sum_{j\neq j'}^J \frac{\mme[D_i^\rb D_{i'}^{\rb}] \mme[D_j^\rs D_{j'}^\rs] \dyij(\ctt) \dydy{i'j}{\rb\rs}(\ctt)}{{I^2_\rt J^2_\rt}} = \frac{\left( \frac{I_\rt I_\rc}{I^2(I-1)} \right) \left( \frac{J_\rt J_\rc}{J^2(J-1)} \right) }{I^2_\rt J^2_\rt} \sum_{i=1}^I \sum_{j=1}^J  (\dyij(\ctt))^2.
\end{align*}
Plugging these back in $ \mmv_{\ctt}^{\rb\rs}$,
\begin{align*}
    \mmv_{\ctt}^{\rb\rs}&= \frac{1}{I^2_\rt J^2_\rt} \frac{I_{\rc}I_{\rt} J_{\rc}J_{\rt}}{I^2 J^2} \left[ 1 + \frac{1}{I-1} + \frac{1}{J-1} + \frac{1}{(I-1)(J-1)}\right] \sum_{i=1}^I \sum_{j=1}^J (\dyij(\ctt))^2 \\
    &= \frac{1}{I^2_\rt J^2_\rt} \frac{I_{\rc}I_{\rt} J_{\rc}J_{\rt}}{I^2 J^2} \left[ \frac{IJ}{(I-1)(J-1)}\right] \sum_{i=1}^I \sum_{j=1}^J (\dyij(\ctt))^2 
    =\frac{I_\rc}{I_\rt }\frac{1}{I-1} \frac{J_\rc}{J_\rt } \frac{1}{J-1}\axisvariancepopulation{\ctt}{\rb\rs}.
\end{align*}
\else
Leveraging independence of $D^\rb_i$ and $D^\rs_j$, expanding, completing squares, and observing that $\sum_{i=1}^I \sum_{j=1}^J \sum_{j'=1}^J  \dyij(\ctt) \dydy{ij'}{\rb\rs}(\ctt) = \sum_{i=1}^I \sum_{i'=1}^J \sum_{j=1}^J  \dyij(\ctt) \dydy{i'j}{\rb\rs}(\ctt) = 0$ yields
\begin{align*}
    \mmv_{\ctt}^{\rb\rs}&= \frac{1}{I^2_\rt J^2_\rt} \frac{I_{\rc}I_{\rt} J_{\rc}J_{\rt}}{I^2 J^2} \left[ 1 + \frac{1}{I-1} + \frac{1}{J-1} + \frac{1}{(I-1)(J-1)}\right] \sum_{i=1}^I \sum_{j=1}^J (\dyij(\ctt))^2 \\
    &= \frac{1}{I^2_\rt J^2_\rt} \frac{I_{\rc}I_{\rt} J_{\rc}J_{\rt}}{I^2 J^2} \left[ \frac{IJ}{(I-1)(J-1)}\right] \sum_{i=1}^I \sum_{j=1}^J (\dyij(\ctt))^2 
    =\frac{I_\rc}{I_\rt }\frac{1}{I-1} \frac{J_\rc}{J_\rt } \frac{1}{J-1} \axisvariancepopulation{\ctt}{\rb\rs}.
\end{align*}
The proofs of the  other parts of the lemma are similar and are omitted.
\fi
\end{proof}

In order to characterize the variance of the spillover effects, we need to characterize the covariance between the estimators $\meanestimate{\type}, \meanestimate{\type'}$, for $\type,\type'\in\types$. 
Recall the definitions provided in \Cref{sec:estimation}: 
for all $\type \neq \type'\in \types$ for {\buyer}s and the {\seller}s
\[ 
    \axisvariancecrosspopulation{\type,\type'}{\rb}:=\sum_{i=1}^I \frac{\left( \dyi(\type) - \dyi(\type')\right)^2}{I}, \quad 
    \axisvariancecrosspopulation{\type,\type'}{\rs}:=\sum_{j=1}^J \frac{\left( \dyj(\type) - \dyj(\type')\right)^2}{J}, \quad 
    \axisvariancecrosspopulation{\type,\type'}{\rb\rs}:=\sum_{i=1}^I \sum_{j=1}^J  \frac{\left( \dyij(\type)-\dyij(\type')\right)^2}{IJ}.
\]
\begin{lemma}\label{lemma_appendix3}
For $\type \neq \type' \in \types$, covariances of type estimators are
\ifarxiv
\else 
of the form
\fi
\begin{align*}
     \mmc_{\ctt,\icb}&:= 
     \mmc\left(\meanestimate{\ctt},\meanestimate{\icb}\right) \\
     &= 
     \frac{I_\rc}{2 I_\rt (I-1)}\left(\axisvariancepopulation{\ctt}{\rb}+\axisvariancepopulation{\icb}{\rb}-\axisvariancecrosspopulation{\ctt,\icb}{\rb} \right)
    - \frac{1}{2(J-1)}\left(\axisvariancepopulation{\ctt}{\rs}+\axisvariancepopulation{\icb}{\rs}-\axisvariancecrosspopulation{\ctt,\icb}{\rs} \right)
    \\
    &-\frac{I_\rc}{2I_\rt (I-1)(J-1)}\left(\axisvariancepopulation{\ctt}{\rb\rs}+\axisvariancepopulation{\icb}{\rb\rs}-\axisvariancecrosspopulation{\ctt,\icb}{\rb\rs} \right).
\end{align*}
\ifarxiv
Similarly,
\begin{align*}
     \mmc_{\ctt,\ics}&:= \mmc\left(\meanestimate{\ctt},\meanestimate{\ics}\right) \\
     &= -\frac{1}{2 (I-1)}\left(\axisvariancepopulation{\ctt}{\rb}+\axisvariancepopulation{\ics}{\rb}-\axisvariancecrosspopulation{\ctt,\ics}{\rb}\right) 
    + \frac{J_{\rc}}{2J_{\rt}(J-1)}\left(\axisvariancepopulation{\ctt}{\rs}+\axisvariancepopulation{\ics}{\rs}-\axisvariancecrosspopulation{\ctt,\ics}{\rs}\right) 
    \\
    &-\frac{J_\rc}{2IJ_\rt (J-1) }\left(\axisvariancepopulation{\ctt}{\rb\rs}+\axisvariancepopulation{\ics}{\rb\rs}-\axisvariancecrosspopulation{\ctt,\ics}{\rb\rs}\right),
\end{align*}

\begin{align*}
     \mmc_{\ctt,\ccc}&:= \mmc\left(\meanestimate{\ctt},\meanestimate{\ccc}\right) \\
     &= -\frac{1}{2 (I-1)}\left(\axisvariancepopulation{\ctt}{\rb}+\axisvariancepopulation{\ccc}{\rb}-\axisvariancecrosspopulation{\ctt,\ccc}{\rb}\right) 
    - \frac{1}{2 (J-1)}\left(\axisvariancepopulation{\ctt}{\rs}+\axisvariancepopulation{\ccc}{\rs}-\axisvariancecrosspopulation{\ctt,\ccc}{\rs}\right)  
    \\
    &+\frac{1}{2 (I-1) (J-1)}\left(\axisvariancepopulation{\ctt}{\rb\rs}+\axisvariancepopulation{\ccc}{\rb\rs}-\axisvariancecrosspopulation{\ctt,\ccc}{\rb\rs}\right),
\end{align*}

\begin{align*}
     \mmc_{\icb,\ics}&:= \mmc\left(\meanestimate{\icb},\meanestimate{\ics}\right) \\
     &= -\frac{1}{2 (I-1)}\left(\axisvariancepopulation{\icb}{\rb}+\axisvariancepopulation{\ics}{\rb}-\axisvariancecrosspopulation{\icb,\ics}{\rb}\right) 
    - \frac{1}{2(J-1)}\left(\axisvariancepopulation{\icb}{\rs}+\axisvariancepopulation{\ics}{\rs}-\axisvariancecrosspopulation{\icb,\ics}{\rs}\right)  
    \\
    &+\frac{1}{2(I-1)(J-1)}\left(\axisvariancepopulation{\icb}{\rb\rs}+\axisvariancepopulation{\ics}{\rb\rs}-\axisvariancecrosspopulation{\icb,\ics}{\rb\rs}\right),
\end{align*}

\begin{align*}
     \mmc_{\icb,\ccc}&:= \mmc\left(\meanestimate{\icb},\meanestimate{\ccc}\right) \\
     &= -\frac{1}{2 (I-1)}\left(\axisvariancepopulation{\icb}{\rb}+\axisvariancepopulation{\ccc}{\rb}-\axisvariancecrosspopulation{\icb,\ccc}{\rb}\right)
    -\frac{J_{\rc}}{2J_{\rt}(J-1)}\left(\axisvariancepopulation{\icb}{\rs}+\axisvariancepopulation{\ccc}{\rs}-\axisvariancecrosspopulation{\icb,\ccc}{\rs}\right) 
    \\
    &-\frac{J_\rc}{2(I-1)J_{\rt}(J-1)}\left(\axisvariancepopulation{\icb}{\rb\rs}+\axisvariancepopulation{\ccc}{\rb\rs}-\axisvariancecrosspopulation{\icb,\ccc}{\rb\rs}\right),
\end{align*}
and last
\begin{align*}
     \mmc_{\ics,\ccc}&:= \mmc\left(\meanestimate{\ics},\meanestimate{\ccc}\right) \\
     &= 
     \frac{I_\rc}{2 I_\rt (I-1)}\left(\axisvariancepopulation{\ics}{\rb}+\axisvariancepopulation{\ccc}{\rb}-\axisvariancecrosspopulation{\ics,\ccc}{\rb}\right) 
    - \frac{1}{2(J-1)}\left(\axisvariancepopulation{\ics}{\rs}+\axisvariancepopulation{\ccc}{\rs}-\axisvariancecrosspopulation{\ics,\ccc}{\rs}\right)  
    \\
    &-\frac{I_\rc}{2I_\rt (I-1)(J-1)}\left(\axisvariancepopulation{\ics}{\rb\rs}+\axisvariancepopulation{\ccc}{\rb\rs}-\axisvariancecrosspopulation{\ics,\ccc}{\rb\rs}\right).
\end{align*}
\else
Results for other pairwise comparisons are similar and omitted.
\fi
\end{lemma}

\begin{proof}[Proof of Lemma \ref{lemma_appendix3}]
\ifarxiv
We show the three following equalities:
\begin{align}\label{3een}
    \mmc_{\ctt,\icb}^{\rb}&:= \mmc\left(\oee_{\ctt}^{\rb},\oee_{\icb}^{\rb}\right) =
    \frac{I_\rc}{2 I_\rt (I-1)}\left(\axisvariancepopulation{\ctt}{\rb}+\axisvariancepopulation{\icb}{\rb}-\axisvariancecrosspopulation{\ctt,\icb}{\rb} \right), 
\end{align}
\begin{align}\label{3twee}
    \mmc_{\ctt,\icb}^{\rs}&:= \mmc\left(\oee_{\ctt}^{\rs},\oee_{\icb}^{\rs}\right)
    =\frac{1}{2(J-1)}\left(\axisvariancepopulation{\ctt}{\rs}+\axisvariancepopulation{\icb}{\rs}-\axisvariancecrosspopulation{\ctt,\icb}{\rs} \right), 
\end{align}
and
\begin{align}\label{3drie}
    \mmc_{\ctt,\icb}^{\rb\rs}&:= \mmc\left(\ooe_{\ctt}^{\rb\rs},\ooe_{\icb}^{\rb\rs}\right)
    =\frac{I_\rc}{2I_\rt (I-1)(J-1)}\left(\axisvariancepopulation{\ctt}{\rb\rs}+\axisvariancepopulation{\icb}{\rb\rs}-\axisvariancecrosspopulation{\ctt,\icb}{\rb\rs} \right). 
\end{align}
In combination with the fact that
\[  
    \mmc\left(\meanestimate{\ctt},\meanestimate{\icb}\right)=
    \mmc\left(\oee_{\ctt}^{\rb},\oee_{\icb}^{\rb}\right)
    -
    \mmc\left(\oee_{\ctt}^{\rs},\oee_{\icb}^{\rs}\right)
    -
    \mmc\left(\ooe_{\ctt}^{\rb\rs},\ooe_{\icb}^{\rb\rs}\right),
\]
this proves the first result.

First (\ref{3een}):
\begin{align*}
    \mmc_{\ctt,\icb}^{\rb}
    &=\mme\left[\left(\frac{1}{I_\rt}\sum_{i=1}^I  D^\rb_i\dyi(\ctt)\right)\left(\frac{1}{I_\rt}\sum_{i=1}^I D^\rb_{i}\dyi(\icb)\right)\right]
    =\mme\left[\frac{1}{I_\rt^2}\sum_{i=1}^I\sum_{i'=1}^I  D^\rb_iD^\rb_{i'}\dyi(\ctt)\dyii(\icb)\right]\\
    &=\frac{1}{I_\rt^2}\sum_{i=1}^I\sum_{i'=1}^I \mme\left[D^\rb_iD^\rb_{i'}\right]\dyi(\ctt)\dyii(\icb)
    \\ 
    &
    =-\frac{1}{I_\rt^2}\sum_{i=1}^I\sum_{i'=1}^I \frac{I_\rc I_\rt}{I^2 (I-1)} \dyi(\ctt) \dyii(\icb) 
     +\frac{1}{I_\rt^2}\sum_{i=1}^I \left(\frac{I_\rc I_\rt}{I^2(I-1)}+\frac{I_\rc I_\rt}{I^2}\right)\dyi(\ctt)\dyi(\icb).
\end{align*}
Because $\sum_i\dyi(\ctt)=0$ the first term is equal to zero. 
Thus,
\begin{align*}
    \mmc_{\ctt,\icb}^{\rb} &= \frac{I_\rc}{I_\rt I}\left(\frac{1}{I-1}\sum_{i=1}^I \dyi(\ctt) \dyi(\icb)\right).
\end{align*}
Because
\begin{align*}
    \axisvariancepopulation{\ctt,\icb}{\rb} &= \frac{1}{I}\sum_{i=1}^I \left( \dyi(\ctt)-\dyi(\icb)\right)^2 \\
    &=\frac{1}{I}\sum_{i=1}^I \left(\dyi(\ctt)\right)^2
    +\frac{1}{I}\sum_{i=1}^I \left(\dyi(\icb)\right)^2
    -\frac{2}{I}\sum_{i=1}^I \left(\dyi(\ctt)\dyi(\icb)\right)^2\\
    &=\axisvariancepopulation{\ctt}{\rb}+\axisvariancepopulation{\icb}{\rb}-\frac{2I_\rt (I-1)}{I_\rc}\mmc_{\ctt,\icb}^{\rb},
\end{align*}
we have
\[ 
    \mmc_{\ctt,\icb}^{\rb}=\frac{I_\rc}{2 I_\rt (I-1)}\left(\axisvariancepopulation{\ctt}{\rb}+\axisvariancepopulation{\icb}{\rb}-\axisvariancecrosspopulation{\ctt,\icb}{\rb} \right).
\]
This completes the proof of (\ref{3een}). Similarly, to prove (\ref{3twee}), we have
\begin{align*}
    \mmc_{\ctt,\icb}^{\rs}
    &=\mme\left[\left(\frac{1}{J_\rt}\sum_{j=1}^JD^\rs_j\dyj(\ctt)\right)\left(\frac{1}{J_\rc}\sum_{j=1}^JD^\rs_j\dyj(\icb)\right)\right]
    =\mme\left[\frac{1}{J_\rt J_\rc}\sum_{j=1}^J\sum_{j'=1}^J  D^\rs_iD^\rs_{j'}\dyj(\ctt)\dyjj(\icb)\right]\\
    &=\frac{1}{J_\rt J_\rc}\sum_{j=1}^J\sum_{j'=1}^J \mme\left[ D^\rs_jD^\rs_{j'}\right]\dyj(\ctt) \dyjj(\icb)\\
    &=-\frac{1}{J_\rt J_\rc}\sum_{j=1}^J\sum_{j'=1}^J\frac{J_\rc J_\rt}{J^2 (J-1)} \dyj(\ctt)\dyjj(\icb) 
    +\frac{1}{J_\rt J_\rc}\sum_{j=1}^J \left(\frac{J_\rc J_\rt}{J^2(J-1)}+\frac{J_\rc J_\rt}{J^2}\right)\dyj(\ctt)\dyj(\icb) \\
    &=\frac{1}{J_\rt J_\rc}\sum_{j=1}^J \left(\frac{J_\rc J_\rt}{J^2(J-1)}+\frac{J_\rc J_\rt}{J^2}\right)\dyj(\ctt)\dyj(\icb) 
    =\frac{1}{ J}\left(\frac{1}{J-1}\sum_{j=1}^J \dyj(\ctt) \dyj(\icb)\right).
\end{align*}
Because
\begin{align*}              
    \axisvariancecrosspopulation{\ctt,\icb}{\rs}&=
    \frac{1}{J}\sum_{j=1}^J \left( \dyj(\ctt)- \dyj(\icb)\right)^2 \\
    &=\frac{1}{J}\sum_{j=1}^J \left(\dyj(\ctt)\right)^2 
    +\frac{1}{J}\sum_{j=1}^J \left(\dyj(\icb)\right)^2
    -\frac{2}{J}\sum_{j=1}^J \left(\dyj(\ctt)\dyj(\icb)\right)^2 \\
    &=\axisvariancepopulation{\ctt}{\rs}+\axisvariancepopulation{\icb}{\rs}+2(J-1)\mmc_{\ctt,\icb}^{\rs},
\end{align*}
it follows that
\[ 
    \mmc_{\ctt,\icb}^{\rs}=\frac{1}{2(J-1)}\left(\axisvariancepopulation{\ctt}{\rs}+\axisvariancepopulation{\icb}{\rs}-\axisvariancecrosspopulation{\ctt,\icb}{\rs} \right).
\]
This finishes the proof of (\ref{3twee}). Third, consider (\ref{3drie}):
\[ 
    \mmc_{\ctt,\icb}^{\rb\rs}=\mme\left[\frac{1}{I_\rt^2 J_\rt  J_\rc}\sum_{i, i'=1}^I\sum_{j, j'=1}^J  D^\rb_i D^\rs_j D^\rb_{i'} D^\rs_{j'}\dyij(\ctt)\dyiijj(\icb) \right].
\]
By independence of $D^\rb_i$ and $D^\rs_j$, this is equal to
\[ 
    \mmc_{\ctt,\icb}^{\rb\rs}= \frac{1}{I^2_\rt J_\rc J_\rt}\sum_{i, i'=1}^I \sum_{j, j'=1}^J \mme\left[ D^\rb_i D^\rb_{i'}\right]\mme\left[ D^\rs_i D^\rs_{j'}\right]\dyij(\ctt) \dyiijj(\icb).
\]

Using the covariances and variances for $D^\rb_i$ and $D^\rb_{i'}$ and for 
$D^\rs_j$ and $D^\rs_{j'}$ this is equal to

\begin{align*}
    \mmc_{\ctt,\icb}^{\rb\rs}&=\frac{1}{I^2_\rt J_\rc   J_\rt} \sum_{i=1}^I \sum_{i'=1}^I \sum_{j=1}^J \sum_{j'=1}^J \frac{I_\rc I_\rt}{I^2 (I-1)} \frac{J_\rc J_\rt}{J^2 (J-1)} \dyij(\ctt) \dyiijj(\icb) \\
    &-\frac{1}{I^2_\rt  J_\rc  J_\rt}\sum_{i=1}^I \sum_{j=1}^J \sum_{j'=1}^J \frac{I_\rc I_\rt}{I^2(I-1)} \frac{J_\rc J_\rt}{J(J-1)} \dyij(\ctt) \dyijj(\icb) \\
    &-\frac{1}{I^2_\rt  J_\rc  J_\rt}\sum_{i=1}^I \sum_{i'=1}^I\sum_{j=1}^J \frac{I_\rc I_\rt}{I(I-1)} \frac{J_\rc J_\rt}{J^2(J-1)} \dyij(\ctt) \dyiij(\icb) \\
    &+\frac{1}{I^2_\rt  J_\rc  J_\rt}\sum_{i=1}^I \sum_{j=1}^J \frac{I_\rc I_\rt}{I(I-1)} \frac{J_\rc J_\rt}{J(J-1)} \dyij(\ctt) \dyij(\icb).
\end{align*}

Because $\sum_i\sum_j \dyij(\gamma)=0$, the first three terms are equal to zero, and so 
\begin{align*}
    \mmc_{\ctt,\icb}^{\rb\rs}&=\frac{1}{I^2_\rt  J_\rc  J_\rt}\sum_{i=1}^I \sum_{j=1}^J \frac{I_\rc I_\rt}{I(I-1)} \frac{J_\rc J_\rt}{J(J-1)}  \dyij(\ctt)\dyij(\icb) \\
    &=\frac{I_\rc}{I_\rt  IJ}\left(\frac{1}{(I-1)(J-1)}\sum_{i=1}^I \sum_{j=1}^J 
    \dyij(\ctt)\dyij(\icb)\right).
\end{align*}
Because
\begin{align*}
    \axisvariancecrosspopulation{\ctt,\icb}{\rb\rs}
    &=\frac{1}{IJ}\sum_{i=1}^I \sum_{j=1}^J \left( \dyij(\ctt)- \dyij(\icb)\right)^2 \\
    &=\frac{1}{IJ}\sum_{i=1}^I \sum_{j=1}^J \left( \dyij(\ctt)\right)^2 + \frac{1}{IJ}\sum_{i=1}^I \sum_{j=1}^J \left( \dyij(\icb)\right)^2 \\
    &-\frac{2}{IJ}\sum_{i=1}^I \sum_{j=1}^J  \dyij(\ctt) \dyij(\icb) \\
    &= \axisvariancepopulation{\ctt}{\rb\rs}+\axisvariancepopulation{\icb}{\rb\rs}-\frac{2I_\rt (I-1)(J-1)}{I_\rc}\mmc_{\ctt,\icb}^{\rb\rs},
\end{align*}

it follows that
\[ 
    \mmc_{\ctt,\icb}^{\rb\rs}=
    \frac{I_\rc}{2I_\rt (I-1)(J-1)}\left(\axisvariancepopulation{\ctt}{\rb\rs}+\axisvariancepopulation{\icb}{\rb\rs}-\axisvariancecrosspopulation{\ctt,\icb}{\rb\rs} \right).
\]
This finishes the proof of (\ref{3drie}). 
The proofs for the other pairwise comparisons follow the same pattern and are omitted.
\else
Note: the proof of this result is simple, but tedious. For a detailed, step-by-step derivation, see \citet[Lemma A.8]{masoero2024multiple}.
\[  
    \mmc\left(\meanestimate{\ctt},\meanestimate{\icb}\right)=
    \mmc\left(\oee_{\ctt}^{\rb},\oee_{\icb}^{\rb}\right)
    -
    \mmc\left(\oee_{\ctt}^{\rs},\oee_{\icb}^{\rs}\right)
    -
    \mmc\left(\ooe_{\ctt}^{\rb\rs},\ooe_{\icb}^{\rb\rs}\right).
\]
For $\mmc_{\ctt,\icb}^{\rb}:= \mmc\left(\oee_{\ctt}^{\rb},\oee_{\icb}^{\rb}\right)$,
\begin{align*}
    \mmc_{\ctt,\icb}^{\rb}
    &=\mme\left[\left(\frac{1}{I_\rt}\sum_{i=1}^I  D^\rb_i\dyi(\ctt)\right)\left(\frac{1}{I_\rt}\sum_{i=1}^I D^\rb_{i}\dyi(\icb)\right)\right]
    \\ 
    &
    =-\frac{1}{I_\rt^2}\sum_{i=1}^I\sum_{i'=1}^I \frac{I_\rc I_\rt}{I^2 (I-1)} \dyi(\ctt) \dyii(\icb) 
     +\frac{1}{I_\rt^2}\sum_{i=1}^I \left(\frac{I_\rc I_\rt}{I^2(I-1)}+\frac{I_\rc I_\rt}{I^2}\right)\dyi(\ctt)\dyi(\icb)\\
     &= \frac{I_\rc}{I_\rt I}\left(\frac{1}{I-1}\sum_{i=1}^I \dyi(\ctt) \dyi(\icb)\right)
     = \frac{I_\rc}{2 I_\rt I}\left(\axisvariancepopulation{\ctt}{\rb}+\axisvariancepopulation{\icb}{\rb}-\axisvariancecrosspopulation{\ctt,\icb}{\rb} \right),
\end{align*}
where we used the fact that 
\(
    \axisvariancepopulation{\ctt,\icb}{\rb} 
    =\axisvariancepopulation{\ctt}{\rb}+\axisvariancepopulation{\icb}{\rb}-\frac{2I_\rt (I-1)}{I_\rc}\mmc_{\ctt,\icb}^{\rb}.
\)
A similar proof yields
\begin{align*}
    \mmc_{\ctt,\icb}^{\rs}
    :=\mmc\left(\oee_{\ctt}^{\rs},\oee_{\icb}^{\rs}\right)
    =\frac{1}{2(J-1)}\left(\axisvariancepopulation{\ctt}{\rs}+\axisvariancepopulation{\icb}{\rs}-\axisvariancecrosspopulation{\ctt,\icb}{\rs} \right).
\end{align*}
Last, leveraging the independence of $D^\rb_i$ and $D^\rs_j$ and $\sum_i\sum_j \dyij(\gamma)=0$
\[ 
    \mmc_{\ctt,\icb}^{\rb\rs}
    =
    \mmc\left(\ooe_{\ctt}^{\rb\rs},\ooe_{\icb}^{\rb\rs}\right)
    =
    \frac{I_\rc}{2I_\rt (I-1)(J-1)}\left(\axisvariancepopulation{\ctt}{\rb\rs}+\axisvariancepopulation{\icb}{\rb\rs}-\axisvariancecrosspopulation{\ctt,\icb}{\rb\rs} \right).
\]
Proofs for the other pairwise comparisons follow the same pattern and are omitted.
\fi
\end{proof}
\begin{theorem}[\Cref{thm:covariance_characterization} in the main paper] \label{thm:covariance_characterization_app}
For a SMRD where \Cref{sutva_local} holds, 
\begin{equation}
    \mmc\left[ \meanestimate{\type}, \meanestimate{\type'} \right]
    = 
        \nu^{\rb}_{\type,\type'} \zeta^{\rb}_{\type,\type'} + 
        \nu^{\rs}_{\type,\type'} \zeta^{\rs}_{\type,\type'} + 
        2 \nu^{\rb}_{\type,\type'} \nu^{\rs}_{\type,\type'} \zeta^{\rb\rs}_{\type,\type'},
   \label{eq:app_covariance_characterization}
\end{equation}
where for $x \in \{\rb,\rs,\rb\rs\}$,
\(
    \zeta^{x}_{\type,\type'} :=  \axisvariancepopulation{\type}{x}+ \axisvariancepopulation{\type'}{x}- \axisvariancecrosspopulation{\type,\type'}{x}
\)
and
\[
        \nu^{\rb}_{\type,\type'} := 
        \begin{cases}
            \alpha^\rb_{\type}/2\mbox{ if } \type=\type',\text{ or } (\type,\type') \in \{(\ccc,\ics), (\ics,\ccc), (\icb,\ctt), (\ctt,\icb)\} \\
            - 1/(2(I-1)) \mbox{ otherwise,}
            \end{cases}
\]
and
\[
        \nu^{\rs}_{\type,\type'}:=
        \begin{cases}
             \alpha^\rs_{\type}/2 \mbox{ if } \type = \type' \text{ or } (\type,\type') \in \{(\ccc,\icb),(\icb,\ccc), (\ics,\ctt), (\ctt,\ics)\} \\
            - 1/(2(J-1)) \mbox{ otherwise.}
        \end{cases}
\]
\end{theorem}
\begin{proof}[Proof of \Cref{thm:covariance_characterization}] 
\Cref{lemma_appendix2} (for $\type = \type'$) and \Cref{lemma_appendix3} (for $\type \neq \type'$) prove this result. 
\ifarxiv
We spell these cases out and verify that the expressions derived in \cref{lemma_appendix2,lemma_appendix3} match with the compact representation provided in \cref{eq:app_covariance_characterization}.
\begin{itemize}
    \item if $\type = \type'$, using \Cref{lemma_appendix2} and the definitions of $\alpha_\type^\rb$ and $\alpha_\type^\rs$, 
    \begin{align*}
        \mmc\left[ \meanestimate{\type}, \meanestimate{\type} \right] 
        =
        \mmv\left(\meanestimate{\type}\right)
        &=
        \frac{I - I_\type}{I_\type} \frac{1}{I-1} \axisvariancepopulation{\type}{\rb} 
        + \frac{J - J_\type}{J_\type} \frac{1}{J-1} \axisvariancepopulation{\type}{\rs} 
        + \frac{I - I_\type}{I_\type }\frac{1}{I-1} \frac{J - J_\type}{J_\type } \frac{1}{J-1} \axisvariancepopulation{\type}{\rb\rs} 
        \\
        &=
        \alpha_\type^\rb \axisvariancepopulation{\type}{\rb} 
        + \alpha_\type^\rs \axisvariancepopulation{\type}{\rs} 
        + \alpha_\type^\rb  \alpha_\type^\rs \axisvariancepopulation{\type}{\rb\rs}
        .
    \end{align*}
    We verify that \cref{eq:app_covariance_characterization} is correct by spelling out $\nu_{\type,\type}^{\rb}$, $\nu_{\type,\type}^{\rs}$, $\zeta_{\type, \type}^{\rb}$, $\zeta_{\type,\type}^{\rs}$ --- and check that we get the same result as above:
    \begin{align*}
    \mmc\left[ \meanestimate{\type}, \meanestimate{\type'} \right]
    &= 
        \nu^{\rb}_{\type,\type'} \zeta^{\rb}_{\type,\type'} + 
        \nu^{\rs}_{\type,\type'} \zeta^{\rs}_{\type,\type'} + 
        2 \nu^{\rb}_{\type,\type'} \nu^{\rs}_{\type,\type'} \zeta^{\rb\rs}_{\type,\type'}
    \\
    &=
        \frac{\alpha_\type^{\rb}}{2}  (2 \axisvariancepopulation{\type}{\rb})+ 
        \frac{\alpha_\type^{\rs}}{2}  (2 \axisvariancepopulation{\type}{\rs})+ 
         2 \frac{\alpha_\type^{\rb}}{2} \frac{\alpha_\type^{\rs}}{2} (2 \axisvariancepopulation{\type}{\rb\rs}) \\
    &=  
    {\alpha_\type^{\rb}} \axisvariancepopulation{\type}{\rb}+ 
    {\alpha_\type^{\rs}}  \axisvariancepopulation{\type}{\rs}+ 
    \alpha_\type^{\rb} {\alpha_\type^{\rs}}  \axisvariancepopulation{\type}{\rb\rs}.
    \end{align*}
    \item if $\type \neq \type'$, use \Cref{lemma_appendix3}, 
    and consider any of the treatment pairs (e.g., $(\icb, \ics)$):
    \begin{align*}
        \mmc\left[ \meanestimate{\icb}, \meanestimate{\ics} \right] 
         &= -\frac{1}{2 (I-1)}\left(\axisvariancepopulation{\icb}{\rb}+\axisvariancepopulation{\ics}{\rb}-\axisvariancecrosspopulation{\icb,\ics}{\rb}\right) 
        - \frac{1}{2(J-1)}\left(\axisvariancepopulation{\icb}{\rs}+\axisvariancepopulation{\ics}{\rs}-\axisvariancecrosspopulation{\icb,\ics}{\rs}\right)  
        \\
        &+\frac{1}{2(I-1)(J-1)}\left(\axisvariancepopulation{\icb}{\rb\rs}+\axisvariancepopulation{\ics}{\rb\rs}-\axisvariancecrosspopulation{\icb,\ics}{\rb\rs}\right)
        \\
        &=
        \nu_{\type,\type'}^{\rs} \zeta_{\type,\type'}^{\rb} 
        + \nu_{\type,\type'}^{\rs} \zeta_{\type,\type'}^{\rs}
        + 2 \nu_{\type}^{\rb} \nu_{\type'}^{\rs} \zeta_{\type,\type'}^{\rb\rs},
    \end{align*}
    which matches the compact representation.
\end{itemize}
\else 
Simple re-writing of expressions in \cref{lemma_appendix2,lemma_appendix3} yield the compact representation provided in \cref{eq:app_covariance_characterization}.
See \citet[Theorem A.9]{masoero2024multiple} for details.
\fi
\end{proof}

\ifarxiv
To present our results on estimates of the variance, we first review a classic result for variances of a simple two-arms experiment, when a single population is present.

\begin{lemma} \label{lemma:pure_experiments}
    Let $y_i$, $i=1,\ldots,I$ be a population of $I$ units with (non-random) potential outcomes $y_i(\ccc)$ (if unit $i$ is in the control group) and $y_i(\ctt)$ (if unit $i$ is in the treatment group).
    Let the treatment group be identified by the index set $\II_{\ctt} = \{i_1,\ldots,i_{I_\ctt}\} \subset \{1,\ldots,I\}$, of size $|\II_\ctt| = I_\ctt$, with $2\le I_\ctt \le I-2$. Let $\II_\ccc = \{1,\ldots,I\} \setminus \II_\ctt$ be the index set of the $I_\ccc := I-I_\ctt$ units assigned to the control group. For $\type \in \{\ccc,\ctt\}$, let
    \[
        \bar{y}_{\type} = \frac{1}{I} \sum_{i=1}^I y_i(\type)        \quad\text{and}\quad
        s_{\type} = \frac{1}{I} \sum_{i=1}^I \left(y_i(\type) - \bar{y}_{\type} \right)^2. 
    \]
    be the mean and variance of the potential outcomes in the population. Define the corresponding plug-in estimates for these  to be
    \[
        \widehat{\bar{Y}}_{\type} = \frac{1}{I_{\type}} \sum_{i \in \II_{\type}} y_i(\type),
        \quad\text{and}\quad
        \hat{S}_{\type} = \frac{1}{I_{\type}} \sum_{i \in \II_{\type}} \left(y_i(\type) - \widehat{\bar{Y}}_{\type} \right)^2. 
    \]
    Then it holds
    \begin{align}
        \mme \left[ \widehat{\bar{Y}}_{\type} \right] = \bar{y}_{\type}, 
    \quad\text{and}\quad
        \mmv\left( \widehat{\bar{Y}}_{\type} \right) = \frac{I-I_{\type}}{I_{\type}} \frac{1}{I-1} s_{\type}, \nonumber
    \quad\text{and}\quad
        \mme\left( \hat{S}_{\type} \right) = \frac{I_{\type}-1}{I_{\type}} \frac{I}{I-1} s_{\type}. \label{eq:expectations_single_axis}
    \end{align}
    I.e., $\widehat{\bar{Y}}_{\type}$ is an unbiased estimate of the population mean $\bar{y}_{\type}$. We can obtain an unbiased estimate of the variance of this estimator by reweighing $\hat{S}_{\type}$:
    \begin{align}
        \widehat{\mmv}\left( \widehat{\bar{Y}}_{\type} \right) :=  \frac{I-I_{\type}}{I_{\type}-1}\frac{1}{I} \hat{S}_{\type} 
        \quad \text{satisfies} \quad 
        \mme\left[\widehat{\mmv}\left( \widehat{\bar{Y}}_{\type} \right)\right] = \mmv\left( \widehat{\bar{Y}}_{\type} \right). 
    \end{align}
    \begin{proof}[Proof of 
    \Cref{lemma:pure_experiments}]
        See \eg\ \citet[Theorems 2.1, 2.2, 2.4]{cochran1977sampling1}.
    \end{proof}
\end{lemma}
\else
\fi
\subsection{Variance estimation in SMRDs: proofs}\label{sec:proofs_variance_estimation}

Here we provide lower and upper bounds on the variance of causal effects in SMRDs (\cref{thm:sample_variance_avg_effs,thm:sample_variance_spillover}). 
For a SMRD in which local interference holds, given an assignment matrix $\bww$, denote by $\II_{\type} \subseteq \{1,\ldots,I\}$ the subset of {\buyer}s' indices for which there exists at least one {\seller} $j$ such that unit $(i,j)$ has type $\type$: $\II_{\type}:=\{i \in \{1,\ldots,I\}\;:\; \type_{ij} = \type\;\text{for some }j\}$. 
Symmetrically, let $\JJ_{\type} \subseteq \{1,\ldots,J\}$ the subset of {\seller}s' indices for which there exists at least one {\buyer} $i$ such that unit $(i,j)$ has type $\type$.
Consistent with \cref{app-sec-def}, $I_{\type} = |\II_{\type}|$ and $J_{\type} = |\JJ_{\type}|$ denote the sizes of these index sets. 
Exactly $I_{\type}J_{\type}$ units are assigned type $\type$. 
Define now, the (nonrandom) row and column \emph{partial} mean of the matrix of potential outcomes: for a given row $i$, the average over a fixed index set of columns $\JJ_{\type} \subseteq \mmj$ --- symmetrically, for a given column $j$, the average over a fixed set of rows $\II_{\type} \subseteq \mmi$:
\[
    \axismeanfixedindex{i}{\JJ_{\type}}{\rb}(\type) = \frac{1}{J_{\type}}\sum_{j \in \JJ_{\type}} y_{i,j}(\type)
    \quad\text{and}\quad 
    \axismeanfixedindex{\II_{\type}}{j}{\rs}(\type) = \frac{1}{I_{\type}}\sum_{i \in \II_{\type}} y_{i,j}(\type).
\]

For a given SMRD, with (random) assignment matrix $\bw$ and characterized by (random) index sets $\II_{\type}, \JJ_{\type}$ for each $\type \in \types$, $i\in \II_{\type}, j \in \JJ_{\type}$, define the random average over the columns selected by the set $\JJ_{\type}$ (or the rows selected by $\II_{\type}$):
\[
    \axismeanestimate{i}{\rb}(\type):=\frac{1}{J_{\type}}\sum_{j \in \JJ_{\type}} y_{i,j}(\type), 
    \quad \text{and} \quad \axismeanestimate{j}{\rs}(\type):=\frac{1}{I_{\type}}\sum_{i \in \II_{\type}} y_{i,j}(\type).
\]
\ifarxiv
\paragraph{Remark} The quantities $\axismeanestimate{i}{\rb}(\type)$ and $\axismeanfixedindex{i}{\JJ_{\type}}{\rb}(\type)$ are \emph{both} averages over $J_{\type}$ elements of the $i$-th row of the matrix of potential outcomes $Y(\type)$. However, $\axismeanestimate{i}{\rb}(\type)$ is random: it is an estimator resulting from the random selection of $J_{\type}$ distinct columns, whereas $\axismeanfixedindex{i}{\JJ_{\type}}{\rb}(\type)$ is a fixed population value, obtained by averaging over the fixed $J_{\type}$ distinct indices $\{j_1,\ldots,j_{J_{\type}} \} = \JJ_{\type}$.
\else 
\fi
Define the sample ``plug-in'' counterparts of the population quantities $\axisvariancepopulationscaled{\type}{\rb}, \axisvariancepopulationscaled{\type}{\rs}$ and $\axisvariancepopulationscaled{\type}{\rb\rs}$:
\begin{equation*}
    \axisvarianceestimatescaled{\type}{\rb}= \frac{1}{I_{\type}}\sum_{i\in\II_{\type}} \left(\axismeanestimate{i}{\rb}(\type) - \meanestimate{\type} \right)^2,
    \quad 
    \axisvarianceestimatescaled{\type}{\rs} = \frac{1}{J_{\type}}\sum_{j\in \JJ_{\type}} \left(\axismeanestimate{j}{\rs}(\type) - \meanestimate{\type} \right)^2,
\end{equation*}
and
\begin{equation*}
    \axisvarianceestimatescaled{\type}{\rb\rs}
    :=
    \frac{1}{I_{\type}J_{\type}}\sum_{i \in \II_{\type}} 
    \sum_{ j \in \JJ_{\type}} 
    \left(y_{i,j}(\type) - \axismeanestimate{i}{\rb}(\type) - \axismeanestimate{j}{\rs}(\type) + \meanestimate{\type} \right)^2.
\end{equation*}
$\axisvarianceestimatescaled{\type}{\rb}, \axisvarianceestimatescaled{\type}{\rs}, \axisvarianceestimatescaled{\type}{\rb\rs}$ are stochastic and depend on the (random) assignment $\bw$ through the index sets $\II_{\type}, \JJ_{\type}$. Last, define the variances of partial averages over subsets $\JJ_\type$ and $\II_\type$
    \begin{align}
        \axisbias{\type}{\rb}:=
        \frac{
            \sum_{\JJ_{\type}} \sum_i \left\{ \axismeanfixedindex{i}{\JJ_{\type}}{\rb}(\type) - \axismeanpopulation{i}{\rb}(\type) \right\}^2
            }{
            I \binom{J}{J_{\type}}
            }, 
        \quad
        \axisbias{\type}{\rs}:=
        \frac{
            \sum_{\II_{\type}} \sum_j \left\{ \axismeanfixedindex{\II_{\type}}{j}{\rs}(\type) - \axismeanpopulation{j}{\rs}(\type)\right\}^2
            }{
            {J \binom{I}{I_{\type}}}}. 
        \label{eq:bsMV}
    \end{align}
Notice that the sums in \cref{eq:bsMV} are over all subsets of $J_\type$ disjoint indices in $[J]$ ($\axisbias{\type}{\rb}$) or $I_\type$ disjoint indices in $[I]$ ($\axisbias{\type}{\rs}$).
In \cref{lemma:rows,lemma:columns,lemma:cross}, we analyze the expectation of each term $\axisvarianceestimatescaled{\type}{\rb},\axisvarianceestimatescaled{\type}{\rs}, \axisvarianceestimatescaled{\type}{\rb\rs}$ separately. First, we state a useful result in \cref{lemma:chi_variance}.

\begin{lemma} \label{lemma:chi_variance}
    Let 
    \begin{align}
        \chi^{2,\rb}_{\type}:=\mme\left[\sum_{i \in \II_{\type}} \left(\axismeanestimate{i}{\rb}(\type) - \meanpopulation{\type} \right)^2\right]. \label{eq:chi_2M}
    \end{align}
    It holds
        $\chi^{2,\rb}_{\type}= I_{\type} \left( \axisvariancepopulationscaled{\type}{\rb} +  \axisbias{\type}{\rb} \right)$,
    where $\axisbias{\type}{\rb}$ was defined in \cref{eq:bsMV}.
\end{lemma}

\begin{proof}[Proof of \cref{lemma:chi_variance}]
    Consider $\chi^{2,\rb}_{\type}$ as defined in \cref{eq:chi_2M}, where the expectation is taken with respect to the random assignment matrices $\bw$. 
    Under (simple) double randomization, every assignment matrix $\bw$ supported on $\mmw$ is equivalently characterized by the index sets $\II_{\type}, \JJ_{\type}$, for $\type \in \types$. 
    \ifarxiv 
    That is, to each $\bw$, there is one and only one collection of index sets $\II_{\type}, \JJ_{\type}$ for $\type \in \types$, and viceversa. Notice that there are exactly $\binom{I}{I_{\type}}\binom{J}{J_{\type}}$ such assignments. Each assignment can be determined by forming index set $\II_{\type}$ by selecting at random $I_{\type}$ rows and index set $\JJ_{\type}$ by selecting at random $J_{\type}$ columns. Every row $i \in \{1,\ldots, I\}$ appears in exactly $\binom{I-1}{I_{\type}-1}$ index sets $\II_{\type}$. 
    \else 
    \fi 
    Hence,
\begin{align} \label{eq:var_dMRD_M0}
    \chi^{2,\rb}_{\type} &=  \frac{\binom{I-1}{I_{\type}-1}}{\binom{I}{I_{\type}}\binom{J}{J_{\type}}}  \sum_{i=1}^I  \sum_{\JJ_{\type}} \left\{    \left( \axismeanfixedindex{i}{\JJ_{\type}}{\rb}(\type) - \meanpopulation{\type}\right)^2\right\} 
    = \frac{I_{\type}}{I\binom{J}{J_{\type}}}  \sum_{i=1}^I  \sum_{\JJ_{\type}} \left\{    \left( \axismeanfixedindex{i}{\JJ_{\type}}{\rb}(\type) - \meanpopulation{\type}\right)^2\right\},
\end{align}
where the second sum is over all $\binom{J}{J_{\type}}$ subsets $\JJ_{\type}$ of $J_{\type}$ distinct indices in $\{1,\ldots,J\}$.

We further decompose $\chi^{2,\rb}_{\type}$: fix a row $i$ and disjoint indices $\JJ_{\type} = \{j_1,\ldots,j_{J_{\type}}\} \subseteq [J]$:
\begin{align*}
  \{ \axismeanfixedindex{i}{\JJ_{\type}}{\rb}(\type) - \meanpopulation{\type} \}^2 &= \left\{ \axismeanfixedindex{i}{\JJ_{\type}}{\rb}(\type) - \axismeanpopulation{i}{\rb}(\type)  + \axismeanpopulation{i}{\rb}(\type) - \meanpopulation{\type}  \right\}^2 \\
    &= \left\{\axismeanfixedindex{i}{\JJ_{\type}}{\rb}(\type) - \axismeanpopulation{i}{\rb}(\type)\right\}^2 + \left\{\axismeanpopulation{i}{\rb}(\type) - \meanpopulation{\type}\right\}^2 
    +2 \left\{\axismeanfixedindex{i}{\JJ_{\type}}{\rb}(\type) - \axismeanpopulation{i}{\rb}(\type)\right\} \left\{\axismeanpopulation{i}{\rb}(\type) - \meanpopulation{\type} \right\}.
\end{align*}
Summing over all choices  $\JJ_{\type}$ of $J_{\type}$ disjoint indices in the set $\{1,\ldots,J\}$,
\begin{align*}
    \sum_{\JJ_{\type}} \left\{ \axismeanfixedindex{i}{\JJ_{\type}}{\rb}(\type) - \meanpopulation{\type} \right\}^2 &= \sum_{\JJ_{\type}}  \left\{ \axismeanfixedindex{i}{\JJ_{\type}}{\rb}(\type) - \axismeanpopulation{i}{\rb}\right\}^2 
    + \binom{J}{J_{\type}}  \left\{ \axismeanpopulation{i}{\rb}(\type) - \meanpopulation{\type} \right\}^2 ,
\end{align*}
using
    $\sum_{\JJ_{\type}}  \left\{ \axismeanfixedindex{i}{\JJ_{\type}}{\rb}(\type) - \axismeanpopulation{i}{\rb}(\type) \right\} = 0,
    \;\text{since} \quad \axismeanpopulation{i}{\rb}(\type) = \frac{\sum_{\JJ_{\type}}\axismeanfixedindex{i}{\JJ_{\type}}{\rb}}{\binom{J}{J_{\type}}}$.
Summing over {\buyers}:
\begin{align*}
    \sum_{i=1}^I \sum_{\JJ_{\type}} \left\{ \axismeanfixedindex{i}{\JJ_{\type}}{\rb}(\type) - \meanpopulation{\type} \right\}^2 &= \sum_{\JJ_{\type}} \sum_{i=1}^I \left\{ \axismeanfixedindex{i}{\JJ_{\type}}{\rb}(\type) - \axismeanpopulation{i}{\rb}(\type)\right\}^2 
    + \binom{J}{J_{\type}} \sum_i \left\{ \axismeanpopulation{i}{\rb}(\type) - \meanpopulation{\type} \right\}^2 \\
    &= I \binom{J}{J_{\type}} \left( \axisbias{\type}{\rb} + \axisvariancepopulationscaled{\type}{\rb} \right),
\end{align*}
where $\axisbias{\type}{\rb}:=\frac{1}{I}\binom{J}{J_{\type}}^{-1}\sum_{\JJ_{\type}} \sum_i \left\{ \axismeanfixedindex{i}{\JJ_{\type}}{\rb}(\type) - \axismeanpopulation{i}{\rb}(\type) \right\}^2$.
Hence, plugging this in \cref{eq:var_dMRD_M0},
\begin{align}
    \chi^{2,\rb}_{\type} &= 
    \sum_{i=1}^I  \sum_{\JJ_{\type}} 
    \frac{\left\{    \left( \axismeanfixedindex{i}{\JJ_{\type}}{\rb}(\type) - \meanpopulation{\type}\right)^2\right\}}{{\frac{I}{I_\type} \binom{J}{J_{\type}}}}
    =\frac{\left[ I \binom{J}{J_{\type}} \left( \axisbias{\type}{\rb} + \axisvariancepopulationscaled{\type}{\rb} \right) \right]}{\frac{I}{I_\type} \binom{J}{J_{\type}}}  
    = I_{\type} \left( \axisvariancepopulationscaled{\type}{\rb} +  \axisbias{\type}{\rb} \right). \label{eq:variance_row} 
\end{align} 
\end{proof}

\begin{lemma} \label{lemma:rows}
    It holds
    \begin{align}
        \mme\left[\axisvarianceestimatescaled{\type}{\rb}\right]= \axisvariancepopulationscaled{\type}{\rb} - \mmv\left(\meanestimate{\type}\right) + \axisbias{\type}{\rb}. \label{eq:var_dMRD_M}
    \end{align}
\end{lemma}
\begin{proof}[Proof of \cref{lemma:rows}]
\ifarxiv 
\begin{align}
    \mme\left[\axisvarianceestimatescaled{\type}{\rb}\right]&=\mme\left[\frac{1}{I_{\type}}\sum_{ i \in \II_{\type}} \left(\axismeanestimate{i}{\rb}(\type) - \meanestimate{\type} \right)^2\right] 
    = \frac{1}{I_{\type}}\mme\left[\sum_{i \in \II_{\type}} \left\{ \left[\axismeanestimate{i}{\rb}(\type) - \meanpopulation{\type}\right] - \left[\meanestimate{\type} - \meanpopulation{\type}\right]  \right\}^2\right] 
    \nonumber  \\
    &=\frac{1}{I_{\type}}
    \mme
    \left[
        \sum_{i \in \II_{\type}} \left[\axismeanestimate{i}{\rb}(\type) - \meanpopulation{\type}\right]^2 
        +
        \sum_{i \in \II_{\type}} \left[\meanestimate{\type} - \meanpopulation{\type}\right]^2
        -2 
        \sum_{i \in \II_{\type}} \left(\meanestimate{\type} - \meanpopulation{\type}\right) \left(\axismeanestimate{i}{\rb}(\type) - \meanpopulation{\type}\right)\right] 
        \nonumber.
\end{align}
\begin{align}
    \mme\left[\axisvarianceestimatescaled{\type}{\rb}\right]
    &= 
    \frac{1}{I_{\type}} 
    \left\{
        \mme\left[
            \sum_{i \in \II_{\type}} \left\{ \axismeanestimate{i}{\rb}(\type) - \meanpopulation{\type} \right\}^2
            + 
            I_{\type} \left\{\meanestimate{\type} - \meanpopulation{\type}\right\}^2 
            - 2 
            \left\{\meanestimate{\type} - \meanpopulation{\type}\right\}I_{\type}\left\{\meanestimate{\type} - \meanpopulation{\type}\right\} 
            \right] 
    \right\} \nonumber 
    \\
    &= \frac{1}{I_{\type}}\mme\left[\sum_{i \in \II_{\type}} \left\{\axismeanestimate{i}{\rb}(\type) - \meanpopulation{\type} \right\}^2\right] \nonumber 
    - \mme\left[\left\{\meanestimate{\type} - \meanpopulation{\type} \right\}^2\right] \nonumber.
\end{align}
\else 
\begin{align*}
    \mme\left[\axisvarianceestimatescaled{\type}{\rb}\right]
    &=
    \mme\left[\frac{1}{I_{\type}}\sum_{ i \in \II_{\type}} \left(\axismeanestimate{i}{\rb}(\type) - \meanestimate{\type} \right)^2\right] \\
    &= 
    \frac{1}{I_{\type}} 
    \left\{
        \mme\left[
            \sum_{i \in \II_{\type}} \left\{ \axismeanestimate{i}{\rb}(\type) - \meanpopulation{\type} \right\}^2
            + 
            I_{\type} \left\{\meanestimate{\type} - \meanpopulation{\type}\right\}^2 
            - 2
            \left\{\meanestimate{\type} 
            -  \meanpopulation{\type}\right\}I_{\type}\left\{\meanestimate{\type} - \meanpopulation{\type}\right\} 
            \right] 
    \right\} \nonumber 
    \\
    &= \frac{1}{I_{\type}}\mme\left[\sum_{i \in \II_{\type}} \left\{\axismeanestimate{i}{\rb}(\type) - \meanpopulation{\type} \right\}^2\right] \nonumber 
    - \mme\left[\left\{\meanestimate{\type} - \meanpopulation{\type} \right\}^2\right] \nonumber.
\end{align*}
\fi 
Hence, we write
\begin{align}
    \mme\left[\axisvarianceestimatescaled{\type}{\rb}\right]= \frac{1}{I_{\type}} \chi_{\type}^{2,\rb} - \mmv\left(\meanestimate{\type}\right), \label{eq:variance_step_1}
\end{align}
where we used the definition of $\chi^{2,\rb}_{\type}$ given in \cref{eq:chi_2M}. 
Now plugging in \cref{eq:variance_row} in \cref{eq:variance_step_1} it follows
\ifarxiv
\begin{align*} 
    \mme\left[\axisvarianceestimate{\type}{\rb}\right]&= \frac{1}{I_{\type}}\chi^{2,\rb}_{\type} - \mmv\left(\meanestimate{\type}\right) = \axisvariancepopulationscaled{\type}{\rb} - \mmv\left(\meanestimate{\type}\right) + \axisbias{\type}{\rb}.
\end{align*}
\else 
\(
    \mme\left[\axisvarianceestimate{\type}{\rb}\right] 
    = 
    \frac{1}{I_{\type}}\chi^{2,\rb}_{\type} - \mmv\left(\meanestimate{\type}\right) = \axisvariancepopulationscaled{\type}{\rb} - \mmv\left(\meanestimate{\type}\right) + \axisbias{\type}{\rb}.
\)
\fi 
\end{proof}

\begin{lemma} \label{lemma:columns}
    Recall $\axisbias{\type}{\rs}$ defined in \cref{eq:bsMV}. It holds
    \begin{align}
        \mme\left[\axisvarianceestimatescaled{\type}{\rs}\right]=\axisvariancepopulationscaled{\type}{\rs} - \mmv\left(\meanestimate{\type}\right) + \axisbias{\type}{\rs}, \label{eq:var_dMRD_V}
    \end{align}
\begin{proof}[Proof of \cref{lemma:columns}]
\ifarxiv
    The proof is identical to \cref{lemma:rows}, where we let $\chi^{2,\rs}_{\type}$ be the column counterpart to \cref{eq:variance_row},
        $\chi^{2,\rs}_{\type}:=\mme\left[\sum_{j \in \JJ_{\type}} \left(\axismeanestimate{j}{\rs}(\type) - \meanpopulation{\type} \right)^2\right]$,
    where, by the same argument of \cref{lemma:chi_variance}, it holds
    \begin{align}
        \chi^{2,\rs}_{\type} &= J_{\type} \axisvariancepopulationscaled{\type}{\rs} + \frac{J_{\type}}{J}\binom{I}{I_{\type}}^{-1}\sum_{\II_{\type}} \sum_{j=1}^J (\axismeanfixedindex{\II_{\type}}{j}{\rs} - \axismeanestimate{j}{\rs})^2 
        = J_{\type} \left( \axisvariancepopulationscaled{\type}{\rs} + \axisbias{\type}{\rs} \right) \label{eq:variance_col},
    \end{align}
    in which we sum over all $\binom{I}{I_{\type}}$ index sets $\II_{\type}$ of $I_{\type}$ disjoint indices in $\{1,\ldots,I\}$.
\else
    The proof is identical to \cref{lemma:rows}, and omitted.
\fi
\end{proof}
\end{lemma}

We now characterize 
$\axisvarianceestimatescaled{\type}{\rb\rs}$. We first state a useful decomposition for matrices. 
\begin{lemma} \label{lemma:square}
    Let $\bm{x} \in \mmr^{I \times J}$ be a matrix, and 
        $\bar{\bar{x}}:=(IJ)^{-1}\sum_{i,j} x_{i,j}$
    be the grand  mean of the matrix, where averaging is uniform across entries. Let 
        $\bar{x}_i^\rb:=J^{-1}\sum_j x_{i,j}$
        and
        $\bar{x}_j^\rs:=I^{-1}\sum_i x_{i,j}$
    be the average of the $i$-th row and of the $j$-th column respectively.
    It holds
    \begin{align*}
        \sum_{i,j}\left(x_{i,j} - \bar{\bar{x}}\right)^2 &= J \sum_{i} \left(\bar{x}_i^\rb - \bar{\bar{x}}\right)^2 + J \sum_{j} \left(\bar{x}_j^\rs - \bar{\bar{x}}\right)^2 
        + \sum_{i,j}\left(x_{i,j} - \bar{x}_i^\rb - \bar{x}_j^\rs + \bar{\bar{x}}\right)^2.
    \end{align*}
\end{lemma}

    \begin{proof}[Proof of \cref{lemma:square}]
    \ifarxiv
        \begin{align*}
            \sum_{i,j}\left(x_{i,j} - \bar{\bar{x}}\right)^2 &= \sum_{i,j}\left( x_{i,j} \pm \bar{x}_i^\rb \pm \bar{\bar{x}} \pm \bar{x}_j^\rs \pm \bar{\bar{x}} + \bar{\bar{x}}\right)^2  \\
            &= \sum_{i,j}\left\{\left( \bar{x}_i^\rb - \bar{\bar{x}}\right) + \left( \bar{x}_s^\rs - \bar{\bar{x}}\right) + \left( x_{i,j} - \bar{x}_i^\rb - \bar{x}_j^\rs + \bar{\bar{x}}\right) \right\}^2 \\
            &= \sum_{i,j} \left( \bar{x}_i^\rb - \bar{\bar{x}}\right)^2 + \sum_{i,j} \left( \bar{x}_j^\rs - \bar{\bar{x}}\right)^2 + \sum_{i,j} \left( x_{i,j} - \bar{x}_i^\rb - \bar{x}_j^\rs + \bar{\bar{x}}\right)^2,
        \end{align*}
        where we have noted that all the cross terms in the square cancel since
        \[
            \sum_{i,j} \left( \bar{x}_i^\rb - \bar{\bar{x}}\right) = 0, \quad 
            \sum_{i,j} \left( \bar{x}_j^\rs - \bar{\bar{x}}\right) = 0, \quad \sum_{i,j} \left( x_{i,j} - \bar{x}_i^\rb - \bar{x}_j^\rs + \bar{\bar{x}}\right) = 0.
        \]
        Hence,
        \begin{align*}
            \sum_{i,j}\left(x_{i,j} - \bar{\bar{x}}\right)^2 &= J \sum_{i} \left(\bar{x}_i^\rb - \bar{\bar{x}}\right)^2 + I \sum_{j} \left(\bar{x}_j^\rs - \bar{\bar{x}}\right)^2 + \sum_{i,j}\left(x_{i,j} - \bar{x}_i^\rb - \bar{x}_j^\rs + \bar{\bar{x}}\right)^2.
    \end{align*} 
    \else
    Since
        \(
            \sum_{i,j} \left( \bar{x}_i^\rb - \bar{\bar{x}}\right) =  
            \sum_{i,j} \left( \bar{x}_j^\rs - \bar{\bar{x}}\right) = 
            \sum_{i,j} \left( x_{i,j} - \bar{x}_i^\rb - \bar{x}_j^\rs + \bar{\bar{x}}\right) = 0,
        \)
    it holds
        \(
            \sum_{i,j}\left(x_{i,j} - \bar{\bar{x}}\right)^2 = 
            \sum_{i,j} 
            \left\{ 
                \left( \bar{x}_i^\rb - \bar{\bar{x}}\right)^2 
                + \left( \bar{x}_j^\rs - \bar{\bar{x}}\right)^2 
                + \left( x_{i,j} - \bar{x}_i^\rb - \bar{x}_j^\rs + \bar{\bar{x}}\right)^2
            \right\}.
        \)
        Hence,
        \begin{align*}
            \sum_{i,j}\left(x_{i,j} - \bar{\bar{x}}\right)^2 &= J \sum_{i} \left(\bar{x}_i^\rb - \bar{\bar{x}}\right)^2 + I \sum_{j} \left(\bar{x}_j^\rs - \bar{\bar{x}}\right)^2 + \sum_{i,j}\left(x_{i,j} - \bar{x}_i^\rb - \bar{x}_j^\rs + \bar{\bar{x}}\right)^2.
    \end{align*}     
    \fi
    \end{proof}

For our matrix of potential outcomes $\bm{y}(\type)$, direct application of \cref{lemma:square} gives us
\begin{align}
    \frac{1}{IJ}\sum_{i,j} \left[y_{i,j}(\type) - \meanpopulation{\type}\right]^2 &= \frac{1}{I}\sum_i \left\{  \axismeanpopulation{i}{\rb}(\type) - \meanpopulation{\type} \right\}^2 + \frac{1}{J}\sum_j \left\{\axismeanpopulation{j}{\rs}(\type) - \meanpopulation{\type} \right\}^2  \nonumber \\
    &+ \frac{1}{IJ}\sum_{i,j}  \left\{ y_{i,j}(\type) - \axismeanpopulation{i}{\rb}(\type) - \axismeanpopulation{j}{\rs}(\type) + \meanpopulation{\type} \right\}^2  
    = \axisvariancepopulationscaled{\type}{\rb} + \axisvariancepopulationscaled{\type}{\rs} + \axisvariancepopulationscaled{\type}{\rb\rs}. \label{eq:sum_squares}
\end{align}
\ifarxiv
We now analyze the expectation of the crossed term $\axisvarianceestimatescaled{\type}{\rb\rs}$.
\else 
\fi 
\begin{lemma} \label{lemma:cross}
\ifarxiv
    It holds
    \[
        \mme\left[\axisvarianceestimatescaled{\type}{\rb\rs}\right] = \axisvariancepopulationscaled{\type}{\rb\rs} + \mmv\left(\meanestimate{\type}\right) - \axisbias{\type}{\rb} - \axisbias{\type}{\rs}.
    \]
\else 
    It holds
    \(
        \mme\left[\axisvarianceestimatescaled{\type}{\rb\rs}\right] = \axisvariancepopulationscaled{\type}{\rb\rs} + \mmv\left(\meanestimate{\type}\right) - \axisbias{\type}{\rb} - \axisbias{\type}{\rs}.
    \)
\fi 
\end{lemma}

\begin{proof}[Proof of \cref{lemma:cross}]
\begin{align*}
    \axisvarianceestimatescaled{\type}{\rb\rs}&= \frac{1}{I_{\type}J_{\type}}\sum_{i \in \II_{\type}}\sum_{j \in \JJ_{\type}}\left\{y_{i,j}(\type) - \axismeanestimate{i}{\rb}(\type) - \axismeanestimate{j}{\rs}(\type) + \meanestimate{\type}\right\}^2.
\end{align*}
Expanding the square,
\ifarxiv
\begin{align*}
    \axisvarianceestimatescaled{\type}{\rb\rs}&=\frac{1}{I_{\type}J_{\type}}\sum_{i \in \II_{\type}}\sum_{j \in \JJ_{\type}} \left\{ \left(y_{i,j}(\type) - \meanpopulation{\type}\right) - \left( \axismeanestimate{i}{\rb}(\type) - \meanpopulation{\type}\right) - \left( \axismeanestimate{j}{\rs}(\type) - \meanpopulation{\type} \right) + \left( \meanestimate{\type} - \meanpopulation{\type} \right) \right\}^2 \\
    &= \frac{1}{I_{\type}J_{\type}} \sum_{i \in \II_{\type}}\sum_{j \in \JJ_{\type}} \left(y_{i,j}(\type) - \meanpopulation{\type}\right)^2 + \frac{1}{I_{\type}} \sum_{i \in \II_{\type}} \left( \axismeanestimate{i}{\rb}(\type) - \meanpopulation{\type}\right)^2 
    \\
    &
    + \frac{1}{J_{\type}} \sum_{j \in \JJ_{\type}} \left( \axismeanestimate{j}{\rs}(\type) - \meanpopulation{\type}\right)^2 + \left(\meanestimate{\type} - \meanpopulation{\type}\right)^2\\
    &-\frac{2}{I_{\type}J_{\type}} \sum_{i \in \II_{\type}} \left( \axismeanestimate{i}{\rb}(\type) - \meanpopulation{\type}\right) \sum_{j \in \JJ_{\type}} \left(y_{i,j}(\type) - \meanpopulation{\type}\right) -  \frac{2}{I_{\type}J_{\type}} \sum_{j \in \JJ_{\type}} \left( \axismeanestimate{j}{\rs}(\type) - \meanpopulation{\type}\right) \sum_{i \in \II_{\type}} \left(y_{i,j}(\type) - \meanpopulation{\type}\right) \\
    & + \frac{2}{I_{\type}J_{\type}} \left(\meanestimate{\type} - \meanpopulation{\type}\right)\sum_{i \in \II_{\type}}\sum_{j \in \JJ_{\type}}  \left( y_{i,j}(\type) - \meanpopulation{\type}\right) + \frac{2}{I_{\type}J_{\type}} \sum_{i \in \II_{\type}} \left( \axismeanestimate{i}{\rb}(\type) - \meanpopulation{\type}\right)\sum_{j \in \JJ_{\type}}\left(\axismeanestimate{j}{\rs}(\type) - \meanpopulation{\type}\right) \\
    &- \frac{2}{I_{\type}J_{\type}} \left(\meanestimate{\type} - \meanpopulation{\type}\right)J_{\type}\sum_{i \in \II_{\type}} \left( \axismeanestimate{i}{\rb}(\type) - \meanpopulation{\type}\right) - \frac{2}{I_{\type}J_{\type}} \left(\meanestimate{\type} - \meanpopulation{\type}\right)I_{\type}\sum_{j \in \JJ_{\type}} \left( \axismeanestimate{j}{\rs}(\type) - \meanpopulation{\type}\right)\\
    &= \frac{1}{I_{\type}J_{\type}} \sum_{i \in \II_{\type}}\sum_{j \in \JJ_{\type}} \left(y_{i,j}(\type) - \meanpopulation{\type}\right)^2 + \frac{1}{I_{\type}} \sum_{i \in \II_{\type}} \left( \axismeanestimate{i}{\rb}(\type) - \meanpopulation{\type}\right)^2 \\
    &+ \frac{1}{J_{\type}} \sum_{j \in \JJ_{\type}} \left( \axismeanestimate{j}{\rs}(\type) - \meanpopulation{\type}\right) +  \left(\meanestimate{\type} - \meanpopulation{\type}\right)^2 
    -  \frac{2}{I_{\type}} \sum_{i \in \II_{\type}} \left( \axismeanestimate{i}{\rb}(\type) - \meanpopulation{\type}\right)^2 - \frac{2}{J_{\type}} \sum_{j \in \JJ_{\type}} \left( \axismeanestimate{j}{\rs}(\type) - \meanpopulation{\type}\right),
\end{align*}
so
\begin{align*}
    \axisvarianceestimatescaled{\type}{\rb\rs}
    &= \frac{1}{I_{\type}J_{\type}} \sum_{i \in \II_{\type}}\sum_{j \in \JJ_{\type}} \left(y_{i,j}(\type) - \meanpopulation{\type}\right)^2 - \frac{1}{I_{\type}} \sum_{i \in \II_{\type}} \left( \axismeanestimate{i}{\rb}(\type) - \meanpopulation{\type}\right)^2 
    \\
    &
    - \frac{1}{J_{\type}} \sum_{j \in \JJ_{\type}} \left( \axismeanestimate{j}{\rs}(\type) - \meanpopulation{\type}\right)^2 + \left(\meanestimate{\type} - \meanpopulation{\type}\right)^2. 
\end{align*}
\else 
\begin{align*}
    \axisvarianceestimatescaled{\type}{\rb\rs}
    &= \frac{1}{I_{\type}J_{\type}} \sum_{i \in \II_{\type}}\sum_{j \in \JJ_{\type}} \left(y_{i,j}(\type) - \meanpopulation{\type}\right)^2 + \frac{1}{I_{\type}} \sum_{i \in \II_{\type}} \left( \axismeanestimate{i}{\rb}(\type) - \meanpopulation{\type}\right)^2 \\
    &+ \frac{1}{J_{\type}} \sum_{j \in \JJ_{\type}} \left( \axismeanestimate{j}{\rs}(\type) - \meanpopulation{\type}\right) +  \left(\meanestimate{\type} - \meanpopulation{\type}\right)^2 
    -  \frac{2}{I_{\type}} \sum_{i \in \II_{\type}} \left( \axismeanestimate{i}{\rb}(\type) - \meanpopulation{\type}\right)^2 - \frac{2}{J_{\type}} \sum_{j \in \JJ_{\type}} \left( \axismeanestimate{j}{\rs}(\type) - \meanpopulation{\type}\right),
\end{align*}
so
\begin{align*}
    \axisvarianceestimatescaled{\type}{\rb\rs}
    &= 
    \sum_{i \in \II_{\type}}\sum_{j \in \JJ_{\type}} \frac{\left(y_{i,j}(\type) - \meanpopulation{\type}\right)^2}{I_{\type}J_{\type}}    
    - \sum_{i \in \II_{\type}} \frac{ \left( \axismeanestimate{i}{\rb}(\type) - \meanpopulation{\type}\right)^2 }{I_{\type}} 
    - \sum_{j \in \JJ_{\type}} \frac{\left( \axismeanestimate{j}{\rs}(\type) - \meanpopulation{\type}\right)^2}{J_{\type}}   
    + \left(\meanestimate{\type} - \meanpopulation{\type}\right)^2. 
\end{align*}
\fi 

Under the expectation operator,
\ifarxiv
\begin{align*}
    \mme \left[ \axisvarianceestimatescaled{\type}{\rb\rs} \right] &
    = 
    \frac{1}{I_{\type}J_{\type}}\mme\left[\sum_{i \in \II_{\type}} \sum_{j \in \JJ_{\type}} \left\{y_{i,j}(\type) - \meanpopulation{\type}\right\}^2\right] - \frac{1}{I_{\type}} \mme\left[\sum_{i \in \II_{\type}} \left( \axismeanestimate{i}{\rb}(\type) - \meanpopulation{\type}\right)^2\right] \\ & - \frac{1}{J_{\type}} \mme\left[ \sum_{j \in \JJ_{\type}} \left( \axismeanestimate{j}{\rs}(\type) - \meanpopulation{\type}\right) \right] + \mmv\left(\meanestimate{\type}\right)  \\
    &=  \frac{1}{IJ} \sum_{i=1}^I\sum_{j=1}^J \left\{y_{i,j}(\type) - \meanpopulation{\type}\right\}^2  - \frac{1}{I_{\type}} \chi^{2,\rb}(\type) - \frac{1}{J_{\type}} \chi^{2,\rs}(\type) + \mmv\left(\meanestimate{\type}\right).
\end{align*}
\else
\begin{align*}
    \mme \left[ \axisvarianceestimatescaled{\type}{\rb\rs} \right] 
    &
     =  \frac{1}{IJ} \sum_{i=1}^I\sum_{j=1}^J \left\{y_{i,j}(\type) - \meanpopulation{\type}\right\}^2  - \frac{1}{I_{\type}} \chi^{2,\rb}(\type) - \frac{1}{J_{\type}} \chi^{2,\rs}(\type) + \mmv\left(\meanestimate{\type}\right).
\end{align*}
\fi
\ifarxiv
Now, leveraging \cref{eq:sum_squares} for the first summation, \cref{eq:variance_row} for the second summation, and \cref{eq:variance_col} for the third summation,
\begin{align}
    \mme \left[ \axisvarianceestimatescaled{\type}{\rb\rs} \right]  
    &= 
    \axisvariancepopulationscaled{\type}{\rb}+\axisvariancepopulationscaled{\type}{\rs}+\axisvariancepopulationscaled{\type}{\rb\rs}  -  \left[\axisvariancepopulationscaled{\type}{\rb} + \axisbias{\type}{\rb}\right] -  \left[ \axisvariancepopulationscaled{\type}{\rs} + \axisbias{\type}{\rs}\right] +  \mmv\left(\meanestimate{\type}\right) \nonumber 
    \\
    &
    =  \axisvariancepopulationscaled{\type}{\rb\rs} + \mmv\left(\meanestimate{\type}\right) - \axisbias{\type}{\rb} - \axisbias{\type}{\rs}. \label{eq:var_dMRD_MV}
\end{align}
\else 
Leveraging \cref{eq:sum_squares} for the first summation, \cref{eq:variance_row} for the second summation (and its ``column'' counterpart for the third summation)
\begin{align}
    \mme \left[ \axisvarianceestimatescaled{\type}{\rb\rs} \right]  &= 
    \pm \axisvariancepopulationscaled{\type}{\rb}
    \pm \axisvariancepopulationscaled{\type}{\rs}
    + \axisvariancepopulationscaled{\type}{\rb\rs}  
    - \axisbias{\type}{\rb} 
    - \axisbias{\type}{\rs}
    +  \mmv\left(\meanestimate{\type}\right) 
    =  \axisvariancepopulationscaled{\type}{\rb\rs} + \mmv_\type - \axisbias{\type}{\rb} - \axisbias{\type}{\rs}. \label{eq:var_dMRD_MV}
\end{align}
\fi 
\end{proof}
\ifarxiv
We now use the characterizations \cref{eq:var_dMRD_M,eq:var_dMRD_V,eq:var_dMRD_MV}, to define an unbiased estimator for $\mmv\left(\meanestimate{\type}\right)$, as stated in \cref{thm:sample_variance_avg_effs}.
\else 
\fi 
\begin{theorem}[Already  \cref{thm:sample_variance_avg_effs} in the main paper]  \label{lemma:sample_variance_avg_effs_long}
For a SMRD where \cref{sutva_local} holds, for all $\type\in \types$,
\ifarxiv
    \[
        \mme\left[\widehat{\Sigma}_{\type}\right]= \mmv\left(\meanestimate{\type}\right), 
    \]
\else 
    \(
        \mme\left[\widehat{\Sigma}_{\type}\right]= \mmv\left(\meanestimate{\type}\right), 
    \)
\fi 
where
\ifarxiv
    \begin{align*}
        \hat{\Sigma}_{\type}&:= \frac{\alpha^{\rb}_{\type}\axisvarianceestimatescaled{\type}{\rb}
        +\alpha^{\rs}_{\type}\axisvarianceestimatescaled{\type}{\rs}
        +(\alpha^{\rb}_{\type}\alpha^{\rs}_{\type})\axisvarianceestimatescaled{\type}{\rb\rs}}{1 - \alpha^\rb_{\type} - \alpha^\rs_{\type} + \alpha^\rb_{\type} \alpha^\rs_{\type}} \\
        &-\frac{\alpha^\rb_{\type}}{1-\alpha^\rb_{\type}} \frac{J-J_\type}{(J-1)(J_{\type}-1)}\frac{1}{ J_\type I_{\type}} \sum_{i \in \II_{\type}}\sum_{j \in \JJ_{\type}} \left(y_{i,j}(\type) - \axismeanestimate{i}{\rb}(\type)\right)^2 \\
        & - \frac{\alpha^\rs_{\type}}{1-\alpha^\rs_{\type}} \frac{I - I_{\type}}{(I-1)(I_{\type}-1)}\frac{1}{I_{\type}J_{\type}} \sum_{j \in \JJ_{\type}} \sum_{i \in \II_{\type}} \left(y_{i,j}(\type) - \axismeanestimate{j}{\rs}(\type)\right)^2,
    \end{align*}
    and where we have used the previously defined (non-random) coefficients $\alpha^\rb_{\type}$ and $\alpha^\rs_{\type}$.
\else 
    \begin{align*}
        \hat{\Sigma}_{\type}&:= 
        \frac{\alpha^{\rb}_{\type}\axisvarianceestimatescaled{\type}{\rb}
        +\alpha^{\rs}_{\type}\axisvarianceestimatescaled{\type}{\rs}
        +(\alpha^{\rb}_{\type}\alpha^{\rs}_{\type})\axisvarianceestimatescaled{\type}{\rb\rs}}{1 - \alpha^\rb_{\type} - \alpha^\rs_{\type} + \alpha^\rb_{\type} \alpha^\rs_{\type}} 
        - \sum_{i \in \II_{\type}}\sum_{j \in \JJ_{\type}} 
        \left\{
        \frac{\left(y_{i,j}(\type) - \axismeanestimate{i}{\rb}(\type)\right)^2}
        {\left( \frac{1-\alpha^\rb_{\type}}{\alpha^\rb_{\type}} \frac{J(J_{\type}-1)}{J-J_\type} ({ J_\type I_{\type}}) \right)} 
        -  
        \frac{ \left(y_{i,j}(\type) - \axismeanestimate{j}{\rs}(\type)\right)^2}{\frac{1-\alpha^\rs_{\type}}{\alpha^\rs_{\type}} 
        \frac{I(I_{\type}-1)}{I - I_{\type}} ({I_{\type}J_{\type}})}
        \right\}.
    \end{align*}
\fi 
\begin{proof}[Proof of \cref{thm:sample_variance_avg_effs} and \cref{lemma:sample_variance_avg_effs_long}]
Given $\alpha^\rb_{\type}, \alpha^\rs_{\type}$, \cref{lemma_appendix2} allows us to write
\ifarxiv 
\[    \mmv\left( \meanestimate{\type} \right) = \alpha_{\type}^\rb \axisvariancepopulationscaled{\type}{\rb} +\alpha_{\type}^\rs \axisvariancepopulationscaled{\type}{\rs} +  \alpha_{\type}^\rb \alpha_{\type}^\rs \axisvariancepopulationscaled{\type}{\rb\rs}.
\]
Define
\[
    \hat{G}_{\type} = \alpha^\rb_{\type}\axisvarianceestimatescaled{\type}{\rb} + \alpha^\rs_{\type}\axisvarianceestimatescaled{\type}{\rs} + \alpha^\rb_{\type}\alpha^\rs_{\type}\axisvarianceestimatescaled{\type}{\rb\rs},
\]
\else 
\(
    \mmv\left( \meanestimate{\type} \right) = \alpha_{\type}^\rb \axisvariancepopulationscaled{\type}{\rb} +\alpha_{\type}^\rs \axisvariancepopulationscaled{\type}{\rs} +  \alpha_{\type}^\rb \alpha_{\type}^\rs \axisvariancepopulationscaled{\type}{\rb\rs}.
\)
Define
\(
    \hat{G}_{\type} = \alpha^\rb_{\type}\axisvarianceestimatescaled{\type}{\rb} + \alpha^\rs_{\type}\axisvarianceestimatescaled{\type}{\rs} + \alpha^\rb_{\type}\alpha^\rs_{\type}\axisvarianceestimatescaled{\type}{\rb\rs},
\)
\fi
and apply the expectation operator, leveraging the results in \cref{lemma:rows,lemma:columns,lemma:cross},
    \begin{align*}
    \mme\left[ \hat{G}_{\type} \right] &=  \alpha^\rb_{\type}\mme\left[\axisvarianceestimatescaled{\type}{\rb}\right] + \alpha^\rs_{\type} \mme\left[\axisvarianceestimatescaled{\type}{\rs}\right] + \alpha^\rb_{\type}\alpha^\rs_{\type} \mme\left[\axisvarianceestimatescaled{\type}{\rb\rs}\right] \\
    &= \alpha^\rb_{\type} \left(\axisvariancepopulationscaled{\type}{\rb} - \mmv\left(\meanestimate{\type}\right) + \axisbias{\type}{\rb} \right) + \alpha^\rs_{\type} \left(\axisvariancepopulationscaled{\type}{\rs} - \mmv\left(\meanestimate{\type}\right) + \axisbias{\type}{\rs} \right) \\
    &+ \alpha^\rb_{\type}\alpha^\rs_{\type} \left( \axisvariancepopulationscaled{\type}{\rb\rs} + \mmv(\meanestimate{\type}) - \axisbias{\type}{\rb} - \axisbias{\type}{\rs} \right).
\end{align*}
\ifarxiv
Rearranging,
\begin{align*}
    \mme\left[ \hat{G}_{\type} \right] = \mmv\left(\meanestimate{\type}\right)\left\{1 - \alpha^\rb_{\type} - \alpha^\rs_{\type} + \alpha^\rb_{\type} \alpha^\rs_{\type}\right\}
    + \alpha^\rb_{\type}\{1-\alpha^\rs_{\type}\} {\axisbias{\type}{\rb}} +  \alpha^\rs_{\type} \left\{1-\alpha^\rb_{\type}\right\}{\axisbias{\type}{\rs}}.
\end{align*}
and, observing that
\[
    \frac{x(1-y)}{1-x-y+xy} = \frac{x(1-y)}{(1-x)(1-y)} = \frac{x}{1-x},
\]
and rescaling the quantity above,
\else 
Rearranging,
\(
    \mme\left[ \hat{G}_{\type} \right] = \mmv\left(\meanestimate{\type}\right)\left\{1 - \alpha^\rb_{\type} - \alpha^\rs_{\type} + \alpha^\rb_{\type} \alpha^\rs_{\type}\right\}
    + \alpha^\rb_{\type}\{1-\alpha^\rs_{\type}\} {\axisbias{\type}{\rb}} +  \alpha^\rs_{\type} \left\{1-\alpha^\rb_{\type}\right\}{\axisbias{\type}{\rs}},
\)
and, observing that
\(
    \frac{x(1-y)}{1-x-y+xy} = \frac{x(1-y)}{(1-x)(1-y)} = \frac{x}{1-x},
\)
and rescaling the quantity above,
\fi 
\begin{align*}
    \frac{ \mme\left[ \hat{G}_{\type} \right] }{1 - \alpha^\rb_{\type} - \alpha^\rs_{\type} + \alpha^\rb_{\type} \alpha^\rs_{\type} } 
    &= 
    \mmv\left(\meanestimate{\type}\right) + \frac{\alpha^\rb_{\type}}{1-\alpha^\rb_{\type}}{\axisbias{\type}{\rb}} + \frac{\alpha^\rs_{\type}}{1- \alpha_{\type}^\rs}{\axisbias{\type}{\rs}}.
\end{align*}
We now leverage standard results to obtain unbiased estimates for $\axisbias{\type}{\rb}, \axisbias{\type}{\rs}$.
\ifarxiv
First, the variance of the row-mean estimate follows from \cref{lemma:pure_experiments}:
\else
By standard results on single randomized experiments \citep[Theorems 2.1, 2.2, 2.4]{cochran1977sampling1}:
\fi
\begin{align}
    \mmv\left(\axismeanestimate{i}{\rb}(\type)\right) 
    &= \mme\left[\left(\axismeanestimate{i}{\rb}(\type) - \axismeanpopulation{i}{\rb}(\type)\right)^2\right] 
    = \frac{\left[\frac{1}{J-1} \sum_{j=1}^J \left\{ y_{i,j}(\type) - \axismeanpopulation{i}{\rb}(\type) \right\}^2\right]}{\frac{J_{\type}J}{J-J_{\type}}}, \label{eq:variance2_row}
\end{align}
where \cref{eq:variance2_row} is implied by standard results in sampling theory: in a SMRD we can see each row $i$ as its own population with mean $\axismeanpopulation{i}{\rb}$ and corresponding estimate  $\axismeanestimate{i}{\rb}$. Then, for those rows which feature at least two columns of type $\type$, we can provide an unbiased estimate of the variance term in \cref{eq:variance2_row}. Define the sample estimate
\[
    \widehat{\mmv}\left(\axismeanestimate{i}{\rb}(\type)\right):= \frac{J-J_{\type}}{J_{\type}}\frac{1}{J} \left[ \frac{1}{J_{\type}-1} \sum_{j \in \JJ_{\type}} \left\{ y_{i,j}(\type) - \axismeanestimate{i}{\rb}(\type)\right\}^2 \right].
\]
\ifarxiv
From \cref{eq:expectations_single_axis},
\else
From \citep[Theorem 2.4]{cochran1977sampling1}
\fi
\[
    \mme\left[ \frac{1}{J_{\type}-1} \sum_{j \in \JJ_{\type}} \left\{ y_{i,j}(\type) - \axismeanestimate{i}{\rb}(\type)\right\}^2 \right] = \frac{1}{J-1} \sum_{j=1}^J \left\{ y_{i,j}(\type) - \axismeanpopulation{i}{\rb}(\type) \right\}^2,
\]
which directly implies that
\ifarxiv
\[
    \mme\left[ \widehat{\mmv}\left(\axismeanestimate{i}{\rb}(\type)\right) \right] = \mmv\left(\axismeanestimate{i}{\rb}(\type)\right).
\]
Averaging these estimates over the rows, 
\[
    \axisbiasestimate{\type}{\rb} = \frac{1}{I_{\type}} \sum_{i \in \II_{\type}}  \widehat{\mmv}\left(\axismeanestimate{i}{\rb}(\type)\right),
\]
\else 
\(
    \mme\left[ \widehat{\mmv}\left(\axismeanestimate{i}{\rb}(\type)\right) \right] = \mmv\left(\axismeanestimate{i}{\rb}(\type)\right).
\)
Averaging these estimates over the rows, 
\(
    \axisbiasestimate{\type}{\rb} = \frac{1}{I_{\type}} \sum_{i \in \II_{\type}}  \widehat{\mmv}\left(\axismeanestimate{i}{\rb}(\type)\right),
\)
\fi
satisfying
\[
    \mme\left[ \axisbiasestimate{\type}{\rb} \right] = \mme\left[ \frac{1}{I_{\type}}\sum_{i \in \II_{\type}}  \widehat{\mmv}\left(\axismeanestimate{i}{\rb}(\type)\right) \right] = \frac{1}{I} \binom{J_{\type}}{J}^{-1} \sum_{\JJ_{\type}} \sum_{i=1}^I \left\{ \axismeanestimate{i}{\rb}(\type) - \axismeanpopulation{i}{\rb}(\type)\right\} ^2 =: \axisbias{\type}{\rb}.
\]
Symmetrically for the {\seller}s, 
\ifarxiv
\begin{align*}
    \mmv\left(\axismeanestimate{j}{\rs}(\type)\right) &= \mme\left[\left(\axismeanestimate{j}{\rs}(\type) - \axismeanpopulation{j}{\rs}(\type)\right)^2\right] 
    = \frac{I-I_{\type}}{I_{\type}} \frac{1}{I} \frac{1}{I-1} \sum_{i=1}^I(y_{i,j}(\type) - \axismeanpopulation{j}{\rs}(\type))^2,
\end{align*}
then
\[
    \widehat{\mmv}\left(\axismeanestimate{j}{\rs}(\type)\right):= \frac{I-I_{\type}}{I_{\type}}\frac{1}{I_{\type}-1}\frac{1}{I} \sum_{i \in \II_{\type}} \left(y_{i,j}(\type) - \axismeanpopulation{j}{\rs}(\type)\right)^2.
\]
It holds
\[
    \mme\left[ \widehat{\mmv}\left(\axismeanestimate{j}{\rs}(\type)\right) \right] = \mmv\left(\axismeanestimate{j}{\rs}(\type)\right).
\]
Average these estimates over the columns, 
\[
    \axisbiasestimate{\type}{\rs} = \frac{1}{J_{\type}} \sum_{j \in \JJ_{\type}}  \widehat{\mmv}\left(\axismeanestimate{j}{\rs}(\type)\right),
    \quad \text{satisfying} \quad 
    \mme\left[ \axisbiasestimate{\type}{\rs} \right] = \axisbias{\type}{\rs}.
\]
Therefore,
\[
    \typevariance_{\type} = \hat{G}_{\type} - \frac{\alpha^\rb_{\type}}{1-\alpha^\rb_{\type}} \axisbiasestimate{\type}{\rb} - \frac{\alpha^\rs_{\type}}{1-\alpha^\rs_{\type}} \axisbiasestimate{\type}{\rs}
\quad\text{satisfies}\quad
    \mme\left[ \typevariance_{\type} \right] = \mmv\left( \meanestimate{\type} \right).
\]
\else 
define 
\(
    \axisbiasestimate{\type}{\rs} = 
    \frac{1}{J_{\type}} \sum_{j \in \JJ_{\type}}  \widehat{\mmv}\left(\axismeanestimate{j}{\rs}(\type)\right).
\)
It satisfies
\(
    \mme\left[ \axisbiasestimate{\type}{\rs} \right] = \axisbias{\type}{\rs}.
\)
Therefore,
\(
    \typevariance_{\type} = \hat{G}_{\type} - \frac{\alpha^\rb_{\type}}{1-\alpha^\rb_{\type}} \axisbiasestimate{\type}{\rb} - \frac{\alpha^\rs_{\type}}{1-\alpha^\rs_{\type}} \axisbiasestimate{\type}{\rs}
\quad\text{satisfies}\quad
    \mme\left[ \typevariance_{\type} \right] = \mmv\left( \meanestimate{\type} \right).
\)
\fi 
\end{proof}
\end{theorem}
\begin{theorem}[Already \cref{thm:variance_bounds} in the main paper] \label{proof_thm:variance_bounds}
    Under the assumptions of \cref{thm:sample_variance_avg_effs} a conservative estimator for $ \mmv({{\hattauspillb}})$ is:
    \ifarxiv
    \[
        \widehat{\mmv}^{\rm hi}(\hattauspillb) 
        := 
        2 
        \left\{ 
            \widehat{\Sigma}_{\icb} + \widehat{\Sigma}_{\ccc} 
        \right\}
        .
    \] 
    \else 
    \(
        \widehat{\mmv}^{\rm hi}(\hattauspillb) 
        := 
        2 
        \left\{ 
            \widehat{\Sigma}_{\icb} + \widehat{\Sigma}_{\ccc} 
        \right\}
        .
    \)
    \fi 
\end{theorem}
\begin{proof}[Proof of \cref{thm:variance_bounds}]
    Recall that $\hattauspillb:=\meanestimate{\icb}-\meanestimate{\ccc}$, so that 
    \begin{equation}
        \mmv\left[{{\hattauspillb}}\right]=
        \mmv\left[\meanestimate{\icb}\right] 
        + \mmv\left[\meanestimate{\ccc}\right] 
        - 2 \mmc(\icb, \ccc). \label{eq:exp_var}
    \end{equation}
    We have unbiased estimators $\hat{\Sigma}_{\icb}$ for $\mmv(\meanestimate{\icb})$ and 
    $\hat{\Sigma}_{\ccc}$ for $\mmv(\meanestimate{\ccc})$.
    To obtain a conservative variance estimator for $\hattauspillb$, it remains for us to find a conservative estimator for the covariance term 
    $\mmc(\icb, \ccc)$. Letting $\hat{\Delta}_\type:=\meanestimate{\type} - \meanpopulation{\type}$, we use Young's (AM-GM) inequality as follows:
    \begin{equation}
        \mmc(\icb, \ccc) 
        \le 
        |\mme[\hat{\Delta}_\icb\hat{\Delta}_\ccc]|
        \le 
        \mme[|\hat{\Delta}_\icb| |\hat{\Delta}_\ccc|] 
        \le \frac{\mme[\hat{\Delta}_\icb^2]}{2} \frac{\mme[\hat{\Delta}_\ccc^2]}{2} 
        = \frac{\Sigma_\icb}{2} + \frac{\Sigma_\ccc}{2}.
    \end{equation}
    Therefore, we can obtain from the sample a conservative estimator for the covariance via
    \begin{equation}
        \mme\left[\frac{\hat{\Sigma}_\icb}{2} + \frac{\hat{\Sigma}_\ccc}{2} \right] \ge | \mmc(\icb, \ccc) | . \label{eq:AMGM_conservative}
    \end{equation}
    Finally, defining 
    \ifarxiv
    \[
        \widehat{\mmv}^{\rm hi}(\hattauspillb) := 
        \widehat{\Sigma}({\icb}) + 
        \widehat{\Sigma}({\ccc}) + 2 \left( \frac{\hat{\Sigma}_\icb}{2} + \frac{\hat{\Sigma}_\ccc}{2} \right)
        =
        2 \left( \widehat{\Sigma}({\icb}) + 
        \widehat{\Sigma}({\ccc}) \right)
    \]
    \else 
    \(
        \widehat{\mmv}^{\rm hi}(\hattauspillb) := 
        \widehat{\Sigma}({\icb}) + 
        \widehat{\Sigma}({\ccc}) + 2 \left( \frac{\hat{\Sigma}_\icb}{2} + \frac{\hat{\Sigma}_\ccc}{2} \right)
        =
        2 \left( \widehat{\Sigma}({\icb}) + 
        \widehat{\Sigma}({\ccc}) \right)
    \)
    \fi 
    and applying the expectation operator to each term we prove our thesis:
    \[
        \mme\left[\widehat{\mmv}^{\rm hi}(\hattauspillb)\right] 
        \ge 
        \mmv\left[\meanestimate{\icb}\right] 
        + \mmv\left[\meanestimate{\ccc}\right] 
        - 2 \mmc(\icb, \ccc)
        = 
        \mmv({{\hattauspillb}}). 
    \]
\end{proof}

\begin{lemma} \label{thm:general_variance_sample_causal}
    Let ${\tau}(\coefvec)$ be as in \cref{eq:causal_estimands}, and $\hat{\tau}(\coefvec)$ be its estimator counterpart as per \cref{eq:linear_estimator}. 
    It holds
	\(
        \mme\left[ \hat{\tau}(\coefvec) \right] = {\tau}(\coefvec). 
        \)
    Extending \cref{thm:sample_variance_spillover} yields a conservative estimator  of 
    $\mmv\left( \hat{\tau}(\coefvec)\right)$ via
    \(
    	\widehat{\mmv}^{\rm hi}\left( \hat{\tau}(\coefvec)\right) = 
	\sum_{\type \in \types} \beta_\type^2 \hat{\Sigma}_\type + 
	\sum_{\type \neq \type'}  \beta_\type \beta_{\type'} \left( { \hat{\Sigma}_\type} +  { \hat{\Sigma}_{\type'}} \right).
    \)
\end{lemma}
\begin{proof}
    Unbiasedness of $ \hat{\tau}(\coefvec) $ follows directly from linearity of the expectation and \cref{lemma:unbiasedness}.
    The variance of $\hat{\tau}(\coefvec)$ es given by:
    \begin{equation}
        \mmv\left( \hat{\tau}(\coefvec) \right) 
        = \sum_{\type \in \types} \beta_{\type}^2 
        \mmv\left( \meanestimate{\type} \right) 
        + 2  \sum_{\type \neq \type'}
        \beta_{\type}  \beta_{\type'}
         \mmc\left( 
        \meanestimate{\type},
        \meanestimate{\type'}
        \right). \label{eq:general_variance}
    \end{equation}
    Plug-in estimates $\hat{\Sigma}_\type$ are unbiased for $ \mmv\left( \meanestimate{\type} \right)$  as per \cref{thm:sample_variance_avg_effs}.
    The covariance terms can be bounded as in \cref{eq:AMGM_conservative} ---
    \(
    	| \mmc(\type, \type') | \le \frac{1}{2} \mme[\hat{\Sigma}_\type + \hat{\Sigma}_{\type'}]
    \), yielding the result.
\end{proof}
\subsection{Probability limit}  \label{app_sec:proof_plim}
\ifresponse
    \begin{customthm}{4.7}{A.18}
    Consider any sequence of SMRDs in which $I,J \uparrow \infty$, where the local interference assumption holds, and which satisfy \cref{ass:regularity}.
    Let $\hat{\tau}(\coefvec)$ be the linear estimator introduced in \cref{thm:spillover_unbiasedness}, and let $\widehat{\mmv}^{\rm hi}\left( \hat{\tau}(\coefvec)\right)$ be its conservative variance estimator given in \cref{thm:variance_bounds}.
    Then, if $I^{-2} + J^{-2} = o\left(\mme\left\{ \widehat{\mmv}[ \hat{\tau}(\coefvec)]\right\}\right)$, we have
    \[
    	\frac{\widehat{\mmv}^{\rm hi}\left( \hat{\tau}(\coefvec)\right)}{\mme\left\{\widehat{\mmv}^{\rm hi}\left( \hat{\tau}(\coefvec)\right)\right\}} = 1 + o_p(1).
    \]

\end{customthm}

\else
\begin{theorem} \label{eq:main_theorem}
Consider any sequence of SMRDs in which $I,J \uparrow \infty$, where the local interference assumption holds, and which satisfy \cref{ass:regularity}.
Let $\hat{\tau}(\coefvec)$ be the linear estimator introduced in \cref{thm:spillover_unbiasedness}, and let $\widehat{\mmv}^{\rm hi}\left( \hat{\tau}(\coefvec)\right)$ be its conservative variance estimator given in \cref{thm:variance_bounds}.
Then, if $I^{-2} + J^{-2} = o\left(\mme\left\{ \widehat{\mmv}[ \hat{\tau}(\coefvec)]\right\}\right)$, we have
\[
	\frac{\widehat{\mmv}^{\rm hi}\left( \hat{\tau}(\coefvec)\right)}{\mme\left\{\widehat{\mmv}^{\rm hi}\left( \hat{\tau}(\coefvec)\right)\right\}} = 1 + o_p(1).
\]

\end{theorem}
\fi

\begin{proof}
    By the continuous mapping theorem, given the characterization of $\widehat{\mmv}^{\rm hi}\left( \hat{\tau}(\coefvec)\right)$ in \cref{thm:general_variance_sample_causal}, it suffices to consider the case where $\hat{\tau}(\coefvec) = \meanestimate{\type}$, i.e.~where $\coefvec$ is a standard basis vector in $\mathbb{R}^4$.
    In this case, \cref{thm:sample_variance_avg_effs} shows that $\widehat{\mmv}^{\rm hi}\left( \hat{\tau}(\coefvec)\right)=\mmv\left(\meanestimate{\type}\right)$ is unbiased.

\ifresponse  
    We first recall important definitions and notation.
The (population) averages outcome for each {\buyer} and each {\seller} for a given type $\type$ are
\begin{equation*}
    \axismeanpopulation{i}{\rb}(\type) 
    :=\frac{1}{J}\sum_{j=1}^J y_{ij}(\type) 
    \qquad\text{and}\qquad
    \axismeanpopulation{j}{\rs}(\type) 
    :=\frac{1}{I}\sum_{i=1}^I y_{ij}(\type). 
\end{equation*}
The deviations from population averages  for {\buyer} $i$, {\seller} $j$, and interactions $(i,j)$ are:
\begin{equation*} 
    \axisvariationmeanpopulation{i}{\rb}(\type) 
    := 
        \axismeanpopulation{i}{\rb}(\type)-\meanpopulation{\type}
        ;
    \quad  
    \axisvariationmeanpopulation{j}{\rs}(\type) 
    :=
        \axismeanpopulation{j}{\rs}(\type)-\meanpopulation{\type};
    \quad 
	\axisvariationmeanpopulation{ij}{\rb\rs}(\type) 
	:= 
        y_{ij}(\type)
        -\axismeanpopulation{i}{\rb}(\type)
        -\axismeanpopulation{j}{\rs}(\type)
        +\meanpopulation{\type}.
\end{equation*}
The population variances for each type $\type$ at the \buyer, \seller, and interaction level are:
\[ 
    \axisvariancepopulation{\type}{\rb}
    := \frac{\sum_{i=1}^I\left[\axisvariationmeanpopulation{i}{\rb}(\type)\right]^2}{I}; 
    \;
    \axisvariancepopulation{\type}{\rs}
    := \frac{\sum_{j=1}^J\left[\axisvariationmeanpopulation{j}{\rs}(\type)\right]^2}{J};
    \;
    \axisvariancepopulation{\type}{\rb\rs}
    ~:=~
    \frac{\sum_{i=1}^I 
    \sum_{j=1}^J \left[
    \axisvariationmeanpopulation{ij}{\rb\rs}(\type)
    \right]^2}{IJ}.
\]
\else
\fi
Given weights $\alpha_\type^{\rb}$ and $\alpha_\type^{\rs}$ defined in \cref{eq:alpha_weights}, by \cref{thm:covariance_characterization} (see also \cref{thm:covariance_characterization_app}), 
\begin{align}\label{eq:pop-var}
    \Sigma_\type \coloneqq \mmv\left(\meanestimate{\type}\right)
    =\alpha^{\rb}_{\type}\axisvariancepopulation{\type}{\rb} 
    + \alpha^{\rs}_{\type} 
    \axisvariancepopulation{\type}{\rs}
    + \alpha^{\rb}_{\type}\alpha^{\rs}_{\type}\axisvariancepopulation{\type}{\rb\rs} .
\end{align}
Using the facts that $\alpha_\type^{\rb} = O(I^{-1}) $ and $\alpha_\type^{\rs} = O(J^{-1})$, we may note that as $I,J \uparrow \infty$,
\(
        1 - \alpha^\rb_{\type} - \alpha^\rs_{\type} + \alpha^\rb_{\type} \alpha^\rs_{\type}
        \sim 
        1,
        \quad
         1 - \alpha^{\rs}_{\type} \sim 1, 
         \quad 
         1 - \alpha^{\rb}_{\type} \sim 1.
\)
Hence $\hat{\Sigma}_\type$ simplifies asymptotically as:
\begin{align}
    \hat{\Sigma}_{\type}
    &\sim
    \underbrace{\alpha^{\rb}_{\type}\axisvarianceestimate{\type}{\rb}}_{(\mathsf{a.1})} 
    +
    \underbrace{
    \alpha^{\rs}_{\type}\axisvarianceestimate{\type}{\rs}}_{(\mathsf{a.2})}
    +
    \underbrace{\alpha^{\rb}_{\type}\alpha^{\rs}_{\type}\axisvarianceestimate{\type}{\rb\rs}}_{(\mathsf{b})}
    -
    \underbrace{
    \frac{1}{J^2I} 
    \left[
    \frac{1}{I_\type}
    \sum_{i \in \II_{\type}}
    \hat{v}_{\type, i}^{\rb}
    \right]
    -\frac{1}{I^2J} 
    \left[
    \frac{1}{J_\type}
    \sum_{j \in \JJ_{\type}}
    \hat{v}_{\type, j}^{\rs}
    \right]
    }_{(\mathsf{c})}.
    \label{eq:main_var_decomp}
\end{align}
Here, we have adopted the notation 
\begin{equation}
    \hat{v}_{\type, i}^{\rb}:=
    \frac{1}{J_\type}
    \sum_{j \in \JJ_{\type}}
    {
        \left(y_{i,j}(\type) 
        - \axismeanestimate{i}{\rb}(\type)\right)^2
    },
    \quad
    \text{and}
    \quad 
    \hat{v}_{\type, j}^{\rs}:=
    \frac{1}{I_\type}
    \sum_{i \in \II_{\type}}
    {
        \left(y_{i,j}(\type) 
        - \axismeanestimate{j}{\rs}(\type)\right)^2
    }. \label{eq:plim_var_axis}
\end{equation}
By boundedness \cref{eq:assumption_b} (b),  $|\hat{v}_{\type, i}^{\rb}| \le 4C_2^2$ and $|\hat{v}_{\type, j}^{\rs}| \le 4C_2^2$.
We conclude immediately that $(\mathsf{c})$ is $O(I^{-2}J^{-1} + I^{-1}J^{-2})$ almost surely. 
Thus, combining \cref{eq:pop-var,eq:main_var_decomp},
\[
    \frac{\hat{\Sigma}_{\type}}{\Sigma_\type} = 
    1 + 
    \frac{
    	\alpha^{\rb}_{\type}(\axisvarianceestimate{\type}{\rb}-\sigma^\rb_\type) 
	+ \alpha^{\rs}_{\type}(\axisvarianceestimate{\type}{\rs}-\sigma^\rs_\type) 
	+ \alpha^{\rb}_{\type}\alpha^{\rs}_{\type}(\axisvarianceestimate{\type}{\rb\rs}-\sigma_\type^{\rb\rs}) 
	+ O(I^{-2}J^{-1} + I^{-1}J^{-2})
	}{
	\Sigma_\type
	}.
\]
Then, to show the thesis it suffices for us to bound $(\mathsf{a.1})$, $(\mathsf{a.2})$ and $(\mathsf{b})$ in \cref{eq:main_var_decomp} above. 

By \Cref{lem:variance_estimator_consistency_buyer} along with the facts that $\Sigma_\type \ge \sigma_\type^\rb$ and $\alpha_\type^\rb = O(1/I)$, 
\[
    \alpha^{\rb}_{\type}\left(\axisvarianceestimate{\type}{\rb}-\sigma^\rb_\type\right) = O_p\left(I^{-1}\sqrt{\Sigma_\type[I^{-1}+ J^{-1}]} + (IJ)^{-1}\right).
\]
Analogously, \Cref{lem:variance_estimator_consistency_seller} together with $\Sigma_\type \ge \sigma_\type^\rs$ and $\alpha_\type^\rs = O(1/J)$ gives
\[
    \alpha^{\rs}_{\type}\left(\axisvarianceestimate{\type}{\rs}-\sigma^\rs_\type\right) = O_p\left(J^{-1}\sqrt{\Sigma_\type[I^{-1}+ J^{-1}]} + (IJ)^{-1}\right).
\]
Lastly, \Cref{lem:interaction_variance_consistency} together with $\Sigma_\type \ge \sigma_\type^\rs+\sigma_\type^\rb$ and $\alpha^{\rb}_{\type}\alpha^{\rs}_{\type} = O(I^{-1}J^{-1})$ gives 
\[
	\alpha^{\rb}_{\type}\alpha^{\rs}_{\type}\left(\axisvarianceestimate{\type}{\rb\rs}-\sigma_\type^{\rb\rs}\right) = O_p\left((IJ)^{-1}\sqrt{\Sigma_\type[I^{-1}+ J^{-1}]}  + (IJ)^{-1}[I^{-1}+ J^{-1}] \right).
\]
Omitting lower-order terms and simplifying fractions, we arrive at 
\[
	\frac{\hat{\Sigma}_{\type}}{\Sigma_\type} = 1 + O_p\left( \sqrt{\frac{[I^{-1}+J^{-1}]^3}{\Sigma_\type}} + \frac{(IJ)^{-1}}{\Sigma_\type} \right) 
	= {1 + o_p(1)},
\]
{where the last equality holds because AM-GM ensures $(IJ)^{-1} \le \frac{1}{2}[I^{-2}+ J^{-2}]$}.
\end{proof}

\revision{Given a parameter space $U$ along with random variables $\{X_u:u \in U\}$ and real numbers $\{t_u:u \in U\}$, in this section we write $X_u = \Opfin(t_u)$ if $\{X_u/t_u: u \in U\}$ is tight: $\sup_{u \in U} \mathbb{P}(|X_u/t_u| > r) \downarrow 0$ as $r \uparrow \infty$  \citep{billingsley2008probability}. This immediately implies the usual, sequential definition: given a sequence of elements $u_n \in U$ such that $X_{u_n} = \Opfin(t_{u_n})$, it follows immediately that $X_{u_n} = O_p(t_{u_n})$ in the usual sense, meaning that the sequence ${X_{u_n}/t_{u_n}}$ is tight: $\sup_n \mathbb{P}(|X_{u_n}/t_{u_n}| > r) \downarrow 0$ as $r \uparrow \infty$.
We typically omit reference to the parameter space $U$ as it will be clear from context.
}

\begin{lemma}[Single randomized convergence]\label{lem:single-randomized-redux}
    Let $a_1, a_2, \ldots, a_I$ be bounded real numbers, $|a_i| \le M$ for all $1 \le i \le I$.
    Write $\bar a = I^{-1}\sum_{i=1}^I a_i$ and $\sigma^2_a = I^{-1}\sum_{i=1}^I (a_i - \bar a)^2$. 
    Then, 
    \begin{align}
        \frac{1}{I_\type}\sum_{i \in \II_\type} \left(a_i - \bar a\right)^2 - \sigma^2_a = \Opfin\left(\sqrt{M^2\sigma^2_a/I_\type}\right) 
        \quad\text{and}\quad
        \frac{1}{I_\type}\sum_{i \in \II_\type} a_i- \bar a =  \Opfin\left(\sqrt{\sigma^2_a/I_\type}\right). \label{eq:single-variance-bias2}
    \end{align}
\end{lemma}
\begin{proof}
 \Cref{eq:single-variance-bias2} is taken from the proof of \citet[Proposition 1]{li2017general_app}.
 \ifarxiv 
 In particular, the variance of both sums are bounded there, and \Cref{eq:single-variance-bias2} then follows by Chebyshev's inequality.
 \fi
\end{proof}
\begin{lemma} \label{lem:subg-max} 
Let $X_1, X_2, \ldots, X_n$ be random variables with the following tail bound property: for all $p \in (0,1)$ and all $i \in \{1,2,...,n\}$,
\(
    \mathbb{P}\left(|X_i| > D\sqrt{\log(2/p)}\right) \leq p, 
\)
where $D > 0$ is a constant.
Then for a fixed probability $\eta \in (0,1)$,
\ifarxiv
\[
    \mathbb{P}\left(\max_{1 \leq i \leq n} |X_i| > D\sqrt{\log(2n/\eta)}\right) \leq \eta.
\]
\else
\(
    \mathbb{P}\left(\max_{1 \leq i \leq n} |X_i| > D\sqrt{\log(2n/\eta)}\right) \leq \eta.
\)
\fi
\end{lemma}

\begin{proof} 
Set $p = \eta/n$, then
\ifarxiv
\[ 
    \mathbb{P}\left(|X_i| > D\sqrt{\log(n/\eta)}\right) \leq \eta/n,
\]
\else
\(
    \mathbb{P}\left(|X_i| > D\sqrt{\log(n/\eta)}\right) \leq \eta/n,
\)
\fi
 $\forall i \in [n]$.
By the union bound, for a fixed probability $\eta \in (0,1)$,
\(
    \mathbb{P}\left(\max_{1 \leq i \leq n} |X_i| > D\sqrt{\log(2n/\eta)}\right) \leq \eta.
\)
\end{proof}

In what follows, we consider a bounded array of real numbers $A=(a_{ij})_{i \in [I], j \in [J]}$ such that for all $(i,j)$, $|a_{ij}| \le M$. For this array, we write:
\[
    \bar{\bar{a}} = (IJ)^{-1}\sum_{i=1}^I\sum_{j=1}^J a_{ij},
    \quad
    \text{and}
    \quad
    \bar a_i = J^{-1}\sum_{j=1}^J a_{ij},
    \quad
    \text{and}
    \quad
    \sigma^2_{\bar a} = I^{-1}\sum_{i=1}^I (\bar a_i - \bar{\bar a})^2.
\]

\begin{lemma}\label{lem:double-error-terms}
Let $\II_\type$ be a random selection of $ I_\type \in \{2,\dots, I-2\}$ indices (and symmetrically $\JJ_\type$ a selection of a random selection of $ J_\type \in \{2,\dots, J-2\}$ indices). It holds:
\begin{align}
    \frac{1}{I_\type J_\type} \sum_{i \in \II_\type} \sum_{j \in \JJ_\type} (a_{ij} - \bar{\bar{a}}) = \Opfin\left(M \left[ \sqrt{\frac{\sigma_{\bar a}^2}{I_\type}} + \sqrt{\frac{\log I}{I_\type J_\type }}
    \right] 
    \right). \label{eq:double-bias-bound}
\end{align}
\end{lemma}
\begin{proof}
We write the left-hand side of \cref{eq:double-bias-bound} as
\begin{equation}
    \frac{1}{I_\type J_\type} \sum_{i \in \II_\type} \sum_{j \in \JJ_\type} (a_{ij} - \bar{\bar{a}}) 
    =
    \frac{1}{I_\type}\sum_{i \in \II_\type} 
    \left[
    \left(
    \frac{1}{J_\type} \sum_{j \in \JJ_\type} a_{ij} - \bar{a}_i
    \right) + (\bar{a}_i - \bar{\bar{a}}) 
    \right]
    =:
    \frac{1}{I_\type}\sum_{i \in \II_\type} 
    \left[\epsilon_i + (\bar{a}_i - \bar{\bar{a}})\right].
    \label{eq:error_rewrite}
\end{equation}
Here we have defined \(\epsilon_i := \frac{1}{J_\type} \sum_{j \in \JJ_\type} a_{ij} - \bar{a}_i,\) which depends only on the seller randomization. 
The second summand on the right-hand side of \cref{eq:error_rewrite} is $\Opfin(\sqrt{\sigma_{\bar a}^2/I_\type})$ by \cref{eq:single-variance-bias2}. 
For the first summand, we begin by bounding $\tilde M = \max_{i \in [I]} |\epsilon_i|$, which again depends only upon the seller randomization.
By \cref{lemma:talagrand-sum}, which shows concentration of single-randomized sums, we have 
\begin{equation}\label{eq:epsilon-bound}
     \bb{P}\{|\epsilon_i| > CMJ_\type^{-1/2}\sqrt{\log(2/\eta)}\} \le \eta.
 \end{equation} 
 By \cref{lem:subg-max}, we then find that\begin{equation}\label{eq:epsilon-max-bound}
     \bb{P}
     \left\{\tilde M > CMJ_\type^{-1/2}\sqrt{\log(2I/\eta)}\right\} 
     \le \eta.
 \end{equation}
Let 
	$\mathcal{A}_1 = \left\{\left|\frac{1}{I_\type} \sum_{i\in\II_\type} \epsilon_i \right| > C \tilde M I_\type^{-1/2}\sqrt{\log(2/\eta)}\right\}$,
and
	$\mathcal{A}_2:=\{\tilde M > CMJ_\type^{-1/2}\sqrt{\log(2I/\eta)}\}$,
for $\eta > 0$. Then,
\begin{inumerate}
    \item
    \Cref{lemma:talagrand-sum} (conditional on seller randomization) implies 
    \(
        \mathbb{P} \left\{ \mathcal{A}_1 \right\} \le \eta,
    \)
    and
    \item \cref{eq:epsilon-max-bound} implies
    \(
        \mathbb{P}\{ \mathcal{A}_2 \} \le \eta.
    \)
\end{inumerate}
By  the union bound,
\(
    \mathbb{P}(\mathcal{A}_1^c \cup \mathcal{A}_1^c) \leq \mathbb{P}(\mathcal{A}_1^c) + \mathbb{P}(\mathcal{A}_1^c) \leq 2\eta.
\)
Both bounds hold simultaneously with probability at least $1-2\eta$. When  both hold:
\ifarxiv
\begin{align*}
    \left|\frac{1}{I_\type} \sum_{i\in\II_\type} \epsilon_i\right| 
    &\leq C \tilde{M} I_\type^{-1/2}\sqrt{\log(2/\eta)} \\
    &\leq C \left(CMJ_\type^{-1/2}\sqrt{\log(2I/\eta)}\right) I_\type^{-1/2}\sqrt{\log(2/\eta)} \\
    &= C^2 M (I_\type J_\type)^{-1/2} \sqrt{\log(2I/\eta) \cdot \log(2/\eta)}
\end{align*}
\else
\[
    \left|\frac{1}{I_\type} \sum_{i\in\II_\type} \epsilon_i\right| 
    = C^2 M (I_\type J_\type)^{-1/2} \sqrt{\log(2I/\eta) \cdot \log(2/\eta)}
\]
\fi
For any fixed $\eta$, $\sqrt{\log(2/\eta)} = O(1)$ and $\sqrt{\log(2I/\eta)} = O(\sqrt{\log I})$, giving us
\[
     \frac{1}{I_\type J_\type} \sum_{i \in \II_\type} \sum_{j \in \JJ_\type} (a_{ij} - \bar{\bar{a}}) = \frac{1}{I_\type} \sum_{i\in\II_\type} \epsilon_i {+ \frac{1}{I_\type}\sum_{i \in \II_\type}\bar{a}_i - \bar{\bar{a}}} = \Opfin\left(M \cdot \sqrt{\frac{\log I}{I_\type J_\type}}  \textcolor{black}{+ \sqrt{\frac{\sigma^2_{\bar a}}{I_\type}} } \right).  
\]

\end{proof}

\begin{lemma} \label{lem:double-error-terms_var}
Under the same assumptions of \cref{lem:double-error-terms},
     \begin{equation}
        \frac{1}{I_\type } \sum_{i \in \II_\type} \left( \frac{1}{J_\type} \sum_{j \in \JJ_\type} a_{ij} - \bar{a}_i\right)^2 = \Opfin\left(\frac{M^2}{J_\type} \right) \label{eq:double-var-bound}.
    \end{equation}
\end{lemma}
\begin{proof}
We use facts about Orlicz norms, collected in \cref{defn:orlicz}. 
To prove \cref{eq:double-var-bound}, note that \cref{eq:epsilon-bound} implies $\|\epsilon_i\|_{\psi_2} = O(MJ_\type^{-1/2})$.
So $\|\epsilon_i^2\|_{\psi_1} = \|\epsilon_i\|_{\psi_2}^2 = O(M^2J_\type^{-1})$. 
Given any buyer assignment via $\II_\type$, we have $\|I_\type^{-1}\sum_{i \in \II_\type} \epsilon_i^2\|_{\psi_1} \le I_\type^{-1}\sum_{i \in \II_\type} \|\epsilon_i^2\|_{\psi_1} = O(M^2J_\type^{-1})$ by Jensen's inequality. 
Thus, conditional upon any seller assignment, the left-hand side of \eqref{eq:double-var-bound} is $\Opfin(M^2J_\type^{-1})$, so its marginal distribution is also $\Opfin(M^2J_\type^{-1})$.
\end{proof}

\subsection*{Bounding key terms}

\begin{lemma} \label{lem:variance_estimator_consistency_buyer}
Consider a sequence of SMRDs with sample sizes $I, J \uparrow \infty$ where \cref{sutva_local} and \cref{ass:regularity} hold. Then,
\begin{equation}\label{eq:sigma-b-error}
    \axisvarianceestimatescaled{\type}{\rb} = \sigma_\type^\rb + O_p\left(\sqrt{\sigma_\type^\rb[I_\type^{-1}+ J_\type^{-1}]} + J_\type^{-1}\right).
\end{equation}
\end{lemma}

\begin{proof}

We decompose $\axisvarianceestimatescaled{\type}{\rb}$ as in Lemma A.12,
\begin{align}
    \axisvarianceestimatescaled{\type}{\rb} &= \frac{1}{I_{\type}}\sum_{i\in\II_{\type}} \left(\axismeanestimate{i}{\rb}(\type) \pm \meanpopulation{\type} - \meanestimate{\type} \right)^2 
    = 
    \underbrace{\frac{1}{I_{\type}} \sum_{i \in \II_{\type}} \left[\axismeanestimate{i}{\rb}(\type) - \meanpopulation{\type}\right]^2}_{(\mathsf{a.1.1})} 
    - \underbrace{\left[\meanestimate{\type} - \meanpopulation{\type}\right]^2}_{(\mathsf{a.1.2})}. \nonumber
\end{align}
We analyze $(\mathsf{a.1.1})$ and $(\mathsf{a.1.2})$ separately.
\paragraph{Bounding $(\mathsf{a.1.1})$:} 
The term $(\mathsf{a.1.1})$ can be decomposed as 
\begin{align}
    (\mathsf{a.1.1}) &= 
    \frac{1}{I_{\type}} 
    \sum_{i \in \II_{\type}} 
        \left[
            \axismeanestimate{i}{\rb}(\type) 
            \pm \axismeanpopulation{i}{\rb}(\type)
            - \meanpopulation{\type}
        \right]^2 \nonumber 
    \\
    &
    = 
    \frac{1}{I_{\type}} \left[
        \sum_{i \in \II_{\type}}        
            \left[
                \axisvariationmeanpopulation{i}{\rb}(\type)
            \right]^2 
            + 
            2
            \sum_{i \in \II_{\type}} 
                \axisvariationmeanpopulation{i}{\rb}(\type)
            \epsilon_i^\rb
            + 
            \sum_{i \in \II_{\type}} \left(\epsilon_i^\rb(\type)\right)^2
            \right],
             \label{eq:to_expect}
\end{align}
\ifresponse
where $\epsilon_i^\rb(\type) = \axismeanestimate{i}{\rb}(\type) - \axismeanpopulation{i}{\rb}(\type)$ represents the sampling error in row $i$, 
and $\axisvariationmeanpopulation{i}{\rb}(\type)$ was defined at the beginning of the proof.
\else
where $\epsilon_i^\rb(\type) = \axismeanestimate{i}{\rb}(\type) - \axismeanpopulation{i}{\rb}(\type)$ represents the sampling error in row $i$.
\fi
By \cref{eq:single-variance-bias2} in \Cref{lem:single-randomized-redux} and \cref{eq:double-var-bound} in \Cref{lem:double-error-terms} , with $a_{ij} = y_{ij}(\type)$ and $M = C_2$, we have 
\begin{equation}\label{eq:row_var_new}
\frac{1}{I_{\type}}  \sum_{i \in \II_{\type}}        
            \left[
                \axisvariationmeanpopulation{i}{\rb}(\type)
            \right]^2  = \sigma_\type^\rb + O_p\left(\sqrt{C_2^2 \sigma_\type^\rb/I_\type}\right), \quad \frac{1}{I_{\type}} 
            \sum_{i \in \II_{\type}} \left(\epsilon_i^\rb(\type)\right)^2 = O_p(J_\type^{-1}).
\end{equation}
By Cauchy-Schwarz,
\(
	\frac{1}{I_{\type}} 
            \sum_{i \in \II_{\type}} 
            \left[
                \axisvariationmeanpopulation{i}{\rb}(\type)
            \right]
            \left(\epsilon_i^\rb\right) \le \left( \frac{1}{I_{\type}}  \sum_{i \in \II_{\type}}        
            \left[
                \axisvariationmeanpopulation{i}{\rb}(\type)
            \right]^2\right)^{1/2} 
            \left(\frac{1}{I_{\type}}  \sum_{i \in \II_{\type}} \left(\epsilon_i^\rb(\type)\right)^2\right)^{1/2},
\)
bounding the cross term in \cref{eq:to_expect}.
Substituting and absorbing lower order terms
\begin{equation}
    (\mathsf{a.1.1}) = \sigma_\type^\rb + O_p\left(\sqrt{\sigma_\type^\rb[I_\type^{-1}+ J_\type^{-1}]} + 1/J_\type\right)
    \label{eq:a11_bound}
\end{equation}

\ifarxiv
\paragraph{Bounding $(\mathsf{a.1.2})$:} Applying
\else
To bound $(\mathsf{a.1.2})$ we apply
\fi 
\cref{eq:double-bias-bound}, with $(a_{ij}) = [y_{ij}(\type)]$ and $M = C_2$, we have
\begin{align}\label{eq:mean-var-new}
    \meanestimate{\type} - \meanpopulation{\type} 
    &
    = 
    O_p\left(\sqrt{\frac{\sigma_\type^\rb}{I_\type}} + \sqrt{\frac{\log I}{I_\type J_\type }}\right).
\end{align}
Taking the square of \cref{eq:mean-var-new}, combining with $(\mathsf{a.1.2})$, suppressing the dependence upon $C_2$, removing lower-order terms and combining with \cref{eq:a11_bound}, we obtain  \cref{eq:sigma-b-error}.
\end{proof}

\begin{lemma} \label{lem:variance_estimator_consistency_seller}
Under the assumptions of \cref{lem:variance_estimator_consistency_buyer},
\begin{equation}
\label{eq:sigma-s-error}
    \axisvarianceestimatescaled{\type}{\rs} = \sigma_\type^\rs + O_p\left(\sqrt{\sigma_\type^\rs[I_\type^{-1}+ J_\type^{-1}]} + I_\type^{-1}\right).
\end{equation}

\end{lemma}
\begin{proof}
Symmetric to the above.
\end{proof}

\begin{lemma} \label{lem:interaction_variance_consistency}
Under the assumptions of \cref{lem:variance_estimator_consistency_buyer},
\begin{equation*}\label{eq:sigma-bs-error}
\axisvarianceestimate{\type}{\rb\rs} = \axisvariancepopulation{\type}{\rb\rs} + O_p\left(\sqrt{ [\sigma_\type^\rs+\sigma_\type^\rb][I_\type^{-1}+ J_\type^{-1}]}  + [I_\type^{-1}+ J_\type^{-1}] \right)
\end{equation*}
\end{lemma}

\begin{proof}
From the decomposition provided in \cref{lemma:cross}, we have:
\begin{align*}
      \axisvarianceestimate{\type}{\rb\rs}
      &=
      \overbrace{
      		 \sum_{i \in \II_{\type}}\sum_{j \in \JJ_{\type}} \frac{\left(y_{i,j}(\type) - \meanpopulation{\type}\right)^2}{{I_{\type}J_{\type}}}}
      ^{(\mathsf{b.1})} 
      - 
      \overbrace{
      \sum_{i \in \II_{\type}}\frac{ \left[ \axismeanestimate{i}{\rb}(\type) - \meanpopulation{\type}\right]^2}{{I_{\type}}}}
      ^{(\mathsf{b.2})} 
    - \overbrace{
     \sum_{j \in \JJ_{\type}} 
     \frac{
     		\left( \axismeanestimate{j}{\rs}(\type) - \meanpopulation{\type}\right)^2
	   }{J_\type}
	}
    ^{(\mathsf{b.3})} 
    + 
    \overbrace{\left(\meanestimate{\type} - \meanpopulation{\type}\right)^2}^{(\mathsf{b.4})}.
\end{align*}
We will analyze each term separately and establish their convergence properties.
\paragraph{Bounding $(\mathsf{b.1})$} 
Applying \cref{eq:double-var-bound} 
with $a_{ij} = \delta_{ij}(\type)^2$ and $M = C_2^2$, we have 
\begin{align*}
    (\mathsf{b.1}) &= \frac{1}{I_{\type}J_{\type}} \sum_{i \in \II_{\type}}\sum_{j \in \JJ_{\type}} \delta_{ij}(\type)^2 
    = \frac{1}{IJ}\sum_{i=1}^I\sum_{j=1}^J\delta_{ij}(\type)^2 + 
    O_p\left( \sqrt{\frac{\sigma_\type^\rb}{I_\type}} +
    \sqrt{\frac{\log I}{I_\type J_\type }} \right).
\end{align*}
By \cref{eq:var_dMRD_MV}, we know:
\(
    \frac{1}{IJ}\sum_{i=1}^I\sum_{j=1}^J\delta_{ij}(\type)^2 = \axisvariancepopulation{\type}{\rb} + \axisvariancepopulation{\type}{\rs} + \axisvariancepopulation{\type}{\rb\rs},
\)
hence
\begin{equation}
    (\mathsf{b.1}) = \axisvariancepopulation{\type}{\rb} + \axisvariancepopulation{\type}{\rs} + \axisvariancepopulation{\type}{\rb\rs} + 
    O_p\left(\sqrt{\frac{\sigma_\type^\rb}{I_\type}} +\sqrt{\frac{\log I}{I_\type J_\type }} \right).
\end{equation}

\paragraph{Bounding $(\mathsf{b.2})$:}
From the analysis of $(\mathsf{a.1.1})$ in \cref{eq:a11_bound}, we know that
\begin{equation*}
    (\mathsf{b.2}) = \sigma_\type^\rb + O_p\left(\sqrt{ \sigma_\type^\rb[I_\type^{-1}+ J_\type^{-1}]} + 1/J_\type\right).
\end{equation*}
\paragraph{Bounding $(\mathsf{b.3})$:}
Symmetrical to $(\mathsf{b.2})$, 
applying the same analysis to columns instead of rows,
\begin{equation*}
    (\mathsf{b.3}) = \sigma_\type^\rs + O_p\left(\sqrt{ \sigma_\type^\rs[I_\type^{-1}+ J_\type^{-1}]} + 1/I_\type\right).
\end{equation*}
\paragraph{Bounding $(\mathsf{b.4})$:}
Taking the square of \cref{eq:mean-var-new} and using $(a+b)^2 \le 2(a^2 + b^2)$ (by am-gm),
\begin{equation*}
    (\mathsf{b.4}) = \left(\meanestimate{\type} - \meanpopulation{\type}\right)^2 = O_p\left(\frac{\sigma_\type^\rb}{I_\type}+\frac{\log I}{I_\type J_\type }\right).
\end{equation*}
\paragraph{Combining $(\mathsf{b})$ terms.}
Using $\sqrt{a} + \sqrt{b} \le \sqrt{a + b}$ and removing lower-order terms,
\begin{align*}
    \axisvarianceestimate{\type}{\rb\rs} &= \axisvariancepopulation{\type}{\rb\rs} + O_p\left(\sqrt{ [\sigma_\type^\rs+\sigma_\type^\rb][I_\type^{-1}+ J_\type^{-1}]}  + [I_\type^{-1}+ J_\type^{-1}] \right). 
\end{align*}
\end{proof}

\section{Proof of Theorem \ref{thm:clt}} \label{proof:clt}
\ifarxiv
In this section we prove \Cref{thm:clt}. 
We consider an SDRD with two populations (buyers, sellers), and a binary treatment assignment at the (buyer-seller) pair level. 
A fixed proportion of buyers $p^\rb:=I_\ct/I \in(0,1)$ are assigned at random $W_i^B=1$, which makes them eligible for treatment. The remaining $I_\rc = I - I_\rt$ are assigned $W_i^B=0$. Similarly,  $p^\rs:=J_\rt/J \in (0,1)$ of sellers are assigned $W_j^\rs=1$ (i.e., are eligible), while the remaining $J_\rc = J - J_\rt$ sellers are assigned $W_j^\rs=0$. Treatment is assigned via $W_{ij} = W_i^\rb W_j^\rs$. 
\paragraph{Remarks on notation}
Recall that $[n]:=\{1,\ldots,n\}$.
Given a $k$-dimensional vector $\bm{a}$, $\|\bm{a}\|_2 = a_1^2+ \ldots + a_k^2$ denotes its 2-norm and $\|A\|_{op} = \max_{\|x\|_2=1} \|Ax\|_2$  its operator norm.
We often use $I_0$ in place of $I_{\rc}$ and $I_1$ in place of $I_{\rt}$ (symmetrically, $J_0$ for $J_{\rc}$ and $J_1$ for $J_{\rt}$) whenever it is more natural to do so.
Last, $C, C', C'', \dots$ denote absolute positive constants whose value may change from line to line.  
Under local interference (\Cref{sutva_local}), as per \Cref{lemma:local_interference}, each buyer-seller pair $(i,j)$ has only $4$ potential outcomes:
    $Y_{i,j} = Y_{i,j}(\type)$, where $\type \in \types$.
We denote with $\type_{i,j}$ the type of the pair $(i,j)$, as per \Cref{eq:types}. 
\else
\fi
\paragraph{Goal of the proof} For a fixed size of the two populations, $\bm{N} = (I, J)$ and for $\type \in \types$, 
we aim to prove joint normality of linear combinations of the random variables
\[
    \meanestimate{\type} = \meanestimate{\type, \bm{N}} = \frac{1}{N_\type}\sum_{i=1}^{I} \sum_{j=1}^{J} Y_{ij}(\type)\mathbf{1}\{\type_{i,j} = \type \},
\]
where $N_\type = \sum_{i,j} 1(\type_{ij} = \type)$. We write 
\begin{equation}
    \widehat{\bm{\tau}} =  \left[\meanestimate{\ccc}, \meanestimate{\icb}, \meanestimate{\ics}, \meanestimate{\ctt}\right]^\top\equiv
    \left[
    \hat{\tau}_{\ccc}, 
    \hat{\tau}_{\icb}, 
    \hat{\tau}_{\ics}, 
    \hat{\tau}_{\ctt}
    \right]^\top \label{eq:hat_tau}
\end{equation}
to denote the (random) vector of group averages, and $\bm{\tau}$ to denote its population counterpart,
\begin{equation}
    {\bm{\tau}} =  
    \left[
    \meanpopulation{\ccc}, 
    \meanpopulation{\icb}, 
    \meanpopulation{\ics}, 
    \meanpopulation{\ctt}
    \right]^\top\equiv
    \left[
    {\tau}({\ccc}), 
    {\tau}({\icb}), 
    {\tau}({\ics}), 
    {\tau}({\ctt})
    \right]^\top. \label{eq:tau_population}    
\end{equation}
Roughly, our proof technique is as follows:
\begin{description}[wide]
    \item[Step 1] In \Cref{sec:proof_clt_conditional} we show that under fixed sellers' assignments $W_{j}^{\rs}$, for $j \in [J]$, 
    standard results of \citet{li2017general_app, shi2022berry_app} extend to SMRDs: a ``conditional'' CLT for $\widehat{\bm{\tau}}$ holds, with the limiting distribution parameterized by the sellers' assignments.
    \item[Step 2] In \Cref{sec:concentration} we prove that when considering the random assignment of sellers, the mean of the limiting distribution in Step 1 is itself normally distributed. Meanwhile, its variance is close to a fixed, deterministic value, independent of both assignments. 
    \item[Step 3] Last, we combine these show in \Cref{sec:tying} that the marginal distribution of $\hat{\tau}(\coefvec)$ is also approximately Gaussian.
\end{description}


\subsection{Step 1: a conditional CLT} \label{sec:proof_clt_conditional}
We now show that conditional upon the seller assignments $(W_j^S)$, we can derive central limit theorems for the MRD estimators in \Cref{sec:estimation} that mirror those known for estimators in standard, single randomized experiments \citep{li2017general_app,shi2022berry_app}. 
Let $\Pi$ denote a uniform random permutation of the seller indices $[J]$
\ifarxiv
, i.e.\ a map $\Pi:[J]\to [J]$ such that $\{\Pi(1), \dots, \Pi(J)\} = [J]$
\fi. Without loss of generality, we can suppose treatment labels are generated according to $W_{j}^S = \mathbbm{1}\{\Pi^{-1}(j) > J_{0}\}$. We proceed in this section by conditioning on a particular realization $\Pi = \pi$.

For a fixed permutation $\pi$, let $\JJ_{0}^{\pi}:=\{\pi(1),\ldots, \pi(J_0)\} \subset [J]$ be the set of sellers with $W_j^{\rs}=0$ and let $\JJ_{1}^{\pi}:=\{\pi(J_0+1),\ldots, \pi(J)\} = [J] \setminus \JJ_{0}^{\pi}$ be the seller indices with $W_j^{\rs}=1$. 
Conditional upon $\Pi = \pi$, each buyer  $i$ has the following ``realizable'' potential outcomes:
\begin{itemize}
    \item $\axismeanfixedindex{i}{\JJ_{0}^{\pi}}{\rb}(\ccc) = \frac{1}{J_0} \sum_{j=1}^{J_0} Y_{i,\pi(j)}(\ccc)$
    and
    $\axismeanfixedindex{i}{\JJ_{0}^{\pi}}{\rb}(\icb) = \frac{1}{J_0} \sum_{j=1}^{J_0} Y_{i,\pi(j)}(\icb)$, which average the unit-level potential outcomes of interactions $(i,j)$ for sellers with $W_j^{\rs} = 0$;
    \item $\axismeanfixedindex{i}{\JJ_{1}^{\pi}}{\rb}(\ics) = \frac{1}{J_1} \sum_{j=J_0+1}^{J} Y_{i,\pi(j)}(\ics)$
    and
    $\axismeanfixedindex{i}{\JJ_{1}^{\pi}}{\rb}(\ctt) = \frac{1}{J_1} \sum_{j=J_0+1}^{J} Y_{i,\pi(j)}(\ctt)$, which average the unit-level potential outcomes of interactions $(i,j)$ for sellers with $W_j^{\rs} = 1$.
\end{itemize}
We can then view our SDRD as a standard randomized experiment with $I$ units, where each buyer $i$ can be thought of as having  potential outcomes corresponding to the above:
\begin{equation}
    \overline{\bm{y}}^{\rb}_{i, \pi}(0) = 
        \begin{pmatrix}
            \axismeanfixedindex{i}{\JJ_{0}^{\pi}}{\rb}(\ccc) &
            0 &
            \axismeanfixedindex{i}{\JJ_{1}^{\pi}}{\rb}(\ics) &
            0
        \end{pmatrix}^\top, 
    \;
    \overline{\bm{y}}^{\rb}_{i, \pi}(1) = 
        \begin{pmatrix}
            0 &
            \axismeanfixedindex{i}{\JJ_{0}^{\pi}}{\rb}(\icb) &
            0 &
            \axismeanfixedindex{i}{\JJ_{1}^{\pi}}{\rb}(\ctt)       
        \end{pmatrix}^\top. \label{eq:mean_unit}
\end{equation}
Notice that we have two potential outcomes for each buyer (since each buyer can either be assigned $W_i^{\rb} = 0$ or $W_i^{\rb} = 1$), and these potential outcomes are vectors in $\mmr^4$ (there is one potential outcome for each type $\type \in \types$). 
The population averages of these vectors are defined as:
\begin{align*}
    \meanonefixedindexvec{\bullet}{\pi}(0) = 
        \begin{pmatrix}
             \frac{1}{I}  \sum_{i=1}^I \axismeanfixedindex{i}{\JJ_{0}^{\pi}}{\rb}(\ccc) &
            0  &
            \frac{1}{I} \sum_{i=1}^I \axismeanfixedindex{i}{\JJ_{1}^{\pi}}{\rb}(\ics)   &
            0
        \end{pmatrix}^\top
        \eqqcolon 
        \begin{pmatrix}
            \meanonefixedindex{\bullet}{\JJ_{0}^{\pi}}(\ccc) &
            0 &
            \meanonefixedindex{\bullet}{\JJ_{1}^{\pi}}(\ics)&
            0
        \end{pmatrix}^\top\!\! ;
        \\
        \label{eq:mean_one_axis} 
    \meanonefixedindexvec{\bullet}{\pi}(1) = 
        \begin{pmatrix}
            0 &
             \frac{1}{I}\sum_{i=1}^I \axismeanfixedindex{i}{\JJ_{0}^{\pi}}{\rb}(\icb) &
            0 &
            \frac{1}{I} \sum_{i=1}^I \axismeanfixedindex{i}{\JJ_{1}^{\pi}}{\rb}(\ctt)   
        \end{pmatrix}^\top
        \eqqcolon
        \begin{pmatrix}
            0 &
            \meanonefixedindex{\bullet}{\JJ_{0}^{\pi}}(\icb) &
            0 &
            \meanonefixedindex{\bullet}{\JJ_{1}^{\pi}}(\ctt)
        \end{pmatrix}^\top. 
\end{align*}
We further define the difference between the outcome at the unit level (\cref{eq:mean_unit}) and the mean across all units (previous display) at the buyer-level treatment $q = 0,1$:
\begin{equation}
    \dot{\bm{\overline{y}}}^{\rb}_{i, \pi}(q) = \overline{\bm{y}}^{\rb}_{i, \pi}(q) - \meanonefixedindexvec{\bullet}{\pi}(q) \in \mmr^4.
\end{equation}
Following \citet{li2017general_app} we define the \buyer{}-level vector of treatment effects $\bm{\tau}_i^{\pi}$
\[
    \bm{\tau}_i^{\pi} = \left[ \overline{\bm{y}}^{\rb}_{i, \pi}(0)+ \overline{\bm{y}}^{\rb}_{i, \pi}(1) \right] = \begin{pmatrix}
            \axismeanfixedindex{i}{\JJ_{0}^{\pi}}{\rb}(\ccc) & 
            \axismeanfixedindex{i}{\JJ_{0}^{\pi}}{\rb}(\icb) &
            \axismeanfixedindex{i}{\JJ_{1}^{\pi}}{\rb}(\ics) &
            \axismeanfixedindex{i}{\JJ_{1}^{\pi}}{\rb}(\ctt)       
        \end{pmatrix}^\top.
\]
In turn, define the $\pi$-conditional population average across all \buyers{}:
\begin{equation}
    \bm{\tau}^{\pi} = \frac{1}{I} \sum_{i=1}^{I} \bm{\tau}_i^{\pi} = \meanonefixedindexvec{\bullet}{\pi}(0) + \meanonefixedindexvec{\bullet}{\pi}(1) = \begin{pmatrix} 
            \meanonefixedindex{\bullet}{\JJ_{0}^{\pi}}(\ccc) &
            \meanonefixedindex{\bullet}{\JJ_{0}^{\pi}}(\icb) &
            \meanonefixedindex{\bullet}{\JJ_{1}^{\pi}}(\ics) &
            \meanonefixedindex{\bullet}{\JJ_{1}^{\pi}}(\ctt)       
        \end{pmatrix}^\top. \label{eq:tau_population_cond}
\end{equation}
We define the centered counterpart of $\bm{\tau}_i^{\pi}$:
\begin{align}
    \dot{\bm{\tau}}_i^{\pi} = {\bm{\tau}}_i^{\pi} - {\bm{\tau}}^{\pi}
    = \begin{psmallmatrix}
            \axismeanfixedindex{i}{\JJ_{0}^{\pi}}{\rb}(\ccc) - \meanonefixedindex{\bullet}{\JJ_{0}^{\pi}}(\ccc) \\ 
            \axismeanfixedindex{i}{\JJ_{0}^{\pi}}{\rb}(\icb) - \meanonefixedindex{\bullet}{\JJ_{0}^{\pi}}(\icb)\\
            \axismeanfixedindex{i}{\JJ_{1}^{\pi}}{\rb}(\ics) - \meanonefixedindex{\bullet}{\JJ_{1}^{\pi}}(\ics)\\
            \axismeanfixedindex{i}{\JJ_{1}^{\pi}}{\rb}(\ctt) - \meanonefixedindex{\bullet}{\JJ_{1}^{\pi}}(\ctt) 
        \end{psmallmatrix}. 
        \label{eq:dot_tau_centered}
\end{align}
Given a (random) assignment of buyers, $W_i^\rb\in\{0,1\}$ for $i=1,\ldots,I$, the natural sample counterpart of $\bm{\tau}^{\pi}$ is $\widehat{\bm{\tau}}^{\pi}$, where we replace each coordinate with the sample mean across units $i$ for which the type was observed, $\estimatefixedindex{\bullet}{\JJ_{\type}^{\pi}}(\type)$:
\begin{equation}
    \widehat{\bm{\tau}}^{\pi} = \begin{psmallmatrix}
        \frac{1}{I_0} \sum_{i:W_i^{\rb}=0}^{I_0} \axismeanfixedindex{i}{\JJ_{0}^{\pi}}{\rb}(\ccc) \\
        \frac{1}{I_1} \sum_{i:W_i^{\rb}=1}^{I_1} \axismeanfixedindex{i}{\JJ_{0}^{\pi}}{\rb}(\icb) \\
        \frac{1}{I_0} \sum_{i:W_i^{\rb}=0}^{I_0} \axismeanfixedindex{i}{\JJ_{1}^{\pi}}{\rb}(\ics) \\
        \frac{1}{I_1} \sum_{i:W_i^{\rb}=1}^{I_1} \axismeanfixedindex{i}{\JJ_{1}^{\pi}}{\rb}(\ctt)
    \end{psmallmatrix} = 
    \begin{psmallmatrix}
        \estimatefixedindex{\bullet}{\JJ_{0}^{\pi}}(\ccc) \\
        \estimatefixedindex{\bullet}{\JJ_{0}^{\pi}}(\icb) \\
        \estimatefixedindex{\bullet}{\JJ_{1}^{\pi}}(\ics) \\
        \estimatefixedindex{\bullet}{\JJ_{1}^{\pi}}(\ctt)
    \end{psmallmatrix}. \label{eq:hat_tau_pi}
\end{equation}
\ifarxiv
The randomness in $\widehat{\bm{\tau}}^{\pi}$ only stems from the assignment of the $I$ buyers via $W_i^{\rb}$. 
\else
\fi
With this characterization in place, we can extend \citet[Theorem 3]{li2017general_app} to our estimator $\widehat{\bm{\tau}}^{\pi}$. 
First, define for $q, r \in \{0,1\}$ the finite population cross-covariance
\begin{align*}
    S^{\pi}_{q,r} &:= \frac{1}{I-1} \sum_{i=1}^I \left\{ \overline{\bm{y}}^{\rb}_{i, \pi}(q) - \meanonefixedindexvec{\bullet}{\pi}(q) \right\} \left\{ \overline{\bm{y}}^{\rb}_{i, \pi}(r) - \meanonefixedindexvec{\bullet}{\pi}(r) \right\}^\top 
    = \frac{1}{I-1} \sum_{i=1}^I \dot{\overline{\bm{y}}}_{i,\pi}^{\rb}(q) \dot{\overline{\bm{y}}}_{i,\pi}^{\rb}(r)^\top,
\end{align*}
and the finite population covariance of the individual effects
\begin{align}
    S_{\widehat{\bm{\tau}}^{\pi} }^2 := \frac{1}{I-1} \sum_{i=1}^I \left\{ \bm{\tau}_i^{\pi} - \bm{\tau}^{\pi} \right\} \left\{ \bm{\tau}_i^{\pi} - \bm{\tau}^{\pi} \right\}^\top. \label{eq:covariance_conditional}
\end{align}

\begin{theorem}[Theorem 3 in  \citet{li2017general_app}] \label{thm:li3} 
    Consider an SDRD under the local interference assumption. Conditionally on the sellers' assignments via $\pi : [J] \to [J]$, the experiment is equivalent to a single-randomized experiment with $I$ units and $Q=2$ treatments indexed by $W_i^{\rb} \in \{0,1\}$ and potential outcomes $\overline{\bm{y}}^{\rb}_{i, \pi}(W_i^{\rb}) \in \mmr^4$.
    The estimator $\widehat{\bm{\tau}}^{\pi}$ is unbiased for $ {\bm{\tau}}^{\pi}$: 
    \[
        \mme\left[\widehat{\bm{\tau}}^{\pi}\right] = {\bm{\tau}}^{\pi}.
    \]
    The covariance of $\widehat{\bm{\tau}}^{\pi}$ is given by
    \begin{equation*}
       V^\pi:= \mmc\left\{  \widehat{\bm{\tau}}^{\pi} \right\} = \sum_{q=1}^Q \frac{1}{I_q} S_{q,q}^{\pi} - \frac{1}{I} S_{{\bm{\tau}}^{\pi}}^{2} ,
    \end{equation*}
    where
    {
    \[
        \frac{S_{0,0}^{\pi}}{I_0} + \frac{S_{1,1}^{\pi}}{I_1} = \sum_{i=1}^I 
        \begin{psmallmatrix}
                \frac{ \left( \t(\ccc) \right)^2}{(I-1) I_0}  
              & 0 
              &   \frac{ \left( \t(\ccc) \t(\ics)  \right)}{(I-1) I_0}
              & 
              0 
              \\
              0 & 
              \frac{ \left( \t(\icb) \right)^2}{(I-1) I_1} 
              & 0 
              & \frac{\left( \t(\icb) \t(\ctt)  \right)}{(I-1) I_1} \\
              \frac{ \left( \t(\ccc) \t(\ics)  \right)}{(I-1) I_0} 
              & 0 
              & \frac{ \left( \t(\ics) \right)^2}{(I-1) I_0}  
              & 
              0 
              \\
              0 & 
            \frac{ \left( \t(\icb) \t(\ctt)  \right)}{(I-1) I_1} 
              & 0 
              & 
            \frac{ \left( \t(\ctt) \right)^2}{(I-1) I_1}
        \end{psmallmatrix},
    \]
    and
\[
    \frac{S_{\bm{\tau}^{\pi}}^2}{I} = \sum_{i=1}^I
    \begin{psmallmatrix}
            \frac{ \left( \t(\ccc) \right)^2}{(I-1) I}  
          & \frac{\left( \t(\ccc) \t(\icb)  \right)}{(I-1) I}
          &   \frac{ \left( \t(\ccc) \t(\ics)  \right)}{(I-1) I}
          & 
          \frac{\left( \t(\ccc) \t(\ctt)  \right)}{(I-1) I} 
          \\
          \frac{ \left( \t(\ccc) \t(\icb)  \right)}{(I-1) I} & 
          \frac{ \left( \t(\icb) \right)^2}{(I-1) I} 
          & \frac{\sum_{i=1}^I \left( \t(\icb) \t(\ics)  \right)}{(I-1) I} 
          & \frac{\left( \t(\icb) \t(\ctt)  \right)}{(I-1) I} \\
          \frac{ \left( \t(\ccc) \t(\ics)  \right)}{(I-1) I} 
          & \frac{ \left( \t(\icb) \t(\ics)  \right)}{(I-1) I} 
          &   \frac{ \left( \t(\ics) \right)^2}{(I-1) I}  
          & 
          \frac{\left( \t(\ics) \t(\ctt)  \right)}{(I-1) I} 
          \\
          \frac{ \left( \t(\ccc) \t(\ctt)  \right)}{(I-1) I} & 
        \frac{ \left( \t(\icb) \t(\ctt)  \right)}{(I-1) I} 
          & \frac{ \left( \t(\ics) \t(\ctt)  \right)}{(I-1) I} 
          & 
        \frac{ \left( \t(\ctt) \right)^2}{(I-1) I}
    \end{psmallmatrix}.
\]}
Hence,
\begin{equation}
{
    \mmc\left\{  \widehat{\bm{\tau}}^{\pi} \right\} = 
    \sum_{i=1}^I
    \begin{psmallmatrix}
            \frac{ \frac{I_1}{I_0} \left( \t(\ccc) \right)^2}{(I-1) I}  
          & - \frac{\left( \t(\ccc) \t(\icb)  \right)}{(I-1) I}
          &   \frac{ \frac{I_1}{I_0} \left( \t(\ccc) \t(\ics)  \right)}{(I-1) I}
          & 
          - \frac{\left( \t(\ccc) \t(\ctt)  \right)}{(I-1) I} 
          \\
          - \frac{ \left( \t(\ccc) \t(\icb)  \right)}{(I-1) I} & 
          \frac{ \frac{I_0}{I_1} \left( \t(\icb) \right)^2}{(I-1) I} 
          & - \frac{\left( \t(\icb) \t(\ics)  \right)}{(I-1) I} 
          &  \frac{ \frac{I_0}{I_1} \left( \t(\icb) \t(\ctt)  \right)}{(I-1) I} \\
          \frac{\frac{I_1}{I_0} \left( \t(\ccc) \t(\ics)  \right)}{(I-1) I} 
          & - \frac{ \left( \t(\icb) \t(\ics)  \right)}{(I-1) I} 
          &   \frac{ \frac{I_1}{I_0}\left( \t(\ics) \right)^2}{(I-1) I}  
          & 
          - \frac{\left( \t(\ics) \t(\ctt)  \right)}{(I-1) I} 
          \\
          - \frac{ \left( \t(\ccc) \t(\ctt)  \right)}{(I-1) I} & 
        \frac{\frac{I_0}{I_1} \left( \t(\icb) \t(\ctt)  \right)}{(I-1) I} 
          & - \frac{ \left( \t(\ics) \t(\ctt)  \right)}{(I-1) I} 
          & 
        \frac{\frac{I_0}{I_1} \left( \t(\ctt) \right)^2}{(I-1) I}
    \end{psmallmatrix}.} \label{eq:cov-pi-exact}
\end{equation}
\begin{proof}
    The proof is given in Theorem 3 in  \citet{li2017general_app}.
\end{proof}
\end{theorem}

Theorem 4 in \citet{li2017general_app} provides a  CLT which relies on the existence of an asymptotic limit for $V^\pi$.
Since $V^\pi$ is random in our context and depends upon the finite population, we instead derive a Berry-Esseen type result following \citet{shi2022berry_app}. In what follows,  $(V^{\pi})^{\frac{1}{2}}$ is defined as the symmetric square root of $V^\pi$, and $(V^{\pi})^{-\frac{1}{2}}$ is its pseudoinverse. In particular, we need not assume that $V^\pi$ has full rank. 

\begin{theorem}[Theorem S4 in \citet{shi2022berry_app}] \label{shiS4}
Let $V^{\pi} := \mmc\{\widehat{\bm{\tau}}^{\pi}\}$ as characterized in \eqref{eq:cov-pi-exact}. Then, there exists a universal constant $C$ such that for all $\bm{\alpha} \in \mathbb{R}^4$ with $\|\bm{\alpha}\|_2 = 1$,
\begin{equation}
    \sup_{t \in \mathbb{R}} \left|\mathbb{P}\left\{\bm{\alpha}^\top (V^{\pi})^{-\frac{1}{2}}(\widehat{\bm{\tau}}^{\pi} -  \bm{\tau}^{\pi})> t  \right\} - \Phi(t)\right| \le C\max_{i \in [I]} \max_{q \in \{0,1\}} \frac{|\bm{\alpha}^\top (V^{\pi})^{-\frac{1}{2}} \dot{\bm{\overline{y}}}^{\rb}_{i, \pi}(q)|}{I_q}. \label{eq:shi-ding-beb}
\end{equation}
\begin{proof}
    See Theorem S4 of \citet{shi2022berry_app}.
\end{proof}
\end{theorem}
\ifarxiv
\Cref{shiS4} provides a Berry-Esseen bound for $\widehat{\bm{\tau}}^{\pi}$, where the upper bound depends on both $\dot{\bm{\overline{y}}}^{\rb}_{i, \pi}$ and $V^{\pi}$ (see the right-hand side of \cref{eq:shi-ding-beb}). 
Now, using \cref{eq:assumption_b}(b) of bounded potential outcomes, we state a slightly different form \Cref{shiS4} where the bound does not depend on $\dot{\bm{\overline{y}}}^{\rb}_{i, \pi}$. 
We use the notation introduced in \Cref{eq:linear_estimator}, so that 	$\hat{\tau}^{\pi}(\coefvec) = \sum_{\type \in \types} \hat{\tau}^{\pi}_\type \beta_\type$.
\else
\fi
\begin{lemma} \label{lem:clt-bounded2}
Under the same setting of \Cref{shiS4} and further assuming bounded potential outcomes as per \cref{eq:assumption_b}(b):
\begin{equation}
		\sup_{t \in \mathbb{R}} 
		\left|\mathbb{P}\left\{\frac{\widehat{{\tau}}^{\pi}(\coefvec) -  {\tau}^{\pi}(\coefvec)}{\sqrt{\mmv\{\hat{\tau}^{\pi}(\coefvec)\}}}> t  \right\} - \Phi(t)\right| 
		\le 
		C \frac{C_2}{\min\{I_0, I_1\}} 
	\frac{\|\coefvec\|_2}{\sqrt{\mmv\{\hat{\tau}^{\pi}(\coefvec)\}}}. \label{eq:shi-ding-beb3}
\end{equation}
\begin{proof}
We first consider the case in which $V^\pi$ is invertible.
\ifarxiv
This case contains the main ideas and is technically simpler than the general case. 
\else
\fi

\noindent\textit{Invertible case.} 
Let 
\(
    \bm{\alpha} = (V^{\pi})^{\frac{1}{2}}\coefvec/{\|(V^{\pi})^{\frac{1}{2}}\coefvec\|_2},
\)
so that $\|\bm{\alpha}\|_2 = 1$ by construction. 
Plugging this choice of $\bm{\alpha}$ in \Cref{eq:shi-ding-beb}, \Cref{shiS4}, 
\[
	\sup_{t \in \mathbb{R}} 
            \left|\mathbb{P}\left\{\frac{\coefvec^\top (\widehat{\bm{\tau}}^{\pi} -  \bm{\tau}^{\pi})}{\|(V^{\pi})^{\frac{1}{2}}\coefvec\|_2}> t  \right\} - \Phi(t)\right| 
        \le \frac{C}{\|(V^{\pi})^{\frac{1}{2}}\coefvec\|_2}
            \max_{i \in [I]} \max_{q \in \{0,1\}} 
            \frac{|\coefvec^\top\dot{\bm{\overline{y}}}^{\rb}_{i, \pi}(q)|}{I_q}.
\]
Applying the Cauchy-Schwarz inequality on the right hand side yields
\[
	\sup_{t \in \mathbb{R}} 
            \left|\mathbb{P}\left\{\frac{\coefvec^\top (\widehat{\bm{\tau}}^{\pi} -  \bm{\tau}^{\pi})}{\|(V^{\pi})^{\frac{1}{2}}\coefvec\|_2}> t  \right\} - \Phi(t)\right| 
        \le 
            \frac{C}{\|(V^{\pi})^{\frac{1}{2}}\coefvec\|_2} 
            \|\coefvec \|_2 \max_{i \in [I]} \max_{q \in \{0,1\}} 
                \frac{\|\dot{\bm{\overline{y}}}^{\rb}_{i, \pi}(q)\|_2}{I_q}.
\]
Last, since by \cref{eq:assumption_b} (b), each entry of $\dot{\bm{\overline{y}}}^{\rb}_{i, \pi}(q)$ has absolute value at most $2C_2$, and since there are exactly 2 non-zero entries in each $\dot{\bm{\overline{y}}}^{\rb}_{i, \pi}(q)$ (cf \Cref{eq:mean_unit}), we conclude that $\max_i \max_q \| \dot{\bm{\overline{y}}}^{\rb}_{i, \pi}(q)\|_2 \le \sqrt{2 \times (2C_2)^2} = \sqrt{8} C_2$. Plugging this in, and noting that $\mmv\{\hat{\tau}^{\pi}(\coefvec)\} = \coefvec^\top V^\pi \coefvec$, so that $\sqrt{\mmv\{\hat{\tau}^{\pi}(\coefvec)\}} = \| (V^\pi)^{1/2} \coefvec\|_2$ yields the thesis:
\begin{align*}
		\sup_{t \in \mathbb{R}} \left|\mathbb{P}\left\{\frac{\widehat{{\tau}}^{\pi}(\coefvec) -  {\tau}^{\pi}(\coefvec)}{\sqrt{\mmv\{\hat{\tau}^{\pi}(\coefvec)\}}}> t  \right\} - \Phi(t)\right| &=  \sup_{t \in \mathbb{R}} \left|\mathbb{P}\left\{\frac{\coefvec^\top (\widehat{\bm{\tau}}^{\pi} -  \bm{\tau}^{\pi})}{\|(V^{\pi})^{\frac{1}{2}}\coefvec\|_2}> t  \right\} - \Phi(t)\right| 
\\
	&\le \frac{\sqrt{8} C C_2}{\min\{I_0, I_1\}} 
	\frac{\|\coefvec\|_2}{\sqrt{\mmv\{\hat{\tau}^{\pi}(\coefvec)\}}}.
\end{align*}

\noindent\textit{Non-invertible case.} In case $V^\pi$ is not invertible, \cref{eq:shi-ding-beb} in \Cref{shiS4} instead gives
\[
    \sup_{t \in \mathbb{R}} \left|\mathbb{P}\left\{\frac{\coefvec^\top (\widehat{\bm{\tau}}^{\pi} -  \bm{\tau}^{\pi})}{\|(V^{\pi})^{\frac{1}{2}}\coefvec\|_2}> t  \right\} - \Phi(t)\right| 
    \le \frac{C}{\|(V^{\pi})^{\frac{1}{2}}\coefvec\|_2}\max_{i \in [I]} \max_{q \in \{0,1\}}
    \frac{|\coefvec^\top (V^{\pi})^{\frac{1}{2}}(V^{\pi})^{-\frac{1}{2}}\dot{\bm{\overline{y}}}^{\rb}_{i, \pi}(q)|}{I_q}. 
\]
Now we use Cauchy-Schwarz and the operator norm inequality to bound the righthand side,
\begin{align*}
    |\coefvec^\top(V^{\pi})^{\frac{1}{2}}(V^{\pi})^{-\frac{1}{2}}\dot{\bm{\overline{y}}}^{\rb}_{i, \pi}(q) | 
    &\le \|\coefvec^\top\|_2\|(V^{\pi})^{\frac{1}{2}}(V^{\pi})^{-\frac{1}{2}}\dot{\bm{\overline{y}}}^{\rb}_{i, \pi}(q)\|_2 \\
    &\le \|\coefvec^\top\|_2\|(V^{\pi})^{\frac{1}{2}}(V^{\pi})^{-\frac{1}{2}}\|_{op}\|\dot{\bm{\overline{y}}}^{\rb}_{i, \pi}(q)\|_2 \\
    &\le \|\coefvec^\top\|_2\|\dot{\bm{\overline{y}}}^{\rb}_{i, \pi}(q)\|_2,
\end{align*}
where in the last step we use the fact that $\|(V^{\pi})^{\frac{1}{2}}(V^{\pi})^{-\frac{1}{2}}\|_{op} \le 1$. Thus,
\[
	\sup_{t \in \mathbb{R}} \left|\mathbb{P}\left\{\frac{\coefvec^\top (\widehat{\bm{\tau}}^{\pi} -  \bm{\tau}^{\pi})}{\|(V^{\pi})^{\frac{1}{2}}\coefvec\|_2}> t  \right\} - \Phi(t)\right| \le \frac{C}{\|(V^{\pi})^{\frac{1}{2}}\coefvec\|_2} \|\coefvec^\top \|_2 \max_{i \in [I]} \max_{q \in \{0,1\}} \frac{
	 \|\dot{\bm{\overline{y}}}^{\rb}_{i, \pi}(q)\|_2
	}{I_q}.
\]
The proof then proceeds as in the invertible case.
\end{proof}
\end{lemma}
\ifarxiv
This concludes the first section.
\else
\fi

\subsection{Analysis of conditional mean and covariance} \label{sec:concentration}

As before, let $\Pi: [J] \to [J]$ denote a permutation chosen uniformly at random. In this section we characterize the distribution of the $\Pi$-conditional mean vector $\bm{\tau}^\Pi$, and the $\Pi$-conditional covariance matrix $V^\Pi:=\mmc(\hat{\bm{\tau}}^\Pi)$ introduced in \Cref{sec:proof_clt_conditional}. This allows us to transfer the results \Cref{sec:proof_clt_conditional}
---which depend on the particular seller assignment $\pi$---to the general case of random seller assignment $\Pi$.

We start by recalling standard results on concentration of random permutations in \Cref{sec:perm-concentration}. 
We use these results to characterize the conditional mean and covariance. We characterize the mean $\mme_{\Pi}\left[\bm{\tau}^{\Pi}\right]$ in \Cref{sec:exp_mean} 
\ifarxiv 
and prove concentration of $\bm{\tau}^{\Pi}$ around $\mme_{\Pi}\left[\bm{\tau}^{\Pi}\right]$ in \Cref{sec:conc_mean}
\fi. We then show that $\bm{\tau}^{\Pi}$ is approximately normal in \Cref{sec:conditional-mean-clt}.
\ifarxiv 
We characterize  $\mme_\Pi[V^\Pi]$ in \Cref{sec:exp_cov} and 
\else 
We
\fi 
show that $V^\Pi$ concentrates around $\mme_\Pi[V^\Pi]$ in \Cref{sec:conc_cov}. Finally in \Cref{sec:simplifying-conditional-clt}, we use this concentration to express the conditional CLT (\cref{lem:clt-bounded2}) in a more convenient form. 

\subsubsection{Useful results on concentration for random permutations}\label{sec:perm-concentration}

We first provide some notation. Let $\PP^J$ be the set of permutations of $[J]$. Given two permutations $\pi_1, \pi_2 \in \PP^J$, let $\delta(\pi_1, \pi_2)$ be their convex distance:
\begin{equation}
    \delta(\pi_1, \pi_2) = 
    \sup_{\|a\|_2=1} 
    \sum_{j=1}^J |a_j| \mathbbm{1}\left\{\pi_1(j) \neq \pi_2(j)\right\}. \label{hamming}
\end{equation}
Moreover, given a set $S \subset \PP^J$ and a permutation $\pi \in \PP^J$, with some slight abuse of notation, we let 
\(
    \d(\pi,S) = \inf_{s \in S}\d(\pi,s),
\) 
i.e.\ the distance of $\pi$ to $S \subset \PP^J$ is the distance to the nearest point in $S$.
To establish concentration of $\bm{\tau}^{\Pi}$ around $\mme_{\Pi}\left[\bm{\tau}^{\Pi}\right]$ (and similarly that the covariance $V^\Pi$ concentrates around $\mme_\Pi[V^\Pi]$), we will use an isoperimetric inequality for uniform random permutations, along with a well-known corollary.
In particular, we will reduce the problem of establishing concentration for the conditional mean and variance to that of establishing concentration for suitable $L$-Lipschitz continuous functions of $\Pi$.
Towards that goal, in what follows we let $X: \PP^J \to \mathbb{R}$ denote an $L$-Lipschitz continuous function with respect to the distance $\delta$ defined in \cref{hamming}: there exists some $L>0$ for which $\forall \;\pi_1, \pi_2 \in \PP^J$
\begin{equation}
    |X(\pi_1) - X(\pi_2)| \le L\delta(\pi_1,\pi_2). \label{eq:lipschitz} 
\end{equation}

\begin{lemma}[\citet{talagrand1995concentration}, Theorem 5.1] \label{lem:talagrand-convex-distance}
Let $\Pi: [J] \to [J]$ be a permutation chosen uniformly at random in $\PP^J$. Then for a set $S \subseteq \PP^J$,
    \[
      \bb{P}(\Pi \in S) \mme\left[\beef\exp\left\{\frac{\d(\Pi, S)^2}{16} \right\}\right] \le 1, \label{symmetric-isoperimetric}
    \]
    where we recall that the distance to the set $S$ is defined as $\d(\pi,S) = \inf_{s \in S}\d(\pi,s)$.
  \begin{proof}
      See \citet[Theorem 5.1]{talagrand1995concentration}.
  \end{proof}
  \end{lemma} 

\Cref{lem:talagrand-convex-distance} has the following well-known corollary. 

  \begin{corollary}[Concentration for random permutations] \label{symmetric-concentration}
      Suppose that $X: \PP^J \to \bb{R}$ is $L$-Lispchitz continuous as per \Cref{eq:lipschitz}.
      Let $\Pi \in \PP^J$ be chosen uniformly at random. 
      Then, for $t > 0$,
    \[
            \bb{P}\{ | X(\Pi) - \mme [X(\Pi)] | >  8  L t \} \le 2 e^{-t^2/8}.
    \]
  \end{corollary}

\begin{proof}
\ifarxiv
Let $\nu \in \mathbb{R}$ be the median of $X(\Pi)$ when $\Pi \sim  \Unif(\PP^J)$, i.e. 
\[
    \nu = \{\inf z \in \mathbb{R} \;:\; \bb{P}[X(\Pi) \le z] \ge 1/2 \}.
\]
Let $S \coloneqq \Set[\pi \in \PP^J]{X(\pi) \le \nu}$.
By Markov's inequality, \Cref{symmetric-isoperimetric}, and the fact that $\mathbb{P}(\Pi \in S) \ge 1/2$,
\begin{align}
    \bb{P}\left( L\d(\Pi, S) \ge s\right) 
    &= \bb{P}\left(\exp\left\{\frac {\d(\Pi, S)^2} {16} \right\} \ge \exp\left\{\f{s^2 }{16 {L}^2}\right\}\right) \nonumber \\
    &\le \bb{E}\left[\beef\exp\left\{\f{\d(\Pi, S)^2}{16} \right\}\right]\exp\left\{-\f{s^2}{16 {L}^2}\right\} \nonumber \\
    &\le 2 \exp\left\{-\f{s^2}{16 {L}^2}\right\}. \label{eq:tail_lipschitz}
\end{align}
${L}$-Lipschitz continuity of $X$ with respect to $\d$ implies $|X(\Pi)-\nu | \le  L\d(\Pi,S)$.
Using \Cref{eq:tail_lipschitz} we then can bound the deviations of $X(\Pi)$ from its median:
\begin{equation}
    \bb{P}(X(\Pi) - \nu \ge s) \le \bb{P}({L}\d(\Pi, S) \ge s)  \le 2 \exp\left\{-\f{s^2}{16 {L}^2}\right\},
    \label{eq:tail_proba}
\end{equation}
and symmetrically,
\begin{equation}
    \bb{P}(X(\Pi) - \nu \le -s) \le \bb{P}( - {L}\delta(\Pi, S) \le -s) \le  2 \exp\left\{-\f{s^2}{16 {L}^2}\right\}. \label{eq:tail_proba_right}
\end{equation}
Finally, we transfer this to concentration around the mean of $X(\Pi)$. 
To avoid confusion, and with an exception to our general notation, let $\boldsymbol{\pi} \approx 3.14$ denote the universal constant.
\begin{align}
  \mme[X(\Pi) - \nu] 
  & \le \mme \left[
    \left(X(\Pi)-\nu \right)
    {1}\left\{X(\Pi) > \nu \right\}
    \right] \nonumber 
    = \int_0^\infty \bb{P}\left( X(\Pi) - \nu > t\right)\mathrm{d}t \nonumber \\
  &\le \int_0^\infty 2e^{-t^2/(16 {L}^2)} \mathrm{d}t = \sqrt{16 \boldsymbol{\pi} }  L , \label{eq:median_bound}
\end{align}
where in the first equality we have used the fact that $\left(X(\Pi)-\nu \right) {1}\left\{X(\Pi) > \nu \right\}\ge 0$ is a non-negative random variable (for which the tail probability formula of its expected value holds), and in the last inequality we have used \Cref{eq:tail_proba}; the integral is computed by noting it coincides with that of a scaled Gaussian density. Symmetrically, $\bb{E}[\nu - X(\Pi)] \le \sqrt{16\boldsymbol{\pi} }  L$. 
The two combined yield an upper and lower bound on the mean $\mme[X(\Pi)]$ in terms of the median and the Lipschitz constant:
\begin{equation}
    \mme[X(\Pi)] \le \nu + \sqrt{16 \boldsymbol{\pi} }  L 
    \quad \text{and} \quad 
    \mme[X(\Pi)] \ge \nu - \sqrt{16\boldsymbol{\pi} }  L. \label{eq:mean_median}
\end{equation}

Using the lower bound on $\mme[X(\Pi)]$ in \Cref{eq:mean_median} we obtain
\begin{align}
    \bb{P}\{X(\Pi) - \mme [X(\Pi)] > t  \} &\le \bb{P}\{X(\Pi) - (\nu - \sqrt{16 \boldsymbol{\pi} }  L) > t \}  
    \\ &= \bb{P}\{X(\Pi) - \nu > t - \sqrt{16 \boldsymbol{\pi} }  L \}, \nonumber
    \intertext{and now applying \Cref{eq:tail_proba_right}}
    \bb{P}\{X(\Pi) - \mme [X(\Pi)] > t + \sqrt{16 \boldsymbol{\pi} }  L \} &\le 2 \exp\left\{-\f{t^2}{16  {L}^2} \right\}. \label{eq:conc_right}
\end{align}
Symmetrically, we use the upper bound on $\mme[X(\Pi)]$ in \Cref{eq:mean_median} to obtain
\begin{equation}
    \bb{P}\{X(\Pi) - \mme [X(\Pi)] < - t - \sqrt{16 \boldsymbol{\pi} }  L \} \le  2 \exp\left\{-\f{t^2}{16  {L}^2}\right\}.\label{eq:conc_left}
\end{equation}
Hence, combining \Cref{eq:conc_right,eq:conc_left} via a union bound and choosing $t = \sqrt{16\boldsymbol{\pi}}Lz$,
\[
    \bb{P}\{ | X(\Pi) - \mme [X(\Pi)] | >  (z + 1)\sqrt{16 \boldsymbol{\pi} } L \} \le 4 e^{-\boldsymbol{\pi}z^2} \le 4e^{-z^2}.
\]
Finally, note that for any $u = z+1 \ge 0$, hence for $u \ge 2$, we may rewrite this as 
\[
    \bb{P}\{ | X(\Pi) - \mme [X(\Pi)] | >  u\sqrt{16 \boldsymbol{\pi} } L \} \le 4 e^{-(u-1)^2} \le 4e^{-u^2/4},
\]
since $u^2/(u-1)^2 \le 4$ for $u \ge 2$.
Meanwhile for $0 < u \le 2$ the bound is larger than $1$, hence it holds trivially. Similarly, we simplify $4e^{-u^2/4} \vee 1 \le 2e^{-u^2/8} \vee 1$ and note that $\sqrt{16\boldsymbol{\pi}} \le 8$. 
\else
See \citet{masoero2024multiple}.
\fi 
\end{proof}

Finally, we apply the result to our context. The following lemma allows us to show concentration of sums of potential outcomes under simple random sampling. 

\begin{lemma}\label{lemma:talagrand-sum}
For some $M >0$, let $\boldb=[b_1, \ldots, b_J]^\top \in [-M, M]^J$ be a vector of scalars, and let $\JJ_\type$ be one of the sets $\{\Pi(1), \Pi(2), \ldots, \Pi(J_0)\}$ (if $\type \in \{\ccc, \icb\}$) or $\{\Pi(J_0+1), \Pi(J_0+2), \ldots, \Pi(J)\}$ (if $\type \in \{\ics, \ctt\}$), so that $|\JJ_\type| = J_\type$.
Put
\(X_{\boldb}(\Pi) ~=~ \sum_{j\in\JJ_\type} b_{\Pi(j)},\) for $\Pi \sim \Unif(\PP^J)$. 
Then we have the bound 
\begin{align}
    \mathbb{P}\left\{|X_{\boldb}(\Pi) - \mathbb{E} \left[X_{\boldb}(\Pi)\right] | >  (8M\sqrt{J_\type})t \right\} & \le 2e^{-t^2/8}. \label{eq:subG-perm-tail-bound} 
\end{align}
\end{lemma}

\begin{proof} 
    Without loss of generality, we will consider the case that $\JJ_\type = \{\Pi(1), \Pi(2), \ldots, \Pi(J_0)\}$; the other case is symmetric.
    Put $a_j = J_0^{-1/2}$ for $1 \le j \le J_0$ and $a_j = 0$ for $j > J_0$, and note that $\|\mathbf{a}\|_2 = 1$. We will apply \Cref{symmetric-concentration} using the weights $\mathbf a$.
    For any two permutations $\pi_1, \pi_2 \in \PP^J$,
\begin{align*}
    |X_{\boldb}(\pi_1) - X_{\boldb}(\pi_2)| = \left|  \sum_{j=1}^{J_0} b_{\pi_1(j)} - b_{\pi_2(j)}\right| 
    &\le \sum_{j=1}^{J_0}  |b_{\pi_2(j)}|\mathbbm{1}\{\pi_1(j) \ne \pi_2(j)\}\\
    &= \sum_{j=1}^{J} \sqrt{J_0}|a_j||b_{\pi_2(j)}|\mathbbm{1}\{\pi_1(j) \ne \pi_2(j)\} \\
    &\le \sqrt{J_0}M \sum_{j=1}^{J}   |a_j| \mathbbm{1}\{\pi_1(j) \ne \pi_2(j)\} \\
    &\le \sqrt{J_0} M \delta(\pi_1,\pi_2).
\end{align*}
These steps follow by the triangle inequality, by our choice of $a_j$, by the fact $|b_j| \le M$, and by the definition of $\delta$ given in \Cref{hamming}. 
The inequality \eqref{eq:subG-perm-tail-bound} then follows by \Cref{symmetric-concentration}, as we have just shown that $X_{\bm{b}}$ is $L$-Lipschitz with respect to the convex distance $\delta(\pi_1,\pi_2)$, with $L = \sqrt{J_0}M$.
\end{proof}

Finally, we state a technical lemma which will help us apply \cref{lemma:talagrand-sum} to expressions which depend on potential outcomes $y_{ij}(\type)$ for multiple types $\type$.

\begin{lemma}\label{lem:helpful-rewrite}
    Under \cref{eq:assumption_a}(a), for any $\type \in \types$ and any $b_1(\type), \dots, b_J(\type) \in [-M, M]$, there exists a collection of numbers $\tilde{b}_1(\type), \ldots, \tilde{b}_{J_0}(\type)$ with absolute value at most $2C_1M$, such that for any $\Pi \in \PP^J$,
    \ifarxiv
    \[
        \frac{1}{J_\type} \sum_{j\in \J_\type^\Pi} b_j(\type) = \frac{1}{J_{0}}\sum_{j=1}^{J_0} \tilde b_{\Pi(j)}(\type).
    \]
    \else
    \(
        \frac{1}{J_\type} \sum_{j\in \J_\type^\Pi} b_j(\type) = \frac{1}{J_{0}}\sum_{j=1}^{J_0} \tilde b_{\Pi(j)}(\type).
    \)
    \fi    
    Note, the left-hand side is a sum over $J_\type$ terms.The right-hand side is a sum over $J_0$ terms, irrespective of  $\type$.
\end{lemma}
\begin{proof}
We proceed by cases, first considering $\type \in \{\ccc, \icb\}$ and then $\type \in \{\ics, \ctt\}$.
Let $\Pi \in \PP^J$ be arbitrary. 
Recall that by construction, we have $\J_\type^\Pi = \{\Pi(1), \dots, \Pi(J_0)\}$ if $\type \in \{\ccc,\icb\}$ and $\J_\type^\Pi ~=~ \{\Pi(J_0+1), \dots, \Pi(J)\}$ if $\type \in \{\ics,\ctt\}$.
Then, if $\type \in \{\ccc,\icb\}$:
\[
    \frac{1}{J_\type} \sum_{j\in \J_\type^\Pi} b_j(\type) 
    = \frac{1}{J_0}\sum_{j=1}^{J_0} b_{\Pi(j)}(\type).
\] 
The claim then directly holds by taking $\tilde b_{\Pi(j)}(\type) = b_{\Pi(j)}(\type)$ for all $j$. 

On the other hand, if $\type \in \{\ics,\ctt\}$, letting $\bar{b}(\type):=\sum_{j=1}^J b_j(\type) / J$, we have 
\begin{align*}
    \frac{1}{J_{\type}}\sum_{j\in \J_\type^\Pi} b_j (\type)
    &= \frac{1}{J_1} \sum_{j=J_0 + 1}^{J} b_{\Pi(j)}(\type) 
    = \frac{1}{J_1} \sum_{j=1}^{J} b_{\Pi(j)}(\type) \left[ 1 - \mathbbm{1}(j \le J_0)\right] 
    = \frac{J}{J_1}\bar{b}(\type) - \frac{1}{J_1} \sum_{j=1}^{J_0} b_{\Pi(j)}(\type) \\
    &= \frac{1}{J_0} \sum_{j=1}^{J_0} \left( \frac{J}{J_1} \bar{b}(\type) - \frac{J_0}{J_1} b_{\Pi(j)} (\type)\right).
\end{align*}
We then take $\tilde b_j(\type) = (J/J_1) \bar{b}(\type) - (J_0/J_1) b_{\Pi(j)}(\type)$; in either case,  $|\tilde b_j(\type)| \le 2C_1M$.
\end{proof}
\subsubsection{Computing the expectation of the mean} \label{sec:exp_mean}
Recall the definition of $\bm{\tau}^{\pi}$,   
given in \Cref{eq:tau_population_cond}.
In what follows, with a slight abuse of notation, we let ${\tau}^\pi(\type)$ be the entry of  $\bm{\tau}^\pi$ referring to type $\type \in \types$. We let ${\tau}^\Pi(\type)$ represent the same quantity, now indexed by a random $\Pi\sim\Unif(\PP^{J})$. 

We note that the expectation of ${\tau}^\Pi(\type)$ coincides with the population mean $\meanpopulation{\type}$:
\begin{align*}
    \mme_\Pi\left[{\tau}^\Pi(\type)\right]&=\mme_\Pi \left[\frac{1}{I} \sum_{i=1}^I \frac{1}{J_{\type}} \sum_{j=1}^{J_{\type}}  y_{i,\Pi(j)}(\type)\right] 
    =  \frac{1}{I J_\type} \sum_{i=1}^I  \mme_\Pi \left[ \sum_{j=1}^{J_{\type}} y_{i,\Pi(j)}(\type)\right] 
    \\ &
     = \frac{1}{I J_\type} \sum_{i=1}^I \frac{\binom{J-1}{J_{\type}-1}}{\binom{J}{J_\type}} \sum_{j=1}^J y_{i,j}(\type) 
     = \frac{1}{I J_\type} \sum_{i=1}^I \frac{J_{\type}}{J} \sum_{j=1}^J y_{i,j}(\type) 
     = \meanpopulation{\type}.
    \end{align*}
By linearity of the expectation operator, for any $\coefvec$ it also holds that
\begin{equation}
    \mme_\Pi\left[\tau^\Pi(\coefvec)\right] = \tau(\coefvec). \label{eq:expect_conditional}
\end{equation}

\ifarxiv
\subsubsection{Concentration of the mean} \label{sec:conc_mean}

We now use \cref{lemma:talagrand-sum,lem:helpful-rewrite} to show concentration of $\tau^\Pi(\coefvec)$ around its expectation.

\begin{lemma}\label{lemma-mean-concentration}
Let $\Pi$ be a uniform random permutation of $[J]$. Then under \cref{eq:assumption_a}, which imposes a balanced experiment with bounded potential outcomes, it holds
\begin{equation}
    \mathbb{P} \left(|\tau^{\Pi}(\coefvec)- \tau(\coefvec)| > \left[32 J_0^{-1/2}C_2C_1\right]\|\coefvec\|_2t  \right) \le 2 e^{-t^2/8}. \label{eq:tail_mean}
\end{equation}
\end{lemma}
\begin{proof}
\Cref{eq:tail_mean} is a bound around deviations of $\tau^{\Pi}(\coefvec)$ around its mean, since $\mathbb{E}[\tau^\Pi(\coefvec)] = \tau(\coefvec)$ as per \cref{eq:expect_conditional}. 
To show \cref{eq:tail_mean} we will first prove that $\tau^\Pi(\coefvec) ~=~ J_0^{-1}\sum_{j=1}^{J_0} b_{\Pi(j)}$ for some suitably bounded numbers $(b_1, \ldots, b_J)$ via \cref{lem:helpful-rewrite}, and then conclude using \cref{lemma:talagrand-sum}. 
From the definition of $\tau^\Pi(\coefvec)$,
\begin{align*}
    \tau^\Pi(\coefvec) &= \sum_{\type \in \types} \beta_\type \frac{1}{I}  \frac{1}{J_{\type}} \sum_{i=1}^I \sum_{j\in \J_{\type}^\Pi} y_{i,j} 
    = \sum_{\type \in \types} \beta_\type  \frac{1}{J_{\type}} \sum_{j\in \J_{\type}^\Pi} 
    \left(\frac{1}{I} \sum_{i=1}^Iy_{i,j}(\type)\right) \\
    &= \sum_{\type \in \types} \beta_\type  \frac{1}{J_{\type}} \sum_{j\in \J_{\type}^\Pi} \axismeanpopulation{j}{\rs}(\type)
    .
\end{align*}
Because of the boundedness assumption (b) in \cref{eq:assumption_b}, $|\axismeanpopulation{j}{\rs}(\type)|\le C_2$ for each $j \in [J]$.
By \cref{lem:helpful-rewrite}, we may find numbers $|\tilde b_j(\type)| \le 2C_1C_2$ which allow us to rewrite $\tau^\Pi(\coefvec)$ as
\[
    \tau^\Pi(\coefvec) 
    = \sum_{\type \in \types} \beta_\type  \frac{1}{J_{0}} \sum_{j=1}^{J_{0}} \tilde b_{\Pi(j)}(\type) 
    = \frac{1}{J_{0}} \sum_{j=1}^{J_{0}} \left\{\sum_{\type \in \types} \beta_\type  \tilde b_{\Pi(j)}(\type)\right\}.
\]
By H\"{o}lder's inequality and the fact that $\|v\|_1 ~\le~ \sqrt{4}\|v\|_2$ for $v \in \mathbb{R}^4$, we can further bound the bracketed terms in the equation above as
\[ 
    \left| \sum_{\type \in \types} \beta_\type  \tilde b_{j}(\type)\right| 
    \le \|\coefvec\|_1\max_\type|\tilde{b}_{j}(\type)|
    \le \|\coefvec\|_1 2C_1C_2 
    \le 4C_1C_2\|\coefvec\|_2.
\]
Then by considering the bounded vector $\bm{b}=[b_1,\dots,b_J]$ in which each $b_j = \sum_{\type \in \types} \beta_\type  \tilde b_{j}(\type) \le 4C_1C_2\|\coefvec\|_2$, we can apply \cref{lemma:talagrand-sum} to $\tau^\Pi(\coefvec)$ (in turn, \cref{eq:tail_mean} follows):
\[
    \tau^\Pi(\coefvec) = \frac{1}{J_0}\sum_{j=1}^{J_0} b_j \le 4C_1C_2\|\coefvec\|_2.
\]
\end{proof}
\fi

\subsubsection{A CLT for the conditional mean}\label{sec:conditional-mean-clt}

Finally, we note an unconditional normal approximation for $\bm{\tau}^\Pi = \mathbb{E}[\bm{\tau} \mid \Pi]$ which mirrors \Cref{lem:clt-bounded2} above. 
\revision{
It follows from the observation that $\bm{\tau}^\Pi$ is the standard mean estimator corresponding to a completely randomized experiment in which $J_1$ out of $J$ units are treated, with the following vector-valued potential outcomes: 
\begin{align*}
    \overline{\bm{y}}^{\rs}_{j}(0) &= 
        \begin{pmatrix}
              \axismeanpopulation{j}{\rs}(\ccc) &
            0  &
            \axismeanpopulation{j}{\rs}(\ics)   &
            0
        \end{pmatrix}^\top;
        \quad 
        \overline{\bm{y}}^{\rs}_{j}(1) =
         \begin{pmatrix}
            0 &
             \axismeanpopulation{j}{\rs}(\icb) &
            0 &
           \axismeanpopulation{j}{\rs}(\ctt)   
        \end{pmatrix}^\top.
\end{align*}}

\begin{lemma}\label{lem:clt-conditional-mean}
Assuming bounded potential outcomes as per \cref{eq:assumption_b}(b):
\begin{equation}
		\sup_{t \in \mathbb{R}} 
		\left|\mathbb{P}\left\{\frac{{{\tau}}^{\Pi}(\coefvec) -  {\tau}(\coefvec)}{\sqrt{\mmv\{{\tau}^{\Pi}(\coefvec)\}}} \le t  \right\} - \Phi(t)\right| 
		\le 
		\frac{C C_2}{\min\{J_0, J_1\}} 
	\frac{\|\coefvec\|_2}{\sqrt{\mmv\{{\tau}^{\Pi}(\coefvec)\}} }. 
\end{equation}
\begin{proof} Identical to \Cref{lem:clt-bounded2}.
\end{proof}
\end{lemma}

In general, the fluctuations of the conditional mean $\tau^\Pi(\coefvec)$ may not be negligible; this may occur, e.g.,  if the number of sellers is small. \Cref{lem:clt-conditional-mean} shows that $\tau^\Pi(\coefvec)$ is itself approximately Gaussian, so that we can derive the CLT for $\hattau(\coefvec)$ despite this possibility. Our use of \cref{lem:clt-conditional-mean} in this capacity was initially suggested by  \citet{sudijono2025private}.

\ifarxiv
\subsubsection{Computing the expectation of the covariance} \label{sec:exp_cov}

Mirroring \Cref{sec:exp_mean} we now compute $\mme_{\Pi}\left[ \mmc\left\{  \widehat{\bm{\tau}}^{\Pi} \right\} \right]$ --- the \emph{expectation} of $\mmc\left\{  \widehat{\bm{\tau}}^{\Pi} \right\}$ over the uniform measure on the space of permutations $\PP^J$ of $[J]$.
\begin{lemma} \label{lemma:exp_cov}
    For $\type,\type' \in \types$ and $i, i' \in [I]$ define
    \[
        \rho_{i,i'}^{\type,\type'} := \mme_{\Pi}\left[ \axismeanfixedindex{i}{\JJ_{\type}^{\Pi}}{\rb}(\type) \axismeanfixedindex{i'}{\JJ_{\type'}^{\Pi}}{\rb}(\type') \right].
    \]
    Now, we characterize the entry $\mmc\left\{  \widehat{\bm{\tau}}^{\Pi} \right\}_{\type,\type'}$ associated with types $\type, \type'$:
    \[
        \mme_{\Pi}\left[ \mmc\left\{  \widehat{\bm{\tau}}^{\Pi} \right\}  \right]_{\type,\type'} = \kappa_{\type,\type'} \left\{\sum_{i=1}^I  \rho_{i, i}^{\type,\type'}   - \frac{1}{I} \sum_{i=1}^I \sum_{i'=1}^I  \rho_{i, i'}^{\type,\type'} \right\},
    \]
where 
\[
    \kappa_{\type,\type'} = 
    \begin{cases} 
        1 &\mbox{ if } \type = \type' \\
        \frac{I - I_{\type}}{I_\type} \frac{1}{I(I-1)} &\mbox{ if } \type \neq \type' \land \JJ_{\type}^{\Pi} = \JJ_{\type'}^{\Pi} \text{(e.g., $\type =\ccc, \type' =\ics$)} \\
        - \frac{I - I_{\type}}{I_\type} \frac{1}{I(I-1)} &\mbox{ if } \type \neq \type' \land \JJ_{\type}^{\Pi} \neq \JJ_{\type'}^{\Pi} \text{(e.g., $\type = \ccc, \type'=\icb$)}.
    \end{cases}
\]
\begin{proof}
The main contribution of this proof is just to make the coeficients $\rho_{i,i'}^{\type,\type'}$ explicit; to do so, we specialize them into three different cases: (i) $\type = \type'$, (ii) $\type \neq \type'$ and $\JJ_\type^{\Pi} = \JJ_{\type'}^{\Pi}$, and (iii) $\type \neq \type'$ and $\JJ_\type^{\Pi} \neq \JJ_{\type'}^{\Pi}$. 

\paragraph{(i) $\type = \type'$.} We have
\begin{align}
    \rho_{i,i'}^{\type,\type}  &= \mme\left[ \axismeanfixedindex{i}{\JJ_{\type}^{\Pi}}{\rb}(\type) \axismeanfixedindex{i'}{\JJ_{\type}^{\Pi}}{\rb}(\type) \right] \nonumber 
    =\left(\frac{1}{J_{\type}}\right)^2 \mme_{\Pi}\left[\sum_{j, j' \in \JJ_{\type}^{\Pi}}  y_{i,j}(\type) y_{i', j'}(\type) \right]. \nonumber
    \intertext{Observing that among the total $\binom{J}{J_\type}$ selection of indices $\JJ_\type^{\Pi}$, exactly $\binom{J-1}{J_\type-1}$ of these index sets contain index $j$ and exactly $\binom{J-2}{J_\type-2}$ of these index sets contain the pair $(j,j')$ for $j\neq j$,}
    \rho_{i,i'}^{\type,\type} &=\left(\frac{1}{J_{\type}}\right)^2 \left[ \frac{J_\type}{J} \sum_{j=1}^J y_{i,j}(\type)^2 + \frac{J_\type(J_\type-1)}{J(J-1)} \sum_{j=1}^J \sum_{j'\neq j} y_{i,j}(\type)y_{i,j'}(\type) \right] \nonumber \\
    &= \frac{1}{J_{\type}J} \left[ \sum_{j=1}^J y_{i,j}(\type) y_{i',j}(\type) + \frac{J_\type-1}{J-1} \sum_{j=1}^J \sum_{j'\neq j} y_{i,j}(\type)y_{i',j'}(\type) \right].
\end{align}

\paragraph{(ii) $\type \neq \type'$ and $\JJ_\type^{\Pi} = \JJ_{\type'}^{\Pi}$.}

The derivation is analogous to (i), and just requires swapping the argument of the second column-wise mean $ \axismeanfixedindex{i}{\JJ_{\type}^{\Pi}}{\rb}(\type)$ with $\type'$, i.e.\ consider $ \axismeanfixedindex{i}{\JJ_{\type}^{\Pi}}{\rb}(\type')$; because $\JJ_\type^{\Pi} = \JJ_{\type'}^{\Pi}$, this change is only in the argument and not in the set indexing this mean:

\begin{align}
    \rho_{i,i'}^{\type,\type'}  &= \mme\left[ \axismeanfixedindex{i}{\JJ_{\type}^{\Pi}}{\rb}(\type) \axismeanfixedindex{i'}{\JJ_{\type}^{\Pi}}{\rb}(\type') \right] 
    =\left(\frac{1}{J_{\type}}\right)^2 \mme_{\Pi}\left[\sum_{j, j' \in \JJ_{\type}^{\Pi}}  y_{i,j}(\type) y_{i', j'} (\type') \right], \nonumber
    \intertext{and now observing that among the total $\binom{J}{J_\type}$ selection of indices $\JJ_\type^{\Pi}$, exactly $\binom{J-1}{J_\type-1}$ of these index sets contain index $j$ and exactly $\binom{J-2}{J_\type-2}$ of these index sets contain the pair $(j,j')$ for $j\neq j$,}
    \rho_{i,i'}^{\type,\type'} &=\left(\frac{1}{J_{\type}}\right)^2 \left[ \frac{J_\type}{J} \sum_{j=1}^J y_{i,j}(\type) y_{i',j}(\type') + \frac{J_\type(J_\type-1)}{J(J-1)} \sum_{j=1}^J \sum_{j'\neq j} y_{i,j}(\type)y_{i',j'}(\type') \right] \nonumber \\
    &= \frac{1}{J_{\type}J} \left[ \sum_{j=1}^J y_{i,j}(\type) y_{i',j}(\type') + \frac{J_\type-1}{J-1}     \sum_{j=1}^J \sum_{j'\neq j} y_{i,j}(\type)y_{i',j'}(\type')
     \right].
\end{align}

\paragraph{(iii) $\type \neq \type'$ and $\JJ_\type \neq \JJ_{\type'}$.}

The derivation is analogous to (i), and just requires swapping the argument and index set of the second column-wise mean $ \axismeanfixedindex{i}{\JJ_{\type}^{\Pi}}{\rb}(\type)$ with $\type'$, i.e.\ consider $ \axismeanfixedindex{i}{\JJ_{\type'}^{\Pi}}{\rb}(\type')$; because $\JJ_\type^{\Pi} \neq \JJ_{\type'}^{\Pi}$, this change is affecting both the argument and the index set defining this mean:

\begin{align}
    \rho_{i,i'}^{\type,\type'}  &= \mme\left[ \axismeanfixedindex{i}{\JJ_{\type}^{\Pi}}{\rb}(\type) \axismeanfixedindex{i'}{\JJ_{\type'}^{\Pi}}{\rb}(\type') \right] 
    =\left(\frac{1}{J_{\type} J_{\type'}}\right) \mme_{\Pi}\left[\sum_{j \in \JJ_{\type}} \sum_{j' \in \JJ_{\type'}}  y_{i,j}(\type) y_{i', j'} (\type') \right] \nonumber
    \intertext{and now observing that among the total $\binom{J}{J_\type}$ selection of indices, there are exactly $\binom{J-2}{J_\type-2}$ selections such that $j, j' \in \JJ_{\type}^{\Pi}$, and exactly $\binom{J-2}{J-J_\type-2}$ selection of indices such that $j, j' \notin \JJ_{\type}^{\Pi}$, then there are exactly $\binom{J}{J_0} -\binom{J-2}{J_0-2} - \binom{J-2}{J_1-2}$ of the total $\binom{J}{J_0}$ such that $j$ and $j'$ do not both belong to $\JJ_{\type}^{\Pi}$} 
    \rho_{i,i'}^{\type,\type'}   &=\left(\frac{1}{J_{\type} J_{\type'}}\right) \left( 1 - \frac{J_0(J_0-1)}{J(J-1)} - \frac{J_1(J_1-1)}{J(J-1)} \right) \sum_{j=1}^J \sum_{j'\neq j} y_{i,j}(\type) y_{i',j'}(\type').
\end{align}
\end{proof}
\end{lemma}
\fi 

\subsubsection{Concentration of the covariance} \label{sec:conc_cov}

Recall the definition of the covariance matrix $S_{\hat{\bm{\tau}}^{\pi}}^2$ given in \Cref{eq:covariance_conditional}:
\[
    S_{\widehat{\bm{\tau}}^{\pi} }^2 := \frac{1}{I-1} \sum_{i=1}^I \left\{ \bm{\tau}_i^{\pi} - \bm{\tau}^{\pi} \right\} \left\{ \bm{\tau}_i^{\pi} - \bm{\tau}^{\pi} \right\}^\top,
\]
and let $S_{\pi}(\type, \type'):=S_{\widehat{\bm{\tau}}^{\pi} }^2(\type, \type')$ be its entry associated with types $\type,\type'$, where we drop the dependence on $\hat{\bm{\tau}}$ for the ease of notation.
Moreover, we have by \cref{eq:cov-pi-exact},
\begin{equation}
    \mmv\{\hattau(\coefvec)|\Pi=\pi\} = \mmv\{\hattau^\pi(\coefvec)\} = \sum_{\type, \type' \in \types}  \beta_\type \beta_{\type'}  \frac{A(\type,\type')}{I(I-1)} \sum_{i=1}^I\t(\type)\t(\type'),\label{eq:linear-estimator-conditional-variance}
\end{equation}
with $\t(\type)$ defined as an entry in the vector $\dot{\bm{\tau}}_i^{\pi}$ of \cref{eq:dot_tau_centered}, and
\begin{equation}
    A = 
    \begin{psmallmatrix}
        A(\ccc,\ccc) & A(\ccc,\icb) & A(\ccc,\ics) & A(\ccc,\ctt) \\
        A(\icb,\ccc) & A(\icb,\icb) & A(\icb,\ics) & A(\icb,\ctt) \\
        A(\ics,\ccc) & A(\ics,\icb) & A(\ics,\ics) & A(\ics,\ctt) \\
        A(\ctt,\ccc) & A(\ctt,\icb) & A(\ctt,\ics) & A(\ctt,\ctt)
    \end{psmallmatrix}
    =
    \{A(\type,\type')\}_{\type,\type' \in \types} \in \mathbb{R}^{4 \times 4} \label{eq:matrix_A}
\end{equation}
a matrix with entries $A(\type,\type')$ indexed by types $\type, \type'$, similar to $S_{\widehat{\bm{\tau}}^{\pi} }^2$. Explicitly, they are
\[
   \t(\type) = \frac{1}{J_{\type}}\sum_{j\in \J_\type^\pi} [y_{i,j}(\type)-\meanonefixedindex{\bullet}{\pi}(\type)] 
   \qquad\text{and}\qquad
    A = \begin{psmallmatrix}
    I_1/I_0 & -1 & I_1/I_0 & -1 \\
     -1 & I_1/I_0 & -1 & I_1/I_0  \\
     I_1/I_0 & -1 & I_1/I_0 & -1 \\
     -1 & I_1/I_0 & -1 & I_1/I_0 
\end{psmallmatrix}.
\]
Our arguments will also make use of certain facts about Orlicz norms, which are collected in the following definition. We refer to \citet[Chapter 2]{vershynin2018high} for proofs.
\begin{definition}\label{defn:orlicz}
For a real random variable $X$, and $k \ge 1$, we define  $\|X\|_{\psi_p}$ as the smallest $t > 0$ such that $\mme[\exp\{(X/t)^p\}] \le 2$, provided that such a $t$ exists (otherwise, it is $\infty$). 
For $X$, $Y$ real, Borel random variables:
\ifarxiv
\begin{enumerate}
    \item $\|aX + bY\|_{\psi_k} \le a\|X\|_{\psi_k} + b\|Y\|_{\psi_k}$ for $k=1,2$;
    \item $\|X - \mme[X]\|_{\psi_k} \le 2\|X\|_{\psi_k}$;
    \item if $\mathbb{P}(|X|>t) \le 2e^{-t^2/s_0^2}$ then $\|X\|_{\psi_2} \le Cs_0$
    \item \begin{enumerate}
        \item if $\|X\|_{\psi_2} \le s_1$ then  $\mathbb{P}(|X|>t) \le 2e^{-t^2/Cs_1^2}$, so $\bb{P}\{|X| \le s_1 \sqrt{C\log(2/\eta)}\} \ge 1-\eta$;
        \item if $\|X\|_{\psi_1} \le s_2$ then  $\mathbb{P}(|X|>t) \le 2e^{-t/Cs_2}$, so $\mathbb{P}\{|X| \le C s_2 \log(2/\eta)\} \ge 1-\eta$,
        \end{enumerate}
    \item $\|X^2\|_{\psi_1} \le \|X\|_{\psi_2}^2$.

\end{enumerate}
\else
\begin{inumerate}
    \item $\|aX + bY\|_{\psi_k} \le a\|X\|_{\psi_k} + b\|Y\|_{\psi_k}$ for $k=1,2$;
    \item $\|X - \mme[X]\|_{\psi_k} \le 2\|X\|_{\psi_k}$;
    \item if $\mathbb{P}(|X|>t) \le 2e^{-t^2/s_0^2}$ then $\|X\|_{\psi_2} \le Cs_0$
    \item if $\|X\|_{\psi_2} \le s_1$ then  $\mathbb{P}(|X|>t) \le 2e^{-t^2/Cs_1^2}$, so $\bb{P}\{|X| \le s_1 \sqrt{C\log(2/\eta)}\} \ge 1-\eta$;
    \item if $\|X\|_{\psi_1} \le s_2$ then  $\mathbb{P}(|X|>t) \le 2e^{-t/Cs_2}$, so $\mathbb{P}\{|X| \le C s_2 \log(2/\eta)\} \ge 1-\eta$,
    \item $\|X^2\|_{\psi_1} \le \|X\|_{\psi_2}^2$.
\end{inumerate}
\fi
\end{definition}
\begin{lemma}\label{lemma-concentration-cov} Under \cref{eq:assumption_b}, for a sufficiently large universal constant $C$, we have with probability at least $1-\eta$
    \begin{align*}
    \bigg|\mmv\{\hattau(\coefvec)|\Pi\} - \mme \left[ \mmv\{\hattau(\coefvec)|\Pi\} \right] \bigg| &\le \frac{CC_1^4C_2^2\|\coefvec\|_2^2}{J(I-1)}\log(4/\eta) \\
    & \qquad + \frac{CC_1^2C_2\|\coefvec\|_2}{\sqrt{J(I-1)}} 
    \sqrt{\mme[\mmv\{\hattau(\coefvec)|\Pi\}]\log(4/\eta)}. \\
\end{align*}
\end{lemma}
\begin{proof}
Throughout the proof, let $C$ be a sufficiently large universal constant.
Using \Cref{lem:helpful-rewrite} and the fact that potential outcomes are bounded as per \cref{eq:assumption_b} (b), we can rewrite
\[
    \T(\type) = \frac{1}{J_{\type}}\sum_{j\in \J_\type^\Pi} [y_{i,j}(\type)-\meanonefixedindex{\bullet}{\Pi}(\type)] = \frac{1}{J_{0}}\sum_{j=1}^{J_0} \tilde b_{i,\Pi(j)}(\type),
\]
for $|\tilde{b}_{i,j}(\type)| \le 4C_1C_2$, since $|y_{i,j}(\type)-\meanonefixedindex{\bullet}{\Pi}(\type)| \le 2C_2$. 
By using \cref{eq:linear-estimator-conditional-variance}, which allows us to express $\mmv\{\hattau(\coefvec)|\Pi=\pi\}$ as a sum over $\T(\type)$, the decomposition above leads us to
\[
   \mmv\{\hattau^\Pi(\coefvec)\} = \sum_{\type, \type' \in \types}  \beta_\type \beta_{\type' } \frac{A(\type,\type')}{I(I-1)} \sum_{i=1}^I\left(\frac{1}{J_{0}}\sum_{j=1}^{J_0} \tilde{b}_{i,\Pi(j)}(\type)\right)\left(\frac{1}{J_{0}}\sum_{j=1}^{J_0} \tilde{b}_{i,\Pi(j)}(\type') \right).
\]
To re-write this as a sum of squares, we use the eigenvector decomposition for the matrix $A$ in \cref{eq:matrix_A}:
\(
    A = 2(I_1/I_0 - 1) u u^\top + 2(I_1/I_0 + 1)v v^\top
\)
for $u = (1,1,1,1)^\top$ and $v = (-1,1,-1,1)^\top$. Therefore we can write the $(\type, \type')$-th entry of $A$ as $A(\type,\type') = A_1(\type)A_1(\type') + A_2(\type)A_2(\type')$ for $|A_1(\type)|, |A_2(\type)| \le \sqrt{2(C_1 + 1)}$. 
Thus,
$\mmv\{\hattau^\Pi(\coefvec)\} = V_1 + V_2$
where for $k \in \{1,2\}$
\begin{align*}
    V_k &= \sum_{\type, \type' \in \types}  \beta_\type \beta_{\type' } \frac{A_k(\type)A_k(\type')}{I(I-1)} \sum_{i=1}^I\left(\frac{1}{J_{0}}\sum_{j=1}^{J_0} \tilde{b}_{i,\Pi(j)}(\type)\right)\left(\frac{1}{J_{0}}\sum_{j=1}^{J_0} \tilde{b}_{i,\Pi(j)}(\type') \right) \\
    &=   \frac{1}{I(I-1)} \sum_{i=1}^I\left(\frac{1}{J_{0}}\sum_{j=1}^{J_0} \sum_{\type} \beta_\type A_k(\type) \tilde{b}_{i,\Pi(j)}(\type)\right)\left(\frac{1}{J_{0}}\sum_{j=1}^{J_0} \sum_{\type'}\beta_{\type'}A_k(\type') \tilde{b}_{i,\Pi(j)}(\type') \right).
\end{align*}
Hence, defining $X_{i,k} \coloneqq  \frac{1}{J_{0}}\sum_{j=1}^{J_0} \left(\sum_{\type} \beta_\type A_k(\type) \tilde{b}_{i,\Pi(j)}(\type)\right)$,
\(
    V_k = \frac{1}{I(I-1)} \sum_{i=1}^I X_{i,k}^2.
\)
Next, we can use Cauchy-Schwarz and our previous bounds on $A_k$ and $\tilde{b}_{i,j}$ to bound these terms constituting $X_{i,k}$s in parentheses as  $|\sum_\type  \beta_\type A_k(\type) \tilde{b}_{i,j}(\type)| \le C\|\coefvec\|_2C_1^{3/2}C_2$ for $k = 1,2$. 
This gives us the trivial bound $|X_{i,k}| \le \sqrt{32}\|\coefvec\|_2(C_1+1)^{3/2}C_2$, so
\begin{equation}\label{eq:trivial-conditional-var-estimate}
    \mmv\{\hattau(\coefvec)|\Pi\} = \sum_{k\in\{1,2\}} \frac{1}{I(I-1)} \sum_{i=1}^I X_{i,k} ^2 \le \frac{64\|\coefvec\|_2^2(C_1+1)^{3}C_2^2}{I-1};
\end{equation}
clearly, the right-hand side of \cref{eq:trivial-conditional-var-estimate}  also bounds $\mathbb{E}[\mmv\{\hattau(\coefvec)|\Pi\}]$.

Thus, applying \Cref{lemma:talagrand-sum} with $b_{i,j} = \sum_\type  \beta_\type A_k(\type) \tilde{b}_{i,j}(\type)$, we find that for a large enough universal constant $C$, we have for all $i \in [I]$ and $k=1,2$:
\begin{equation}
    \|
        X_{i,k}-\mme[X_{i,k}]
    \|_{\psi_2} 
    \le 
        CC_1^{3/2}C_2\|\coefvec\|_2 J_0^{-1/2}.
    \label{eq:variance-concentration-basic}
\end{equation}
Finally, by linearity of expectation and the identity
\[X^2 - \mme[X^2] = \{(X-\mme[X])^2 - \mme[(X-\mme[X])^2]\} + 2(X-\mme[X])\mme[X] \]
which is seen by expanding $X^2 = \{\mme[X] + (X - \mme[X])\}^2$ and its expectation, we have
\begin{align*}
    &\quad |\mmv\{\hattau^\Pi(\coefvec)\} - \mme[\mmv\{\hattau^\Pi(\coefvec)\}]|  \\
    &= |V_1 + V_2 - \mathbb{E}[V_1 + V_2]| \\ \displaybreak
    &= \left|\frac{1}{I(I-1)}\sum_{i=1}^I \sum_{k\in\{1,2\}} (X_{i,k} - \mme[X_{i,k}])^2 - \mme[(X_{i,k} - \mme[X_{i,k}])^2] + 2(X_{i,k}-\mme[X_{i,k}])\mme[X_{i,k}] \right| \\ 
    &\le \underbrace{\left|\frac{1}{I(I-1)}\sum_{i=1}^I \sum_{k\in\{1,2\}} (X_{i,k} - \mme[X_{i,k}])^2 - \mme[(X_{i,k} - \mme[X_{i,k}])^2] \right|}_{s_1} \\ 
    & \qquad + \underbrace{\left|\frac{1}{I(I-1)}\sum_{i=1}^I \sum_{k\in\{1,2\}} 2(X_{i,k}-\mme[X_{i,k}])\mme[X_{i,k}]\right|}_{s_2}
\end{align*}
For the first summand $s_1$, we have by the triangle inequality
\begin{align*}
    \|s_1\|_{\psi_1} &= \left|\frac{1}{I(I-1)}\sum_{i=1}^I \sum_{k\in\{1,2\}} (X_{i,k} - \mme[X_{i,k}])^2 - \mme[(X_{i,k} - \mme[X_{i,k}])^2] \right| \\ 
    &\le \frac{1}{I(I-1)}\sum_{i=1}^I \sum_{k\in\{1,2\}} \|(X_{i,k} - \mme[X_{i,k}])^2 - \mme[(X_{i,k} - \mme[X_{i,k}])^2]\|_{\psi_1}.
\end{align*}
By \cref{defn:orlicz} (ii) and (vi), $\|U^2 - \mme[U^2]\|_{\psi_1} \le 2\|U^2\|_{\psi_1} \le 2\|U\|^2_{\psi_2}$, we can last apply \cref{eq:variance-concentration-basic} to obtain
\begin{equation*}
    \|s_1\|_{\psi_1} \le \frac{2}{I(I-1)}\sum_{i=1}^I \sum_{k\in\{1,2\}} \|X_{i,k} - \mme[X_{i,k}]\|_{\psi_2}^2 \le \frac{C C_1^3 C_2^2\|\coefvec\|_2^2}{J_0(I-1)}.
\end{equation*}
For the second summand $s_2$, we similarly have by the triangle inequality and the same bound used above  (\cref{eq:variance-concentration-basic}):
\begin{align*}
    \|s_2\|_{\psi_2} &= \left|\frac{1}{I(I-1)}\sum_{i=1}^I \sum_{k\in\{1,2\}} 2(X_{i,k}-\mme[X_{i,k}])\mme[X_{i,k}]\right| \\ 
    &\le \frac{1}{I(I-1)}\sum_{i=1}^I \sum_{k\in\{1,2\}} 2\|X_{i,k}-\mme[X_{i,k}]\|_{\psi_2}|\mme[X_{i,k}]| \\
    &\le \frac{2CC_1^{3/2}C_2\|\coefvec\|_2}{\sqrt{J_0}(I-1)}\left(\frac{1}{2I}\sum_{i=1}^I \sum_{k\in\{1,2\}} |\mathbb{E}[X_{i,k}]|\right).
    \intertext{By two applications of Jensen's inequality to the term in parentheses, this is}
    &\le \frac{2CC_1^{3/2}C_2\|\coefvec\|_2}{\sqrt{J_0}(I-1)}\left(\frac{1}{2I}\sum_{i=1}^I \sum_{k\in\{1,2\}} \sqrt{\mathbb{E}[X_{i,k}^2]}\right) \\ \displaybreak
    &\le \frac{2CC_1^{3/2}C_2\|\coefvec\|_2}{\sqrt{J_0}(I-1)}\left(\sqrt{\frac{1}{2I}\sum_{i=1}^I \sum_{k\in\{1,2\}} \mathbb{E}[X_{i,k}^2]}\right) \\ 
    &=  \frac{\sqrt{2}CC_1^{3/2}C_2\|\coefvec\|_2}{\sqrt{J_0(I-1)}}\left(\sqrt{\frac{1}{I(I-1)}\sum_{i=1}^I \sum_{k\in\{1,2\}} \mathbb{E}[X_{i,k}^2]}\right) \\
    &= \frac{\sqrt{2}CC_1^{3/2}C_2\|\coefvec\|_2}{\sqrt{J_0(I-1)}}\sqrt{\mme[\mmv\{\hattau^\Pi(\coefvec)\}]},
\end{align*}
It follows from \cref{defn:orlicz} (iv) and (v) 
that, for a universal constant $C'$ possibly larger than $C$, the following events each have probability at least $1-\eta$
\begin{align*}
    s_1 &\le \frac{C'C_1^{3}C_2^2\|\coefvec\|_2^2}{J_0(I-1)}\log(2/\eta) \\
    s_2 &\le \frac{C'C_1^{3/2}C_2\|\coefvec\|_2}{\sqrt{J_0(I-1)}} 
    \sqrt{\mme[\mmv\{\hattau(\coefvec)|\Pi\}]}\log(2/\eta)^{1/2} .
\end{align*}
Thus, replacing $\eta$ by $\eta/2$ and using a union bound, it holds with probability $1-\eta$ that
\begin{align*}
    |\mmv\{\hattau^\Pi(\coefvec)\} - \mme[ \mmv\{\hattau^\Pi(\coefvec)\}]| &\le s_1+s_2 \\
    & \le \frac{C'C_1^{3}C_2^2\|\coefvec\|_2^2}{J_0(I-1)}\log(4/\eta) \\
    & \qquad + \frac{C'C_1^{3/2}C_2\|\coefvec\|_2}{\sqrt{J_0(I-1)}} 
    \sqrt{\mme[\mmv\{\hattau^\Pi(\coefvec)\}]\log(4/\eta)},
\end{align*}
which, after noting $J_0 \ge J/C_1$ and simplifying, gives us the claimed inequality. 
\end{proof}

\subsubsection{Simplifying the conditional CLT}\label{sec:simplifying-conditional-clt}

We now simplify \Cref{lem:clt-bounded2}, which, as stated, involves normalization by the random conditional variance $\mmv\{\hattau(\coefvec)|\Pi\}$. In \Cref{lemma:simplified-conditional-clt}, we show that $\mmv\{\hattau(\coefvec)|\Pi\}$ can be replaced by the deterministic quantity $\mme[\mmv\{\hattau(\coefvec)|\Pi\}]$, which simplifies the analysis in \Cref{sec:tying} to follow. 
We first state three helper lemmas, namely  \cref{lemma:coupling-to-kolmogorov,lemma:combine-coupling-with-kolmogorov,gaussian-continuity}.

\begin{lemma}[Lemma 2.1 of \cite{chernozhukov2016empirical}]\label{lemma:coupling-to-kolmogorov} 
Let $X$, $Y$ be real, Borel random variables.
Suppose that $\mathbb{P}(|X-Y| > \nu) \le \eta$. Then 
\[
    \sup_{t \in \mathbb{R}}\left|\mathbb{P}(X \le t) - \mathbb{P}(Y \le t) \right| \le \eta + \sup_{t \in \mathbb{R}} \mathbb{P}(|Y - t| \le \nu). 
\]
If $Y \sim N(0,1)$ is standard normal then the RHS of the bound above simplifies to $\eta + 2\nu$.
\end{lemma}
\begin{proof} The first statement is exactly Lemma 2.1 in \cite{chernozhukov2016empirical}; the second claim follows as the density of a standard Gaussian random variable is bounded by $1$.
\end{proof}

\begin{lemma}\label{lemma:combine-coupling-with-kolmogorov} Let $X$, $Y$ be real, Borel random variables. Let $\Phi_\sigma(t) = \Phi(t/\sigma)$ be the CDF of a zero-mean Gaussian with variance $\sigma^2$. 

If $\bb{P}(|X-Y| > \nu ) \le \eta$ and $\sup_{t \in \bb{R}}|\bb{P}(Y \le t) - \Phi_\sigma(t)| \le \epsilon$, then
\[
    \sup_{t \in \bb{R}}|\bb{P}(X \le t) - \Phi_\sigma(t)| \le \eta + 3\epsilon + 2\nu/\sigma.
\]
\end{lemma}
\begin{proof}
Given any $t \in \mathbb{R}$ our assumptions imply $\bb{P}(Y \le t) \le \Phi_\sigma(t) + \epsilon$. Using continuity of $\Phi_\sigma$, they also imply 
 \(\bb{P}(Y < t) = \lim_{u \uparrow t} \bb{P}(Y \le u) \ge \lim_{u \uparrow t} \Phi_\sigma(u) - \epsilon = \Phi_\sigma(t) - \epsilon.\)
Thus,
\begin{align*}
    \sup_{t \in \mathbb{R}} \mathbb{P}(|Y - t| \le \nu) 
    &= \sup_{t \in \mathbb{R}} \left\{ \mathbb{P}(Y  \le t + \nu) -  \mathbb{P}(Y  < t - \nu) \right\} \\
    &\le \sup_{t \in \mathbb{R}} \left\{\Phi_\sigma(t + \nu) + \epsilon - [\Phi_\sigma(t - \nu) - \epsilon]\right\}  \\
    &= 2\epsilon + \sup_{t \in \bb{R} }\left\{ \int_{t-\nu}^{t+\nu} \Phi_\sigma'(u)\mathrm{d}u  \right\}\le  2\epsilon + 2\nu/\sigma,
\end{align*}
where we have used $|\Phi'_\sigma| \le 1/\sigma$ in the last step. 
We then use Lemma \ref{lemma:coupling-to-kolmogorov} and the triangle inequality for the norm $\|F\|_\infty = \sup_{t \in \bb{R}}|F(t)|$ to conclude:
\begin{align*}
    \sup_{t \in \bb{R}}|\bb{P}(X \le t) - \Phi(t)|
    &\le \sup_{t \in \bb{R}}|\bb{P}(X \le t) - \bb{P}(Y \le t)| + \sup_{t \in \bb{R}}|\bb{P}(Y \le t) - \Phi(t)| \\
    &\le \eta + \sup_{t \in \mathbb{R}} \mathbb{P}(|Y - t| \le \nu) + \epsilon \le 3\epsilon + 2\nu/\sigma + \eta.
\end{align*} 
\end{proof}

\begin{lemma}\label{gaussian-continuity}
Let $\Phi$ be the standard Gaussian CDF. For any $\eta \in (0,1)$
\begin{align}
\sup_{s \in \mathbb{R}} \Big|\Phi\left(\frac{s-\mu_1}{\sigma_1}\right) - \Phi\left(\frac{s-\mu_2}{\sigma_2}\right) \Big| \le \eta + \frac{|\mu_1 - \mu_2|}{\sigma_2} + \frac{|\sigma_1-\sigma_2|}{\sigma_2}\sqrt{2\log(e/\eta)}. 
\label{eq:thesis-gaus-cont} 
\end{align}
\end{lemma}
\begin{proof}
    Note that by making the substitution $t = \sigma_2 s + \mu_2$ we have 
    \[\sup_{s \in \mathbb{R}} \Big|\Phi\left(\frac{s-\mu_1}{\sigma_1}\right) - \Phi\left(\frac{s-\mu_2}{\sigma_2}\right) \Big| = \sup_{t \in \mathbb{R}} \Big|\Phi\left(\frac{t - (\mu_1 - \mu_2)/\sigma_2}{\sigma_1/\sigma_2}\right) - \Phi\left(t\right) \Big|\]
    Let $Y$ be standard normal. Define $X = (\sigma_1/\sigma_2)Y + (\mu_1 - \mu_2)/\sigma_2$.
    By construction, then, 
\begin{equation*}
    |Y - X| \le \sigma_2^{-1}|\mu_1 - \mu_2| + |Y||\sigma_1/\sigma_2 - 1|.
\end{equation*}
By the Gaussian concentration inequality $\bb{P}(|Y| > t) = 2\Phi(-t) \le 2e^{-t^2/2}$ where the last inequality holds for $t \ge 1$, we can choose $t = \sqrt{2\log(e/\eta)} \ge 1$ to deduce that with probability at least $1-\eta$ we have $ |Y| \le \sqrt{2\log(e/\eta)}$. On this event, hence with probability $1-\eta$,
\begin{equation}
    |Y - X| \le \frac{|\mu_1 - \mu_2|}{\sigma_2} + \frac{|\sigma_1 - \sigma_2|}{\sigma_2}\sqrt{2\log(e/\eta)}. 
\end{equation} 
The proof then follows immediately by applying \Cref{lemma:coupling-to-kolmogorov}.
\end{proof}

\begin{lemma}\label{lem:pi-measurable-substitution}
Let $V$ be a discrete, real random variable. For a real, Borel random variable $X$, a cumulative distribution function $F$, and $\varepsilon: \bb{R} \to \bb{R}$, suppose that
\[
    \sup_{t \in \bb{R}}|\bb{P}(X \le t \mid V) - F(t)| \le \varepsilon(V).
\]
Then, if $U = u(V)$ is a $\sigma(V)$-measurable random variable, it also holds that
\[
    |\bb{P}(X \le U \mid V) - F(U)| \le \varepsilon(V).
\]
\end{lemma}
\begin{proof}
For any fixed $v$ in the support of $V$ and $u = u(v)$ we have 
\[
    |\bb{P}(X \le u \mid V = v) - F(u)| 
    \le \sup_{t \in \bb{R}}|\bb{P}(X \le u \mid V = v) - F(t)| \le \varepsilon(v).
\]
\end{proof}

\begin{lemma}\label{lemma:simplified-conditional-clt} Under assumptions (a) and (b), it holds with probability $1 - \eta$  that
    \begin{equation}
    \sup_{t \in \mathbb{R}} \left|\mathbb{P}\left\{\frac{\widehat{{\tau}}^{\Pi}(\coefvec) -  {\tau}^{\Pi}(\coefvec)}{\sqrt{\mme[\mmv\{\hattau(\coefvec)|\Pi\}]}} \le t \middle|\Pi \right\} - \Phi(t)\right| 
		\le \eta + 
		\frac{CC_1^2C_2\|\coefvec\|_2(I^{-1} + J^{-1})}{\sqrt{\mme[\mmv\{\hattau(\coefvec)|\Pi\}]}}\log(C/\eta)
    \end{equation}
\end{lemma}
\begin{proof}
Put $\sigma_\Pi^2 = \mmv\{\hattau(\coefvec)|\Pi\}$ and $\sigma_2^2 = \mme[\mmv\{\hattau(\coefvec)|\Pi\}]$. \Cref{lem:clt-bounded2} states that 
\[
	\sup_{t \in \mathbb{R}} \left|\mathbb{P}\left\{\sigma_\Pi^{-1}[\widehat{{\tau}}^{\Pi}(\coefvec) -  {\tau}^{\Pi}(\coefvec)] \le t \middle|\Pi \right\} - \Phi(t)\right| 
		\le 
		\frac{C C_2\|\coefvec\|_2}{\min\{I_0, I_1\}} 
	\frac{1}{\sigma_{\Pi} } \eqqcolon \frac{B_1}{\sigma_\Pi}
 \]
$\Pi$ is discrete and $U_t = \sigma_{\Pi}^{-1}t$ is $\sigma(\Pi)$-measurable, and 
 \(\widehat{{\tau}}^{\Pi}(\coefvec) -  {\tau}^{\Pi}(\coefvec) \le t \iff \sigma_\Pi^{-1}[\widehat{{\tau}}^{\Pi}(\coefvec) -  {\tau}^{\Pi}(\coefvec)] \le U_t
 \).
 Thus, combining the previous display with \Cref{lem:pi-measurable-substitution} gives 
\[  \left|\mathbb{P}\left\{\sigma_\Pi^{-1}[\widehat{{\tau}}^{\Pi}(\coefvec) -  {\tau}^{\Pi}(\coefvec)] \le U_t \middle|\Pi \right\} - \Phi(U_t)\right| = \left|\mathbb{P}\left\{\widehat{{\tau}}^{\Pi}(\coefvec) -  {\tau}^{\Pi}(\coefvec) \le t \middle|\Pi \right\} - \Phi(U_t)\right| 
\le 
		 \sigma_{\Pi}^{-1}B_1.
\]
for any $t \in \mathbb{R}$.
 By the triangle inequality for $\|f\|_\infty = \sup_{t \in \mathbb{R}}|f(t)|$, the above, and using \cref{gaussian-continuity} to bound 
 \(\sup_{t \in \mathbb{R}}\left|\Phi(\sigma_{\Pi}^{-1}t) - \Phi (\sigma_2^{-1}t)\right|\),
 it holds for $\eta \in (0,1)$ that
 \begin{align*}
 & \sup_{t\in\mathbb{R}} \left|\mathbb{P}\left\{[\widehat{{\tau}}^{\Pi}(\coefvec) -  {\tau}^{\Pi}(\coefvec)] \le t \middle|\Pi \right\} - \Phi(\sigma_2^{-1}t)\right| \\
 & \qquad \qquad \le \sup_{t\in\mathbb{R}} \left|\mathbb{P}\left\{[\widehat{{\tau}}^{\Pi}(\coefvec) -  {\tau}^{\Pi}(\coefvec)] \le t \middle|\Pi \right\} - \Phi(\sigma_\Pi^{-1}t)\right| + \sup_{t \in \mathbb{R}}\left|\Phi(\sigma_{\Pi}^{-1}t) - \Phi (\sigma_2^{-1}t)\right| \\
 & \qquad \qquad \le  \sigma_{\Pi}^{-1}B_1 + \sup_{t \in \mathbb{R}}\left|\Phi(\sigma_{\Pi}^{-1}t) - \Phi (\sigma_2^{-1}t)\right|  \\
 & \qquad \qquad \le \sigma_{\Pi}^{-1}B_1 + \eta + \left|\frac{\sigma_\Pi - \sigma_2}{\sigma_2}\right|\sqrt{2\log(e/\eta)}.
 \end{align*}
 We now manipulate the bound above to remove its dependence on the random quantity $\sigma_{\Pi}$.
  We start by considering the first term, $\sigma_{\Pi}^{-1}B_1$.
 
 \paragraph{First term, case $\sigma_{\Pi} \ge \frac{1}{2}\sigma_2$:}
 If $\sigma_{\Pi} \ge \frac{1}{2}\sigma_2$ then $\sigma_{\Pi}^{-1}B_1 \le 2\sigma_2^{-1}B_1$, implying the bound
 \begin{equation}
 \sup_{t\in\mathbb{R}} \left|\mathbb{P}\left\{\widehat{{\tau}}^{\Pi}(\coefvec) -  {\tau}^{\Pi}(\coefvec) \le t \middle|\Pi \right\} - \Phi(\sigma_2^{-1}t)\right|
 \le \frac{2B_1}{\sigma_2} + \eta + 2\left|\frac{\sigma_\Pi - \sigma_2}{\sigma_2}\right|\sqrt{2\log(e/\eta)}.\label{eq:intermediate-simplify-cov}
 \end{equation}
 
\paragraph{First term, case $\sigma_{\Pi} < \frac{1}{2}\sigma_2$:} 
On the other hand, if $\sigma_{\Pi} < \frac{1}{2}\sigma_2$, then $|\sigma_2 - \sigma_\Pi| \ge \frac{1}{2}\sigma_2$, so  the third term satisfies $2\sigma_2^{-1}|\sigma_{\Pi} - \sigma_2|\sqrt{2\log(e/\eta)} \ge 1$ and the above becomes trivial. We conclude that \cref{eq:intermediate-simplify-cov} holds.

\paragraph{Third term:}
Next, we handle the third term $\sigma_2^{-1}|\sigma_{\Pi} - \sigma_2|$. Combining the inequality $|x-1| \le |x-1||x+1| = |x^2 - 1|$ for $x = \sigma_\Pi/\sigma_2 > 0$ with \Cref{lemma-concentration-cov} gives us that with with probability ~$1-\eta'$,
 \[\left|\frac{\sigma_\Pi - \sigma_2}{\sigma_2}\right|  \le \left|\frac{\sigma_\Pi^2 - \sigma_2^2}{\sigma_2^2} \right| \le \frac{B_2^2\log(4/\eta')}{\sigma_2^2} + \frac{B_2\sigma_2\sqrt{\log(4/\eta')}}{\sigma_2^2}; \quad B_2 \coloneqq \frac{CC_1^2C_2\|\coefvec\|_2}{\sqrt{J(I-1)}}\]
The first term is the square of the second term, and that the bound becomes trivial if either exceeds $1$, so we may assume that the second term is larger and deduce that with probability $1-\eta'$
 \[\sup_{t\in\mathbb{R}} \left|\mathbb{P}\left\{\widehat{{\tau}}^{\Pi}(\coefvec) -  {\tau}^{\Pi}(\coefvec) \le t \middle|\Pi \right\} - \Phi(\sigma_2^{-1}t)\right|
 \le \frac{2B_1}{\sigma_2} + \eta + \frac{B_2\sqrt{8\log(e/\eta)\log(4/\eta'})}{\sigma_2}.\]
 Finally, we make the substitution $t' = \sigma_2 t$ above, rearrange, simplify $\min\{I_0,I_1\} \ge I/C_1 \ge I/C_1^3$ in the definition of $B_1$ and $(IJ)^{-1/2} \le (I^{-1} + J^{-1})/2$ in $B_2$, and finally take $\eta' = \eta$
 to deduce the claimed inequality.
\end{proof}

\subsection{Final result} \label{sec:tying}
In this subsection we combine the results of \cref{sec:proof_clt_conditional} and \cref{sec:concentration} and  finally state and prove the CLT presented in \Cref{thm:clt}.

\subsubsection{Combining CLTs}
For any permutation $\Pi$, we have the decomposition
\begin{equation}
    \hattau(\coefvec) - \tau(\coefvec) = \{\hattau(\coefvec) - \tau^\Pi(\coefvec)\} + \{\tau^\Pi(\coefvec)- \tau(\coefvec)\}. \label{eq:decomposition}
\end{equation}
\Cref{lem:clt-conditional-mean,lemma:simplified-conditional-clt} yield the two following Gaussian approximations
\begin{align*}
    \frac{\{\hattau(\coefvec) - \tau^\Pi(\coefvec)\}}{\mme[\mmv\{\hattau(\coefvec)|\Pi\}]^{-1/2}} \overset{\mathrm{d}} \approx N(0,1), 
    \quad\text{and}\quad
    \frac{\{\tau^\Pi(\coefvec)- \tau(\coefvec) \}}{\mmv\{\tau^\Pi(\coefvec)\}^{-1/2}} \overset{\mathrm{d}} \approx N(0,1).
\end{align*}
The rate of convergence in both cases is a function of sample sizes $I$ and $J$ and was characterized in \cref{lem:clt-conditional-mean,lemma:simplified-conditional-clt}. 
We use the decomposition in \cref{eq:decomposition} to combine these approximations to recover a Gaussian approximation of $\hattau(\coefvec) - \tau(\coefvec)$ (the left-hand side of \cref{eq:decomposition}). 
This is accomplished with the following technical lemma.

\begin{lemma}\label{lem:tensorize-kolmogorov}
Let $V$ be a discrete random variable, and $X$ be different $\sigma(V)$-measurable random variable. 
Let $Y$ be a real valued random variable.
For $F, G$ conditional distribution functions with bounded densities $F'$ and $G'$, 
if for $\eta \in (0,1)$ and some numbers $\Delta(\eta), \Delta' > 0$, 
\begin{equation}
    \mathbb{P}\left\{\sup_{t \in \bb{R}} \left|\bb{P}(Y \le t \mid V) - F(t)\right| \le \Delta(\eta) \right\} \ge 1-\eta,
    \quad\text{and}\;
    \sup_{t \in \bb{R}} \left|\bb{P}(X \le t) - G(t)\right| \le \Delta'       \label{eq:condition_tensorize_kolmogorov}
\end{equation}
then, for $Z \sim F$ and $W \sim G$ independent of each other and of $(X, Y, V)$, we have
\begin{align}
     \sup_{t \in \bb{R}}\left| \bb{P}\{X + Y \le t\} - \bb{P}\{X + Z \le t\}\right| &\le \eta + \Delta(\eta) \label{eq:tensorize-conditional}\\
     \sup_{t \in \bb{R}}\left| \bb{P}\{X + Z \le t\} - \bb{P}\{W + Z \le t\}\right| &\le \Delta' \label{eq:tensorize-unconditional}\\
    \sup_{t \in \bb{R}}\left| \bb{P}\{X + Y \le t\} - \bb{P}\{W + Z \le t\}\right| &\le \eta + \Delta(\eta) + \Delta'. \label{eq:tensorize-combined}
\end{align}
\end{lemma}
\begin{proof}
Let $t \in \bb{R}$ be given. We have 
\begin{align}
    |\bb{P}\{X + Y \le t\} - \bb{E}[F(t-X)]|
    &= \bigg|\bb{E}\left[\bb{P}(Y \le t - X\mid V) - F(t-X)\right]\bigg| \nonumber\\
    &\le \bb{E}\left[ \big|\bb{P}(Y \le t - X\mid V) - F(t-X) \big|\right] \nonumber\\
    &\le \bb{E}\left[\sup_{u \in \bb{R}}\big|\bb{P}(Y \le u\mid V) - F(u)\big|\right],\label{eq:kolmogorov_proof_exp}
    \end{align}
where we have used Jensen's inequality and then \Cref{lem:pi-measurable-substitution} (as $X$ is $\sigma(V)$-measurable). 
Define now the event $\mathcal{E}_\eta:= \left\{ \left|\bb{P}(Y \le u\mid V) - F(u)\right| \le \Delta(\eta)\right\}$.  
We use $\mathcal{E}_\eta$ to bound the argument of the expectation in \cref{eq:kolmogorov_proof_exp} as follows: 
(i) on $\mathcal{E}_\eta$, by its definition, \[\sup_{u \in \bb{R}}\big|\bb{P}(Y \le u\mid V) - F(u)\big| \le \Delta(\eta);\]
(ii) on $\mathcal{E}_\eta^c$, $\sup_{u \in \bb{R}}\big|\bb{P}(Y \le u\mid V) - F(u)\big| \le 1$, since it is an absolute difference of probabilities.
Hence, we conclude that
\begin{equation}
|\bb{P}\{X + Y \le t\} - \bb{E}[F(t-X)]|\le \bb{E}\left[\mathbbm{1}\{\mathcal{E}_\eta^c\} +\mathbbm{1}\{\mathcal{E}_\eta\}\Delta(\eta) \right] \le \eta + \Delta(\eta), \label{eq:conditional-tensorization-bound}
\end{equation}
where the very last inequality follows from  noting that by \cref{eq:condition_tensorize_kolmogorov}, $\bb{P}(\mathcal{E}_\eta^c) \le \eta$.

Next, we can simplify $\mme[F(t-X)]$ using the fundamental theorem of calculus, Fubini's theorem, and the convolution formula for sums of independent random variables as follows:
\begin{align}
    \mme[F(t-X)] = \mme\left[\int_{-\infty}^{t-X} F'(u)\mathrm{d}u\right] \nonumber 
    &= \mme\left[\int_{-\infty}^\infty F'(u)\mathbbm{1}\{u \le t-X\}\mathrm{d}u \right] \nonumber \\
    &=  \int_{-\infty}^\infty F'(u) \bb{P}\{X \le t - u\} \mathrm{d}u  \label{eq:rewrite_proba} \\
    &= \bb{P}(X + Z \le t) \label{eq:proba_equiv}.
\end{align} 
Since $t \in \bb{R}$ was arbitrary, substituting \cref{eq:proba_equiv} in \cref{eq:conditional-tensorization-bound} proves \cref{eq:tensorize-conditional}.

Further, expanding from \cref{eq:rewrite_proba} allows us to write:
\begin{align*}
    \mme[F(t-X)] &= \int_{-\infty}^\infty F'(u) \bb{P}\{X \le t - u\} \mathrm{d}u 
    \\
    &= \int_{-\infty}^\infty F'(u)G(t-u)\mathrm{d}u + \int_{-\infty}^\infty F'(u)[\bb{P}\{X \le t - u\} - G(t-u)]\mathrm{d}u \\
    &= \bb{P}(Z + W \le t) + \int_{-\infty}^\infty F'(u)[\bb{P}\{X \le t - u\} - G(t-u)]\mathrm{d}u. 
\end{align*}
It follows that we may bound
\begin{align}
    \bigg|\mme[F(t-X)]  - \bb{P}(Z + W \le t) \bigg| &= \left| \int_{-\infty}^\infty F'(u)[\bb{P}\{X \le t - u\} - G(t-u)]\mathrm{d}u \right| \nonumber \\
    &\le \int_{-\infty}^\infty 
  F'(u)\left|\bb{P}\{X \le t - u\} - G(t-u)\mathrm{d}u \right| \nonumber \\
  \bigg|\bb{P}(X + Z \le t)  - \bb{P}(Z + W \le t) \bigg|&\le\int_{-\infty}^\infty 
  F'(u) \Delta' \mathrm{d}u = \Delta',\label{eq:unconditional-tensorization-bound}
\end{align}
where we have plugged in the equality $\mme[F(t-X)] = \bb{P}(X + Z \le t)$ from \cref{eq:proba_equiv}. Because $t \in \bb{R}$ was arbitrary, this proves \cref{eq:tensorize-unconditional}.

Finally, we use the triangle inequality and \cref{eq:conditional-tensorization-bound,eq:unconditional-tensorization-bound} to deduce:
\begin{align*}
    \left| \bb{P}\{X + Y \le t\} - \bb{P}\{Z + W  \le t\}\right| 
    &\le \left| \bb{P}\{X + Y \le t\} - \mme[F(t-X)]\right|  \\
    & \;\;+ \left|\mme[F(t-X)]- \bb{P}\{Z + W \le t\}\right| \\
    &\le \eta + \Delta(\eta) + \Delta'
\end{align*}
After noting that $t \in \bb{R}$ was arbitrary, this proves \cref{eq:tensorize-combined}.  
\end{proof}

\subsubsection{Final bound}

Finally we prove \Cref{thm:clt}.
Next, we give the main result, which shows how we combine the normal approximations in \Cref{lem:clt-conditional-mean,lemma:simplified-conditional-clt} with \Cref{lem:tensorize-kolmogorov}.
\begin{lemma}\label{lemma:tying-together}
    For some $\Delta_1>0$, suppose that
    \begin{align}
        \sup_{t \in \mathbb{R}} 
		\left|\mathbb{P}\left\{\frac{{{\tau}}^{\Pi}(\coefvec) -  {\tau}(\coefvec)}{\sqrt{\mmv\{{\tau}^{\Pi}(\coefvec)\}}} \le  t  \right\} - \Phi(t)\right| &\le \frac{\Delta_1}{\sqrt{\mmv\{{\tau}^{\Pi}(\coefvec)\}}} \label{eq:ks-unocnditional}
  \intertext{and that with probability at least $1-\eta$,}
       \ \sup_{t \in \mathbb{R}} \left|\mathbb{P}\left\{\frac{\widehat{{\tau}}(\coefvec) -  {\tau}^{\Pi}(\coefvec)}{\sqrt{\mme[\mmv\{\hat{\tau}(\coefvec)|\Pi\}]}} \le  t  \middle| \Pi \right\} - \Phi(t)\right| &\le \eta + \frac{\Delta_2\log(C/\eta)}{\sqrt{\mme[\mmv\{\hat{\tau}(\coefvec)|\Pi\}]}} . \label{eq:ks-conditional}
 \intertext{Then, with $\xi(C,t) = C t \log(C / t)$ and some $\Delta_2>0$, we have}
  \sup_{u \in \mathbb{R}} \left|\mathbb{P}\left\{\frac{\widehat{{\tau}}(\coefvec) -  \tau(\coefvec)}{\sqrt{\mmv\{\hat{\tau}(\coefvec)\}}} \le  u \right\} - \Phi(u)\right| &\le \xi\left(C, \frac{(\Delta_1 + \Delta_2)^{1/3}}{\mmv\{\hat{\tau}(\coefvec)\}^{1/6}}\right).  \label{eq:generic-clt}
    \end{align} 
    \end{lemma}
Before we prove \Cref{lemma:tying-together}, we show how it implies \Cref{thm:clt}.

\begin{corollary}[\Cref{thm:clt} in the main paper]\label{cor:clt-appendix}
Under Assumptions (a) and (b), it holds for universal constants $C, C' > 0$ that
\[\sup_{t \in \mathbb{R}} \left|\mathbb{P}\left\{\frac{\widehat{{\tau}}(\coefvec) -  \tau(\coefvec)}{\sqrt{\mmv\{\hat{\tau}(\coefvec)\}}} \le  t \right\} - \Phi(t)\right| \le C'\Delta^{1/3}\log(C'/\Delta); \quad \Delta = \left(\frac{C C_1^2C_2\|\coefvec\|_2(I^{-1}+J^{-1})}{\sqrt{\mmv\{\hat{\tau}(\coefvec)\}}}\right).\]
\end{corollary}
\begin{proof}[Proof of \Cref{cor:clt-appendix}]
By \Cref{lem:clt-conditional-mean} and \Cref{lemma:simplified-conditional-clt}, \cref{eq:ks-unocnditional,eq:ks-conditional} hold with
\[\Delta_1 = CC_1C_2J^{-1}; \quad \Delta_2 = CC_1^2C_2\|\coefvec\|_2(I^{-1} + J^{-1})\]
Thus, using $C_1^2 \ge C_1$ since $C_1 \ge 1$ by definition, \cref{eq:generic-clt} holds with $(\Delta_1 + \Delta_2) = CC_1^2C_2\|\coefvec\|_2(I^{-1} + J^{-1})$. Thus, by \Cref{lemma:tying-together}
\[
    \sup_{t \in \mathbb{R}} \left|\mathbb{P}\left\{\frac{\widehat{{\tau}}(\coefvec) -  \tau(\coefvec)}{\sqrt{\mmv\{\hat{\tau}(\coefvec)\}}} \le  t \right\} - \Phi(t)\right| \le C\Delta^{1/3}\log(C/\Delta).
\]
\end{proof}

\begin{proof}[Proof of \Cref{lemma:tying-together}]
We introduce the shorthand $\sigma_1^2 = \mmv\{{\tau}^{\Pi}(\coefvec)\}$, $\sigma_2^2 = \mathbb{E}[\mmv\{\hattau(\coefvec)|\Pi\}]$, and $\sigma^2 = \mmv{\hattau(\coefvec)}$, so in particular $\sigma_1^2 + \sigma_2^2 = \sigma^2$, and we write $\Phi_s(t) = \Phi(t/s)$ for the Gaussian CDF with scale $s$. Finally, note that we may assume $(\Delta_1 + \Delta_2)/\sigma \le 1$, or else the final bound becomes trivially true.

After substituting $u = (\sigma/\sigma_2)t$ and rearranging, \cref{eq:ks-conditional} gives that
\begin{equation}\label{eq:ks-conditional-scaled}
\sup_{t \in \mathbb{R}} \left|\mathbb{P}\left\{\hattau(\coefvec) - \tau^{\Pi}(\coefvec) \le  t  |\Pi \right\} - \Phi_{\sigma_2}(t)\right| \le \eta + (\Delta_2/\sigma_2)\log(1/\eta)
\end{equation}
with probability $1-\eta$.
Substituting $u = (\sigma/\sigma_1)t$ in \cref{eq:ks-unocnditional} similarly gives 
\begin{equation}\label{eq:ks-unconditional-scaled}
    \sup_{t \in \mathbb{R}} \left|\mathbb{P}\left\{\tau^{\Pi}(\coefvec) - {\tau}(\coefvec) \le  t \right\} - \Phi_{\sigma_1}(t)\right| \le \Delta_1/\sigma_1.
\end{equation}
Now, let $g_1 \sim N(0, \sigma_1^2)$ and $g_2 \sim N(0,\sigma_2^2)$ be independent Gaussian random variables, which are also independent of the random assignment. Applying \Cref{lem:tensorize-kolmogorov} with $X = \tau^{\Pi}(\coefvec) - \tau(\coefvec)$, $Y = \hattau(\coefvec) - \tau^{\Pi}(\coefvec)$, $V = \Pi$, $W = g_1$ and $Z = g_2$ gives the bounds
\begin{align}
    \sup_{t \in \bb{R}}\left|\bb{P}\left\{\hattau(\coefvec) - {\tau}^\Pi(\coefvec)  \le t\right\} - \bb{P}\{\tau^\Pi(\coefvec) - \tau(\coefvec) + g_2 \le t\}\right| \le 2\eta +  \frac{\Delta_2\log(C/\eta)}{\sigma_2}, \label{eq:conditional-bound}\\
    \sup_{t \in \bb{R}}\left|\bb{P}\left\{\tau^\Pi(\coefvec) - \tau(\coefvec) + g_2 \le t\right\} -\Phi_\sigma(t) \right| \le \Delta_1/\sigma_1 \label{eq:unconditional-bound},\\
    \sup_{t \in \bb{R}}\left|\bb{P}\left\{\hattau(\coefvec) - {\tau}(\coefvec) \le t\right\} - \Phi_\sigma(t)\right| 
    \le 2\eta + (\Delta_1/\sigma_1) + \frac{\Delta_2\log(C/\eta)}{\sigma_2}, \label{eq:both-bound}
\end{align}
where we used $W + Z = g_1 + g_2 \sim N(0,\sigma^2)$ since $g_1, g_2$ are independent and $\sigma_1^2 + \sigma_2^2 = \sigma^2$.  We then consider cases, first assuming that $\Delta_1/\sigma_1 \le \sigma_1/\sigma$ and $\Delta_2/\sigma_2 \le \sigma_2/\sigma$; otherwise we will show that the proof simplifies. 

\paragraph{Case 1, $\Delta_1/\sigma_1 \le \sigma_1/\sigma$ and $\Delta_2/\sigma_2 \le \sigma_2/\sigma$.} 
In this case, we start from \cref{eq:both-bound}. Our assumption that $\Delta_1/\sigma_1 \le \sigma_1/\sigma$ and $\Delta_2/\sigma_2 \le \sigma_2/\sigma$ implies
\(\frac{\Delta_1}{\sigma} = \frac{\Delta_1}{\sigma_1} \frac{\sigma_1}{\sigma} \ge \frac{\Delta_1^2}{\sigma_1^2}\) and \(\frac{\Delta_1}{\sigma} = \frac{\Delta_1}{\sigma_1} \frac{\sigma_1}{\sigma} \ge \frac{\Delta_1^2}{\sigma_1^2}\)
Plugging this into the above, putting $\Delta_{1+2} \coloneqq \Delta_1 + \Delta_2$, and using $\sqrt a + \sqrt b \le 2\sqrt{a+b}$, we get
\[ \sup_{t \in \bb{R}}\left|\bb{P}\left\{\hattau(\coefvec) - {\tau}(\coefvec) \le t\right\} - \Phi_\sigma(t) \right| \le  2\eta + \frac{\sqrt{\Delta_1} + \sqrt{\Delta_2}\log(C/\eta)}{\sqrt{\sigma}} \le 2 \eta + 2\sqrt{\frac{\Delta_{1+2}}{\sigma}} \log(C/\eta).\]
Since we may assume $\Delta_{1+2}/\sigma \le 1$ or else the final bound is trivial, we may plug in $\eta = \Delta_{1+2}/\sigma$ to obtain the simplified bound
\begin{align*}
    \sup_{t \in \bb{R}}\left|\bb{P}\left\{\hattau(\coefvec) - {\tau}(\coefvec) \le t\right\} - \Phi_\sigma(t) \right| 
    & \le  C'(\Delta_{1+2}/\sigma)^{1/2}\log\{C'/(\Delta_{1+2}/\sigma)\} \\
    & \le C'(\Delta_{1+2}/\sigma)^{1/3}\log\{C'/(\Delta_{1+2}/\sigma)\}.
\end{align*} 
This is precisely our claim, after taking $u = \sigma t$.
\paragraph{Case 2: $\Delta_2/\sigma_2 > \sigma_2/\sigma$.} 
 Multiplying both sides by $\sigma_2/\sigma$ gives $\Delta_2/\sigma > \sigma_2^2/\sigma^2$. 
Moreover, we may assume that $\Delta_1 / \sigma_1 \le \sigma_1/\sigma$, since otherwise the same reasoning gives $\Delta_1/\sigma > \sigma_1^2 / \sigma^2$, implying
\(\Delta_{1+2}/\sigma = (\Delta_1 + \Delta_2)/\sigma > (\sigma_1^2 + \sigma_2^2)/\sigma^2 = 1,\)
in which case the bound is trivial. 
Multiplying both sides of the inequality $\Delta_1 / \sigma_1 \le \sigma_1/\sigma$ by $\Delta_1 / \sigma_1$ gives $(\Delta_1/\sigma_1)^2 \le \Delta_1/\sigma$. To summarize, we may assume
\begin{equation}\label{eq:case-analysis-sigma2}
    \sigma_2/\sigma < \sqrt{\Delta_2/\sigma}; \quad \Delta_1/\sigma_1 \le  \sqrt{\Delta_1/\sigma}.
\end{equation}

By definition, 
\(\mme\{[\hattau(\coefvec) - \tau^\Pi(\coefvec)]^2\} = \mme[\mme\{[\hattau(\coefvec) - \tau^\Pi(\coefvec)]^2| \Pi\}] = \mme[\mmv\{\hattau(\coefvec)|\Pi\}] = \sigma_2^2.\) By Chebyshev's inequality, using $\mmv\{\hattau(\coefvec) - \tau^\Pi(\coefvec) - g_2\} = \mmv\{\hattau(\coefvec) - \tau^\Pi(\coefvec)\} + \mmv\{g_2\} = 2\sigma_2^2$ due to independence of $g_2$, it holds with probability at least $1-\eta$ that
\begin{align*}
    |\{\tau^\Pi(\coefvec) - \tau(\coefvec) + g_2\} - \{\hattau(\coefvec) - \tau(\coefvec)\}| = |\hattau(\coefvec) - \tau^\Pi(\coefvec) - g_2| \le \sigma_2\sqrt{\frac{2}{\eta}}.
\end{align*}
Applying \Cref{lemma:combine-coupling-with-kolmogorov} with the above bound and \cref{eq:unconditional-bound}, we have
\begin{align*}
    \sup_{t \in \bb{R}}\left|\bb{P}\left\{\hattau(\coefvec) - \tau(\coefvec)\le t\right\} -\Phi_\sigma(t) \right| 
    &\le C\left\{\eta + (\Delta_1/\sigma_1) + \frac{\sigma_2}{\sigma}\sqrt{\frac{2}{\eta}}\right\}
\intertext{By \cref{eq:case-analysis-sigma2} and $\Delta_1, \Delta_2 \le \Delta_{1+2}$, this simplifies to}
  &\le C\left\{\eta + \sqrt{\Delta_{1+2}/\sigma} + \sqrt{\Delta_{1+2}/\sigma}\sqrt{\frac{2}{\eta}}\right\}.
    \intertext{Plugging in $\eta = (\Delta_{1+2}/\sigma)^{1/3}$, which we may assume is at most $1$, this is}
    &\le C'(\Delta_{1+2}/\sigma)^{1/3} \\
    &\le C''(\Delta_{1+2}/\sigma)^{1/3}\log\{C''/(\Delta_{1+2}/\sigma)\}.
\end{align*}

\paragraph{Case 3: $\Delta_1/\sigma_1 > \sigma_1/\sigma$.} 
This is completely analogous to Case 2, with $(\Delta_1, \sigma_1)$ swapped with $(\Delta_2, \sigma_2)$, and \cref{eq:ks-unocnditional} replaced by \cref{eq:ks-conditional}. By a symmetric argument, we may assume
\begin{equation}\label{eq:case-analysis-sigma1}
    \sigma_1/\sigma < \sqrt{\Delta_1/\sigma}; \quad \Delta_2/\sigma_2 \le  \sqrt{\Delta_2/\sigma}.
\end{equation}
By definition, \(\mmv\{\tau^{\Pi}(\coefvec) - \tau(\coefvec)\} = \mmv\{\tau^{\Pi}(\coefvec)\} = \sigma_1^2\), and $\mmv\{\tau^{\Pi}(\coefvec) - \tau(\coefvec) - g_1\} = 2\sigma_1^2$ by independence of $g_1$, so by Chebyshev's inequality,
\begin{align*}
    |\{\tau^\Pi(\coefvec) - \tau(\coefvec) + g_2\} - \{g_1 + g_2\}| = |\tau^\Pi(\coefvec) - \tau(\coefvec) - g_1| \le \sigma_1\sqrt{\frac{2}{\eta}}.
\end{align*}
Applying \Cref{lemma:combine-coupling-with-kolmogorov} with the above bound and \cref{eq:conditional-bound}, we have
\begin{align*}
    \sup_{t \in \bb{R}}\left|\bb{P}\left\{\hattau(\coefvec) - \tau(\coefvec) + g_2 \le t\right\} -\Phi_\sigma(t) \right| 
    &\le C\left\{\eta + (\Delta_2/\sigma_2)\log(C/\eta) + \frac{\sigma_1}{\sigma}\sqrt{\frac{2}{\eta}}\right\}.
    \intertext{By \cref{eq:case-analysis-sigma1} and $\Delta_1, \Delta_2 \le \Delta_{1+2}$, this simplifies to}
    &\le C\left\{\eta + \sqrt{\Delta_{1+2}/\sigma}\log(C/\eta) + \sqrt{\Delta_{1+2}/\sigma}\sqrt{\frac{2}{\eta}}\right\}.
    \intertext{Plugging in $\eta = (\Delta_{1+2}/\sigma)^{1/3}$, which we may assume is at most $1$, this is}
    &\le C'(\Delta_{1+2}/\sigma)^{1/3}\log\{C'/(\Delta_{1+2}/\sigma)\}.
\end{align*}
\end{proof}

\section{Additional simulations} \label{sec:additional_experiments}

We show simulations for results of \cref{sec:estimation} for SMRDs under  local interference. 
Fix $1\leq I_T \le I-1, 1\leq J_T \le J-1$, and let $P_{\bm{\Type}}$ be the distribution over the matrix of types $\bm{\Type}$ induced by sampling $\bw$ from a SMRD  as per \Cref{eq:types}. 
We draw data via: 
\begin{align}\label{eq:dgp}
\begin{split}
    \bm{\Type}  &\sim P_{\bm{\Type}}(\cdot),
    \quad\text{and}\quad
    Y_{ij} \mid  \bm{\Type}  {\sim} F_{\type_{ij}}(\cdot).
\end{split}
\end{align}
Here, potential outcomes are distributed as follows:
\begin{align} \label{eq:parametric_ALI}
    Y_{ij}\mid \Type_{ij} &\overset{ind}{\sim} 
    \begin{cases}
        F_0(\cdot) &\mbox{ if } \type = \ccc, \\
        F_0(\cdot) + F_{\rb}(\cdot) &\mbox{ if } \type = \icb, \\
        F_0(\cdot) + F_{\rs}(\cdot) &\mbox{ if } \type = \ics, \\
        F_1(\cdot) + F_{\rb}(\cdot) + F_{\rs}(\cdot) &\mbox{ if } \type = \ctt.
    \end{cases}
\end{align}
$F_{\ell}$ are distributions, $\ell \in \{0 , \rb, \rs, 1\}$. By construction, data drawn from \Cref{eq:parametric_ALI} satisfies the local interference assumption \eqref{sutva_local}.
In our illustration the $F_\ell$ are Gaussian, although this is not required --- indeed, we do not need to impose any parametric assumption on the specification \Cref{eq:parametric_ALI} for our simulations to be consistent with the theory proved in \cref{sec:estimators}. 
We set $p_0 := 1, p_1 := 1$, and the proportions of treated {\buyer}s and {\seller}s in the MRD be $p_\rb:=I_{\ct}/I$ and $p_\rs:=J_{\ct}/J$, and $F_\ell(\cdot) = \NN(p_{\ell} \mu_\ell,\sigma_\ell^2)$, $\ell \in \{0 , \rb, \rs, 1\}$. 
We set $I = 200, J = 150$, $p_{\rb} = 0.45$, $p_{\rs} = 0.55$ and $\mu_{0} = 3, \mu_{B} = -1, \mu_{S} = -1, \mu_{1}= 6$ and $\sigma_{x} = 1 \;\forall\;x \in \{0,B,S,1\}$.

To assess validity of the results presented in \Cref{sec:estimation}, we draw matrices $\bm{Y}(\type) = [Y_{ij}(\type)]$ of $I\times J$ fixed potential outcomes  $\forall \type \in \types$ via \Cref{eq:parametric_ALI}. We  sample 10,000 assignment matrices $\bw$ i.i.d.\ at random from the SMRD $\mmw$ (equivalently, we sample types $\bm{\Gamma}$ from $P_{\bm{\Gamma}}$ in \Cref{eq:dgp}). Each assignment corresponds to a matrix of types and hence which potential outcomes are observed. To each assignment corresponds an observed matrix of $I \times J$ realized potential outcomes. We use the collection of outcomes from the 10,000 re-randomizations to empirically verify the properties of the proposed estimators. 

For the type estimator defined in \Cref{eq:meantypeestimate} we check that  $\meanestimate{\type}$ is an unbiased estimate of $\meanpopulation{\type}$ (\Cref{lemma:unbiasedness}) and that $
\widehat{\Sigma}_{\type}$ is an unbiased estimator of the variance of the type estimator (\Cref{thm:sample_variance_avg_effs}). 
\Cref{fig:type_effs_ALI} reports the histogram of the values attained by $\meanestimate{\ccc}$ across the 10,000 Monte Carlo replicates. From \Cref{eq:parametric_ALI} (and, under mild assumptions, from the CLT), the type estimator is normally distributed, and from \Cref{lemma:unbiasedness}, it is centered at the true population value $\meanpopulation{\type}$. 
Moreover, the distance between the 2.5\% and 97.5\% quantiles of the distribution of the type estimator is close to the length of our 95\% confidence interval.
In the right panel, we show that $\widehat{\Sigma}_{\ccc}$ is an unbiased estimator for the variance of the type estimator, as proved in \Cref{thm:sample_variance_avg_effs}. Analogous results hold for $\icb, \ics, \ctt$.

\begin{figure}[ht]
    \centering   
    \includegraphics[width=.75\textwidth]{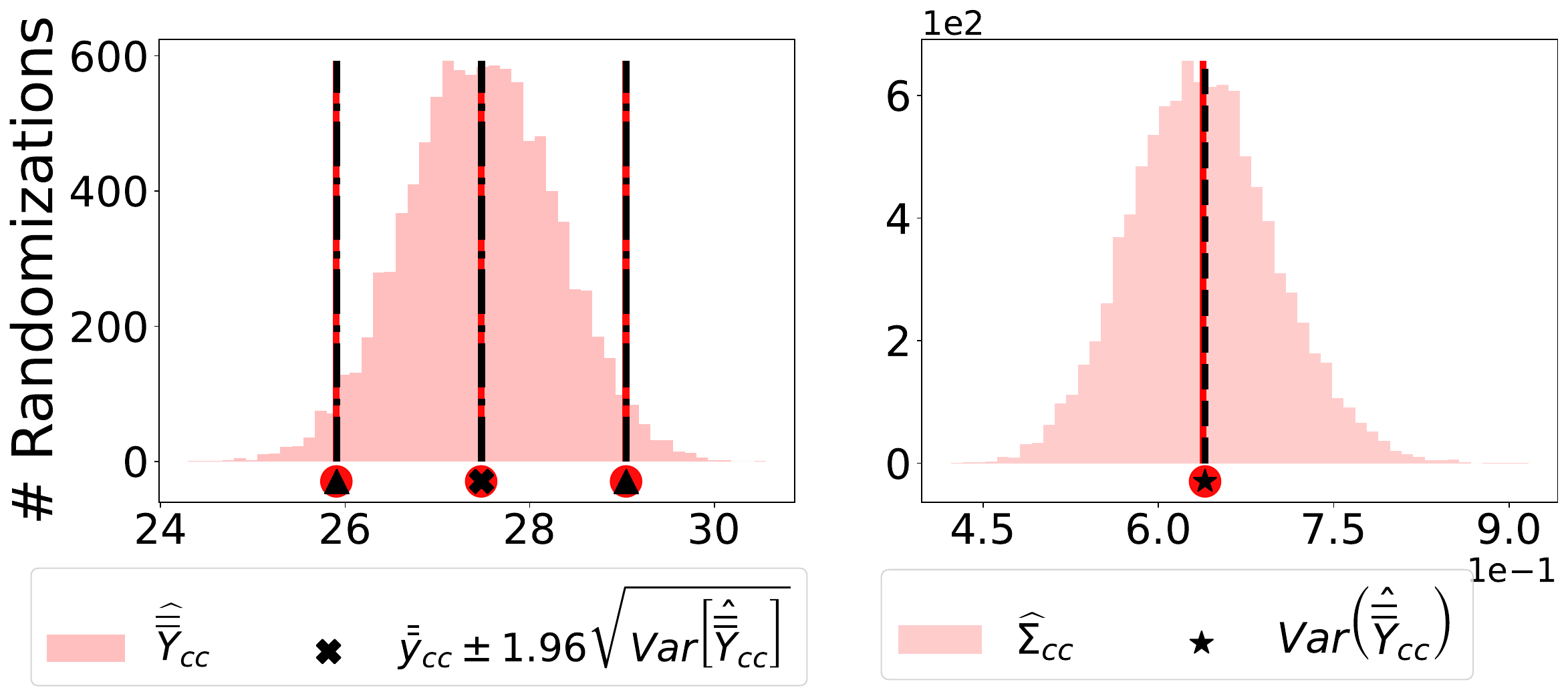}
    \caption{
   Distribution of $\meanestimate{\ccc}$ (left) and of the variance estimator $\widehat{\Sigma}_{\ccc}$ (right). Black lines are plotted in correspondence of the population quantities $\meanpopulation{\ccc}$, $\mmv\left(\meanestimate{\ccc}\right)$.
    }
    \label{fig:type_effs_ALI}
\end{figure}

\Cref{fig:spill_ALI} focuses on the spillover effect $\tauspillb$: the left panel shows the distribution of the unbiased estimator $\hattauspillb$ (\Cref{thm:spillover_unbiasedness}). $\hattauspillb$ is a linear combination of Gaussians, and usual confidence intervals can be derived. The right panel contains the distribution of the upper bound $\widehat{\mmv}^{\rm hi}(\hattauspillb)$ for the variance ${\mmv}(\hattauspillb)$ (\Cref{thm:sample_variance_spillover}). 

\begin{figure}[ht]
    \centering   
    \includegraphics[width=.75\textwidth]{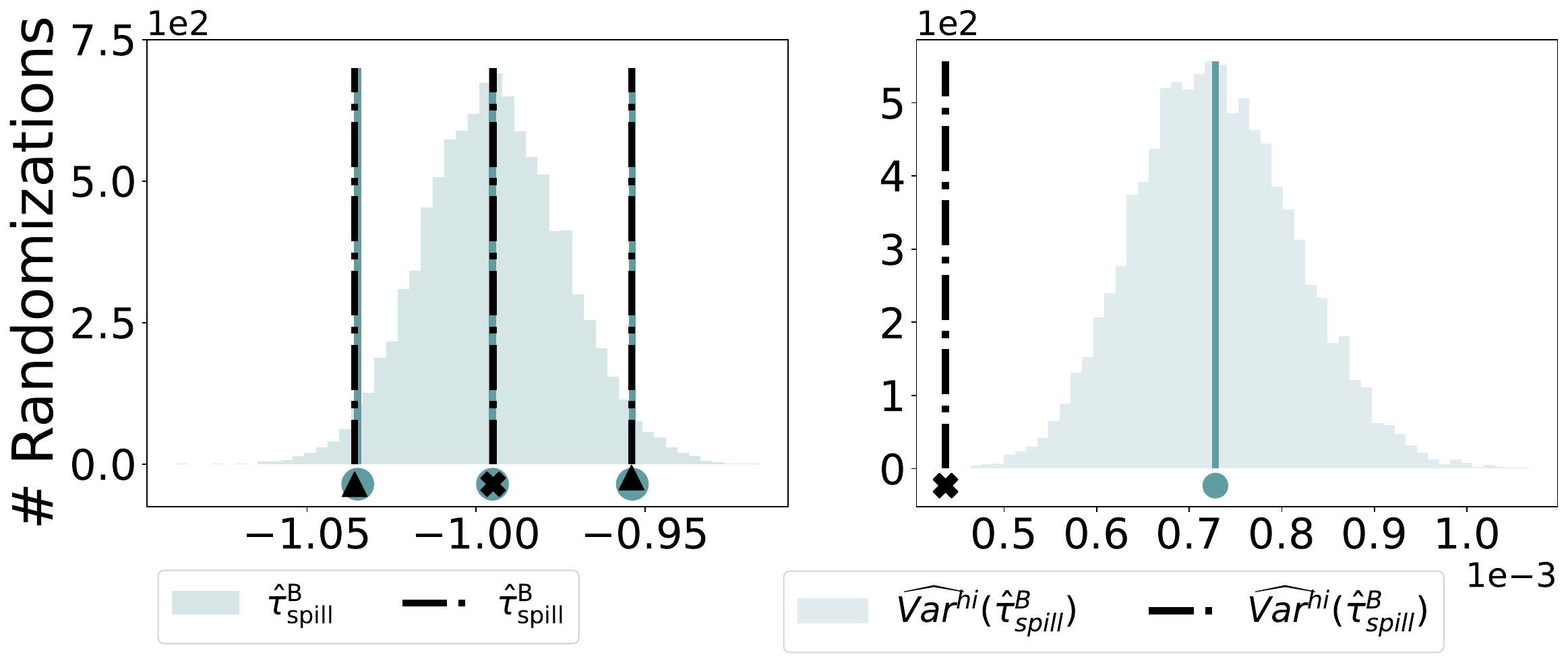}
    \caption{
   Distribution of the estimator for the spillover effect $\hat{\tau}_{\rm{spill}}^B$ (left) and corresponding variance estimator $\widehat{\mmv}^{\rm{hi}}(\hat{\tau}_{\rm{spill}}^B)$ (right). Black lines correspond to the population quantities.
    }
    \label{fig:spill_ALI}
\end{figure}

\subsection{Figures for the average-type and spillover effects}

For each $\type \in \types$, we report properties of $\meanestimate{\type}$ similar to \cref{fig:type_effs_ALI} in \cref{fig:type_effs_icb,fig:type_effs_ics,fig:type_effs_tr}. In \cref{fig:spill_direct_ALI,fig:spill_seller_ALI,fig:spill_ate_ALI} we provide plots for estimators of spillover effects.

\begin{figure}[ht]
    \centering   
    \includegraphics[width=.85\textwidth]{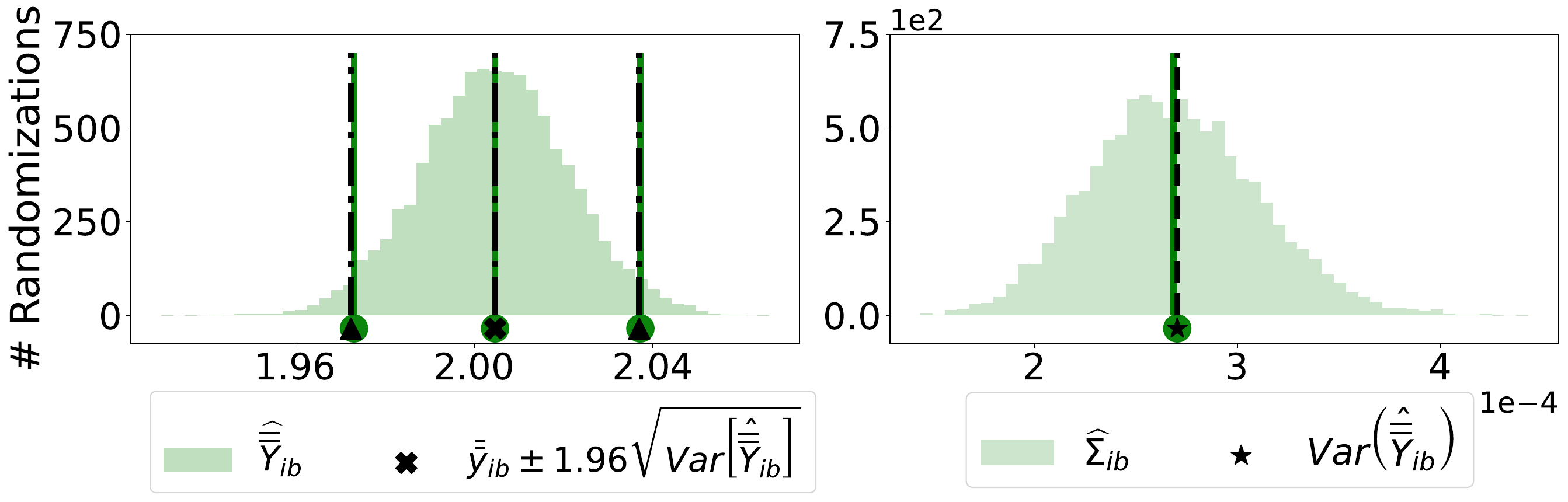}
    \caption{
    Same as \Cref{fig:type_effs}, now for $\icb$.
    }
    \label{fig:type_effs_icb}
\end{figure}

\begin{figure}[ht]
    \centering   
    \includegraphics[width=.85\textwidth]{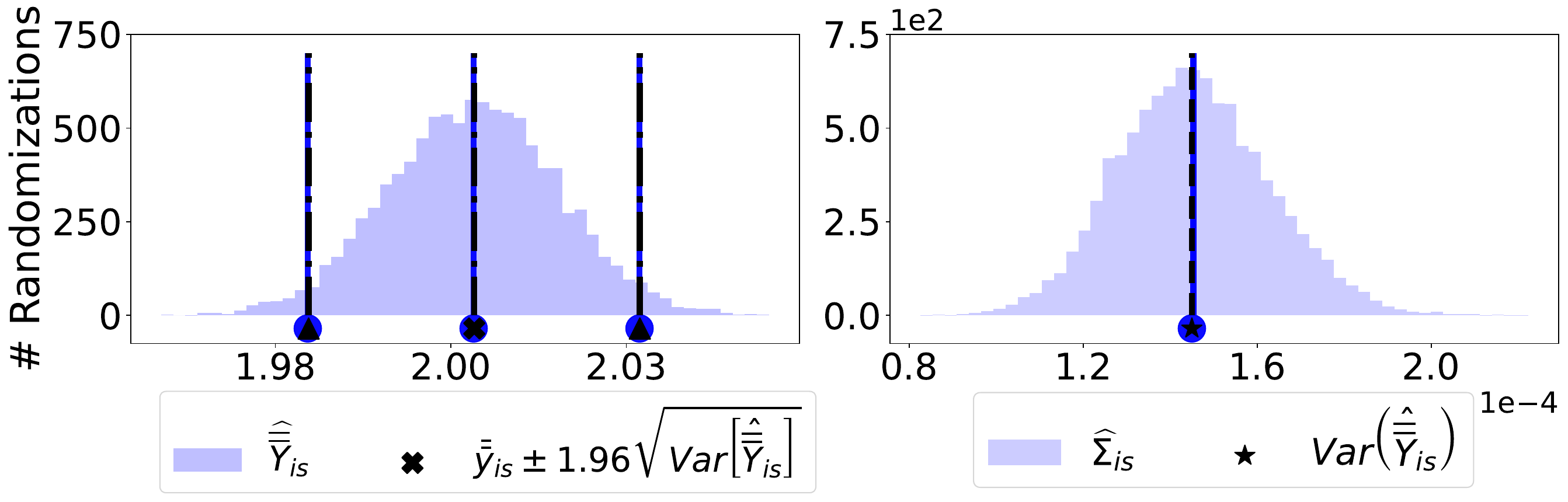}
    \caption{
    Same as \Cref{fig:type_effs}, now for $\ics$
    }
    \label{fig:type_effs_ics}
\end{figure}

\begin{figure}[ht]
    \centering   
    \includegraphics[width=.85\textwidth]{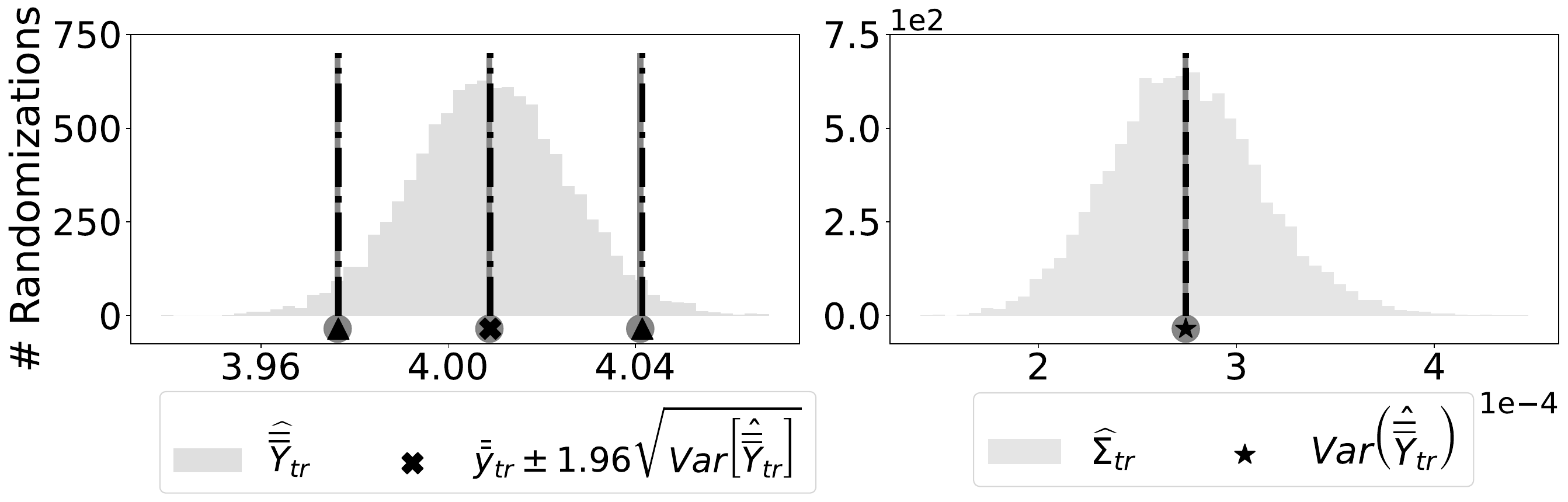}
    \caption{
    Same as \Cref{fig:type_effs}, now for $\ctt$
    }
    \label{fig:type_effs_tr}
\end{figure}


\begin{figure}[ht]
    \centering   \includegraphics[width=.85\textwidth]{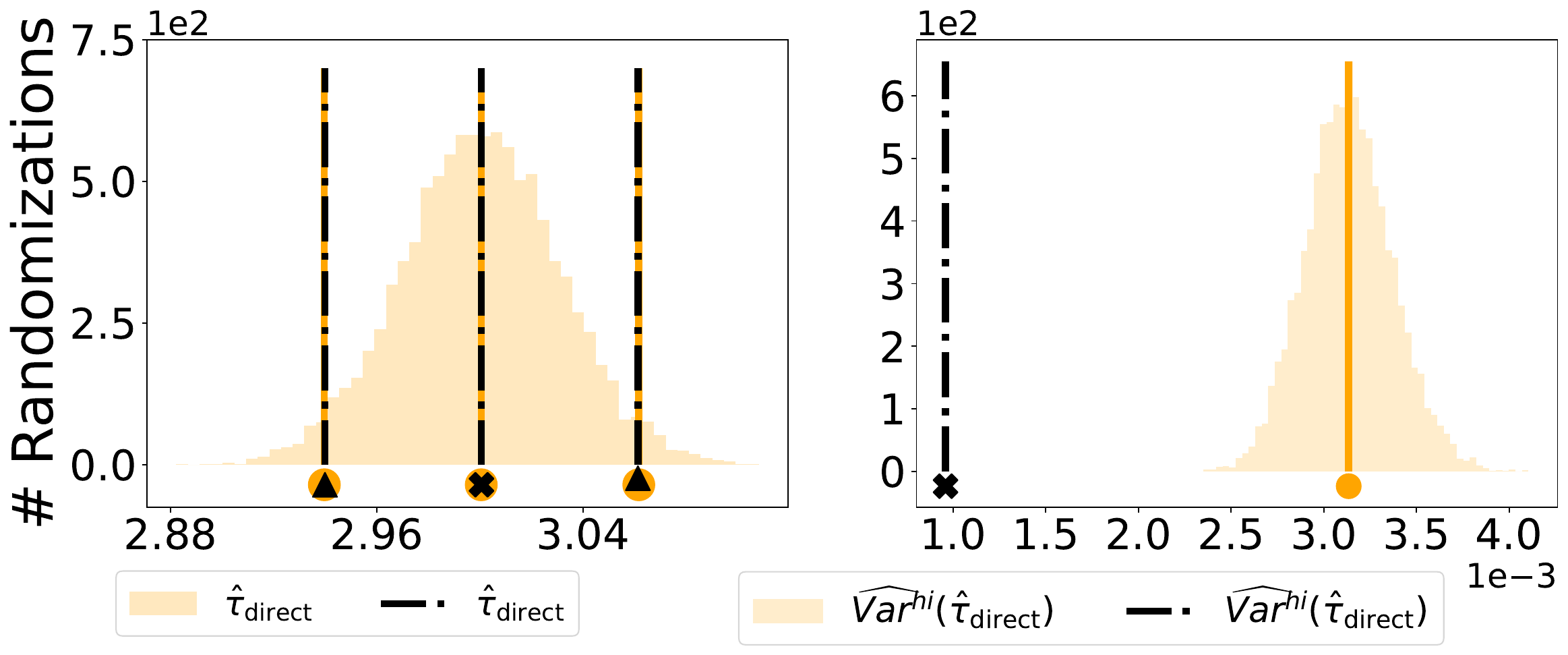}
    \caption{
    Same as \Cref{fig:spill_ALI}, now for $\hat{\tau}_{\rm{direct}}$.
    }
    \label{fig:spill_direct_ALI}
\end{figure}

\begin{figure}
    \centering   \includegraphics[width=.85\textwidth]{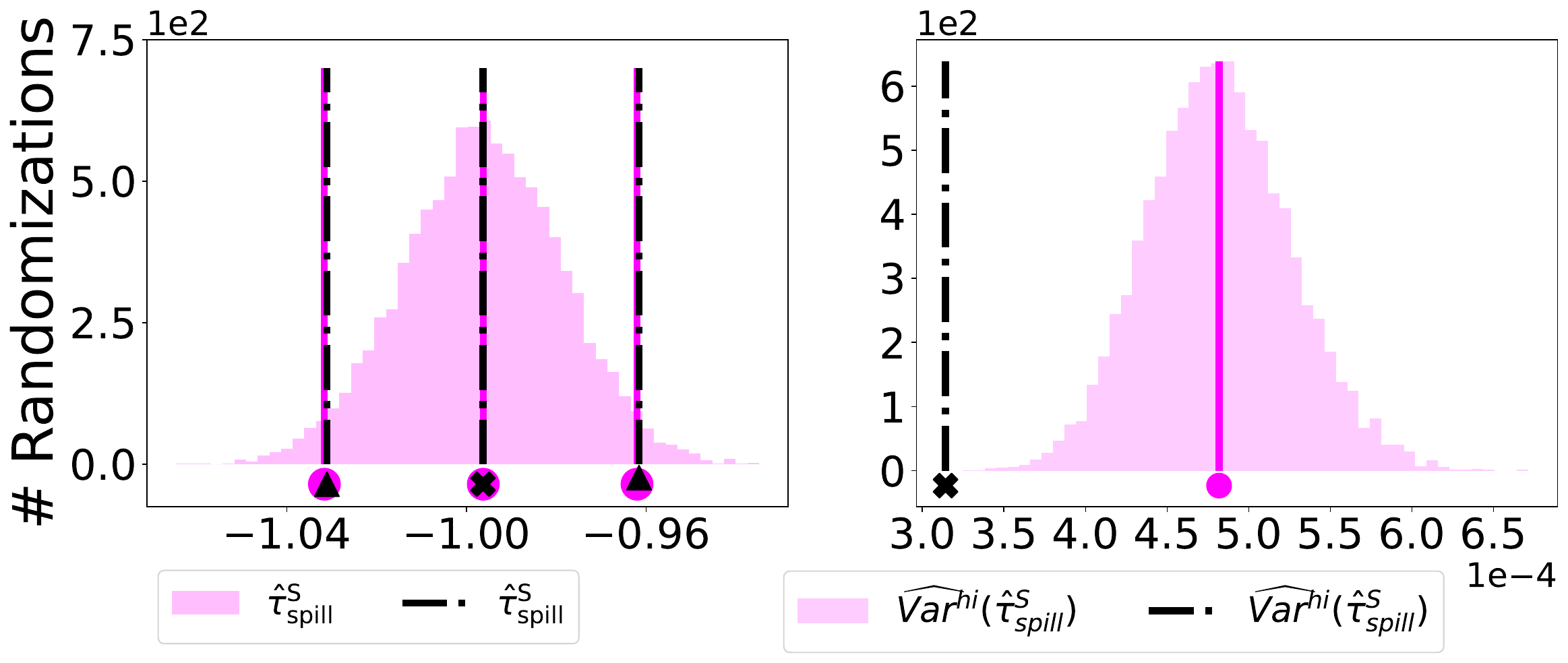}
    \caption{
   Same as \Cref{fig:spill_ALI}, now for $\hat{\tau}_{\rm{spill}}^S$.
    }
    \label{fig:spill_seller_ALI}
\end{figure}

\begin{figure}[ht]
    \centering   \includegraphics[width=.85\textwidth]{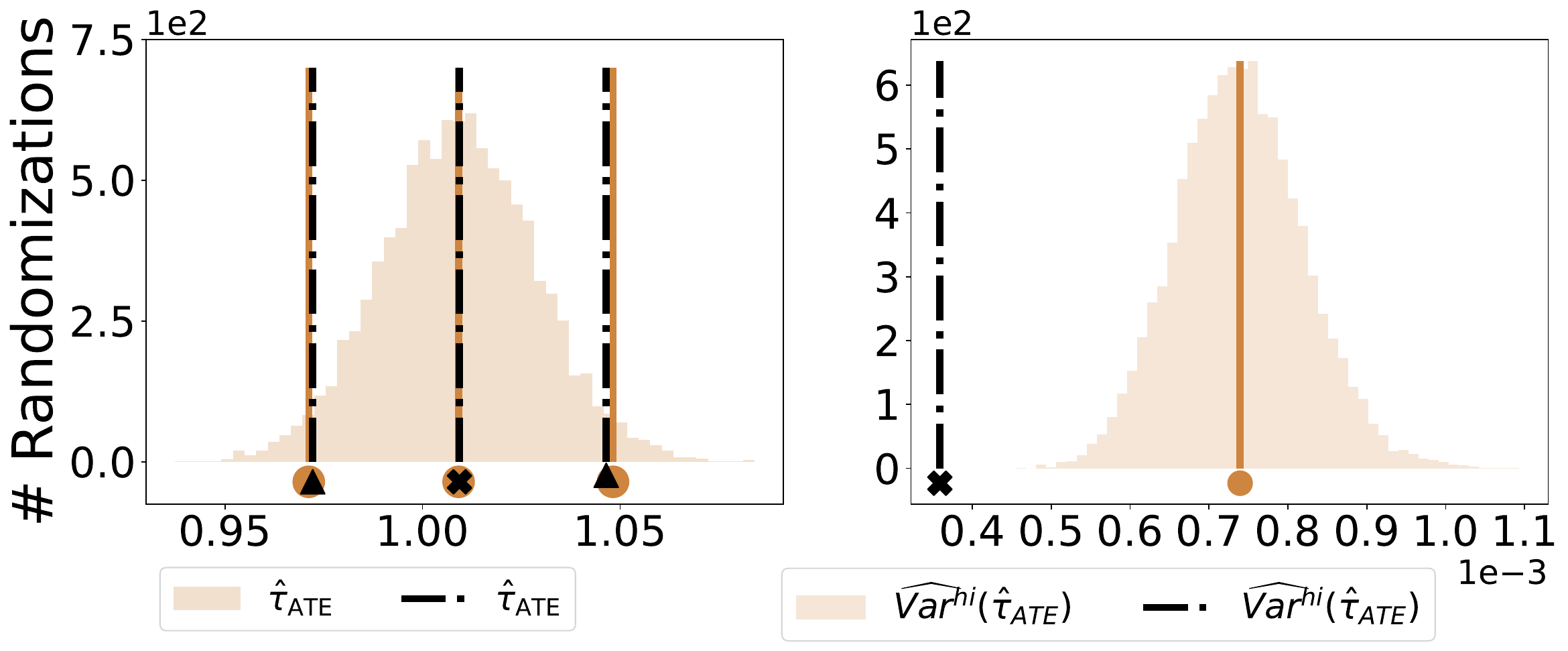}
    \caption{
    Same as \Cref{fig:spill_ALI}, now for $\hat{\tau}_{\rm{ATE}}$.
    }
    \label{fig:spill_ate_ALI}
\end{figure}

\subsection{Testing under the null hypothesis}

Similar to \cref{sec:experiment}, we here provide additional results where we show that, under the null hypothesis of no effect, we can use our derived variance formulae to construct valid test statistics. We consider again the case of $I=200$ and $J=150$, and let $\mu_0 = \mu_1 = 3$, $\sigma_0 = \sigma_1 = 1$. 
We let $\mu_{\rs} = \mu_{\rb} = 0$ and $\sigma_{\rs} = \sigma_{\rb} = 0$, leading to potential outcomes $Y_{i,j}(\gamma) = Y_{i,j}(\gamma')$ for $\gamma, \gamma' \in \{\ccc, \icb, \ics\}$ and $Y_{ij}(\ccc) \overset{d}{=} Y_{ij}(\ctt)$. We run 10,000 Monte Carlo draws, keeping the underlying potential outcomes fixed and randomizing over the assignments $\bw$ using $I_1 = 90$ and $J_1 = 85$. We report results from this simulation in \cref{fig:ate_null}.

\begin{figure}[ht]
    \centering   
    \includegraphics[width=.85\textwidth]{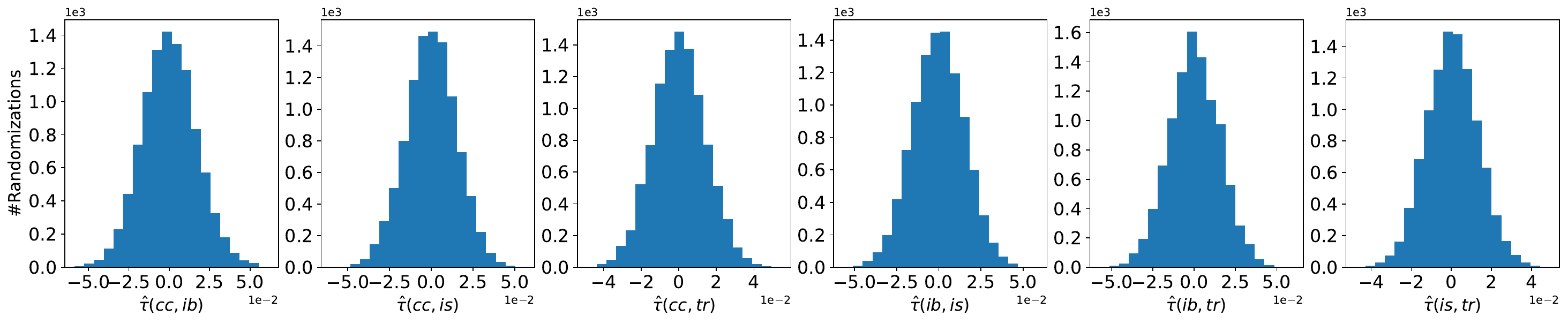}
    \caption{Empirical distribution of $\hat{\tau}$ over 10{,}000 Monte-Carlo  draws }
    \label{fig:ate_null}
\end{figure}
We also compute the test statistics $\hat{t} = \hat{\tau} / \sqrt{\mmv({\hat{\tau}}})$, where in the denominator we use the true (unknown) variance of the estimator, given in \cref{eq:general_variance}. Under this enforced null, we test the null hypothesis that a pair of types $\type, \type'$ is associated with no (average) effects on the outcome $Y$, 
$
    H_0^{(\type, \type')} = \left\{ \tau(\type,\type') = 0 \right\}
$
using the test statistic $\hat{t}$ defined above, and leveraging the normality of the CLT derived in \cref{thm:clt}. The corresponding p-values obtained by a standard two sided t-test are uniformly distributed, as expected (\cref{fig:pvals}).

\begin{figure}[ht]
    \centering   
    \includegraphics[width=\textwidth]{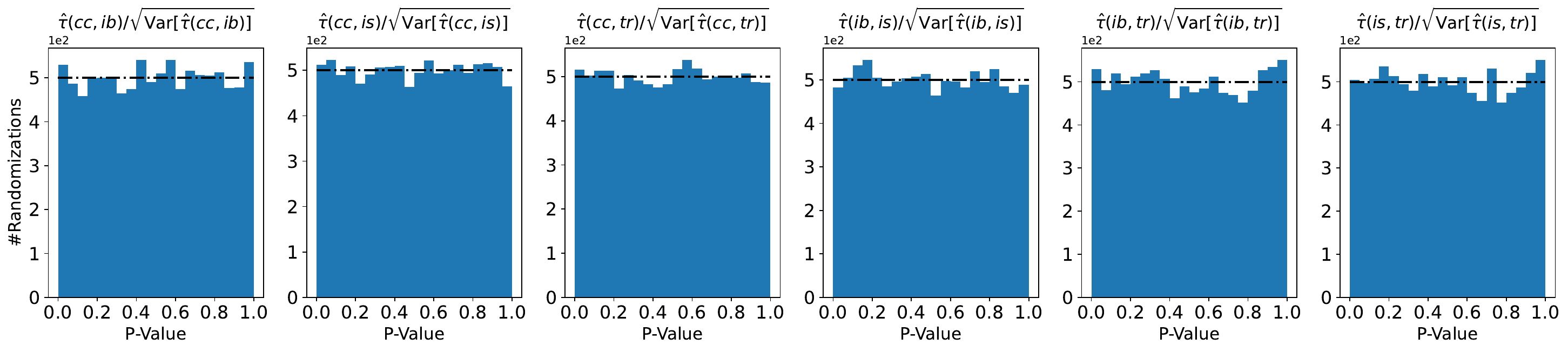}
    \caption{Empirical distribution of the p-values for the test of the null hypothesis $H_0^{(\type, \type')}$ that the means are identical in two different types $\type, \type'$.}
    \label{fig:pvals}
\end{figure}

\bibliographystyle{abbrvnat}
\bibliography{references}

\end{document}